\documentclass[a4paper,english,numberwithinsect]{lipics}
\usepackage{paralist}

  \usepackage{lipsum}
  \usepackage{float}
  \newcounter{mycounter}

\usepackage{amsmath, amsthm}
\usepackage[bbgreekl]{mathbbol}
\usepackage{xspace}
\usepackage{dsfont}
\usepackage{chngcntr}
\usepackage{float}
\usepackage[all]{xy}
\SelectTips {cm}{}
\newdir{ >}{{}*!/-5pt/@{>}}
\newdir{ (}{{}*!/-5pt/@{(}}

\DeclareFontFamily{U}{bbold}{}\DeclareFontShape{U}{bbold}{m}{n}{%
  <4.25>bbold5<5>bbold5<6>bbold6<7>bbold7<8>bbold8<8.5>bbold9<9.5>bbold10<10>bbold10<12>bbold12}{}

\theoremstyle{plain}
\newtheorem{proposition}[theorem]{Proposition}
\newtheorem*{theorem*}{Theorem}
\newtheorem*{lemma*}{Lemma}
\newtheorem*{claim}{Claim}

\theoremstyle{definition}
\newtheorem{assumptions}[theorem]{Assumptions}
\newtheorem{notation}[theorem]{Notation}
\newtheorem{construction}[theorem]{Construction}
\newtheorem{rem}[theorem]{Remark}

\def\D{\DCat}
\def\C{\Cat}
\def\A{\mathds{A}}
\def\V{\VCat}

\def\epsilon{\varepsilon}

\renewcommand{\rho}{\varrho}
\def\ol{\overline}

\def\o{\cdot}

\def\S{\mathscr{S}}

\newcommand{\Quo}{\mathsf{Quo}}

\newcommand{\takeout}[1]{\empty}

\DeclareMathOperator{\Lim}{Lim}

\renewcommand{\phi}{\varphi}

\newcommand{\xra}{\xrightarrow}

\newcommand{\Ra}{\Rightarrow}

\newcommand{\Lra}{\Leftrightarrow}
\newcommand{\seq}{\subseteq}

\newcommand{\T}{\mathcal{T}}

\renewcommand{\P}{\mathscr{P}}

\newcommand\epidownarrow{\mathrel{\rotatebox[origin=c]{90}{$\twoheadleftarrow$}}}

\newcommand{\Stone}{\mathbf{Stone}}

\newcommand{\Priest}{\mathbf{Priest}}

\newcommand{\Alg}{\mathbf{Alg}\,}

\newcommand{\Set}{\mathbf{Set}}
\newcommand{\E}{\mathcal{E}}
\newcommand{\Pos}{\mathbf{Pos}}

\newcommand{\BA}{\mathbf{BA}}

\newcommand{\DL}{\mathbf{DL}_{01}}

\newcommand{\ACat}{\mathscr{A}}

\newcommand{\Cat}{\mathscr{C}}

\newcommand{\DCat}{\mathscr{D}}

\newcommand{\hatD}{\widehat\D}

\newcommand{\MT}{\mathbf{T}}

\newcommand{\MR}{\mathbf{R}}
\newcommand{\hatT}{{\widehat\MT}_\F}
\newcommand{\hatt}{\hat{T}_\F}
\newcommand{\hateta}{{\hat\eta}^\F}
\newcommand{\hatmu}{{\hat\mu}^\F}

\newcommand{\id}{\mathit{id}}
\newcommand{\Id}{\mathsf{Id}}

\newcommand{\U}{\mathds{U}}

\newcommand{\FSigma}{{\mathbb{\Sigma}}}
\newcommand{\FGamma}{{\mathbb{\Gamma}}}
\newcommand{\FDelta}{{\mathbb{\Delta}}}

\newcommand{\VCat}{\mathscr{V}}

\newcommand{\CatMon}{\mathbf{Mon}}
\newcommand{\Mon}[1]{\mathbf{Mon}(#1)}

\newcommand{\FAlg}[1]{\mathbf{Alg}_{f}\,#1}

\newcommand{\Vect}{\mathbf{Vec}}

\newcommand{\JSL}{{\mathbf{JSL}_0}}

\newcommand{\Rec}{{\mathsf{Rec}}}

\newcommand{\under}[1]{{|#1|}}

\newcommand{\epito}{\twoheadrightarrow}

\newcommand{\monoto}{\rightarrowtail}

\newcommand{\br}[1]{{#1^+}}

\newcommand{\M}{\mathcal{M}}

\newcommand{\F}{\mathcal{F}}

\newcommand{\B}{\mathscr{B}}

\newcommand{\Ind}{\mathsf{Ind}\,}
\newcommand{\Pro}{\mathsf{Pro}\,}

\newcommand{\K}{\mathscr{K}}
\newcommand{\Field}{K}

\newcommand{\Nat}{\mathbb{N}}

\newcommand{\one}{\mathds{1}}

\newcommand{\To}{\Rightarrow}

\mathchardef\ordinarycolon\mathcode`\:
\mathcode`\:=\string"8000
\begingroup \catcode`\:=\active
  \gdef:{\mathrel{\mathop\ordinarycolon}}
\endgroup
\newcommand{\defeq}{:=}
\mathchardef\hyph="2D
\newcommand{\dash}{\mathord{-}}

\usepackage{cite}
\usepackage[kerning]{microtype}
\title{Eilenberg Theorems for Free}
\titlerunning{Eilenberg Theorems For Free}

\author[1]{Henning Urbat\footnote{Supported by Deutsche Forschungsgemeinschaft (DFG) under project AD 187/2-1}$^,$}
\author[1]{Ji\v{r}\'{i} Ad\'{a}mek$^{*,}$}
\author[1]{Liang-Ting Chen}
\author[2]{Stefan~Milius\footnote{Supported by Deutsche Forschungsgemeinschaft (DFG) under project MI~717/5-1}$^,$}
\affil[1]{Institut f\"{u}r Theoretische Informatik \\
  Technische Universit\"{a}t Braunschweig, Germany}
\affil[2]{Lehrstuhl f\"{u}r Theoretische Informatik\\
  Friedrich-Alexander Universit\"{a}t Erlangen-N\"{u}rnberg, Germany}
\authorrunning{H.~Urbat, J.~Ad\'{a}mek, L.-T.~Chen, S.~Milius}
\Copyright{Henning Urbat, Ji\v{r}\'{i} Ad\'{a}mek, Liang-Ting Chen, and Stefan
  Milius}


\DOI{10.4230/LIPIcs.xxx.yyy.p}

\begin{document}
\maketitle
\begin{abstract}
 Eilenberg-type correspondences, relating varieties of languages (e.g.\ of 
 finite words, infinite words, or trees) to pseudovarieties of finite algebras, 
 form the backbone of algebraic language theory. Numerous such correspondences 
 are known in the literature. We demonstrate that they all arise from the same recipe: one models languages and the algebras recognizing them by 
 monads on an algebraic category, and applies a Stone-type duality. Our main 
 contribution is a variety theorem that covers e.g.~Wilke's and Pin's 
 work on $\infty$-languages,  the variety theorem for cost functions of 
 Daviaud, Kuperberg, and Pin, and unifies the two previous categorical 
 approaches of  Boja\'nczyk and of Ad\'amek et al. In addition we derive a number of new results, including an extension of the local variety theorem of Gehrke, 
 Grigorieff, and Pin from finite to infinite words.
\end{abstract}

\section{Introduction}

Algebraic language theory investigates the behaviors of finite machines
by relating them to finite algebraic structures. This has 
proved 
extremely
fruitful. For example, regular languages (the behaviors of finite automata) can be described  purely algebraically as the 
languages recognized by finite 
monoids, and the decidability of 
star-freeness rests on Sch\"utzenberger's theorem~\cite{sch65}: a regular 
language is star-free iff it is recognized by a finite aperiodic 
monoid. To organize and study correspondence results of this kind systematically, Eilenberg \cite{eilenberg76} introduced the concepts of a \emph{variety of regular languages} and a \emph{pseudovariety of monoids}. By a pseudovariety of monoids is meant a class of finite monoids closed under homomorphic images, submonoids and finite products. A variety of regular languages associates to every finite alphabet $\Sigma$ a set $W_\Sigma$ of regular languages over $\Sigma$ such that
\begin{itemize}
\item $W_\Sigma$ is closed under finite union, finite intersection and complement;
\item $W_\Sigma$ is closed under derivatives: for any language $L\in W_\Sigma$ and $y\in\Sigma^*$, the languages
$y^{-1}L = \{\,x\in\Sigma^*: yx\in L\,\}$ and $Ly^{-1} = \{\,x\in\Sigma^*: xy\in L\,\}$ lie in $W_\Sigma$.
\item for any language $L\in W_\Sigma$ and any monoid homomorphism $g\colon \Delta^*\to \Sigma^*$, the preimage
$g^{-1}L = \{\,x\in\Delta^*: g(x)\in L \,\} \seq \Delta^*$
lies in $W_\Delta$.
\end{itemize}
 Eilenberg's celebrated \emph{variety 
theorem}~\cite{eilenberg76} asserts that varieties of regular languages and pseudovarieties of monoids are in bijective correspondence. This together with Reiterman's 
theorem~\cite{Reiterman1982}, 
stating that pseudovarieties of monoids can be specified by profinite 
equations, establishes a firm connection between automata, languages, and 
algebras.

Over the past four decades, Eilenberg's theorem inspired a lot of research and has been extended into two orthogonal directions. On the one hand, several authors established Eilenberg-type correspondences for classes of regular languages with modified closure properties, e.g. dropping closure under complement or intersection \cite{pin95,polak01,straubing02,ggp08}. On the other hand, variety theorems were discovered for machine behaviors beyond  finite words, including weighted languages over a 
field \cite{reu80}, infinite words \cite{wilke91,pin98}, words on linear orderings 
\cite{Bedon1998,Bedon2005}, ranked trees \cite{almeida90}, binary trees 
\cite{salehi07}, and cost functions \cite{pin16}. This plethora 
of structurally similar results (often with similar proofs) has raised interest in 
category-theoretic approaches which allow to derive all the 
above results as special instances of only \emph{one} general variety theorem. Two first steps in this direction, corresponding to the two orthogonal extensions of Eilenberg's theorem, were achieved by Boja\'nczyk~\cite{boj15} and in our previous work \cite{ammu14,ammu15,amu15, cu15}. Boja\'nczyk extends the 
notion of language recognition by monoids to algebras for a \emph{monad} on 
sorted sets and proves a generic Eilenberg theorem, improving earlier results of Almeida \cite{almeida90} and Salehi and Steinby \cite{sal04} on algebras over a finitary signature. In contrast, in our previous work one keeps monoids as the recognizing structures, but considers them in categories $\D$ 
of (ordered) algebras such as posets, semilattices, and vector spaces. This led us to a uniform presentation of five Eilenberg theorems for languages of finite words 
\cite{eilenberg76,pin95,polak01,reu80,straubing02}.

To obtain the desired general Eilenberg theorem, a unification of the two 
approaches is required. On the one hand, one needs to take the step from sets 
to more general categories $\D$ to capture the proper notion of 
language recognition; e.g. for the treatment of weighted 
languages~\cite{reu80} one needs to work over the category of vector spaces. 
On the other hand, to deal with machine 
behaviors beyond finite words, one has to replace monoids by other algebraic 
structures. The main contribution of this paper is a variety theorem that
achieves 
the 
desired unification, and in addition directly encompasses many Eilenberg-type
correspondences that follow from neither of the previous generic results, including 
the 
work in
\cite{Bedon1998,Bedon2005,wilke91,pin98,salehi07,pin16}. 

To motivate our approach let us start with the observation that, traditionally, Eilenberg-type correspondences are established in essentially the same four steps:
\begin{enumerate}
\item Identify an algebraic theory T such that the languages in mind are the 
ones recognized by finite T-algebras. For example, for 
regular languages one takes the theory of monoids.
\item Describe the syntactic T-algebras, i.e.~the minimal algebraic recognizers of languages.
\item Infer the form of the derivatives under which varieties of languages are closed. 
\item Establish the Eilenberg correspondence between varieties of languages and 
pseudovarieties of algebras by relating the languages of a variety to their corresponding syntactic algebras.
\end{enumerate}
The key insight provided by this paper is that all four steps can be simplified or even  completely
automatized.

\smallskip
For Step~1, putting a common roof over Boja\'nczyk's and our own previous 
work, we consider an algebraic category $\D$ and algebras for
a \emph{monad} $\MT$ on $\D^S$, the category of $S$-sorted $\D$-algebras for 
a 
 finite set $S$ of sorts. This is illustrated by the following diagram:
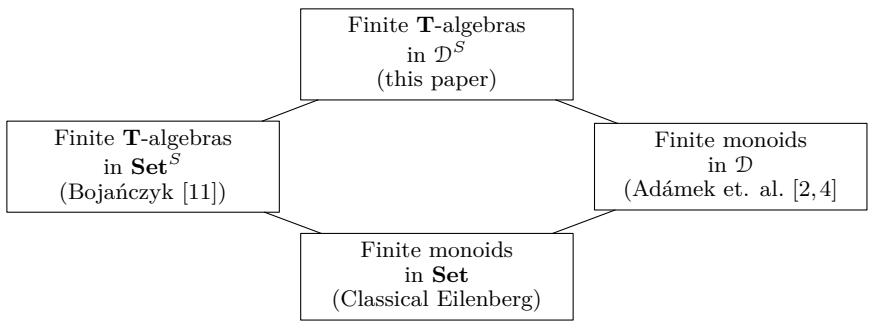
\footnotesize{
\[
\xymatrix@=8pt{
& *+[F]{\txt<8pc>{Finite $\MT$-algebras\\ in $\D^S$\\ (this paper)}} \ar@{-}[dl] \ar@{-}[dr] &\\
*+[F]{\txt<8pc>{Finite $\MT$-algebras \\ in $\Set^S$ \\ (Boja\'nczyk \cite{boj15})}} \ar@{-}[dr] & & *+[F]{\txt<8pc>{Finite monoids\\ in $\D$  \\ (Ad\'amek et. al. \cite{ammu14,ammu15}}} \ar@{-}[dl]\\
& *+[F]{\txt<8pc>{Finite monoids\\ in $\Set$ \\ (Classical Eilenberg)}} &
}
\]
}
\normalsize
 For example, to capture regular languages one
 takes the free-monoid monad $\MT\Sigma = \Sigma^*$ on $\Set$. For regular $\infty$-languages one takes the monad $\MT(\Sigma,\Gamma)=(\Sigma^+,\Sigma^\omega+\Sigma^*\times\Gamma)$ on $\Set^2$ representing $\omega$-semigroups. And for weighted languages over a finite field $\Field$, ones takes the monad  $\MT$ on the category of vector spaces representing associative algebras over $\Field$.

For Steps~2 and 3, we develop our main technical tool, the concept of a \emph{unary presentation} of a monad. Intuitively, a unary presentation expresses in categorical terms how to present $\MT$-algebras by  unary operations; for example, a monoid $M$ can be presented by the translation maps $x\mapsto yx$ and $x\mapsto xy$ for $y\in M$. It turns out that unary presentations are, in a precise sense, both necessary and sufficient for constructing syntactic algebras (Theorem \ref{thm:synalgconst}). This observation clarifies the role of syntactic algebras in classical work on Eilenberg-type theorems, as well as the nature of derivatives appearing in varieties of languages. In our paper, syntactic algebras are \emph{not} used for proving the variety theorem: we rely solely on the elementary algebraic concept of a unary presentation.

We emphasize that, in general, nontrivial work lies in 
finding a good unary presentation for a monad $\MT$. However, our work here shows that 
then Steps 3 and 4 are entirely generic: after choosing a unary 
presentation, we obtain a notion of \emph{variety of $\MT$-recognizable languages} (involving a notion of derivatives directly inferred from on the presentation) and the following

\smallskip
\noindent\textbf{Variety Theorem.} \emph{Varieties of $\MT$-recognizable languages are in bijective correspondence with pseudovarieties of $\MT$-algebras.}
\smallskip

\noindent The proof of the Variety Theorem has two main ingredients. The first 
one is duality: besides $\D$ we also consider a variety $\C$ 
that is dual to $\D$ on the level of finite algebras. Varieties of 
languages live in $\C$, while over $\D^S$ we form pseudovarieties of 
$\MT$-algebras. This is much inspired by the work of Gehrke, Grigorieff, and Pin \cite{ggp08} who interpreted the original Eilenberg theorem \cite{eilenberg76} in terms of Stone duality. In our categorical setting, this corresponds to choosing $\C =$ boolean algebras and $\D=\Set$, and the free-monoid monad $\MT\Sigma=\Sigma^*$ on $\Set$. 

The second ingredient is the \emph{profinite monad} $\widehat\MT$ associated to $\MT$, 
introduced in \cite{camu16}. It extends the classical construction of 
the free profinite monoid, and allows for the introduction of profinite topological 
methods to our setting.  For example, Pippenger's result~\cite{Pippenger1997} that the boolean 
algebra of regular languages dualizes to the free profinite monoid is shown to hold at the level of monads. Moreover, the profinite monad will lead us to a categorical Reiterman theorem (Theorem \ref{thm:reiterman}), asserting that pseudovarieties of $\MT$-algebras correspond to profinite equational theories. The proof of the Variety Theorem then boils down to the fact that 
\begin{enumerate}[(i)]
\item varieties of $\MT$-recognizable languages dualize to profinite equational theories, and 
\item by the Reiterman theorem, profinite equational theories correspond to pseudovarieties of $\MT$-algebras.
\end{enumerate}
Our Eilenberg-Reiterman theory for a monad establishes a conceptual and highly 
parametric, yet easily applicable framework for algebraic language theory. In fact, to establish an Eilenberg correspondence using our framework, the traditional four steps indicated above are replaced by the following simple three-step procedure:
\begin{enumerate}
\item Find a monad $\MT$ whose finite algebras recognize the desired languages.
\item Choose a unary presentation for $\MT$.
\item Spell out what a variety of $\MT$-recognizable languages and a pseudovariety of $\MT$-algebras is (by instantiating our general definitions), and invoke the Variety Theorem.
\end{enumerate}
To illustrate the strength of this approach, we will show that
roughly a dozen Eilenberg theorems in the literature emerge as special cases of our results. In addition, we will get several new correspondence results for free, e.g.\ an extension of the local 
variety 
theorem of \cite{ggp08} from finite words to infinite 
words and trees, and a variety theorem for $\omega$-regular languages.

Our paper is structured as follows. In Section \ref{sec:setting} we recall the construction of the profinite monad associated to a monad $\MT$. In Section \ref{sec:pseudovar} we investigate pseudovarieties of $\MT$-algebras and establish our categorical Reiterman theorem: pseudovarieties corrrespond to profinite theories. In Section \ref{sec:unpres} we introduce unary presentations for $\MT$-algebras, our key concept for defining the derivatives of a language. In Section \ref{sec:reclang} we set up our duality framework for algebraic language theory and demonstrate that some important topological characterizations of regular languages hold at the level of monads. In Section \ref{sec:varietythm} we present the main result of our paper, the variety theorem for varieties of $\MT$-recognizable languages and pseudovarieties of $\MT$-algebras. In Section 7 the variety theorem is generalized to reduced languages and reduced $\MT$-algebras, a useful variant in a many-sorted setting. Finally, in Section \ref{sec:applications}, we derive about twenty concrete Eilenberg correspondences as direct instances of our general results, including a number of correspondences that are not yet know in the literature.

Readers should be familiar with basic category theory, topology and universal algebra. A toolkit collecting some important definitions and facts is given in the Appendix.

\section{The Profinite Monad}\label{sec:setting}

In the section, we recall from \cite{camu16} the construction of the \emph{profinite monad} $\widehat\MT$ associated to a monad $\MT$. The profinite monad will play a key role for our categorical approach to Eilenberg-type theorems. 

Recall that for a finitary one-sorted signature $\Gamma$, a \emph{variety of algebras} is a class of $\Gamma$-algebras presented by equations $s=t$ between $\Gamma$-terms. A \emph{variety of ordered algebras} is a class of ordered $\Gamma$-algebras (i.e. $\Gamma$-algebras on a poset with monotone $\Gamma$-operations) presented by inequations $s\leq t$ between $\Gamma$-terms. Morphisms of (ordered) $\Gamma$-algebras are (monotone) $\Gamma$-homomorphisms. In our applications we will work with the following varieties:
\begin{table}[h]
\begin{tabularx}{\columnwidth}{lX}
\hline \vspace{-0.4cm}\\
$\Set$         & sets and functions \\
$\Pos$ & posets and monotone maps \\
$\BA$ &  boolean algebras and boolean homomorphisms\\
$\DL$ & distributive lattices with $0$ and $1$ and lattice homomorphisms preserving $0$ and $1$\\
$\JSL$ & join-semilattices with $0$ and semilattice homomorphisms preserving $0$\\
$\Vect_{\Field}$ & vector spaces over a finite field $\Field$ and linear maps\\
\hline
\end{tabularx}
\end{table}

\noindent We view $\Pos$ as a variety of ordered algebras over the empty signature, and the other categories above as varieties of algebras. All these varieties are \emph{locally finite}, i.e. every finitely generated algebra is finite.

\begin{notation}
For a category $\D$ and a set $S$ of sorts, we denote by $\D^S$ the $S$-fold product category of $\D$. Thus an object of $\D^S$ is a family $D=(D_s)_{s\in S}$ of objects in $\D$ and a morphism $f\colon D\to D'$ is a family $(f_s\colon D_s\to D'_s)_{s\in S}$ of morphisms in $\D$, with composition and identities taken sortwise. We often write $f$ for $f_s$ if the sort $s$ is clear from the context. If $\D$ is a variety of algebras or ordered algebras, we call a morphism $f$ of $\D^S$ \emph{surjective (injective, order-reflecting)} if this property holds sortwise, i.e. for every component $f_s$. 
\end{notation}
For the rest of this paper let us fix the following data:
\begin{enumerate}[(i)]
\item a variety $\D$ of algebras or ordered algebras;
\item a finite set $S$ of sorts;
\item a monad $\MT=(T,\eta,\mu)$ on $\D^S$;
\item a class $\F$ of finite $\MT$-algebras.
\end{enumerate}
We denote by $\FAlg{\MT}$ and $\Alg{\MT}$
 the categories of (finite) $\MT$-algebras and $\MT$-homomorphisms, and we view $\F$ as a full subcategory of $\FAlg{\MT}$.

\smallskip
\noindent The purpose of the monad $\MT$ is to represent formal languages and the possibly many-sorted algebras recognizing them. Languages are modeled as morphisms from free $\MT$-algebras into a suitable object of outputs in $\D^S$. The class $\F$ consists of those finite $\MT$-algebras that we use to recognize languages. In most applications, we will chose $\F =$ all finite $\MT$-algebras.

\begin{example}[Monoids]\label{ex:monads_monoids}
 Languages of finite words are represented by the free-monoid monad $\MT_\ast\Sigma=\Sigma^*$ on $\Set$ whose algebras are monoids. A language $L\seq\Sigma^*$ corresponds (via its characteristic function) to a function from $T_\ast\Sigma$ into the set $\{0,1\}$.
 \end{example}
 
 \begin{example}[$\omega$-semigroups]\label{ex:monads_omegasem}
Languages of finite and infinite words (\emph{$\infty$-languages} for short) are 
represented by the monad  
$\MT_\infty$ on $\Set^2$ associated to 
the 
algebraic theory 
of $\omega$-semigroups \cite{pinperrin04}. By an \emph{$\omega$-semigroup} is meant a 
two-sorted 
set $A=(A_+,A_\omega)$ 
equipped with a binary product $A_+\times A_+ \xra{\o} A_+$, a mixed 
binary product $A_+\times A_\omega \xra{\o} A_\omega$ and an 
$\omega$-ary 
product $A_+^\omega \xra{\pi} A_\omega$ satisfying
the following (mixed) associative laws:
\begin{enumerate}[(i)]
\item $(a \o b)\o c = a\o (b\o c)$ for all $a,b,c\in A_+$;
\item $(a \o b)\o z = a\o (b\o z)$ for all $a,b\in A_+$ and $z\in A_\omega$;
\item $\pi(a_0, a_1, a_2,\ldots) = a_0\o \pi(a_1,a_2,\ldots)$ for all $a_0,a_1,a_2,,\ldots\in A_+$;
\item $\pi(a_0,a_1,a_2,a_3,\ldots) = \pi(a_0\ldots a_{k_0-1}, a_{k_0}\ldots a_{k_1-1}, a_{k_1}\ldots a_{k_2-1}, \ldots)$ for all $a_0,a_1,a_2,\ldots\in A_+$ and all strictly increasing sequences $0<k_0<k_1<k_2<\ldots$ of natural numbers.
\end{enumerate}
A \emph{morphism} between two $\omega$-semigroups $A$ and $B$ is a two-sorted map $h\colon A\to B$ that preserves all finite and infinite products. The free $\omega$-semigroup generated by the two-sorted set
 $(\Sigma,\Gamma)$ is carried by $(\Sigma^+, \Sigma^\omega+\Sigma^*\times\Gamma)$, with products given by concatenation of words. Thus $T_\infty(\Sigma,\Gamma)=(\Sigma^+,\Sigma^\omega+\Sigma^*\times \Gamma)$ and, in particular, $T_\infty(\Sigma,\emptyset)=(\Sigma^+,\Sigma^\omega)$. An $\infty$-language $L\seq (\Sigma^+,\Sigma^\omega)$ thus corresponds to a two-sorted function from $T_\infty(\Sigma,\emptyset)$ to $(\{0,1\},\{0,1\})\in \Set^2$.
\end{example}

\begin{example}[$\Field$-algebras]\label{ex:monads_kalgs}
Weighted languages over a finite field 
$\Field$ are represented 
by the monad $\MT_\Field$ on $\Vect_\Field$ constructing free 
$\Field$-algebras (= associative algebras over the field $\Field$). For the vector space $\Field^\Sigma$ with finite basis $\Sigma$ we 
have 
$\MT_\Field 
(\Field^\Sigma) = 
\Field[\Sigma]$, the $\Field$-algebra of all polynomials $\sum_{i<n} k_iw_i$ with $k_i\in \Field$ and
$w_i\in\Sigma^*$.  Since $\Field[\Sigma]$ has basis $\Sigma^*$, a weighted language $L\colon \Sigma^*\to \Field$ corresponds to a linear map from $T_\Field(\Field^\Sigma)$ to $\Field$.
\end{example}

\begin{assumptions}\label{ass:dt}
From now on, we assume that 
\begin{enumerate}[(i)]
\item the variety $\D$ is locally finite;
\item all epimorphisms in $\D$ are surjective;
\item $T$ preserves epimorphisms;
\item $\F$ is closed under quotients, subalgebras and finite products.
\end{enumerate}
\end{assumptions}

\begin{rem}
\begin{enumerate}\label{rem:tpressurj}
\item The variety $\D$ has the factorization system of surjective morphisms and injective (resp. order-reflecting) morphisms, cf. \ref{app:factsystem}. It extends sortwise to $\D^S$. \emph{Quotients} and \emph{subobjects} is $\D$ and $\D^S$ are taken in this factorization system.
\item  Since
$T$ 
preserves surjections, the factorization system of $\D^S$ lifts to the category $\Alg{\MT}$: every $\MT$-homomorphism factorizes into a 
 surjective homomorphism followed by an injective (resp. 
order-reflecting) one. 
\emph{Quotients} and \emph{subalgebras} in $\Alg{\MT}$ are taken in this
factorization system.
\item The \emph{homomorphism theorem} holds for $\MT$-algebras: let $e\colon A\epito B$ and $h\colon A\to C$ be $\MT$-homomorphisms with $e$ surjective, and suppose that, for every sort $s$ and any two elements $a,a'\in \under{A_s}$,
\[ e(a) = e(a') \quad\text{implies}\quad h(a)=h(a'),\]
respectively (in the ordered case)
\[ e(a) \leq e(a') \quad\text{implies}\quad h(a)\leq h(a').\]  
Then there exists a unique $\MT$-homomorphism $\ol h\colon B\to C$ with $h=\ol h \o e$. Indeed, the homomorphism theorem for the $S$-sorted variety $\D^S$ (see \ref{app:homtheorem}) yields a morphism $\ol h$ in $\D^S$ with $h=\ol h \o e$, and since $T$ preserves surjections, $\ol h$ is a $\MT$-homomorphism. 
\end{enumerate}
\end{rem}

\begin{notation}
For a variety $\D$ of algebras, let $\Stone(\D)$ be the category of \emph{Stone $\D$-algebras}; its objects are $\D$-algebras equipped with a Stone topology and continuous $\D$-operations, and its morphisms are continuous $\D$-morphisms. Similarly, for a variety $\D$ of ordered algebras, let $\Priest(\D)$ be the category of \emph{Priestley $\D$-algebras}; its objects are $\D$-algebras with a Priestley topology and continuous monotone $\D$-operations, and its morphisms are continuous monotone $\D$-morphisms. We denote by $\hatD$ the full subcategory of $\Stone(\D)$ (resp. $\Priest(\D)$) on \emph{profinite $\D$-algebras}, i.e. inverse limits of algebras in $\D_f$. Here $\D_f$ is viewed as a full subcategory of $\hatD$ by equipping objects of $\D_f$ with the discrete topology.
\end{notation}

\begin{example}
For the varieties $\D =\Set$, $\Pos$, $\JSL$ and $\Vect_\Field$ that we encounter in our applications, every Stone (resp. Priestley) $\D$-algebra is profinite: we have 
\begin{enumerate}
\item $\widehat{\Set}=\Stone$ (Stone spaces);
\item $\widehat{\Pos} = \Priest$ (Priestley spaces);
\item
$\widehat{\JSL} = \Stone(\JSL)$ (Stone semilattices);
\item  $\widehat{\Vect_\Field} = 
\Stone(\Vect_\Field)$ (Stone vector spaces).
\end{enumerate}
See e.g. Johnstone \cite{Johnstone1982}.
\end{example}
In the following we establish some important properties of the category $\hatD$. Recall from \ref{app:procomp} the concept of a \emph{pro-completion} (= free completion under cofiltered limits) of a category.
\begin{lemma}\label{lem:hatD-procomp}
$\hatD$ is the pro-completion of $\D_f$.
\end{lemma}
A proof is sketched in \cite[Remark VI.2.4]{Johnstone1982} for the unordered case, and the argument is similar for the ordered case. For convenience, we present a complete proof for the latter.

\begin{proof}
Suppose that $\D$ is a variety of ordered algebras. By \cite[Prop. 1]{ban72}, the category $\hatD$ is complete (with 
limits formed on the level of $\Set$) and, by definition, every object of 
$\hatD$ is a cofiltered limit of objects in $\D_f$. Thus, by 
\ref{app:procomp}, it only remains to show that every object $D\in \D_f$ is 
finitely copresentable in $\hatD$: given a cofiltered limit cone $(p_i\colon X\to 
X_i)_{i\in I}$ in $\hatD$ and a morphism $f\colon X\to D$, we need to show that $f$ 
factors through the cone essentially uniquely. The uniqueness part is clear since, 
forgetting the $\D$-algebraic structure, $D$ is finitely copresentable in 
$\Priest$. The existence of a factorization is established as follows:
\begin{enumerate}[(1)]
\item We first consider the case where all $X_i$'s are finite. Since $(p_i)$ is a cofiltered 
limit cone in $\Priest$ and $D$ is a finite poset with discrete topology (and 
thus finitely copresentable in $\Priest$), there exists an $i\in I$ and a 
monotone map $f'\colon X_i\to D$ with $f'\o p_i = f$. Choose $j\in I$ and a 
connecting morphism $g\colon X_j\to X_i$ with $g[X_j]=p_i[X]$, see Lemma 
\ref{lem:connecting-map-image}. We claim that the composite $h = f'\o g$ is a 
morphism of $\hatD$, i.e. preserves all $\D$-operations. Indeed, given an 
$n$-ary operation symbol $\sigma$ in the signature of $\D$ and $x_1,\ldots, 
x_n\in X_j$, choose $x_k'\in X$ with $g(x_k) = p_i(x_k')$ for each $k=1,\ldots,n$. 
Then
\begin{align*}
h(\sigma(x_1,\ldots, x_n)) &= f'(g(\sigma(x_1,\ldots, x_n))) & (\text{$h=f'g$})\\
&= f'(\sigma(g(x_1),\ldots, g(x_n))) & (\text{$g$ morphism of $\hatD$})\\
&= f'(\sigma(p_i(x_1'),\ldots,p_i(x_n'))) & (\text{$g(x_k')=p_i(x_k')$})\\
&= f'(p_i(\sigma(x_1',\ldots, x_n'))) & (\text{$p_i$ morphism of $\hatD$})\\
&= f(\sigma(x_1',\ldots, x_n')) & (\text{$f'p_i=f$})\\
&= \sigma(f(x_1'),\ldots, f(x_n')) & (\text{$f$ morphism of $\hatD$})\\
&= \sigma(f'\o p_i(x_1'),\ldots, f'\o p_i(x_n')) & (\text{$f=f'p_i$})\\
&= \sigma(f'\o g(x_1),\ldots, f'\o g(x_n)) & (\text{$p_i(x_k')=g(x_k)$})\\
&= \sigma(h(x_1),\ldots, h(x_n)) & (\text{$h=f'g$})
\end{align*}
Thus $h$ is a morphism in $\hatD$. Moreover, we have $f= f'\o p_i = f'\o g \o p_j = h\o p_j$, i.e. $f$ factors through $p_j$.
\item Now let the $X_i$'s be arbitrary. We may assume that $I$ is a codirected 
poset, see \ref{app:cofiltered_limits}. The connecting morphism for $i\leq j$ is denoted by $g_{ij}\colon X_i\to 
X_j$. Form the codirected poset
\[ J = \{\,(i,e): i\in I \text{ and } e\colon X_i\epito A_e \text{ is a surjective morphism in $\hatD$ with $A_e\in \D_f$ } \,\} \]
ordered by
\[ (i,e) \leq (j, q) \quad\text{iff}\quad i\leq j \text{ and } q\o g_{ij} = 
g\o e \text{ for some } g\colon A_e\to A_q. \]
Note that $g$ is necessarily unique because $e$ is surjective. A routine verification shows that the diagram $Q\colon J\to \hatD$ given by 
\[ (i,e)\mapsto A_e \quad\text{and}\quad ((i,e)\leq (j,q)) \mapsto g \]
has the limit cone
\[ (e\o p_i\colon X \to A_e)_{(i,e)\in J}.\]
By (1), there exists an $(i,e)\in J$ and a morphism $f'\colon A_e\to D$ with $f'\o e\o p_i = f$.  This shows that $f$ factors through $p_i$. \endproof
\end{enumerate}
\end{proof}

\begin{rem}\label{rem:hatdlfcp}
 From the above lemma and
\ref{app:procomp}, it follows the category $\hatD$ is 
locally 
finitely copresentable, and its finitely copresentable objects are the objects 
of $\D_f$. Since the set $S$ of sorts is assumed to be finite, this implies that the product
category $\D^S$ is also locally finitely copresentable and its finitely 
copresentable objects
are the objects of $\D_f^S$.
\end{rem}

\begin{notation}
We have the two forgetful functors 
\[V\colon \hatD\to \D\quad \text{and}\quad V^S\colon \hatD^S\to \D^S.\]  We usually drop the superscript and write $V$ for $V^S$. Moreover, we write $D$ for $VD$ if $D$ is finite, and likewise $h$ for $Vh$ if $h$ is a morphism between finite objects.
\end{notation}

\begin{rem}\label{rem:vrightadjoint}
The forgetful functor $V\colon \hatD \to \D$ is a right adjoint, see \cite[Proposition 2.8]{camu16}. Consequently, it preserves limits and in particular monomorphisms. Since limits in 
$\hatD^S$ and $\D^S$ are computed sortwise, the same holds for the $S$-sorted 
forgetful functor $V\colon \hatD^S\to\D^S$.
\end{rem}

\begin{lemma}\label{lem:episurj}
All epimorphisms in $\hatD$ are surjective.
\end{lemma}
\begin{proof}
\begin{enumerate}[(1)]
\item First, observe that all monomorphisms in $\hatD$ are injective because 
 $V\colon\hatD\to\D$ preserves 
monomorphisms, see Remark \ref{rem:vrightadjoint}, and monomorphisms in the variety $\D$ are injective.
\item We show that every epimorphism $e\colon A\epito B$ in $\D_f$ is an epimorphism in $\D$ (and therefore surjective by Assumption \ref{ass:dt}(ii)). To this end, suppose that 
$f,g\colon B\to C$ are morphisms
  in $\D$ with $f\o e = g\o e$. Express $C$ as a directed union $(c_i\colon
  C_i\monoto C)_{i\in I}$ of finite subobjects, using that $\D$ is a locally
  finite variety. Since $B$ is finite and the union is directed, the morphisms
  $f$ and $g$ factor through some $c_i$, i.e. there exist morphisms $f',g'$ 
  with $f=c_i\o f'$ and $g=c_i\o
  g'$. Then
  \[
    c_i\o f'\o e = f \o e = g \o e = c_i \o g' \o e,
  \]
  and since $c_i$ is monic, it follows that $f'\o e = g' \o e$.
  Since $e$ is an epimorphism in $\D_f$ and $f',g'$ are morphisms in $\D_f$, this implies
  $f'=g'$ and therefore $f=g$.
\item We prove that every epimorphism $e\colon A\epito B$ in $\hatD$ with $B\in \D_f$ is surjective. To see this, factorize $e$ as $e=m\o q$ with $q$ surjective and $m$ injective/order-reflecting. The morphism $m$ has finite domain and is thus a morphism in $\D_f$.  Moreover, since $e=m\o q$ and $e$ is an epimorphism in $\hatD$, so is $m$. In particular, $m$ is an epimorphism in $\D_f$. Then part (2) shows that $m$ is surjective, which by the equation $e=m\o q$ implies that $e$ is surjective.
\item Now let $e\colon A\epito B$ be an arbitrary epimorphism in $\hatD$. By 
\ref{app:lfcp_cat}, one can express $e$ in the locally finitely copresentable 
category $\hatD^\to$ as a cofiltered limit $((a_i,b_i)\colon e\to f_i)_{i\in I}$ 
of morphisms $f_i\colon 
A_i\to B_i$ in $\D_f$. Take the (epi, strong mono) 
factorizations of $a_i$ and $b_i$, see \ref{app:factsystem}(c). Diagonal 
fill-in gives a morphism $e_i$ as in the 
diagram 
below:
\[
\xymatrix{
A \ar@/_2em/[dd]_{a_i} \ar@{->>}[r]^e \ar@{->>}[d]_{a_i'} & B 
\ar@/^2em/[dd]^{b_i} \ar@{->>}[d]^{b_i'}\\
A_i' \ar@{-->>}[r]^{e_i} \ar@{ >->}[d]_{a_i''} & B_i' \ar@{ >->}[d]^{b_i''}\\
A_i \ar[r]_{f_i} & B_i
}
\]
The objects 
$A_i'$ 
and $B_i'$ are finite by part (1) of the proof. Moreover, since $e$ and $b_i'$ 
are 
epimorphic, so is $e_i$, and thus part (2) shows that $e_i$ is 
surjective. Moreover, by part (3), also $a_i'$ is surjective. Finally, observe that $((a_i',b_i')\colon e\to e_i)_{i\in I}$ is a 
cofiltered limit cone in $\hatD^\to$ by \ref{app:factsystem}(a),(c). Since limits in $\hatD^\to$ are computed componentwise, 
Lemma~\ref{lem:surjections-between-cofiltered-diagrams} shows
  that~$e$ is surjective.\qedhere
\end{enumerate}
\end{proof}
In the algebraic theory of regular languages, topology enters the stage via the Stone space $\widehat{\Sigma^*}$ of \emph{profinite words} over an alphabet $\Sigma$. It is formed as the inverse (a.k.a. cofiltered) limit of all finite quotient monoids of the free monoid $\Sigma^*$. In \cite{camu16} we generalized this construction from monoids to algebras for a monad, leading to the concept of the profinite monad $\widehat\MT$ and, more generally, the pro-$\F$ monad $\hatT$ associated to the monad $\MT$. In the following, we review the construction of pro-$\F$ monads and their basic properties. Recall from \ref{app:codensity_monad} the concept of a \emph{codensity monad} of a functor.

\begin{definition}[see \cite{camu16}]
The \emph{pro-$\F$ monad} of $\MT$ is the codensity monad of the functor
\[K\colon \F\to \hatD^S,\quad (A,\alpha)\mapsto A.  \]
We denote this monad on $\hatD^S$ by $\hatT=(\hatt,\hateta,\hatmu)$. If $\F=$ all finite $\MT$-algebras, we write $\widehat{\MT}$ for $\hatT$ and call $\widehat{\MT}$ the \emph{profinite monad} of $\MT$.
\end{definition}

\begin{rem}[see \cite{camu16}]\label{rem:hattconst}
\begin{enumerate}
\item A more explicit construction of $\hatT$ can be described as follows. For any $\D\in\D^S_f$, let $\MT D = (TD,\mu_D)$ be the free $\MT$-algebra on $D$ (cf. \ref{app:algmonad}).  Form the comma category $(\MT D\downarrow \F)$ of all
$\MT$-homomorphisms $h\colon\MT D \to (A,\alpha)$ with $(A,\alpha)\in \F$, see
\ref{app:commacat}. Then the object $\hatt D$ is the limit of the cofiltered diagram 
\[
  (\MT D\downarrow \F) \to \hatD^S,\quad (h\colon\MT D \to (A,\alpha)) \to A.
\]
The cofilteredness of $(\MT D\downarrow \F)$ follows from our assumption that $\F$ is closed under subalgebras and finite products. We denote the limit projection associated to $h\colon\MT D \to (A,\alpha)$ by \[h^+\colon \hatt D \to A.\] In particular, for any $\MT$-algebra $(A,\alpha)\in \F$ we have the projection \[\alpha^+\colon \hatt A\to A\] because $\alpha\colon \MT A \to (A,\alpha)$ is a $\MT$-homomorphism by the associative law. The unit $\hateta_D$ and multiplication $\hatmu_D$ are uniquely determined by the commutative diagrams
 \begin{equation} \label{eq:hat-etasurj}
   \vcenter{
     \xymatrix@R-1em{
       {D} 
       \ar[r]^{\hateta_D}
       \ar[dr]_{h \eta_D}
       & 
       \hatt {D} \ar[d]^<<<{\br{h}} 
       &
       \hatt\hatt D \ar[l]_-{\hatmu_D}
       \ar[d]^<<<{\hatt \br{h}} 
       \\
       & 
       A
       &
       \hatt A 
       \ar@{->>}[l]^{\br{\alpha}}
     }
   }
\qquad \text{for all $h\colon \MT D\to (A,\alpha)$ in  $(\MT D\downarrow \F)$}.
\end{equation}
\item 
In the above explicit construction of $\hatt D)$ ($D\in\D_f^S$) as a cofiltered limit, it actually suffices to form $\hatt D$ as the limit of all \emph{quotients} of $\MT D$ with codomain in $\F$. More precisely, let $\Quo(\MT D,\F)$ be the full subcategory of $(\MT D\downarrow \F)$ on all surjective $\MT$-homomorphisms $e\colon \MT D\epito (A,\alpha)$ with codomain in $\F$. Then the inclusion functor $\Quo(\MT D,\F)
\hookrightarrow (\MT D \downarrow \F )$ is final (see
\ref{app:final_func}). Therefore $\hatt D$ is also the cofiltered limit of the smaller diagram
\[
  \Quo(\MT D,\F) \to \hatD^S,\quad (e\colon\MT D \epito (A,\alpha)) \to A,
\]
and the above limit cone $(\, h^+\colon \hatt D\to A\, )_{h\in (\MT D\downarrow F)}$ restricts to a limit cone \[(\, e^+\colon \hatt D\epito A\, )_{e\in \Quo(\MT D, \F)}.\] Note that the projections $e^+$ are surjective by Lemma \ref{lem:projection-is-surjective}.
\end{enumerate}
\end{rem}

\begin{example}
\begin{enumerate}
\item Let $\MT=\MT_\ast$ on $\Set$ and $\F=$ all finite monoids.
The monad $\widehat{\MT_*}$ on $\Stone$ assigns to each finite set (= finite discrete space)
$\Sigma$ the  
Stone space $\widehat{\Sigma^*}$ of profinite words, obtained as the inverse limit of all finite quotient monoids of the free monoid $\Sigma^*$.
\item Let $\MT=\MT_\Field$ on $\Vect_\Field$ and $\F=$ all finite vector spaces. The monad $\widehat{\MT_\Field}$ on
$\Stone(\Vect_\Field)$ assigns
to each finite vector space $\Field^\Sigma$ the Stone vector space obtained as 
the limit of all finite linear quotient spaces of $\Field[\Sigma]$.
\end{enumerate}
\end{example} 

\begin{rem}\label{rem:basic} We list some important technical properties of $\hatT$ that we will use throughout our paper. See \cite{camu16} for proofs.
\begin{enumerate}
\item \label{rem:iota} There is a natural transformation
  $\iota\colon TV \to V\hatt$ (where $V\colon \hatD^S\to \D^S$ is the forgetful functor) whose component
  $\iota_D\colon TVD \to V\hatt D$ for $D\in\D_f^S$ is uniquely determined by the commutative triangles
\[
\vcenter{
\xymatrix{
TVD \ar[dr]_{h} \ar[r]^{\iota_D} \ar[r] & V\hatt D \ar[d]^{Vh^+}\\
& A
}
}
\qquad \text{for all $h\colon \MT D \to (A,\alpha)$ in $(\MT D\downarrow \F)$}.
\]
The morphism $\iota_D$ is \emph{dense}, i.e. its image forms (sortwise) a dense subset of $\hatt D$ in $\hatD^S$. The density of $\iota_D$ implies two statements that we will use frequently:
\begin{enumerate}[(i)]
\item For any two morphisms $f,g\colon \hatt D\to A$ in $\hatD^S$ with $Vf\o \iota_D = Vg\o \iota_D$, one has $f=g$. Indeed, this follows from the fact that the topology of $A$ is Hausdorff.
\item For any surjective morphism $e\colon \hatt D \epito A$ in
    $\hatD^S$ with finite codomain, the restriction $Ve\o \iota_D\colon TD\epito
    A$ is also surjective, as this map is dense and the topology of $A$ is discrete.
\end{enumerate}
\item \label{rem:finite_talgs} If $(A,\alpha)\in\F$, then $(A, \alpha^+)$ is a $\hatT$-algebra: putting $h=\alpha$ in the commutative diagram \eqref{eq:hat-etasurj} gives the
  unit and associative law. Conversely, every finite $\hatT$-algebra $(B,\beta)$ restricts to the finite $\MT$-algebra $(B,V\beta\o \iota_B)$ in $\F$. This yields
  an isomorphism \[\F \cong \FAlg{\hatT}\] given by
  $(A,\alpha)\mapsto (A,\alpha^+)$ and $h\mapsto h$; its inverse is given by $(B,\beta)\mapsto (B,V\beta\o \iota_B)$ and $h\mapsto h$. We will therefore tacitly identify finite $\hatT$-algebras with $\MT$-algebras in $\F$.
\item \label{rem:hatthomrestrict} Every $\hatT$-homomorphism $h\colon \hatT D\to (B,\beta)$ with $D\in \D_f^S$ and finite codomain $(B,\beta)$ restricts to the $\MT$-homomorphism $Vh\o \iota_D\colon \MT D\to (B,V\beta\o \iota_B)$.
\item \label{rem:finitelycopres} Every finite $\hatT$-algebra is a finitely
          copresentable object (cf. \ref{app:fin_copres}) in $\Alg{\hatT}$.
\item \label{rem:factorization} $\hatt$ preserves 
    surjections. Therefore, in analogy to Remark \ref{rem:tpressurj}, the factorization system of surjective and injective/order-reflecting morphisms of $\hatD^S$ lifts to $\Alg{\hatT}$. \emph{Quotients} and \emph{subalgebras} of $\hatT$-algebras are taken w.r.t. this factorization system. Moreover, the homomorphism theorem holds for $\hatT$-algebras.
 \end{enumerate}
\end{rem}

\begin{lemma}\label{lem:gat}
Every $\MT$-homomorphism $g\colon \MT D'\to\MT D$ with $D,D'\in \D_f^S$ extends 
uniquely to a 
$\hatT$-homomorphism ${\hat g}\colon\hatT D'\to\hatT D$ such that the following 
diagrams commute for all $\MT$-homomorphisms $h\colon \MT D\to A$ with $A\in 
\F$:
\begin{equation}\label{eq:h_at}
  \vcenter{
    \xymatrix{
    T D' \ar[r]^g \ar[d]_{\iota_{D'}} & TD \ar[d]^{\iota_{D}}\\
    V\hatt D' \ar[r]_{V{\hat g}}  & V\hatt D 
    }
  }
  \qquad\qquad
  \vcenter{
    \xymatrix{
      \hatt D' \ar[r]^{\hat g} \ar[dr]_{\br{(hg)}} & \hatt D 
      \ar[d]^{\br{h}}\\
      & A
    }
  }
\end{equation}
\end{lemma}

\begin{proof}
The morphisms $(hg)^+$ form a compatible family over the diagram defining 
$\hatt D$, i.e.\ for all $\MT$-homomorphisms $k\colon A\to A'$ in $\F$ we 
have $(khg)^+ = k\o (hg)^+$. Indeed, this holds when precomposed with the 
dense map $\iota_{D'}$: 
\[
  V(khg)^+\cdot \iota_{D'} = khg = k \cdot V(hg)^+ \cdot \iota_{D'}. 
\]
Thus there exists a unique $\hat g\colon \hatt D' \to \hatt D$ with $(hg)^+ = h^+\o 
\hat g$ for all $h$, i.e. making the right-hand diagram of~\eqref{eq:h_at} commute. This
also implies that the left-hand diagram commutes, since it commutes when
postcomposed with every limit projection $Vh^+$ (cf. Remark \ref{rem:vrightadjoint}).
\[
  \xymatrix{
    TD'
    \ar[r]^-g
    \ar[d]_{\iota_{D'}}
    \ar `l[d] `[dd] [ddrr]_(.4){hg} 
    &
    TD
    \ar[d]_{\iota_D}
    \ar[rdd]^h
    \\
    V\hatt D'
    \ar[r]^-{V\hat g}
    \ar[rrd]_-{V(hg)^+}
    &
    V\hatt D
    \ar[rd]|(.4){Vh^+}
    \\
    &&
    A
    }
\]
It remains to show that $\hat g$ is a $\hatT$-homomorphism. To this end consider the following diagrams for all $h\colon \MT D\to (A,\alpha)$ with $(A,\alpha)\in \F$:
\[
\xymatrix{
\hatt\hatt D' \ar@/_10ex/[dd]_{\hatt (hg)^+} \ar[r]^{\hatmu_{D'}} \ar[d]_{\hatt\hat g} & \hatt D' \ar[d]^{\hat g} \ar@/^5ex/[dd]^{(hg)^+}\\
\hatt\hatt D \ar[r]_{\hatmu_D} \ar[d]_{\hatt h^+} & \hatt D \ar[d]^{h^+}\\
\hatt A \ar[r]_{\alpha^+} & A
}
\]
The outside and the lower square commute by \eqref{eq:hat-etasurj}, and the left-hand and right-hand parts by the definition of $\hat g$. Therefore the upper square commutes when postcomposed with the limit projections $h^+$. It follows that the upper square commutes, showing that $\hat g\colon \hatT D'\to \hatT D$ is a $\hatT$-homomorphism.
\end{proof}

\section{Pseudovarieties of 
\texorpdfstring{$\MT$}{T}-algebras}\label{sec:pseudovar}

In this section we study pseudovarieties of $\MT$-algebras, the 
algebraic half of any Eilenberg-type correspondence, and their connection to profinite $\hatT$-algebras. The main result is a categorical Reiterman theorem that describes pseudovarieties of $\F$-algebras in terms of \emph{pro-$\F$ theories}. This result will form the basis for our duality approach to Eilenberg-type theorems.

\subsection{Profinite $\hatT$-algebras}
In the last section, we introduced the category $\hatD$ of profinite $\D$-algebras. The concept of profiniteness can be extended to $\hatT$-algebras as follows.
 
\begin{definition}
A $\hatT$-algebra is called \emph{profinite} if it is a cofiltered limit of finite 
$\hatT$-algebras (= finite $\F$-algebras, see Remark \ref{rem:basic}.\ref{rem:finite_talgs}) in $\Alg{\hatT}$.
\end{definition}

\begin{rem}\label{rem:hattprofinite}
Every free $\hatT$-algebra $\hatT D$ with $D\in \D_f^S$ is profinite.  Indeed, 
since the forgetful functor from $\Alg{\hatT}$ to $\hatD^S$ reflects 
limits, see 
\ref{app:limits_of_talgs}, the 
right-hand square of \eqref{eq:hat-etasurj} shows that the $\hatT$-algebra $\hatT D$ 
is the cofiltered limit of the diagram
\[ (\MT D\downarrow \F)\to \Alg{\hatT},\quad (h\colon \MT D 
\to 
(A,\alpha)) \mapsto (A,\alpha^+), \]
with limit projections $h^+$.
\end{rem}

\begin{lemma}\label{lem:profinite_candiagram}
Every profinite $\hatT$-algebra $A$ is the cofiltered limit of its canonical diagram. That is, the cofiltered 
diagram
\[ (A\downarrow \FAlg{\hatT}) \to \Alg{\hatT}, \quad (h\colon A\to A')\mapsto A', 
\] 
has limit $A$ with limit projections $h\colon A\to A'$.
\end{lemma}

\begin{proof}
Let $A$ be a profinite $\hatT$-algebra. Then, by definition, there exists a cofiltered limit cone $(p_i\colon A\to A_i)_{i\in I}$ in 
$\Alg{\hatT}$ with $A_i\in \FAlg{\hatT}$. Since the forgetful functor $\hat U\colon \Alg{\hatT} \to \hatD^S$ reflects 
limits, 
it suffices to show that the cofiltered cone $(h\colon \hat UA\to \hat UA')$ in 
$\hatD^S$ is a 
limit cone. 
To this end 
we verify the criterion of \ref{app:cof_lim_in_lfcp_cat}. 

For (i), let $f\colon \hat UA\to B$ be a morphism in $\hatD^S$ with
$B\in\hatD_f^S$. Since $\hat U$ preserves limits, we have the limit
cone $(p_i\colon \hat U A\to \hat U A_i)$ in $\hatD^S$. Moreover, since $B$
is finitely copresentable in $\hatD^S$, see Remark \ref{rem:hatdlfcp},
there exists an $i$ and morphism $f'\colon \hat UA_i\to B$ with
$f= f'\o p_i$.  This proves that $f$ factors through the cone
$(h\colon \hat U A\to \hat U A')$ via $h=p_i$ and $f'$, as desired.

For (ii), suppose that $h\colon A\to A'$ in $(A\downarrow \FAlg{\hatT})$ and 
$f',f''\colon \hat UA'\to B$ are given with 
$f'\o 
h = f''\o h$. Since $A'$ is finitely copresentable in 
$\Alg{\hatT}$ by Remark \ref{rem:basic}.\ref{rem:finitelycopres}, there 
exists 
an $i$ and a $\hatT$-homomorphism  $h'\colon A_i\to A'$ with $h=h'\o p_i$. Then 
$(f'\o h')\o  p_i = (f''\o h') \o p_i$. Thus, by (ii) applied to 
the cofiltered limit cone $(p_i\colon \hat U A\to \hat U A_i)$, we have a connecting morphism 
$a_{ji}\colon A_j\to A_i$ 
 with $f'\o h'\o  a_{ji} = f''\o h' \o a_{ji}$. Thus $h'\o 
a_{ji}\colon p_j\to h$ is a morphism in  $(A\downarrow \FAlg{\hatT})$ that 
merges $f'$ and $f''$, as desired.
\[
\xymatrix{
& & \hat UA \ar[dll]_{p_j} \ar[dl]|{p_i} \ar[d]^{h} \ar[dr]^f &\\
\hat U A_j \ar[r]_{a_{ji}} & \hat U A_i \ar[r]_{h'} & \hat U A' 
\ar@<0.5ex>[r]^{f'} 
\ar@<-0.5ex>[r]_{f''} & B 
}
\vspace*{-25pt}
\]
\end{proof}

\begin{notation}
Let $(A\epidownarrow 
\FAlg{\hatT})$ be the full subcategory of $(A\downarrow 
\FAlg{\hatT})$ on all surjective $\hatT$-homomorphisms $e\colon A\epito A'$ with 
finite codomain.
\end{notation}

\begin{corollary}\label{cor:profinite_epi}
A $\hatT$-algebra $A$ is profinite iff $A$ is the limit of the cofiltered 
diagram
\[ (A\epidownarrow \FAlg{\hatT}) \to \Alg{\hatT}, \quad (e\colon A\epito 
A')\mapsto A', 
\] 
with limit projections $e\colon A\epito A'$.
\end{corollary}

\begin{proof}
  By Remark \ref{rem:basic}.\ref{rem:factorization}, the inclusion functor $(A\epidownarrow \FAlg{\hatT})\monoto(A\downarrow \FAlg{\hatT})$ is final. This immediately implies the claim, see \ref{app:final_func}.
\end{proof}

\subsection{Local pseudovarieties and the local Reiterman theorem}
In this subsection, we present a local version of a categorical Reiterman theorem. ``Local'' means that we fix a finite $S$-sorted set $\Sigma\in\Set_f^S$  and consider finite $\MT$-algebras and profinite $\hatT$-algebras equipped with a $\Sigma$-indexed set of generators. On the level of languages, this will correspond to considering languages over the alphabet $\Sigma$. 

\begin{notation}
For any $\Sigma\in\Set_f^S$, let $\FSigma\in \D_f^S$ denote the free object in $\D^S$ generated by $\Sigma$, taken w.r.t. the forgetful functor $\under{\dash}: \D^S\to \Set_f^S$. Note that $\FSigma$ is finite because $\D$ is locally finite. Recall that $\Quo(\MT\FSigma,\F)$ denotes the class of all quotients $e\colon \MT\FSigma\epito A$ in $\Alg{\MT}$ with codomain $A\in \F$. We view it as a poset ordered by $e\leq e'$ if $e$ factorizes through $e'$, i.e. $e=q\o e'$ for some $q$. As usual, two quotients $e$ and $e'$ are identified if $e\leq e'$ and $e'\leq e$, i.e. if they are connected by an isomorphism.
\end{notation} 

\begin{definition}\label{def:localpseudovariety} Let $\Sigma\in\Set_f^S$. By a \emph{$\Sigma$-generated $\F$-algebra} is meant an element of $\Quo(\MT\FSigma,\F)$. 
 A \emph{local pseudovariety of 
$\Sigma$-generated 
$\F$-algebras} is a nonempty class $\P\seq \Quo(\MT\FSigma,\F)$ of $\Sigma$-generated $\F$-algebras that forms an ideal of the poset $\Quo(\MT\FSigma,\F)$, i.e.
\begin{enumerate}[(i)]
\item $\P$ is downwards closed: $e_0\leq e_1$ and $e_1\in \P$ implies $e_0\in \P$;
\item $\P$ is directed: for any two elements $e_0,e_1\in \P$, there exists an $e\in \P$ with $e_0,e_1\leq e$.
\end{enumerate}
\end{definition}

\begin{rem}\label{rem:subdirectproduct}
Recall from \ref{app:subdirectproducts} the notion of a \emph{subdirect product} of two quotients.
Clearly condition (ii) in the above definition can be replaced by
\begin{itemize}
\item[(ii')]$\P$ is closed under subdirect products:  for any two elements $e_0,e_1\in \P$, the subdirect product $e$ of $e_0$ and $e_1$ lies in $\P$.
\end{itemize}
\end{rem}

\begin{definition}
A \emph{$\Sigma$-generated profinite $\hatT$-algebra} is a 
quotient
$\phi\colon\hatT\FSigma\epito P$ in
$\Alg{\hatT}$ whose codomain $P$ is a profinite $\hatT$-algebra. $\Sigma$-generated profinite $\hatT$-algebras are
ordered by
$\phi\leq \phi'$ iff $\phi$ factors through $\phi'$.
\end{definition}
Our goal in this subsection is to show that $\Sigma$-generated profinite $\hatT$-algebras and local pseudovarieties of $\Sigma$-generated $\F$-algebras are equivalent concepts (see Theorem \ref{thm:localreiterman}). To this end, we first explain how
  to translate a local pseudovariety into a $\Sigma$-generated
  profinite $\hatT$-algebra and vice versa.
\begin{construction}\label{rem:localvarconst} 
\begin{enumerate}
\item\label{rem:localvarconst1} To each local pseudovariety $\P$ of $\Sigma$-generated $\F$-algebras 
we associate a $\Sigma$-generated profinite $\hatT$-algebra \[\phi^\P\colon 
\hatT\FSigma\epito P^\P\] as follows. 
Viewed as full 
subcategory of the comma category $(\MT\FSigma\downarrow \F)$, the 
category $\P$ is cofiltered because $\P$ is directed w.r.t. $\leq$.
Let $P^\P$ be the cofiltered limit of the diagram
\[ \P \to \Alg{\hatT},\quad (e\colon \MT\FSigma\epito A) \mapsto A,\]
and denote the limit projections by \[e^*_\P\colon P^\P\epito A.\]
The projections are 
surjective by Lemma \ref{lem:projection-is-surjective}.
Thus $P^\P$ is a profinite $\hatT$-algebra. Moreover, the 
$\hatT$-homomorphisms $e^+\colon \hatT\FSigma\epito 
A$ 
(where $e$ ranges over all elements of $\P$) form a compatible family over the 
above diagram.  Hence there exists a unique $\hatT$-homomorphism 
$\phi^\P\colon \hatT\FSigma\epito P^\P$ with $e^+ = e_\P^*\o 
\phi^\P$ for 
all $e\in \P$.
\[
\xymatrix{
\hatT\FSigma \ar@{->>}[dr]_{e^+} \ar@{->>}[r]^{\phi^\P} & P^\P \ar@{->>}[d]^{e^*_\P}\\
& A
}
\]
 Note that $\phi^\P$ is surjective by Lemma 
\ref{lem:surjections-between-cofiltered-diagrams}. Thus $\phi^\P\colon 
\hatT\FSigma\epito P^\P$ is a $\Sigma$-generated profinite $\hatT$-algebra.
\item\label{rem:localvarconst2} Conversely, given a $\Sigma$-generated profinite $\hatT$-algebra 
$\phi\colon 
\hatT\FSigma\epito P$, define $\P^{\phi}$ to be the class of all 
$\Sigma$-generated $\F$-algebras of the form 
\[ e = (\xymatrix@1{\MT\FSigma \ar[r]^{\iota_\FSigma} & V\hatT\FSigma 
\ar@{->>}[r]^{V\phi} & VP \ar@{->>}[r]^{Ve'} & A}), \]
where $e'\colon P\epito A$ is a surjective $\hatT$-homomorphism with $A\in 
\FAlg{\hatT}\cong \F$. Note that any such morphism $e$ is indeed a surjective $\MT$-homomorphism by Remark \ref{rem:basic}.\ref{rem:iota} and \ref{rem:basic}.\ref{rem:hatthomrestrict}.
One easily verifies that $\P^{\phi}$ forms a local pseudovariety of 
$\Sigma$-generated $\F$-algebras. 
\end{enumerate}
\end{construction}

\begin{lemma}\label{lem:localpseudovar_localprofiniteth_1}
For any local pseudovariety $\P$ of $\Sigma$-generated $\MT$-algebras we have 
$\P = 
\P^\phi$ where $\phi := \phi^\P$.
\end{lemma}

\begin{proof}
$\P \seq 
\P^{\phi}$: Let $(e: \MT\FSigma\epito A)\in \P$. Then for the 
corresponding limit projection $e^*_\P: P^\P \epito A$ we have
\[ e = (\xymatrix{T\FSigma \ar[r]^{\iota_\FSigma} & V\hatt\FSigma 
\ar@{->>}[r]^{V\phi} & VP^\P \ar@{->>}[r]^{Ve^*_\P} & A}) 
\]
since $e^+ = e^*_\P \cdot \phi^\P$ by the definition of $\phi=\phi^\P$
and since $e^+ \cdot \iota_\FSigma = e$ by Remark
\ref{rem:basic}.\ref{rem:iota}. Therefore $e\in \P^{\phi}$ by the
definition of $\P^{\phi}$.

$\P^{\phi}\seq \P$: Let $(e: \MT\FSigma\epito A)\in 
\P^{\phi}$. Thus there exists a surjective $\hatT$-homomorphism 
$e': P^\P \epito A$ with $A\in \FAlg{\hatT}$ and
\[ e = (\xymatrix{T\FSigma \ar[r]^{\iota_\FSigma} & V\hatt\FSigma 
\ar@{->>}[r]^{V\phi} & VP^\P \ar@{->>}[r]^{Ve'} & A}). 
\]
Since $A$ is finitely copresentable in $\Alg{\hatT}$, see Remark 
\ref{rem:basic}.\ref{rem:finitelycopres}, the $\hatT$-homomorphism $e'$ 
factors through the the limit cone defining $P^\P$; that is, there 
exists an $h: 
\MT\FSigma\epito A'$ in $\P$ and a $\hatT$-homomorphism $s: A'\epito A$ with 
$e' = s\o h^*_\P$. Then the commutative diagram below shows that $e\leq h$. It follows that $e\in \P$, because $h\in \P$ and $\P$ is downwards closed.
\[
\xymatrix{
T\FSigma \ar[r]^-{\iota_\FSigma} \ar@{->>}@/_5em/[rrr]_h \ar@/^5em/[rrrr]^e & 
V\hatt\FSigma \ar@{->>}@/_2em/[rr]_{h^+} \ar@{->>}[r]^-{V\phi^\P} &  
VP^\P \ar@/^2em/[rr]^{Ve'} \ar@{->>}[r]^-{Vh_\P^*} & A' \ar@{->>}[r]^-s & 
A
}
\vspace*{-15pt}
\]
\end{proof}

\begin{lemma}\label{lem:localpseudovar_localprofiniteth_2}
 For each $\Sigma$-generated profinite $\hatT$-algebra $\phi: 
 \hatT\FSigma\epito 
P$ we have $\phi 
= \phi^\P$ where $\P := \P^\phi$.
\end{lemma}
%
More precisely, the lemma states that $\phi$ and 
$\phi^\P$ are identical as quotients of $\hatT\FSigma$, 
i.e. there exists 
an 
isomorphism $j: 
          P \xra{\cong} P^\P$ with 
          $\phi^\P = j\o 
          \phi$.
%
\begin{proof}
Let $(P\epidownarrow 
\FAlg{\hatT})$ be the full subcategory of $(P\downarrow 
\FAlg{\hatT})$ on all surjective $\hatT$-homomorphisms $e': 
\hatT\FSigma\epito 
A$ with $A\in \FAlg{\hatT}$. Consider the functor
\[ F: (P\epidownarrow 
\FAlg{\hatT}) \to \P \]
that maps $e': P\epito 
A$ to the  $\Sigma$-generated finite $\MT$-algebra
\[ F(e') = (\xymatrix{T\FSigma \ar[r]^{\iota_\FSigma} & V\hatt\FSigma 
\ar@{->>}[r]^{V\phi} & VP \ar@{->>}[r]^{Ve'} & A}). \]
and acts as identity on morphisms. Note that $F(e')\in\P$ by the 
definition of $\P=\P^\phi$, so $F$ is well-defined. We claim that $F$ is 
an isomorphism. Indeed, $F$ is injective on objects because 
$\phi$ is surjective and $\iota_\FSigma$ is dense. The surjectivity on 
objects is the definition of $\P$. The 
bijectivity on morphisms is clear.

Next observe that $F$ commutes with the projection functors $\pi$ and $\pi'$:
\[
\xymatrix{
(P\epidownarrow 
\FAlg{\hatT}) \ar[rr]^F \ar[dr]_{\pi} && \P
\ar[dl]^{\pi'}\\
& \Alg{\hatT} &
}
\]
The limit of $\pi$ is $P$ by 
Corollary \ref{cor:profinite_epi}, and 
the limit of $\pi'$ is $P^{\P}$ by the definition of 
$P^{\P}$. Since 
$F$ is an isomorphism (in particular, a final functor) and limits are unique 
up to isomorphism, there is an 
isomorphism $j: 
P\xra{\cong} P^\P$ with $e' =  F(e')^*_{\P}\o 
j$ for all 
$e': 
P\epito A$ in $(P\epidownarrow 
\FAlg{\hatT})$. Thus in the diagram below the outside and all inner parts 
except, perhaps, for the upper inner triangle commute:
\[
\xymatrix{
\hatt\FSigma  \ar[d]^{\iota_\FSigma} \ar@{->>}@/_4em/[ddd]_{F(e')} & \\
V\hatt \FSigma \ar@{->>}@/^14em/[dd]^{VF(e')^+}
\ar@{->>}[d]_{V\phi} 
\ar@{->>}[dr]|{V\phi^{\P}} &\\
V P \ar[r]|<<<<<{Vj}  
\ar@{->>}[d]_{Ve'} & VP^{\P} 
\ar@{->>}[dl]|{VF(e')^*_{\P}}\\
A
}
\]
It follows that this triangle also commutes, as it 
commutes when precomposed with the dense map 
$\iota_\FSigma$ and postcomposed with the limit projections 
$VF(e')^*_{\P}$.
\end{proof}

\begin{theorem}[Local Reiterman Theorem]\label{thm:localreiterman}
For each $\Sigma\in\Set_f^S$, the maps
\[
\P \mapsto \phi^\P\quad\text{and}\quad \phi \mapsto \P^\phi
\]
give an isomorphism between the lattice of local pseudovarieties of 
$\Sigma$-generated $\F$-algebras (ordered by inclusion) and the lattice of $\Sigma$-generated profinite $\hatT$-algebras.
\end{theorem}
\begin{proof} By Lemma
  \ref{lem:localpseudovar_localprofiniteth_1} and
  \ref{lem:localpseudovar_localprofiniteth_2} the maps
  $\P \mapsto \phi^\P$ and $\phi \mapsto \P^{\phi}$ are mutually
  inverse, and it only
  remains to prove that they are order-preserving. Given two $\Sigma$-generated profinite $\hatT$-algebras $\phi\colon \hatt\FSigma\epito P$ and $\phi'\colon \hatt\FSigma\epito P'$ with $\phi\leq \phi'$, we have
  $\P^{\phi}\seq \P^{\phi'}$ because every quotient of $P$ is
  also a quotient of $P'$. Given two local pseudovarieties
  $\P\seq \P'$, the morphisms $e^*_{\P'}\colon P^{\P'}\epito A$,
  where $e$ ranges over all $e\colon \MT\FSigma\epito A$ in $\P$, form a
  compatible family over the diagram defining $P^\P$. Indeed,
  for each morphism $h\colon e \to e'$ in $\P \subseteq \P'$
  (cf.~Construction~\ref{rem:localvarconst}.\ref{rem:localvarconst1}), the equation
  $h \cdot e^*_{\P'} = (e')^*_{\P'}$ holds for the limit
  projections. Hence there exists a unique morphism
  $q\colon P^{\P'}\to P^\P$ with $e^*_{\P'} = e^*_\P \o q$
  for all $e\in \P$. It follows that $q\o \phi^{\P'} = \phi^\P$,
  because this holds when postcomposed with the limit projections
  $e^*_\P$.
\[
\xymatrix{
\hatt\FSigma \ar@{->>}[d]_{\phi^{\P'}} 
\ar@{->>}[dr]^{\phi^\P} 
& \\
{P^{\P'}} \ar@{->>}[r]^{q} \ar@{->>}[d]_{e^*_{\P'}} & P^\P 
\ar@{->>}[dl]^{e^*_\P} \\
A
}
\]  
Indeed, the outside of the above diagram commutes because both sides
yield $e$ when we apply $V$ and precompose with the dense map
$\iota_\FSigma$. Therefore we have $\phi^{\P}\leq \phi^{\P'}$ as desired.
\end{proof}
Theorem
  \ref{thm:localreiterman} can be interpreted  in terms of {profinite (in-)equations}:

\begin{definition}
\begin{enumerate}
\item If $\D$ a variety of algebras, a \emph{profinite equation} over $\Sigma\in\Set_f^S$ is a pair of elements $u,v\in\under{\hatt\FSigma}_s$ in some sort $s$, denoted by $u=v$. We say that a
  $\Sigma$-generated  $\F$-algebra $e\colon\MT\FSigma\epito A$
  \emph{satisfies $u=v$} if $e^+(u)= e^+(v)$.
\item If $\D$ a variety of ordered algebras, a \emph{profinite inequation} over $\Sigma\in\Set_f^S$ is a pair of elements $u,v\in\under{\hatt\FSigma}_s$ in some sort $s$, denoted by $u\leq v$. We say that a
  $\Sigma$-generated $\F$-algebra $e\colon\MT\FSigma\epito A$
  \emph{satisfies $u \leq v$} if $e^+(u)\leq e^+(v)$.
\end{enumerate}
\end{definition}

\begin{notation}
For a set $E$ of profinite (in-)equations over $\Sigma$, let $\P[E]$ denote the 
class of all $\Sigma$-generated  $\F$-algebras satisfying all (in-)equations in 
$E$. Conversely, for a class $\P$ of $\Sigma$-generated $\F$-algebras, 
let $E[\P]$ be the set of all profinite (in-)equations over $\Sigma$ satisfied by all 
algebras in $\P$.
\end{notation}

\begin{corollary}[Equational Local Reiterman Theorem]
A class $\P$ of $\Sigma$-generated  $\F$-algebras forms a local pseudovariety iff $\P=\P[E]$ for some set $E$ of profinite equations (resp. profinite inequations) over $\Sigma$.
\end{corollary}

\begin{proof}
We consider the case where $\D$ is variety of ordered algebras; for the unordered case, replace all inequations by equations throughout the proof.

The ``if'' direction is a straightforward verification. For the ``only
if'' direction, suppose that $\P$ is a local pseudovariety of
$\Sigma$-generated $\F$-algebras, and let
$\phi^\P\colon \hatT\FSigma\epito P^\P$ be the corresponding
$\Sigma$-generated profinite $\hatT$-algebra. Let us first observe that a profinite inequation
$u\leq v$ lies in $E[\P]$ iff $\phi^\P(u)\leq \phi^\P(v)$. Indeed, for every quotient $e\colon T\FSigma\epito A$ in $\P$ we have the commutative triangle
\[
\xymatrix{
\hatt\FSigma \ar@{->>}[d]_{e^+}\ar@{->>}[r]^{\phi^\P} & P^\P \ar@{->>}[dl]^{e_\P^*}\\
A
}
\]
by the definition of $\phi^\P$. Therefore, $\phi^\P(u)\leq \phi^\P(v)$ implies $e^+(u)\leq e^+(v)$ for every $e\in\P$, i.e. the profinite inequation $u\leq v$ lies in $E[\P]$. Conversely, if $u\leq v$ is a profinite inequation in $E[\P]$, then
 $e_\P^*\o \phi^\P (u) = e^+(u) \leq e^+(v) = e_\P^*\o \phi^\P(v)$ for every $e\in \P$. This implies $\phi^\P(u)\leq \phi^\P (v)$ because the limit projections $e^*_\P$ ($e\in \P$) are jointly order-reflecting.

We claim that
$\P=\P[E[\P]]$. The inclusion $\seq$ is trivial. To prove $\supseteq$,
let $e\colon \MT\FSigma\epito A$ be an element of $\P[E[\P]]$, i.e. $e$
satisfies every profinite inequation that every algebra in $\P$ satisfies. By the
homomorphism theorem, see Remark \ref{rem:basic}.\ref{rem:factorization}, there
exists a (surjective) $\hatT$-homomorphism $h\colon P^\P \epito A$
with $e^+=h\o \phi^\P$. Indeed, by the argument above every pair $u,v$ with
$\phi^\P(u)\leq \phi^\P(v)$ forms a profinite inequation $u\leq v$
satisfied by $\P$. Thus $u\leq v$ is satisfied by $e$,
i.e. $e^+(u)\leq e^+(v)$.

Therefore $e=Vh\o V\phi^\P \o \iota_\FSigma$ lies in $\P^{(\phi^\P)}$ by the definition of $\P^{(\dash)}$. Since $\P^{(\phi^\P)}=\P$ by Theorem \ref{thm:localreiterman}, we thus have $e\in \P$, as required.
\end{proof}

\subsection{Pseudovarieties and Reiterman's theorem}

In this subsection we consider finite $\MT$-algebras that are not quotients of $\MT\FSigma$ for a \emph{fixed} alphabet $\Sigma$, as in the local case, but where the generators $\Sigma$ can be chosen from a subclass $\A\seq\Set_f^S$. On the language side, this corresponds to considering languages over alphabets in $\A$. For example, the classical Eilenberg theorem deals with regular languages $L\seq \Sigma^*$ over all possible finite alphabets $\Sigma$, and thus we will choose $\A=\Set_f$ in this case. However, as Example~\ref{ex:monads_omegasem} shows, in a many-sorted setting one may need to make a \emph{choice} of alphabets in $\Set_f^S$ for capture the desired languages; then $\A$ is chosen to be a proper subclass of $\Set_f^S$.

\begin{notation}
For the rest of our paper, let $\A\seq \Set_f^S$ be a fixed class of alphabets.
\end{notation}

\begin{definition}\label{def:agenerated}
A $\MT$-algebra $(A,\alpha)$ is called \emph{$\A$-generated} if there exists
a 
surjective $\MT$-homomorphism $e\colon \MT\FSigma\epito (A,\alpha)$ with
$\Sigma\in\A$.
\end{definition}
We emphasize that, in contrast to Definition \ref{def:localpseudovariety}, an $\A$-generated $\MT$-algebra $(A,\alpha)$ is not equipped with a \emph{fixed} surjective $\MT$-homomorphism $e\colon \MT\FSigma\epito (A,\alpha)$. Only the existence of $e$ is required.

\begin{example}\label{ex:agenerated_talgs}
Every 
finite $\MT$-algebra 
$(A,\alpha)$ 
is $\Set_f^S$-generated. Indeed, since $A$ is a finite (and thus a finitely generated) object of $\D^S$, there exists a surjective morphism
$e\colon\FSigma\epito A$ with $\Sigma\in \Set_f^S$, and therefore $(A,\alpha)$ is a 
quotient of $\MT\FSigma$ via $(\MT\FSigma\xra{Te} \MT A\xra{\alpha} 
(A,\alpha))$. Note that $Te$ is surjective by our assumption that $T$ preserves surjections, and $\alpha$ is surjective by the unit law for the $\MT$-algebra $(A,\alpha)$.
\end{example}

\begin{example}[$\omega$-semigroups]\label{ex:agenerated_omegasem} Let $\MT=\MT_\infty$ on $\Set^2$. As suggested by Example~\ref{ex:monads_omegasem}, we choose 
$\A=\{\,(\Sigma,\emptyset): \Sigma\in\Set_f\,\}$. 
A finite 
$\MT_\infty$-algebra (= finite $\omega$-semigroup) $A$ is 
$\A$-generated 
iff it is \emph{complete}, i.e.\
every element $a\in A_\omega$ can be expressed as an infinite product 
$a=\pi(a_0,a_1,\ldots)$ for some $a_i\in A_+$. Indeed, for
the ``only if'' direction, suppose that $A$ is $\A$-generated, i.e. there exists a surjective $\omega$-semigroup morphism $e\colon (\Sigma^+, 
\Sigma^\omega)\epito 
(A_+,A_\omega)$ for some $\Sigma\in\Set_f$. For any $a\in A_\omega$, choose 
$s_0s_1\ldots\in\Sigma^\omega$ with $a=e(s_0s_1\ldots)$. Then 
\[ a = e(s_0s_1\ldots) = e(\pi(s_0,s_1,\ldots)) = 
\pi(e(s_0),e(s_1),\ldots), \]
which shows that $A$ is complete. For the ``if'' 
direction, suppose that $A$ is complete. Let $\Sigma:= A_+\in\Set_f$, and 
extend the map 
$(\id,\emptyset)\colon (\Sigma,\emptyset)\to (A_+,A_\omega)$ 
to an 
$\omega$-semigroup morphism $e\colon (\Sigma^+, \Sigma^\omega)\to 
(A_+,A_\omega)$, using that $(\Sigma^+,\Sigma^\omega)$ is the free 
$\omega$-semigroup on $(\Sigma,\emptyset)$. Clearly the component $e\colon 
\Sigma^+\to A_+$ is surjective 
because $e(a)=a$ for all $a\in A_+$. To show that also the component $e\colon 
\Sigma^\omega 
\to A_\omega$ is surjective, let $a\in A_\omega$ and choose elements $a_i\in 
A_+$ with $a=\pi(a_0,a_1,\ldots)$, using the completeness of $A$.
It follows that
\[ a = \pi(a_0,a_1,\ldots) = \pi(e(a_0), e(a_1),\ldots) = 
e(\pi(a_0,a_1,\ldots)). \]
Thus $e$ is surjective, which proves that $A$ is $\A$-generated. 
\end{example}

\begin{definition}\label{def:pseudovar}
 A \emph{pseudovariety of $\F$-algebras} is a full subcategory $\V\seq \F$ such that all algebras in $\V$ are $\A$-generated, and $\V$ is closed under quotients and $\A$-generated subalgebras of finite 
products. That is,
\begin{enumerate}[(i)]
\item for every $A\in\V$ and every quotient $e\colon A\epito B$ in $\Alg{\MT}$, one has $B\in \V$;
\item for every finite family $A_i$ $(i\in I)$ of algebras in $\V$ and every $\A$-generated subalgebra $A\monoto\prod_{i\in I} A_i$, one has $A\in \V$.
\end{enumerate}
\end{definition}

\begin{rem}
In most applications, $\A$ is chosen such that finite products of $\A$-generated finite
$\MT$-algebras are $\A$-generated. In this case the definition of a 
pseudovariety 
simplifies: it is a class $\V\seq \F$ of $\A$-generated algebras closed under 
quotients, 
$\A$-generated subalgebras, and finite products. This holds, e.g., in the following two examples:
\end{rem}

\begin{example}\label{ex:pseudovar_monads}
Let $\A=\Set_f^S$ and $\F=\FAlg{\MT}$. Then every finite $\MT$-algebra is $\A$-generated by Example \ref{ex:agenerated_talgs}. Thus, instantiating Definition \ref{def:pseudovar}, a pseudovariety is
a 
class of finite $\MT$-algebras closed under quotients, subalgebras, and finite 
products. This concept was studied in \cite{camu16}. For 
the 
monad $\MT_*$ on $\Set$ we recover the original concept of Eilenberg: a 
class 
of finite monoids closed under quotients, submonoids, and finite products.
\end{example}

\begin{example}[$\omega$-semigroups]\label{ex:pseudovar_omegssem}
Let $\MT=\MT_\infty$ on $\Set^2$, and choose $\A=\{\,(\Sigma,\emptyset): \Sigma\in\Set_f\,\}$ and $\F=$ all finite $\omega$-semigroups. Then $\A$-generated $\omega$-semigroups are precisely complete $\omega$-semigroups by Example \ref{ex:agenerated_omegasem}, and clearly products of
complete $\omega$-semigroups are complete. Thus, instantiating Definition \ref{def:pseudovar}, a 
\emph{pseudovariety of $\omega$-semigroups} is a class of 
finite 
complete $\omega$-semigroups closed under quotients, complete 
$\omega$-subsemigroups, and
finite products. This concept 
is due to Wilke \cite{wilke91}; see also \cite{pinperrin04}.
\end{example}

\begin{definition}\label{def:profinitetheory}
A \emph{pro-$\F$ theory} is a family \[\rho = (\,\rho_\Sigma\colon \hatT\FSigma 
\epito P_\Sigma\,)_{\Sigma\in \A}\] such that (i) $\rho_\Sigma$ is a $\Sigma$-generated profinite
$\hatT$-algebra for every $\Sigma\in\A$ and (ii) for 
every 
$\MT$-homomorphism $g\colon \MT \FDelta\to \MT\FSigma$ with 
$\Sigma,\Delta\in\A$ there exists a  
$\hatT$-homomorphism $g_P\colon P_\Delta\to P_\Sigma$ with $\rho_\Sigma \o \hat g = 
g_P\o \rho_\Delta$.
\[ 
\xymatrix{
\hatT\FDelta \ar[r]^{\hat g} \ar@{->>}[d]_{\rho_\Delta} & \hatT\FSigma \ar@{->>}[d]^{\rho_\Sigma}\\
P_\Delta \ar@{-->}[r]_{g_P} & P_\Sigma
}
\]
Here $\hat g$ is the continuous extension of $g$, see Lemma \ref{lem:gat}. Pro-$\F$ theories are ordered by $\rho\leq \rho'$ iff 
$\rho_\Sigma\leq \rho_\Sigma'$ for each $\Sigma\in \A$.
\end{definition}
Using the local Reiterman theorem, pro-$\F$ theories can be equivalently described in terms of local pseudovarieties:

\begin{definition}\label{def:ftheory}
An \emph{$\F$-theory} is a family \[\T=(\,\T_\Sigma\,)_{\Sigma\in \A}\] such that (i) $\T_\Sigma$ is a local pseudovariety of $\Sigma$-generated $\F$-algebras for every $\Sigma\in \A$, and (ii) for every $\MT$-homomorphism $g\colon \MT\FDelta\to \MT\FSigma$ with $\Sigma,\Delta\in \A$ and every $e\colon \MT\FSigma\epito A$ in $\T_\Sigma$, there exists an $\ol e\colon \MT\FDelta\epito B$ in $\T_\Delta$ and a $\MT$-homomorphism $\ol g\colon B\to A$ with $\ol g \o \ol e = e\o g$.
\[ 
\xymatrix{
\MT\FDelta \ar[r]^{g} \ar@{->>}[d]_{\ol e} & \MT\FSigma \ar@{->>}[d]^{e}\\
B \ar@{-->}[r]_{\ol g} & A
}
\]
$\F$-theories are ordered by $\T\leq \T'$ iff $\T_\Sigma\seq \T_\Sigma'$ for all $\Sigma \in \A$. 
\end{definition}

\begin{rem}
Since the local pseudovariety $\T_\Delta$ is downwards closed, one can always choose $\ol e$ and $\ol g$ in the above definition as the factorization of $e\o g$ into a surjective and an injective/order-reflecting $\MT$-homomorphism.
\end{rem}

\begin{construction}
\begin{enumerate}
\item For any pro-$\F$ theory $\rho$, form the family \[\T^\rho = (\,\T_\Sigma^\rho\,)_{\Sigma\in\A}\] where $\T_\Sigma^\rho$ is the local pseudovariety of $\Sigma$-generated $\F$-algebras associated to $\rho_\Sigma$ (see Construction \ref{rem:localvarconst}.2).
\item For any $\F$-theory $\T$, form the family \[\rho^\T = (\,\rho_\Sigma^\T\colon \hatT\FSigma\epito P_\Sigma^\T\,)\] where $\rho_\Sigma^\T$ is the $\Sigma$-generated profinite $\hatT$-algebra associated to $\T_\Sigma$ (see Construction \ref{rem:localvarconst}.1). The limit projection associated to $e\colon \MT\FSigma\epito A$ in $\T_\Sigma$ is denoted by
\[ e_\T^*\colon P_\Sigma^\T\to A. \]
\end{enumerate}
\end{construction}

\begin{lemma}\label{lem:profiniteconcrete}
\begin{enumerate}
\item For any pro-$\F$ theory $\rho$, the family $\T^\rho$ forms an $\F$-theory.
\item For any $\F$-theory $\T$, the family $\rho^\T$ forms a pro-$\F$ theory.
\item The maps $\rho\mapsto \T^\rho$ and $\T\mapsto \rho^\T$ give an isomorphism between the poset of pro-$\F$ theories and the poset of $\F$-theories.
\end{enumerate}
\end{lemma}

\begin{proof}
\begin{enumerate}
\item Suppose that $\rho$ is a pro-$\F$ theory, and let $g\colon \MT\Delta\to\MT\FSigma$ be a $\MT$-homomorphism with $\Sigma,\Delta\in \A$ and $e\colon \MT\FSigma\epito A$ in $\T_\Sigma^\rho$. By the definition of $\T_\Sigma^\rho$, the quotient $e$ has the form
\[ e = (T\FSigma \xra{\iota_\FSigma} V\hatt\FSigma \xra{V\rho_\Sigma} VP_\Sigma \xra{Vq} A) \]
for some $\hatT$-homomorphism $q$. Factorize the $\hatT$-homomorphism $q\o g_P$ into a surjective homomorphism $\ol q$ followed by an injective/order-reflecting homomorphism $\ol g$. Then 
\[ \ol e := (T\FDelta \xra{\iota_\FDelta} V\hatt\FDelta \xra{V\rho_\Delta} VP_\Delta \xra{V\ol q} A)\]
lies in $\T_\Delta^\rho$ and $e\o g = \ol g \o \ol e$, as shown by the commutative diagram below:
\[
\xymatrix{
T\FDelta \ar@/_7ex/[ddd]_{\ol e} \ar[r]^g \ar[d]_{\iota_\FDelta} & T\FSigma \ar[d]^{\iota_\FSigma} \ar@/^7ex/[ddd]^e \\
V\hatt\FDelta \ar[r]_{V\hat g}\ar@{->>}[d]_{V\rho_\Delta} & V\hatt\FSigma \ar@{->>}[d]^{V\rho_\Sigma}\\
VP_\Delta \ar@{->>}[d]_{V\ol q} \ar[r]_{Vg_P} & VP_\Sigma \ar@{->>}[d]^{Vq} \\
B \ar[r]_{\ol g} & A
}
\]
This proves that $\T^\rho$ is an $\F$-theory.
\item Suppose that $\T$ is an $\F$-theory, and let $g\colon \MT\FDelta\to\MT\FSigma$ be a $\MT$-homomorphism with $\Sigma,\Delta\in \A$. We need to define $g_P\colon P^\T_\Delta\to P^\T_\Sigma$ with $\rho^\T_\Sigma\o \hat g = g_P\o \rho^\T_\Delta$. For any $e\colon \MT\FSigma\epito A$ in $\T_\Sigma$, choose $\ol e\colon \MT \epito A_e$ in $\T_\Delta$ and $g_e\colon A_e\to A$ with $e\o g = g_e\o \ol e$, using that $\T$ is an $\F$-theory. One readily verifies that the morphisms $g_e\o \ol e^*_\T$ (where $e$ ranges over $\T_\Sigma$) form a cone over the diagram defining the object $P_\Sigma^\T$. Therefore there exists a unique $g_P\colon P_\Delta^\T \to P_\Sigma^\T$ with $e^*_\T \o g_P = g_e\o \ol e^*_\T$ for all $e\in \T_\Sigma$. It follows that $\rho_\Sigma^\T \o \hat g = g_P \o \rho_\Delta$, since this holds when precomposed with the dense map $\iota_\Delta$ and postcomposed with the limit projections $e_\T^*$.
\[
\xymatrix{
T\FDelta \ar@/_10ex/[ddd]_{\ol e} \ar[r]^g \ar[d]_{\iota_\FDelta} & T\FSigma \ar[d]^{\iota_\FSigma} \ar@/^10ex/[ddd]^{e} \\
V\hatt\FDelta \ar[r]_{V\hat g}\ar@{->>}[d]_{V\rho^\T_\Delta} & V\hatt\FSigma \ar@{->>}[d]^{V\rho^\T_\Sigma}\\
VP^\T_\Delta \ar@{->>}[d]_{V\ol e_\T^*} \ar[r]_{Vg_P} & VP^\T_\Sigma \ar@{->>}[d]^{Ve_\T^*} \\
A_e \ar[r]_{g_e} & A
}
\]
Thus $\rho^\T$ is a pro-$\F$ theory.
\item By the local Reiterman theorem (Theorem \ref{thm:localreiterman}), the two maps $\rho\mapsto \T^\rho$ and $\T \mapsto \rho^\T$ give mutually inverse and order-preserving bijections between the two posets.\qedhere
\end{enumerate}
\end{proof}

\begin{rem} $\F$-theories generalize the \emph{varieties of filters of congruences} introduced by Almeida \cite{almeida90} for algebras over a finitary signature, and earlier by Th\'erien \cite{therien80} for monoids.
\end{rem}
Our goal in this subsection is to show that $\F$-theories (and thus also pro-$\F$ theories) correspond to pseudovarietes of $\F$-algebras. First an auxiliary result:

\begin{lemma}\label{lem:projective}
For every $\Sigma\in\Set_f^S$, the free $\MT$-algebra $\MT\FSigma$ is projective. That is, for any two $\MT$-homomorphisms $h\colon \MT\FSigma\to B$ and $e\colon A\epito B$ with $e$ surjective, there exists a $\MT$-homomorphism $g\colon \MT\FSigma\epito A$ with $h=e\o g$.
\end{lemma}

\begin{proof}
Consider the morphism $h' := h\o \eta\colon \FSigma\to B$ in $\D^S$. Since the free $\D^S$-object $\FSigma$ is projective, one can choose a morphism $g'\colon \FSigma\to A$ with $h'=e\o g'$. Let $g\colon \MT\FSigma\epito A$ be the unique extension of $g'$ to a $\MT$-homomorphism (i.e. with $g'=g\o \eta$), using that $\MT\FSigma$ is the free $\MT$-algebra on $\FSigma$. Then $h=g\o e$, as this holds when precomposed with the unit $\eta$.
\[
\xymatrix{
& \FSigma \ar[d]^\eta \ar@/_2ex/[ddl]_{g'} \ar@/^2ex/[ddr]^{h'} & \\
&\MT\FSigma \ar[dr]^h \ar[dl]_g &\\
A \ar@{->>}[rr]_e & & B
}
\]
\end{proof}

\begin{lemma}\label{lem:theory_to_pseudovar_property}
Let $\T$ be an $\F$-theory, and let $(A,\alpha)$ be an 
$\A$-generated algebra in $\F$. Then the following statements are equivalent:
\begin{enumerate}[(i)]
\item There exists a surjective $\MT$-homomorphism $e\colon \MT\FSigma\epito (A,\alpha)$ with $\Sigma\in \A$ and $e\in \T_\Sigma$.
\item Every $\MT$-homomorphism $h\colon \MT\FDelta\to (A,\alpha)$ with $\Delta\in \A$ factors through some $e\in \T_\Delta$.
\end{enumerate}
\end{lemma}
\begin{proof}
(i)$\Ra$(ii) Given a $\MT$-homomorphism $h\colon \MT\FDelta\to A$ with $\Delta\in \A$, choose a 
$\MT$-homomorphism $g\colon \MT\FDelta\to \MT\FSigma$ with $h = e\o g$, see Lemma \ref{lem:projective}. Since $\T$ is an $\F$-theory, we have $e\o g = \ol g \o \ol e$ for some $\ol e\in \T_\Delta$ and some $\MT$-homomorphism $\ol g$. Therefore $h$ factors through $\ol e\in \T_\Delta$, which proves (ii).
\[
\xymatrix{
\MT\FDelta \ar[r]^g \ar[dr]^h \ar@{->>}[d]_{\ol e} & \MT\FSigma \ar@{->>}[d]^e \\
B \ar[r]_{\ol g} & A
}
\]
(ii)$\Ra$(i)  Since $A$ is $\A$-generated, there exists a surjective 
$\MT$-homomorphism $e\colon \MT\FSigma\epito A$ for some $\Sigma\in\A$. 
Thus, by (ii), $e$ factors through some quotient $q\in \T_\Sigma$, i.e. $e\leq q$. Since the local pseudovariety $\T_\Sigma$ is downwards closed, it follows that $e\in \T_\Sigma$. This proves (i).
\end{proof}

\begin{notation}
For any $\F$-theory $\T$, we denote by $\V^\T$ the class of all $\A$-generated algebras $(A,\alpha)\in \F$ satisfying the equivalent properties (i)/(ii) of Lemma \ref{lem:theory_to_pseudovar_property}.
\end{notation}

\begin{lemma}\label{lem:theory_to_pseudovar_welldef}
If $\T$ is an $\F$-theory, then $\V^\T$ is a pseudovariety of $\F$-algebras.
\end{lemma}

\begin{proof}
The closure of $\V^\T$ under quotients follows from
Lemma~\ref{lem:theory_to_pseudovar_property}(i) and the fact that $\T_\Sigma$ is downwards closed. To show that $\V^\T$ is closed under $\A$-generated subalgebras of finite products, let $A_i$ ($i\in I$) be a finite family of objects in $\V^\T$ and let $m\colon 
A\monoto
\prod_{i} A_i$ be an $\A$-generated subalgebra of their product. Denote the product projections by
$\pi_i\colon \prod_{i} A_i\to A_i$. We show that every 
$\MT$-homomorphism 
$h\colon\MT\FDelta\to A$  with $\Delta\in \A$
factors through some element of $\T_\Delta$.  For each $i$, the $\MT$-homomorphism $\pi_i\o m\o h$ has codomain $A_i\in \V^\T$ and thus, by definition of $\V^\T$, factors through some element of $\T_\Delta$; that is, $\pi_i \o m\o h = h_i\o e_i$ for some $e_i\in \T_\Delta$ and some $\MT$-homomorphism $h_i$. Since $\T_\Delta$ is a local pseudovariety, there exists $e\in \T_\Delta$ with $e_i\leq e$ for all $i$, i.e. $e_i= g_i\o e$ for some $g_i$.
It follows that $h$ factors through $e$ by the diagonal fill-in
property:
\[  
\xymatrix@R+1em{
\MT\FDelta \ar@{->>}[r]^{e} \ar[d]_h & B  \ar[d]|{\langle 
h_i\o g_i\rangle}\ar@{-->}[ld] 
\ar[dr]^{h_i\o g_i} &&\\
A \ar@{ >->}[r]_m & \prod_i A_i \ar[r]_{\pi_i} & A_i
}
\]
By the definition of $\V^\T$, this shows $A\in \V^\T$.
\end{proof}

\begin{notation}
For any pseudovariety $\V$ of $\F$-algebras and $\Sigma\in \A$ we form the family
\[  \T^\V = (\,\T_\Sigma^\V\,)_{\Sigma\in \A}\]
where, for each $\Sigma\in \A$,
\[  \T_\Sigma^\V = \{\,e\colon \MT\FSigma\epito A \,:\, A\in \V\,\}.\]
\end{notation}

\begin{lemma}
For any pseudovariety $\V$ of $\F$-algebras, $\T^\V$ is an $\F$-theory.
\end{lemma}

\begin{proof}
Clearly $\T_\Sigma^\V$ is a local pseudovariety of $\Sigma$-generated $\F$-algebras for every $\Sigma\in \A$: it is downwards closed because $\V$ closed under quotients, and it is closed under subdirect products because $\V$ is closed under $\A$-generated subalgebras of finite products.

To show that $\T^\V$ forms an $\F$-theory, let $g\colon \MT\FDelta\to\MT\FSigma$ be a $\MT$-homomorphism with $\Sigma,\Delta\in \A$ and $e\colon \MT\FSigma\epito A$ in $\T_\Sigma^\V$, Factorize $e\o g$ into a surjective $\MT$-homomorphism $\ol e\colon \MT\FDelta\epito B$ followed by an injective/order-reflecting $\MT$-homomorphism $\ol g\colon B\monoto A$:
\[
\xymatrix{
\MT\FDelta \ar[r]^g \ar@{->>}[d]_{\ol e} & \MT\FSigma \ar@{->>}[d]^e\\
B \ar@{>->}[r]_{\ol g} & A
}
\] 
Then $B$ is an $\A$-generated subalgebra of $A$. Since $A\in \V$ and $\V$ is closed under $\A$-generated subalgebras, it follows that $B\in\V$ and therefore $\ol e\in \T_\Delta^\V$. Thus $\T^\V$ is an $\F$-theory.
\end{proof}

\begin{lemma}\label{lem:theory_pseudovar_equivalence1}
For any pseudovariety $\V$ of $\F$-algebras we have $\V = \V^\T$ where $\T := \T^\V$.
\end{lemma}

\begin{proof}
To show $\V\seq\V^\T$, let $A\in \V$. Since $A$ is $\A$-generated, there exists a surjective $\MT$-homomorphism $e\colon T\FSigma\epito A$ with $\Sigma\in \A$. Then $e\in \T_\Sigma$ by the definition of $\T_\Sigma = \T_\Sigma^\V$ and therefore $A\in \V^\T$ by the definition of $\V^\T$.

For the reverse inclusion $\V^\T\seq \V$, let $A\in \V^\T$. Then, for some $\Sigma\in\A$, there exists $e\in \T_\Sigma$ with codomain $A$. But by the definition of $\T_\Sigma = \T_\Sigma^\V$, this means that $A\in \V$. 
\end{proof}

\begin{lemma}\label{lem:theory_pseudovar_equivalence2}
For any $\F$-theory $\T$
we 
have $\T = \T^\V$, where $\V:=\V^\T$.
\end{lemma}

\begin{proof}

 To show $\T\seq\T^\V$, let $e\colon \MT\FSigma\epito A$ in $\T_\Sigma$. Then $A\in \V^\T=\V$ by the definition of $\V^\T$. But then $e\in \T^\V$ be the definition of $\T^\V$.

For the reserve inclusion $\T^\V\seq \T$, let $e\colon \MT\FSigma\epito A$ in $\T_\Sigma^\V$. By the definition of $\T_\Sigma^\V$, this means that $A\in \V=\V^\T$, i.e. there exists $\Delta\in \A$ and $q\colon \MT\FDelta\to A$ in $\T_\Delta$ with codomain $A$. Using projectivity of free $\MT$-algebras (see Lemma \ref{lem:projective}), choose a $\MT$-homomorphism $g\colon \MT\FSigma\to \MT\FDelta$ with $q\o g = e$. Since $\T$ is an $\F$-theory, there exists $\ol q\colon \MT\FSigma\epito B$ in $\T_\Sigma$ and $\ol g\colon B\to A$ with $q\o g = \ol g \o \ol q$.
\[
\xymatrix{
\MT\FSigma \ar[r]^g \ar@{->>}[d]_{\ol q} \ar@{->>}[dr]^e & \MT\FDelta \ar@{->>}[d]^{q}\\
B \ar@{->>}[r]_{\ol g}& A
}
\]
Therefore $e=\ol g \o \ol q$, i.e. $e\leq \ol q$. Since $\T_\Sigma$ is downwards closed and $\ol q\in \T_\Sigma$, it follows that $e\in \T_\Sigma$.
\end{proof}

\begin{lemma}\label{lem:ftheory_vs_pseudovar}
The maps 
\[\T\mapsto \V^\T\quad\text{and}\quad \V\mapsto \T^\V\]
give an isomorphism between the lattice of $\F$-theories and the lattice of pseudovarieties of $\F$-algebras (ordered by inclusion).
\end{lemma}

\begin{proof}
By Lemma \ref{lem:theory_pseudovar_equivalence1} and \ref{lem:theory_pseudovar_equivalence2}, the two maps are mutually inverse. Moreover, they are clearly order-preserving, which proves the claim.
\end{proof}
From this and Lemma \ref{lem:profiniteconcrete} we get the main result of the present section:

\begin{theorem}[Reiterman Theorem]\label{thm:reiterman}
The lattice of
pseudovarieties of $\F$-algebras is isomorphic to the 
lattice of pro-$\F$
theories.
\end{theorem}
As in the local case, we can interpret the above theorem in terms of profinite (in-)equations.
\begin{definition} Let $\Sigma\in\Set_f^S$ and $u,v\in\under{\hatt\FSigma}_s$ for some sort $s$.
\begin{enumerate}
\item Let $\D$ be a variety of algebras. A $\MT$-algebra $(A,\alpha)$ in $\F$
  \emph{satisfies the profinite equation $u=v$} if $h^+(u)= h^+(v)$ for every $\MT$-homomorphism $h\colon \MT\FSigma\to A$.
\item Let $\D$ be a variety of ordered algebras. A  $\MT$-algebra $(A,\alpha)$ in $\F$
  \emph{satisfies the profinite inequation $u\leq v$} if $h^+(u)\leq h^+(v)$ for every $\MT$-homomorphism $h\colon \MT\FSigma\to A$.
\end{enumerate}
\end{definition}

\begin{notation}
For a class $E$ of profinite (in-)equations over sets $\Sigma\in\A$, let $\V[E]$ denote the 
class of all $\A$-generated algebras in $\F$ satisfying all (in-)equations in 
$E$. Conversely, for a subclass $\V\seq\F$
let $E[\V]$ be the class of all profinite (in-)equations over $\A$ satisfied by all 
algebras in $\V$.
\end{notation}

\begin{corollary}[Equational Reiterman Theorem]
A class $\V\seq \F$ of $\A$-generated algebras forms a pseudovariety iff $\V=\V[E]$ for some class $E$ of profinite (in-)equations over $\A$.
\end{corollary}

\begin{proof}
We consider the case where $\D$ is variety of ordered algebras; for the unordered case, replace all inequations by equations.

For the ``if'' direction, we may assume that $E$ consists of a single profinite inequation $u\leq v$ over $\Sigma\in\A$ (since the class of all pseudovarieties forms a complete lattice under intersection). Then
\begin{enumerate} 
\item \emph{$\V[u\leq v]$ is closed under quotients:}  Let $A\in \V[u\leq v]$ and $q\colon A\epito B$ be a surjective $\MT$-homomorphism. Given a $\MT$-homomorphism $h\colon \MT \FSigma\to B$, choose a $\MT$-homomorphism $g\colon \MT\FSigma\to A$ with $q\o g = h$, using projectivity of the free $\MT$-algebra $\MT\FSigma$. Then $g^+(u)\leq g^+(v)$ because $A$ satisfies $u\leq v$, and thus $h^+(u) = q\o g^+(u)\leq q\o g^+(v)=h^+(v)$. It follows that $B$ satisfies $u\leq v$ and thus $B\in \V[u\leq v]$
 \item \emph{$\V[u\leq v]$ is closed under $\A$-generated subalgebras of finite products.} Let $A_i$ ($i\in I$) be a finite family of algebras in $\V[u\leq v]$, and suppose that $m\colon B\monoto \prod A_i$ is an $\A$-generated subalgebra of the product $\prod A_i$. For any $\MT$-homomorphism $h\colon \MT\FSigma\to B$ and $i\in I$, we have the composite $h_i := \pi_i\o m\o h$ where $\pi_i\colon \prod A_i\to A_i$ is the product projection. Note that $h_i^+ = \pi_i \o m\o h^+$ because this holds when precomposed with the dense map $\iota_\FSigma$. Since $A_i$ satisfies $u\leq v$, we have $h_i^+(u)\leq h_i^+(v)$ and thus $\pi_i \o m \o h^+(u)\leq \pi_i\o m\o h^+(u)$ for all $i$. Since the projections $\pi_i$ are jointly order-reflecting and $m$ is order-reflecting, it follows that $h^+(u)\leq h^+(v)$. Therefore $B$ satisfies $u\leq v$, i.e. $B\in \V[u\leq v]$
\end{enumerate}
\smallskip
 For the ``only if'' direction, 
let $\V$ be a pseudovariety of $\F$-algebras, let $\T :=\T^\V$ be the $\F$-theory corresponding to $\V$, and let $\rho_\T = 
(\rho^\T_\Sigma\colon \hatT\FSigma \epito P_\Sigma^\T)_{\Sigma\in\A}$ be the pro-$\F$ theory corresponding to $\T$. Let us first observe that every profinite inequation $u\leq v$ over $\Sigma\in \A$ with $\rho^\T_\Sigma(u)\leq \rho^\T_\Sigma(v)$ is satisfied by every algebra $A\in \V$. Indeed, let $h\colon \MT\FSigma\to A$ be a $\MT$-homomorphism, and factorize $h=m\o e$ with $e$ surjective and $m$ order-reflecting. Then $e\in \T_\Sigma$ because $\V$ is closed under $\A$-generated subalgebras. Thus $\rho^\T_\Sigma(u)\leq \rho^\T_\Sigma(v)$ implies $e^+(u) = e^*_\T\o \rho^\T_\Sigma(u)\leq e^*_\T\o \rho^\T_\Sigma(v) = e^+(v)$ and thus $h^+(u) = m\o e^+(u)\leq m\o e^+(v) = h^+(v)$, showing that $A$ satisfies the profinite inequation $u\leq v$.

We claim that 
$\V=\V[E[\V]]$. The inclusion $\seq$ is trivial. To prove $\supseteq$, let 
$A\in\V[E[\V]]$, i.e. $A$ satisfies every profinite inequation over $\A$ that the algebras in $\V$ satisfy. Since $A$ is $\A$-generated, there is a surjective 
$\MT$-homomorphism $e\colon T\FSigma\epito A$ with $\Sigma\in \A$. By the above argument, any profinite inequation $u\leq v$ over $\Sigma$ with $\rho^\T_\Sigma(u)\leq \rho^\T_\Sigma(v)$ is satisfied by every algebra in $\V$ and thus also by $A$. Hence $e^+(u)\leq e^+(v)$. This implies by the homomorphism theorem that $e^+$ factors through $\rho^\T_\Sigma$. In particular, $A$ is a quotient of $P_\Sigma^\T$. This implies that $A$ lies in $\V^{\T}$ by the definition of $\V^\T$. Since $\V^\T=\V$ by Theorem \ref{thm:reiterman}, we conclude that $A\in \V$.
\end{proof}
The special case $\A=\Set_f^S$ of the previous corollary has been proved in \cite{camu16}.

\section{Unary Presentations}\label{sec:unpres}
In this section we develop the notion of a \emph{unary presentation} of a $\MT$-algebra that later,  in Section \ref{sec:varietythm}, will serve as the key tool for defining the derivatives of a language. For motivation, consider a monoid $M$ and recall that an equivalence relation $\equiv$ on $M$ is a \emph{monoid congruence} if it is stable under multiplication, i.e. $x\equiv x'$ and $y\equiv y'$ implies $xy\equiv x'y'$. To show that $\equiv$ is a monoid congruence, it actually suffices to show that $\equiv$ is stable under left and right multiplication, i.e. $x\equiv x'$ implies $xy\equiv x'y$ and $yx\equiv yx'$ for all $y\in M$. This observation generalizes to algebras $A$ over a finitary $S$-sorted signature $\Gamma$: an $S$-sorted equivalence $\equiv$ on $A$ is a $\Gamma$-congruence iff it is stable under all \emph{elementary translations}, see \ref{app:elementarytranslations}. Analogously for ordered algebras $A$ and $S$-sorted preorders $\preceq$ on $A$.  

The concept of a unary presentation states this phenomenon for algebras for the monad $\MT$ on $\D^S$ (in lieu of just $\Gamma$-algebras) in categorical terms. First, by identifying congruences/stable preorders with their corresponding quotients (see \ref{app:congruences} and \ref{app:stablepreorders}),  we get the following concept:

\begin{definition}\label{def:tquotient}
Let $(A,\alpha)$ be a $\MT$-algebra. A quotient $e\colon A\epito B$ in $\D^S$ is called a \emph{$\MT$-quotient} of $(A,\alpha)$ if it carries a $\MT$-homomorphism, i.e. if there exists a $\MT$-algebra structure $(B,\beta)$  on $B$ such that $e\colon (A,\alpha)\epito(B,\beta)$ is a $\MT$-homomorphism.
\[
\xymatrix{
TA \ar[r]^\alpha \ar@{->>}[d]_{Te} & A \ar@{->>}[d]^e\\
TB \ar@{-->}[r]_\beta & B
}
\]
If moreover $(B,\beta)$ lies in $\F$, we say that $e$ is an \emph{$\F$-quotient} of $(A,\alpha)$.
\end{definition}

\begin{rem}
Clearly the $\MT$-algebra structure $\beta$ on $B$ is uniquely determined because $Te$ is an epimorphism. Moreover, to show that $e$ is a $\MT$-quotient it suffices to find a morphism $\beta$ making the above square commute; then $(B,\beta)$ is always a $\MT$-algebra, see \ref{app:quotients_of_talgs}.
\end{rem}

\begin{rem}
In view of our applications, it is helpful to rephrase the categorical definitions of this section in terms of congruences/stable preorders instead of quotients. Definition \ref{def:tquotient} has the following equivalent formulation:
\begin{enumerate}
\item Suppose that $\D$ is a variety of algebras and view $\D^S$ as a variety of $S$-sorted algebras, see \ref{app:varieties}. Let $A\in \D^S$. Given a $\MT$-algebra $(A,\alpha)$, we call a congruence $\equiv$  on $A$ a \emph{$\MT$-congruence} if there exists a surjective $\MT$-homomorphism $e\colon (A,\alpha)\epito (B,\beta)$ with kernel $\equiv$. If moreover $(B,\beta)\in \F$, we call $\equiv$ an \emph{$\F$-congruence}. In the case $\F =$ all finite $\MT$-algebras, this means precisely that $\equiv$ is a $\MT$-congruence of finite index, i.e. with only finitely many congruence classes. By \ref{app:congruences}, we have the bijective correspondences
\begin{align*}
\text{congruences on $A$ in $\D^S$} &\quad\leftrightarrow\quad \text{quotients of $A$ in $\D^S$}\\
\text{$\MT$-congruences on $(A,\alpha)$} &\quad\leftrightarrow\quad \text{$\MT$-quotients of $(A,\alpha)$}\\
\text{$\F$-congruences on $(A,\alpha)$} &\quad\leftrightarrow\quad \text{$\F$-quotients of $(A,\alpha)$}
\end{align*}
\item   Suppose that $\D$ is a variety of ordered algebras, and view $\D^S$ as a variety of $S$-sorted ordered algebras. Let $A\in \D^S$. Given a $\MT$-algebra $(A,\alpha)$, we call a stable preorder $\preceq$  on $A$ an \emph{$\MT$-stable preorder} if there exists a surjective $\MT$-homomorphism $e\colon (A,\alpha)\epito (B,\beta)$  with ordered kernel $\preceq$. If moreover $(B,\beta)\in \F$, we call $\preceq$ an \emph{$\F$-stable preorder}. In the case $\F =$ all finite $\MT$-algebras, this means precisely that $\preceq$ is $\MT$-stable preorder of finite index, i.e. the induced congruence $\equiv \,=\, \preceq \cap \succeq$ has finite index. By \ref{app:stablepreorders}, we have the bijective correspondences
\begin{align*}
\text{stable preorders on $A$ in $\D^S$} &\quad\leftrightarrow\quad \text{quotients of $A$ in $\D^S$}\\
\text{$\MT$-stable preorders on $(A,\alpha)$} &\quad\leftrightarrow\quad \text{$\MT$-quotients of $(A,\alpha)$}\\
\text{$\F$-stable preorders on $(A,\alpha)$} &\quad\leftrightarrow\quad \text{$\F$-quotients of $(A,\alpha)$}
\end{align*}
\end{enumerate}
\end{rem}
The characterization of congruences/stable preorders of $\Gamma$-algebras in terms of translations has the following categorical formulation:

\begin{definition}\label{def:uquotient} Let $A$ be an object of $\D^S$.
\begin{enumerate}
\item
By a \emph{unary operation} on $A$ is meant a morphism $u\colon A_s\to A_{t}$ in $\D$, where $s$ and $t$ are arbitrary sorts.
\item  Given a set $\U$ of unary operations on $A$, a quotient $e\colon A\epito B$ in $\D^S$ is called a \emph{$\U$-quotient} of $A$ if every unary operation $u\colon A_s\to A_{t}$ in 
$\U$  admits a 
lifting along $e$, i.e. a morphism $u_{B}\colon B_s\to B_{t}$ making the following square commutative:
\[
\xymatrix{
A_s \ar[r]^u \ar@{->>}[d]_e & A_{t} \ar@{->>}[d]^e \\
B_s \ar@{-->}[r]_{u_B} & B_{t}
}
\]
\end{enumerate}
\end{definition}

\begin{rem}
Again, let us state this concept in terms of congruences/stable preorders.
\begin{enumerate}
\item Let $\D$ be a variety of algebras. Given $A\in \D^S$ and a set $\U$ of unary operations on $A$, we call a congruence $\equiv$ on $A$ a \emph{$\U$-congruence} if it stable under all operations in $\U$, that is, $a\equiv_s a'$ implies $u(a)\equiv_{t} u(a')$ for all $u\colon A_s\to A_{t}$ in $\U$ and $a,a'\in A_s$. By the homomorphism theorem (see \ref{app:homtheorem}), we get the bijective correspondence
\begin{align*}
\text{$\U$-congruences on $A$} &\quad\leftrightarrow\quad \text{$\U$-quotients of $A$}.
\end{align*}
\item Let $\D$ be a variety of ordered algebras. Given $A\in \D^S$ and a set $\U$ of unary operations on $A$, we call a stable preorder $\preceq$ on $A$ a \emph{$\U$-stable preorder} if it stable under all operations in $\U$, that is, $a\preceq a'$ implies $u(a)\preceq_t u(a')$ for all $u\colon A_s\to A_{t}$ in $\U$ and $a,a'\in A_s$. By the homomorphism theorem for ordered algebras (see \ref{app:homtheorem}), we get the bijective correspondence
\begin{align*}
\text{$\U$-stable preorders on $A$} &\quad\leftrightarrow\quad \text{$\U$-quotients of $A$}.
\end{align*}
\end{enumerate}
\end{rem}
For $\Gamma$-algebras over an infinitary signature $\Gamma$, the characterization of congruences by elementary translations generally fails, but remains valid for equivalence relations $\equiv$ that are refinable to a $\Gamma$-congruence of finite index, see Examples \ref{ex:unpres_monads} and \ref{ex:refinable}. The concept of refinement also has a categorical formulation:

\begin{definition}
Given a $\MT$-algebra $(A,\alpha)$, a quotient $e\colon A\epito B$ in $\D^S$ is called \emph{$\F$-refinable} if it factorizes through some $\F$-quotient of $(A,\alpha)$.
\end{definition}

\begin{rem}
Again, in terms of congruences/stable preorders we get the following concepts:
\begin{enumerate}
\item Suppose that $\D$ is a variety of algebras. Given a $\MT$-algebra $(A,\alpha)$, a congruence on $A$ in $\D^S$ is called \emph{$\F$-refinable} if it contains an $\F$-congruence. By the homomorphism theorem for $\MT$-algebras (see Remark \ref{rem:tpressurj}), we have the bijective correspondence
\begin{align*}
\text{$\F$-refinable congruences on $A$} &\quad\leftrightarrow\quad \text{$\F$-refinable quotients of $A$}.
\end{align*}
\item   Suppose that $\D$ is a variety of ordered algebras. Given a $\MT$-algebra $(A,\alpha)$, a stable preorder on $A$ in $\D^S$ is called \emph{$\F$-refinable} if it contains an $\F$-stable preorder. By the homomorphism theorem for $\MT$-algebras (see Remark \ref{rem:tpressurj}), we have the bijective correspondence
\begin{align*}
\text{$\F$-refinable stable preorders on $A$} &\quad\leftrightarrow\quad \text{$\F$-refinable quotients of $A$}.
\end{align*}
\end{enumerate}
\end{rem}
This leads to following central definition of our paper, the notion of a \emph{unary presentation} of a $\MT$-algebra.
\begin{definition}\label{def:unpres}
By a \emph{unary presentation} of a $\MT$-algebra $(A,\alpha)$ is meant a set
$\U$ of unary operations on $A$ such that for any $\F$-refinable quotient $e\colon A\epito B$ in $\D^S$,
\[ \text{$e$ is a $\MT$-quotient} \quad\text{iff}\quad \text{$e$ is a $\U$-quotient}. \]
\end{definition}

\begin{rem}
Every $\F$-refinable $\MT$-quotient is an $\F$-quotient because $\F$ is closed under quotients. Therefore the last line of the previous definition could be replaced by
\[ \text{$e$ is an $\F$-quotient of $(A,\alpha)$} \quad\text{iff}\quad \text{$e$ is a $\U$-quotient of $(A,\alpha)$}. \]
\end{rem}

\begin{rem}\label{rem:unprescong}
Here is the equivalent version of Definition \ref{def:unpres} for congruences and stable preorders:
\begin{enumerate}
\item Suppose that $\D$ is a variety of algebras. Then a set $\U$ of unary operations on $A$ forms a unary presentation of $(A,\alpha)$ iff for any $\F$-refinable congruence $\equiv$ on $A$,
\[ \text{$\equiv$ is a $\MT$-congruence on $(A,\alpha)$} \quad\text{iff}\quad \text{$\equiv$ is a $\U$-congruence on $(A,\alpha)$}. \]
\item   Suppose that $\D$ is a variety of ordered algebras. Then a set $\U$ of unary operations on $A$ forms a unary presentation of $(A,\alpha)$ iff for any $\F$-refinable stable preorder $\preceq$ on $A$, 
\[ \text{$\preceq$ is a $\MT$-stable order on $(A,\alpha)$} \quad\text{iff}\quad \text{$\preceq$ is a $\U$-stable preorder on $(A,\alpha)$}. \]
\end{enumerate}
\end{rem}

\begin{rem}
In practice, if the monad $\MT$ represents algebras with finitary operations, the restriction to \emph{$\F$-refinable} quotients can often be dropped, i.e. the equivalence of the two properties in Definition \ref{def:unpres} holds for arbitrary quotients $e$. However, the restriction is usually necessary in the presence of infinitary operations. See the examples in the following section.
\end{rem}

\subsection{Examples of unary presentations}
To illustrate the above definitions, we derive unary presentations for monoids, finitary $\Gamma$-algebras, $\omega$-semigroups and, more generally, algebras for a monad on $\Set_f^S$.

\begin{example}[Monoids]\label{ex:unpres_monoids}
Let $\MT=\MT_\ast$ on $\Set$ and $\F =$ all finite monoids. Every monoid $M$ has a unary presentation given by the unary operations $x\mapsto yx$ and $x\mapsto xy$ on $M$, where $y$ ranges over all elements of $M$.  Indeed, an equivalence relation $\equiv$ on $M$ forms a monoid congruence iff it is stable under left and right multiplication, i.e. $x\equiv x'$ implies $yx\equiv yx'$ and $xy\equiv x'y$ for all $y\in M$. This holds even without the condition that $\equiv$ is $\F$-refinable, i.e. refinable to a monoid congruence of finite index.
\end{example}

\begin{example}[$\Gamma$-algebras]\label{ex:unpres_gammaalgs}
Generalizing the previous example, let $\Gamma$ be a finitary $S$-sorted signature, and let $\MT$ be the free-algebra monad on $\Set^S$ associated to some variety $\ACat$ of $\Gamma$-algebras; that is, $\ACat\cong \Alg{\MT}$. Put $\F =$ all finite algebras in $\ACat$. Every algebra $A\in \ACat$ has a unary presentation given by the set of elementary translations on $A$. Indeed, an $S$-sorted equivalence relation $\equiv$ on $A$ is a $\Gamma$-congruence (= $\MT$-congruence) iff it is stable under all elementary translations,  see \ref{app:elementarytranslations}. Again the condition that $\equiv$ is $\F$-refinable, i.e. refinable to a $\Gamma$-congruence of finite index, is not needed.

Likewise, if $\MT$ is the free-algebra monad on $\Pos^S$ associated to an $S$-sorted variety $\ACat$ of ordered algebras, the elementary translations form a unary presentation.
\end{example}

\begin{example}[$\omega$-semigroups]\label{ex:unpres_omegasem}
Let $\MT=\MT_\infty$ on $\Set^2$ and $\F =$ all finite $\omega$-semigroups. Every $\omega$-semigroup $A=(A_+,A_\omega)$ has a unary presentation $\U$ given by the operations
\begin{enumerate}[(1)]
\item $x\mapsto yx$ and $x\mapsto xy$ on $A_+$, 
\item $x\mapsto xz$ and $x\mapsto x^\omega=\pi(x,x,x,\ldots)$ from $A_+$ to $A_\omega$, and 
\item $z\mapsto yz$ on $A_\omega$,
\end{enumerate}
where $y\in A_+\cup \{1\}$ and $z\in A_\omega$. Here we view $1$ as a neutral element of the semigroup $A_+$, i.e. we put $1x = x1 =: x$ and $1z=:z$. 

To show that $\U$ is a unary presentation, suppose that $\equiv$ is an $\F$-refinable two-sorted equivalence relation on $A$, i.e. $\equiv$ contains an $\omega$-semigroup congruence $\mathord{\sim}\seq \mathord{\equiv}$ of finite index. We need to prove that $\equiv$ is an $\omega$-semigroup congruence (i.e. stable under finite and infinite multiplication) iff $\equiv$ is stable under the operations in $\U$.

\smallskip
\noindent($\To$) If $\equiv$ is an $\omega$-semigroup congruence, then $\equiv$ is stable under finite and infinite multiplication and thus in particular under the operations in $\U$.

\smallskip
\noindent($\Leftarrow$) Suppose that $\equiv$ is stable under the operations in $\U$. Stability under the operation (1), the first one in (2) and the one in (3) implies that $\equiv$ is stable under (mixed) binary products. It remains to show that $\equiv$ is stable under infinite products: given $v_i, w_i\in A_+$  with $v_i\equiv w_i$ ($i=0,1,2,\ldots$) we need to show that $\pi(v_0,v_1,v_2\ldots) \equiv \pi(w_0,w_1,w_2\ldots)$. This rests on the following combinatorial lemma, an instance of Ramsey's theorem:

\begin{lemma*}[see \cite{pinperrin04}, p.77]
Let $X$ be a set, $h\colon X^+\to E$ a function into a finite set $E$ and $x_0,x_1,x_2,\ldots$ a sequence in $X$. Then there exists a strictly increasing sequence $k_0<k_1<k_2<\ldots$ of natural numbers such that \[h(x_{k_i}\ldots x_{k_{i+1}-1})=h(x_{k_0}\ldots x_{k_{1}-1})\quad\text{for all $i\geq 0$}.\] 
\end{lemma*}
Apply the lemma to the set $X=A_+\times A_+$, the semigroup morphism
\[h\colon (A_+\times A_+)^+ \to (A_+/\mathord{\sim})\times (A_+/\mathord{\sim})\]
mapping $(v,w)\in A_+\times A_+$ to $([v]_\sim, [w]_\sim)$,
and the sequence $(v_0,w_0), (v_1,w_1), (v_2,w_2), \ldots$ in $A_+\times A_+$. This gives a strictly increasing sequence $k_0<k_1<k_2<\ldots$ of natural numbers such that
\[ v_{k_i}\ldots v_{k_{i+1}-1} \sim v_{k_0}\ldots v_{k_{1}-1} \quad\text{and}\quad w_{k_i}\ldots w_{k_{i+1}-1} \sim w_{k_0}\ldots w_{k_{1}-1}\quad\]
for all $i\geq 0$. It follows that
\begin{align*}
\pi(v_0,v_1,v_2,\ldots) & = v_0\ldots v_{k_0-1}\cdot \pi(v_{k_0}\ldots v_{k_1-1}, v_{k_1}\ldots v_{k_2-1}, v_{k_2}\ldots v_{k_3-1},\ldots)\\
&\sim v_0\ldots v_{k_0-1}\cdot \pi(v_{k_0}\ldots v_{k_1-1}, v_{k_0}\ldots v_{k_0-1}, v_{k_0}\ldots v_{k_0-1},\ldots)\\
&= v_0\ldots v_{k_0-1} \o (v_{k_0}\ldots v_{k_1-1})^\omega 
\end{align*}
where the first step uses the associative laws, the second one that $\sim$ is an $\omega$-semigroup congruence, and the third one is the definition of $(\dash)^\omega$. In particular, since $\sim\,\seq\, \equiv$, we get
\[ \pi(v_0,v_1,v_2,\ldots) \equiv v_0\ldots v_{k_0-1} (v_{k_0}\ldots v_{k_1-1})^\omega \]
and analogously
\[ \pi(w_0,w_1,w_2,\ldots) \equiv w_0\ldots w_{k_0-1} (w_{k_0}\ldots w_{k_1-1})^\omega. \] 
But since $\equiv$ is stable under left and right multiplication and $\omega$-powers by assumption, $v_i\equiv w_i$ for all $i\geq 0$ implies
\[ v_0\ldots v_{k_0-1} (v_{k_0}\ldots v_{k_1-1})^\omega \equiv w_0\ldots w_{k_0-1} (w_{k_0}\ldots w_{k_1-1})^\omega. \]
Thus $\pi(v_0,v_1,v_2,\ldots)\equiv \pi(w_0,w_1,w_2,\ldots)$, as required.
\end{example}

\begin{example}[Monads on $\Set_f^S$]\label{ex:unpres_monads}
Let $\MT$ be an arbitrary monad on $\Set^S$ and $\F=$ all finite $\MT$-algebras. Then every $\MT$-algebra $(A,\alpha)$ has a generic unary presentation given as follows. Let $1_{s}\in\Set^S$ be
  the $S$-sorted set with one element in sort $s$ and otherwise empty; thus a 
  morphism $1_s\to A$ in $\Set^S$ chooses an element of 
 $A_s$. A 
 \emph{polynomial} over $A$ is a morphism $p\colon1_{t}\to
 T(A+1_{s})$ with $s,t\in S$, i.e.\ a ``term'' of output sort $t$ in a distinguished
 variable of sort $s$. Every polynomial $p\colon1_{t}\to
 T(A+1_{s})$ over $A$ induces an evaluation map $[p]\colon A_s \to A_{t}$ sending an element
 $x\colon 1_s\to A$ of $A_s$ to the element
 \[ 1_{t}\xra{p} T(A+1_s) \xra{T(A+x)} T(A+A) 
 \xra{T\left[\id,\id\right]} TA
 \xra{\alpha} A \]
of $A_{t}$. We claim that the set \[\U = \{\, [p]: \text{$p$ a polynomial over $A$}\,\}\]forms unary presentation of $(A,\alpha)$. 

Thus let $\equiv$ be an $\F$-refinable equivalence relation on $A$. We need to show that $\equiv$ forms a $\MT$-congruence iff it is stable under polynomial evaluation maps.

\smallskip
\noindent\textbf{($\To$)} Suppose that $\equiv$ is a $\MT$-congruence, i.e. there exists a (finite) quotient $e\colon (A,\alpha)\epito (B,\beta)$ in $\Alg{\MT}$ with kernel $\equiv$. Every polynomial $p\colon1_{t}\to
 T(A+1_{s})$ over $A$ yields a polynomial over $B$ given by
 \[ q\colon 1_{t} \xra{p} T(A+1_s)  \xra{T(e + 1_s)} T(B+1_s)  \]
Then the following square commutes (which then implies that $\equiv$ is a $\U$-congruence):
\[
\xymatrix{
A_s \ar[r]^{[p]} \ar@{->>}[d]_e & A_{t} \ar@{->>}[d]^e\\
B_s \ar[r]_{[q]} & B_{t}
}
\]
To see this, let $x\colon 1_s\to A$ be an element of $A_s$ and consider the diagram below:
\[
\xymatrix{
1_{t} \ar[r]^<<<<<p \ar@{=}[d] & T(A+1_s) \ar[rr]^{T(A+x)} \ar[d]^{T(e+1_s)} && T(A+A) \ar[rr]^{T[id,id]}  \ar[d]^{T(e+e)} && TA \ar[r]^\alpha \ar[d]^{Te} & A \ar[r]^e & B \ar@{=}[d]\\
1_{t} \ar[r]_<<<<<q & T(B+1_s) \ar[rr]_{T(A+e\o x)} && T(B+B) \ar[rr]_{T[id,id]} && TB \ar[rr]_\beta & & B
}
\]
The upper horizontal path is precisely the image of $x$ under $e\o [p]$, and the lower horizontal path is the image of $x$ under $[q]\o e$. Therefore, all we need to show is that the above diagram commutes. 
The leftmost square commutes by the definition of $q$, the rightmost square because $e$ is a $\MT$-homomorphism, and the two other squares commute trivially. Thus the diagram commutes, showing that $e\o[p](x) = [q]\o e(x)$.

\smallskip
\noindent\textbf{($\Leftarrow$)} Recall e.g. from \cite{manes76} that every monad $\MT$ on $\Set^S$ arises from a signature and equations. More precisely, there exists an $S$-sorted signature $\Gamma$ (possibly consisting of a \emph{proper class} of operations) and a class $E$ of equations between well-founded $\Gamma$-terms such that $\MT$ is the free-algebra monad associated to $\Gamma$ and $E$. In particular, every $\MT$-algebra $(A,\alpha)$ corresponds to a $(\Gamma,E)$-algebra with the same carrier $A$, and every $\MT$-congruence  corresponds to a $\Gamma$-congruence, i.e.~an equivalence relation stable under all $\Gamma$-operations. To avoid complicated notation, we restrict ourselves to the single-sorted case.

Now suppose that $\equiv$ is an equivalence relation on a $(\Gamma,E)$-algebra $A$ that contains a $\Gamma$-congruence of finite index, and that moreover $\equiv$ is stable under polynomial evaluation maps. We need to show that $\equiv$ is a $\Gamma$-congruence. Thus let $\gamma$ be an operation symbol in $\Gamma$ of (possibly infinite) arity $\kappa$, and let  $x=(x_i)_{i<\kappa}$ and $y=(y_i)_{i<\kappa}$ be two $\kappa$-tuples in $A$ with  $x_i\equiv y_i$ for all $i<\kappa$. It is our task to prove that $\gamma^{A}(x)\equiv \gamma^{A}(y)$.

Since $\equiv$ contains a $\Gamma$-congruence $\sim$ of finite index, we may assume that only finitely many distinct elements of $A$ appear in the tuples $x$ and $y$. Indeed, otherwise choose a fixed representative of every $\sim$-equivalence class and replace every component of $x$ and $y$ by its unique $\sim$-representative. The resulting tuples $x'$ and $y'$ contain only finitely many distinct elements of $A$ because $\sim$ has finite index, and one has $\gamma^{A}(x)\equiv \gamma^{A}(x')$ and $\gamma^{A}(y)\equiv \gamma^{A}(y')$ because $\sim$ is a $\Gamma$-congruence and $\mathord{\sim}\seq\mathord{\equiv}$. Therefore it suffices to work with $x'$ and $y'$ in lieu of $x$ and $y$.

  The finiteness assumption on $x$ and $y$ implies that, in particular, there are only finitely many pairs $(u,v)\in A\times A$ for which the set
\[ I_{u,v} = \{\, i<\kappa : x_i=u \text{ and } y_i=v \,\} \]
is nonempty; say the pairwise distinct pairs $(u_1,v_1), \ldots, (u_n, v_n)$. Then the sets $I_{u_j,v_j}$ form a partition of $\kappa$. For each $j=0,\ldots, n$, consider the $\kappa$-tuple $x^j = (x_i^j)_{i<\kappa}$ arising from $x$ by replacing all occurrences of the values $u_1,\ldots, u_j$ by $v_1,\ldots, v_j$, respectively. Formally, 
\[
x_i^j = \begin{cases}
v_k, & i\in I_{u_k,v_k} \text{ for some $k\in\{1,\ldots,j\}$}\\
x_i, & \text{otherwise}.
\end{cases}
\]
Note that $x^0=x$ and $x^n=y$. For each $j=1,\ldots n$, form the polynomial $p^j$ over $A$ given by \[p^j =\gamma^{T(A+1)}(z^j) \in T(A+1),\]  where the $\kappa$-tuple $z^j=(z^j_i)_{i<\kappa}$ in $T(A+1)$ emerges from $x$ by replacing all occurrences of $u_1,\ldots, u_{j-1}$ by $v_1,\ldots, v_{j-1}$, respectively, and $u_j$ by the unique element $\ast$ of $1$. Here we view $1$ and $TA$ as subsets of $T(A+1)$. Formally,
 \[
 z_i^j = \begin{cases}
 \ast, & i\in I_{u_j,v_j}\\
 x_i^{j-1}, & \text{otherwise.}
 \end{cases}
 \]
Substituting $u_j$ for $\ast$ in $z^j$ yields $x^{j-1}$, and thus the evaluation map $[p^j]\colon A\to A$ maps $u_j$ to $\gamma^{A}(x^{j-1})$. Similarly, substituting $v_j$ for $\ast$ yields $x^j$, and thus $[p^j]$ maps $v_j$ to $\gamma^{A}(x^j)$. Since $u_j\equiv v_j$ and $\equiv$ is stable under polynomial evaluation maps, we get 
\[\gamma^{A}(x^{j-1}) = [p^j](u_j) \equiv [p^j](v_j) = \gamma^{A}(x^j) \quad \text{for $j=1,\ldots, n$}.\] This implies
\[ \gamma^{A}(x) = \gamma^{A}(x^0)\equiv \gamma^{A}(x^1)\equiv\gamma^{A}(x^2)\equiv\cdots\equiv \gamma^{A}(x^n)=\gamma^{A}(y),\]
proving that $\equiv$ is a $\Gamma$-congruence.
\end{example}

\begin{rem}
 Note that in the case of monoids and $\omega$-semigroups, the polynomial presentation is rather unwieldy and much larger than the unary presentations given in Example \ref{ex:unpres_monoids} and \ref{ex:unpres_omegasem}. For example, for a monoid $M$ the polynomial presentation contains all unary operations
$x\mapsto y_0xy_1x\ldots xy_n$ on $M$ with $y_0,\ldots, y_n\in M$. 

As the example of $\omega$-semigroups shows, finding a small unary presentation is generally a nontrivial challenge for algebras with infinitary operations. On the other hand, for algebras with finitary operations one can always take the unary presentation given by the elementary translations (see Example \ref{ex:unpres_gammaalgs}), which is much smaller than the generic presentation by all polynomials.
\end{rem}

\begin{example}\label{ex:refinable}
 In contrast to Example \ref{ex:unpres_monads}, in general not every $\MT$-algebra admits a unary presentation if $\D\neq\Set$. To see this, consider the variety $\D=\Set_{c,d}$ of sets with two constants, i.e. algebras for the signature $\Gamma=\{c,d\}$ with two constant symbols. Let $\Set_{c\neq d}$ be the full reflective subcategory on the terminal object $1$ and all $X\in\Set_{c,d}$ with $c^X\neq d^X$. Then the inclusion $\Set_{c\neq d} \monoto \Set_{c,d}$ is a monadic right adjoint, and thus it induces a monad $\MT$ on $\D$ satisfying $\Alg{\MT}\cong \Set_{c\neq d}$. Explicitly, $\MT$ is given by
\[
TX = \begin{cases}
X, & c^X\neq d^X;\\
1, & c^X=d^X.
\end{cases}
\]
Put $\F$ = all finite $\MT$-algebras. For any $X\in \D$, a nontrivial equivalence relation $\equiv$ on $\under{X}$ is a $\MT$-congruence iff  $c^X\not\equiv d^X$, i.e. the corresponding quotient $X/\mathord{\equiv}$ in $\D$ lies in $\Set_{c\neq d}$.

\begin{claim}
  The $\MT$-algebra corresponding to $X=\{x,c,d\}$ in $\Set_{c\neq d}$ has no unary presentation.
\end{claim}
\begin{proof}
Suppose that, on the contrary, $X$ has a unary presentation $\U$. We consider two cases:
\begin{enumerate}
\item $\U$ contains a non-identity morphism $u$. Then  $u(x)\in\{c,d\}$, and we assume w.l.o.g. that $u(x)=d$. Let $\equiv$ be the smallest equivalence relation on $X$ with $x\equiv c$. Then $\equiv$ is a $\MT$-congruence because $c\not\equiv d$. On the other hand, $\equiv$ is not a $\U$-congruence: we have $x\equiv c$ but $u(x)=d\not\equiv c=u(c)$, where the last equation holds because $u$ is a morphism of $\Set_{c,d}$. This contradicts our assumption that $\U$ is a unary presentation of $X$.
\item $\U$ is contains only identity morphisms. Let $\equiv'$ be the smallest equivalence relation on $X$ with $c\equiv' d$. Then $\equiv'$ is trivially a $\U$-congruence. On the other hand, $\equiv'$ is not a $\MT$-congruence because $c\equiv' d$. This contradicts our assumption that $\U$ is a unary presentation.\qedhere
\end{enumerate}
\end{proof}

\end{example}

\begin{example}
We have seen that in the two finitary Examples \ref{ex:unpres_monoids} and \ref{ex:unpres_gammaalgs} the restriction to $\F$-refinable quotients was not needed. However, for algebras with infinitary operations, this restriction is crucial. To demonstrate this, we devise a monad $\MT$ on $\Set$ and a $\MT$-algebra $A$ such that for the polynomial presentation $\U$, not every $\U$-congruence is a $\MT$-congruence.

Consider the free-algebra monad $\MT$ on $\Set$  associated to the signature $\Gamma$ with a single $\omega$-ary operation symbol $\gamma$ and no equations. The free $\Gamma$-algebra $\MT X$ generated by a set $X$ is carried by the set of well-founded $\Gamma$-trees (i.e.~$\omega$-branching trees containing no infinite path) with leaves labeled by elements of $X$. We call a $\Gamma$-tree $t$ \emph{bounded} if, for some $\ell\in \Nat$, every path in $t$ has length at most $\ell$. Now fix $X\neq \emptyset$ and consider the equivalence relation $\equiv$ on $TX$ with $s\equiv t$ iff either both $s$ and $t$ are bounded, or both $s$ and $t$ are unbounded. 

Let $\U$ be the polynomial presentation of the free $\Gamma$-algebra $A=\MT X$. We claim that the equivalence relation $\equiv$ is a $\U$-congruence, i.e. for $s\equiv t$ in $TX$ and any polynomial $p\in T(TX+\{\ast\})$ over $TX$, we have $[p](s) \equiv [p](t)$. Here $[p](u)$ ($u\in TX$) is the tree arising from $p$ by substituting the tree $u$ for every $\ast$-labeled leaf. There are three cases to consider:
\begin{enumerate}
\item $p$ is unbounded. Then both $[p](s)$ and $[p](t)$ are unbounded, and therefore  $[p](s) \equiv [p](t)$.
\item $p$ is bounded and both $s$ and $t$ are bounded. Then both $[p](s)$ and $[p](t)$ are bounded, and therefore $[p](s) \equiv [p](t)$.
\item $p$ is bounded and both $s$ and $t$ are unbounded. If the variable $\ast$ occurs in $p$ at least once, then both $[p](s)$ and $[p](t)$ are unbounded and therefore $[p](s) \equiv [p](t)$. Otherwise we have $[p](s)=p=[p](t)$ and thus trivially $[p](s)\equiv [p](t)$.
\end{enumerate}
Thus $\equiv$ is a $\U$-congruence. On the other hand, $\equiv$ is not a $\MT$-congruence (= $\Gamma$-congruence) on $\MT  X$. To see this, choose for each $n\geq 0$ a tree $t_n\in TX$ of height $n$. Then $t_0\equiv t_n$ for all $n\geq 0$ because $t_n$ is bounded. But
\[
\gamma^{TX}(t_0,t_0,t_0,\ldots) \not\equiv \gamma^{TX}(t_0,t_1,t_2,\ldots)
\]
because the tree $\gamma^{TX}(t_0,t_0,t_0,\ldots)$ is bounded and the tree $\gamma^{TX}(t_0,t_1,t_2,\ldots)$ is unbounded. Thus $\equiv$ not stable under the operation $\gamma^{TX}$, showing that $\equiv$ is not a $\Gamma$-congruence.
\end{example}

\subsection{Properties of unary presentations}
Next, we establish some important technical properties of unary presentations.

\begin{lemma}\label{lem:homo-preserves-unary}
  Let $D\in \D_f^S$ and let $\U$ be a unary presentation of $\MT D$. Suppose that $e\colon 
  \MT D\epito A$ 
  and $k\colon A\epito B$ are surjective $\MT$-homomorphisms with $A,B\in 
  \F$. Then the following diagram commutes for all $u\colon (T D)_s\to (T D)_{t}$ in $\U$, where $u_A$ and $u_B$ are the liftings of $u$ along $e$ 
  and $k\o e$, respectively.
    \[
      \xymatrix{
        (T D)_s \ar@{->>}[r]^{e} \ar[d]_{u} & A_s \ar@{->>}[r]^{k} 
        \ar[d]^{u_A}& B_s \ar[d]^{u_B}\\
        (T D)_{t} \ar@{->>}[r]_{e} & A_{t} \ar@{->>}[r]_{k} & 
        B_{t} 
      }
    \] 
\end{lemma} 

\begin{proof}
The left-hand and the outward square commute by definition of $u_A$ and $u_B$. Thus the right-hand square commutes, as it does when precomposed with the epimorphism $e$.
\end{proof}

\begin{lemma}\label{lem:subdirectproduct}
Let $A\in \D^S$ and $\U$ a set of unary operations on $A$. Then for any two $\U$-quotients $e_i\colon A\epito A_i$ ($i=0,1$) the following statements hold:
\begin{enumerate}
\item The subdirect product of $e_0$ and $e_1$ in $\D^S$ (see \ref{app:subdirectproducts}) is a $\U$-quotient.
\item The pushout of $e_0$ and $e_1$ in $\D^S$ is a $\U$-quotient.
\end{enumerate}
\end{lemma}

\begin{proof}
\begin{enumerate}
\item Let $e\colon A\epito B$ be the subdirect product of $e_0$ and $e_1$. Thus $\langle e_0, e_1\rangle = m\o e$ for some injective/order-reflecting morphism $m$.
Suppose that $u\colon A_s\to A_t$ is a unary operation in $\U$. We need to show that $u$ has a lifting along $e$. Consider the following diagram, where $u_{A_i}$ is the lifting of $u$ along $e_i$, using that $e_i$ is a $\U_\Sigma$-quotient.
\[
\xymatrix{
A_s \ar[r]^{u} \ar@{->>}@/_8ex/[ddd]_{e_i}\ar@{->>}[d]_e & A_t \ar@{->>}[d]^e \ar@{->>}@/^8ex/[ddd]^{e_i}\\
B_s \ar@{>->}[d]_{m} & B_t \ar@{>->}[d]^{m}\\
(A_0\times A_1)_s \ar[d]_{\pi_i} \ar[r]^{u_{A_0}\times u_{A_1}} & (A_0\times A_1)_t \ar[d]^{\pi_i} \\
(A_i)_s \ar[r]_{u_{A_i}} & (A_i)_t
}
\] 
All parts except possibly the upper rectangle commute. This implies that the upper rectangle also commutes, as is does when postcomposed with the product projections $\pi_i$. Diagonal fill-in yields a morphism $u_B\colon B_s\to B_t$ with $e\o  u = u_B\o e$, i.e. a lifting of $u$ along $e$. This proves that $e$ is a $\U$-quotient.
\item Let $p=p_0\o e_0 = p_1\o e_1\colon A\epito P$ be the pushout of $e_0$ and $e_1$, and let $u\colon A_s\to A_t$ in $\U$. We need to show that $u$ lifts along $p$. The proof is illustrated by the diagram below, where $u_{A_i}$ is the lifting of $u$ along the $\U$-quotient $e_i$.
\[
\xymatrix{
&& A_{t} \ar@{->>}[ddll]_{e_0} \ar@{->>}[ddrr]^{e_1} 
&&\\
&& A_s \ar@{->>}[dd]^{p} \ar[u]_{u} \ar@{->>}[dl]_{
e_0} 
\ar@{->>}[dr]^{e_1} 
&&\\
A_{0,t} \ar@{->>}[ddrr]_{p_0} & A_{0,s} \ar[l]_{u_{A_0}} \ar@{->>}[dr]_{p_0}  & & 
A_{1,s} 
\ar[r]^{u_{A_1}} \ar@{->>}[dl]^{p_1} & A_{1,t} \ar@{->>}[ddll]^{p_1}\\
&& P_s \ar@{-->}[d]^{u_P} &&\\
&& P_{t} &&
}
\]
 Then the morphisms $p_0\o u_{A_0}$ and $p_1\o u_{A_1}$ merge $e_0$ and $e_1$. Therefore, by the universal property of the pushout $p$, there exists a unique $u_P\colon P_s\to P_t$ with $p_0\o u_{A_0} = u_P\o p_0$ and $p_1\o u_{A_1} = u_P \o p_1$. Thus all parts of the above diagram commute. This implies $p\o u = u_P\o p$, i.e. $u$ lifts along $p$.
\end{enumerate}
\end{proof}

\subsection{Unary presentations for $\hatT$-algebras}

In this subsection, we show that any unary presentation of a free $\MT$-algebra $\MT D$ ($D\in \D_f^S$) can be extended to a presentation of the free $\hatT$-algebra $\hatT D$. First, we characterize $\F$-refinable quotients of $\MT D$ in terms of $\hatT$:

\begin{lemma}\label{lem:extensiblevsrefinable}
Let $D\in\D_f^S$. A quotient $e\colon T D\epito A$ in $\D^S$ is $\F$-refinable iff it is extensible to a quotient of $\hatt D$ in $\hatD^S$, i.e. $e = V\hat e \o \iota_ D$ for some $\hat e\colon \hatt D\epito A$ in $\hatD^S$.
\[
\xymatrix{
TD \ar[dr]_{e} \ar[r]^{\iota_D} \ar[r] & V\hatt D \ar[d]^{V\hat e}\\
& A
}
\]
\end{lemma}

\begin{proof}
Suppose the $e$ is $\F$-refinable, i.e. $e = p\o q$ for some quotient $q\colon \MT D\epito B$ in $\Alg{\MT}$ with $B\in \F$ and some morphism $p\colon B\epito A$ in $\D^S$. By Remark \ref{rem:basic}.\ref{rem:iota} it follows that 
\[ e = p\o q = p\o Vq^+\o \iota_ D, \]
i.e. $e$ is extensible to the quotient $p\o q^+\colon \hatt D\epito A$ in $\hatD^S$.

Conversely, suppose that $e$ is extensible to a quotient $\hat e\colon \hatt D\epito A$ in $\hatD^S$. Since $A$ is finitely copresentable in $\hatD^S$, see Remark \ref{rem:hatdlfcp}, the morphism $\hat e$ factors through the cofiltered limit cone defining $\hatt D$. That is, there is a finite quotient $q\colon \MT D\epito B$ in $\Alg{\MT}$ with $B\in\F$ and a morphism $p\colon B\epito A$ with $\hat e = p\o q^+$. Therefore
\[ e = V\hat e\o \iota_ D = p\o Vq^+\o \iota_ D = p\o q, \]
showing that $e$ is $\F$-refinable.
\end{proof}
Next, we show that given a unary presentation of $\MT D$, the unary operations extend to $\hatT D$:

\begin{lemma}\label{lem:extend_fi}
  Let $D\in \D_f^S$ and let $\U$ be a unary presentation of $\MT D$. Then every unary operation 
  $u\colon (T D)_s\to (T D)_{t}$ in $\U$ 
 extends uniquely to a morphism $\hat u\colon (\hatt D)_s \to 
  (\hatt D)_{t}$ in 
  $\hatD$ making the following square commute.
  \[
    \xymatrix{
    (T D)_s \ar[d]_{\iota_ D} \ar[r]^{u} & (T D)_{t}
    \ar[d]^{\iota_ D} \\
    V(\hatt D)_s \ar[r]_{V\hat u} & V(\hatt D)_{t}  
    }
  \]
\end{lemma}

\begin{proof}
For each $u\colon (T D)_s\to (T D)_{t}$ in $\U$, the morphisms 
$u_A\o e^+\colon (\hatt D)_s\to 
A_{t}$  
(where $e$ ranges over
surjective $\MT$-homomorphisms $e\colon \MT D\epito A$ with 
$A\in\F$ and $u_A$ is the lifting of $u$ along $e$) form a 
compatible family over the diagram 
defining $(\hatt D)_{t}$ by Lemma~\ref{lem:homo-preserves-unary}.
Hence there exists a unique morphism $\hat u\colon 
(\hatt D)_s\to(\hatt D)_{t}$ in
$\hatD$ with $e^+\o \hat u = u_A\o e^+$ for all $e$. Therefore in the 
diagram below the outside and all parts except, perhaps, for the upper square 
commute:
  \[
    \xymatrix{
    (T D)_s \ar[d]_{\iota_ D} \ar[r]^{u} 
    \ar@{->>}@/_4em/[dd]_{e}& 
    (T D)_{t}
    \ar[d]^{\iota_ D} \ar@{->>}@/^4em/[dd]^{e}\\
    V(\hatt D)_s \ar@{->>}[d]_{Ve^+} \ar[r]^{V\hat u} & 
    V(\hatt D)_{t} \ar@{->>}[d]^{Ve^+} \\
    A_s \ar[r]_{u_A} & A_{t}
    }
  \]
It follows that the upper square commutes when postcomposed with the morphisms 
$Ve^+$. Since by Remark~\ref{rem:hatdlfcp} the functor $V$ 
preserves limits (and thus the morphisms $Ve^+$ are jointly 
monomorphic),  the upper square commutes. Moreover, $\hat u$ is unique with 
this property because $\iota_ D$ is dense and $(\hatt D)_{t}$ is a Hausdorff space (see Remark 
\ref{rem:basic}.\ref{rem:iota}).
\end{proof}
Therefore we can extend our notion of $\U$-quotients and $\F$-quotients of $\MT D$ as follows:

\begin{definition}
Let $D\in \D_f^S$ and $\phi\colon \hatt D \epito P$ a quotient of $\hatt D$ in $\hatD^S$.
\begin{enumerate}
\item 
 $\phi$ is called a \emph{profinite $\hatT$-quotient} of $\hatT D$ if it carries a profinite quotient algebra of $\hatT D$, i.e. there exists a $\hatT$-algebra structure $(P,\rho)$  on $P$ such that $(P,\rho)$ is a profinite $\hatT$-algebra and $\phi\colon \hatT D\epito(P,\rho)$ is a $\hatT$-homomorphism.
\item If $\U$ is a unary presentation of $\MT D$, then $\phi$ is called a \emph{$\U$-quotient} of $\hatt D$ if for every $u\colon (TD)_s\to (TD)_t$ in $\U$, the extension $\hat  u\colon (\hatt D)_s\to (\hatt D)_t$ has a lifting along $p$:
 \[
\xymatrix{
(\hatt D)_s \ar[r]^{\hat u} \ar@{->>}[d]_{\phi} & (\hatt D)_{t} 
\ar@{->>}[d]^{\phi}\\
P_s \ar@{-->}[r]_{\exists u_P} & P_{t}
}
  \]
\end{enumerate}
\end{definition}

\begin{rem}\label{rem:uextend}
If the object $P$ in the above definition is finite, then $\phi\colon \hatt D\epito P$ is a $\U$-quotient of $\hatt D$ iff its restriction $V\phi\o \iota_D\colon TD\epito A$ is a $\U$-quotient of $TD$. This follows immediately from the definition of $\hat u$ and the fact that $\iota_D$ is dense, see Remark \ref{rem:basic}.\ref{rem:iota}.
\end{rem} 

\begin{lemma}\label{lem:subdirectproducthatd}
Let $D\in \D_f^S$ and suppose that $\U$ is a unary presentation of $\MT D$. 
Then for any two $\U$-quotients $e_i\colon \hatt D\epito P_i$ ($i=0,1$) in $\hatD^S$ the following statements hold:
\begin{enumerate}
\item The subdirect product of $e_0$ and $e_1$ in $\hatD^S$ is a $\U$-quotient.
\item The pushout of $e_0$ and $e_1$ in $\hatD^S$ is a $\U$-quotient.
\item Given a $\U$-quotient $\phi\colon \hatt\FSigma\epito P$ and a finite quotient $e\colon P\epito A$, there exists a finite quotient $h\colon P\epito B$ such that $h\o \phi$ is a $\U$-quotient of $\hatt D$ and $e\leq h$.
\[
\xymatrix{
& \hatt D \ar@{->>}[d]^\phi\\
& P \ar@{->>}[d]^e \ar@{->>}[dl]_h\\
B \ar@{-->}[r]& A
}
\]
\item Given a $\U$-quotient $\phi\colon \hatt\FSigma\epito P$ and a finite quotient $e\colon P\epito A$, let $\S$ be the full subcategory of $(P\downarrow \hatD_f^S)$ on all finite quotients $e\colon P\epito A$ such that $e\o \phi$ is a $\U_\Sigma$-quotient. Then $P$ is the cofiltered limit of the diagram
\[ \pi_S\colon \colon \S\to \hatD^S,\quad (e\colon P\epito A)\mapsto A, \]
with limit projections $e$.
\end{enumerate}
\end{lemma}

\begin{proof}
The first two statements are completely analogous to  Lemma \ref{lem:subdirectproduct} and its proof.
\begin{itemize}
\item[3.] The proof is illustrated by the 
diagram below:
\[
\xymatrix{
\hatt D \ar@{->>}[rr]^{g^+} \ar@{->>}[dr]^q \ar@{->>}[dd]_{\phi} && C 
\ar@{->>}[dl]_{k} \ar[dd]^j \\
& B \ar@{-->}[dr]^{l} &\\
P \ar@{->>}[ur]^{h} \ar@{->>}[rr]_e && A
}
\]
Since the finite object $A$ is finitely copresentable in $\hatD^S$, the morphism $e\o \phi$ factors through the cofiltered limit cone defining $\hatt D$. That is, there exists a surjective $\MT$-homomorphism $g\colon \MT D\epito C$ with $C\in\F$ and a morphism $j\colon C\epito A$ with $e\o \phi = j\o g^+$. Then $g^+$ is a $\U$-quotient of $\hatt D$ by Remark \ref{rem:uextend}, since $\U$ is a unary presentation of $\MT D$. Form the pushout $q=h\o \phi = k\o g^+$ of $\phi$ and $g^+$ in $\hatD^S$. The pushout exists because $\hatD^S$, being a locally finitely copresentable category, is cocomplete (see \ref{app:lfcp_cat}). Note that $k$ and $q$ are surjective because pushouts preserve epimorphisms, and epimorphisms in $\hatD^S$ are surjective by Lemma \ref{lem:episurj}. This implies that $B$ is finite, because $C$ is finite. Moreover, $q$ is a $\U$-quotient by part 2 of this lemma. Since the morphisms $e$ and $j$ merge $\phi$ and $g^+$, the universal property of the pushout yields a morphism $l$ with $l\o h = e$. Therefore $e\leq h$, and $h\o\phi =q$ is a finite $\U$-quotient of $\hatt D$, as required.
\item[4.] We first show that the poset $\S$ is cofiltered. First, $\S$ is nonempty because it contains the image of the unique morphism $h\colon P\to 1$ into the terminal object. Secondly, $\S$ is directed, i.e. for any two elements $e_i\in \S$ ($i=0,1$) there exists an element $q\in \S$ with $e_i\leq q$. To see this, form the subdirect product $e$ of $e_0\o \phi$ and $e_1\o \phi$ in $\hatD^S$. It is a $\U$-quotient of $\hatt D$ by part 1 of this lemma. Moreover since $e_i\o \phi \leq \phi$, the minimality of $e$ gives $e\leq \phi$, i.e. $e=q\o \phi$ for some $q$. Then $q\in \S$ and $e_i\leq q$, as desired.
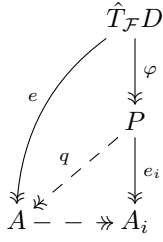
\[
\xymatrix{
& \hatt D \ar@{->>}[d]^\phi \ar@/_3ex/@{->>}[ddl]_e \\
& P \ar@{-->>}[dl]_q \ar@{->>}[d]^{e_i}\\
A \ar@{-->>}[r] & A_i
}
\]
Since the category $\hatD^S$ is locally finitely copresentable, the object $P$ is the limit of its canonical diagram
\[ \pi\colon (P\downarrow \hatD_f^S)\to \hatD^S, \quad (g\colon P\to A)\mapsto A,\]
see \ref{app:can_diagram}. Therefore it suffices to show that the inclusion $I\colon \S\to (P\downarrow \hatD_f^S)$ is a final functor, i.e. every morphism $g\colon P\to A$ with finite codomain factors through some quotient in $\S$. But this follows immediately from part 3 of this lemma.\endproof
\end{itemize}
\end{proof}
The following lemma shows that the lifting property of a unary presentation extends to profinite algebras:

\begin{lemma}\label{lem:profinite_algs}
Let $D\in\D_f^S$ and let $\U$ be a unary presentation of $\MT D$.
Then for
any quotient $\phi \colon \hatt D \epito P$ in $\hatD^S$,
\[ \text{$\phi$ is a profinite $\hatT$-quotient} \quad\text{iff}\quad \text{$\phi$ is a $\U$-quotient}.\]
\end{lemma}

\proof
($\To$) Suppose that $\phi$ is a profinite $\hatT$-quotient, i.e. $P$ carries a $\hatT$-algebra structure such that $P$ is profinite and $\phi\colon \hatT D\epito P$ is $\hatT$-homomorphism.
For any finite quotient $h\colon P\epito A$ in 
$\Alg{\hatT}$ we have the corresponding $\F$-quotient $e \defeq
V(h \o \phi)\o\iota_ D\colon \MT D\epito A$ by Remark 
\ref{rem:basic}.\ref{rem:hatthomrestrict}. Since $\U$ is a 
unary presentation of $\MT D$, each $u\colon (T D)_s\to (T D)_{t}$ in $\U$ 
has a lifting
$u_{A}\colon A_s\to 
A_{t}$ along $e$.

Since $P$ is a  profinite $\hatT$-algebra, $P$ is the cofiltered limit of the diagram of all 
finite quotient algebras
$h\colon P\epito A$, see Lemma \ref{lem:profinite_candiagram}.  The morphisms 
$u_{A}\o h\colon 
P_s\to A_{t}$ form a compatible family
over this diagram
by Lemma \ref{lem:homo-preserves-unary}. Therefore there
exists a morphism $u_P\colon P_s\to P_{t}$ in $\hatD$ with $h\o u_P = 
u_{A}\o h$
for all $h$. It follows that $u_P$ is a lifting of $\hat u$ along $\phi$ (i.e. $\phi\o \hat u = u_P\o \phi$) since this holds when postcomposed with the limit 
projections $h$.
\[
\xymatrix{
(\hatt D)_s \ar[r]^{\hat u} \ar@{->>}[d]_{\phi} & (\hatt D)_{t} 
\ar@{->>}[d]^{\phi} \\
P_s \ar@{-->}[r]^{u_P} \ar@{->>}[d]_{h}  & P_{t} \ar@{->>}[d]^{h} \\
A_s \ar[r]_{u_{A}} & A_{t} 
}
\]
This shows that $\phi$ is a $\U$-quotient of $\hatt D$.

\smallskip
\noindent
($\Leftarrow$) Let $\hat e\colon \hatt D\epito P$ be a $\U$-quotient of $\hatt D$. We need to show that $P$ can be equipped with a 
$\hatT$-algebra structure such that $P$ is profinite and $\hat e$ is a 
$\hatT$-homomorphism. Let $\S$ be the full subcategory of $(P\downarrow \hatD_f^S)$ on all finite quotients $e\colon P\epito A$ such that $e\o \phi\colon \hatt D \epito A$ is a $\U$-quotient of $\hatt D$. By Lemma \ref{lem:subdirectproduct}.4, the object $P$ is the cofiltered limit of the diagram 
\[ \pi\colon \S \to \hatD^S,\quad (e\colon P\epito A)\mapsto A, \]
with limit projections $e$. For each $e\colon P\epito A$ in $\S$ we have that $e\o \phi$ is a $\U$-quotient of $\hatt D$ by the definition of $\S$. Therefore $q:= Ve\o V\phi\o \iota_D$ is a $\U$-quotient of $TD$ by Remark \ref{rem:uextend}, and moreover $q$ is $\F$-refinable by Lemma \ref{lem:extensiblevsrefinable}. Since $\U$ forms a unary presentation of $\MT D$, it follows that there exists a $\MT$-algebra structure $(A,\alpha)$ on $A$ such that $q\colon \MT\FSigma\epito (A,\alpha)$ is a $\MT$-homomorphism. Thus $q^+= e\o \phi\colon \hatT D\epito (A,\alpha^+)$ is a $\hatT$-homomorphism by diagram \eqref{eq:hat-etasurj} in Remark \ref{rem:hattconst}. 

Since the forgetful functor from $\Alg{\hatT}$ to $\hatD^S$ creates 
limits, see \ref{app:limits_of_talgs}, it follows that there is a 
unique $\hatT$-algebra structure 
$\rho\colon 
\hatt P\to P$ making the cone $(e\colon P\epito A)_{e\in \S}$ a cofiltered 
limit cone in $\Alg{\hatT}$. Thus $(P,\rho)$ is profinite. To see that $\phi\colon 
\hatT D\epito (A,\alpha)$ is a $\hatT$-homomorphism, consider the diagram 
below:
\[
\xymatrix{
\hatt \hatt  D \ar[r]^{\hatmu_ D} \ar@{->>}[d]_{\hatt \phi} & 
\hatt D \ar@{->>}[d]^{\phi}\\
\hatt P \ar@{-->}[r]^\rho \ar@{->>}[d]_{\hatt e} &  P \ar@{->>}[d]^e\\
\hatt A \ar[r]_{\alpha^+}  & A 
}
\] 
The lower square commutes for all $e\in \S$ in $\S$ by the 
definition of $\rho$, and the outside commutes because $e \o \phi = q^+$ is a $\hatT$-homomorphism. Thus 
also the upper square commutes, as it commutes when postcomposed with the limit 
projections $e$ in $\hatD^S$.

\section{Recognizable Languages and Duality}\label{sec:reclang}
In this section, we set up our duality framework for algebraic language theory. As explained in the introduction, the idea is to consider a variety $\C$ of algebras that dualizes to the variety $\D$ on the level of finite algebras, and form varieties of languages inside $\C$. To this end, we make the following

\begin{assumptions}\label{ass:catframework}
In addition to the Assumptions \ref{ass:dt} on $\D$ and the monad $\MT$ on $\D^S$, we fix a variety $\C$ of algebras such that
\begin{enumerate}[(i)] 
\item $\C$ is locally finite;
\item the signature of $\C$ contains a 
constant;
\item the full subcategories $\C_f$ and 
$\D_f$ on 
finite algebras are dually equivalent.
\end{enumerate}
\end{assumptions}

\begin{rem}
Assumption (ii) will be used to define derivatives of languages in the many-sorted case, see Definition \ref{def:derivatives_preimages}. If $S=1$, it can be dropped.
\end{rem}

\begin{example}\label{ex:categories} The following pairs of categories $\C$/$\D$ satisfy our assumptions:
\begin{enumerate}
\item $\C = \BA$ and $\D=\Set$: Stone duality 
\cite{Johnstone1982} yields a dual equivalence $\BA_f^{op}\simeq \Set_f$, 
mapping a finite boolean algebra to the set of its atoms.
\item $\C=\DL$ 
and $\D=\Pos$: Birkhoff 
duality 
\cite{birkhoff37} gives a  dual equivalence 
$(\DL)_f^{op}\simeq \Pos_f$, mapping a finite distributive lattice to the 
poset 
of its join-irreducible elements.
\item $\C = \D = \JSL$: the self-duality  $(\JSL)_f^{op}\simeq (\JSL)_f$ maps a finite semilattice $(X,\vee,0)$ to its opposite 
semilattice $(X,\wedge,1)$.
\item $\C =\D = \Vect_{\Field}$ for a finite field $\Field$: 
the familiar self-duality of $(\Vect_{\Field})_f$ maps a finite (= finite-dimensional) 
vector space 
$X$ to its dual space $X^* = \Vect_\Field(X,\Field)$.
\end{enumerate}
\end{example}

\begin{rem}\label{rem:oc_vs_od}
\begin{enumerate}
\item Recall from Lemma \ref{lem:hatD-procomp} that the category $\hatD$ is the
    pro-completion of $\D_f$. Moreover, since the variety $\C$ is locally finite, $\C$ is the 
    \emph{ind-completion} (the free completion under
    filtered colimits) of $\C_f$.  
    Then
    the dual equivalence between $\C_f$ and $\D_f$ extends to a dual 
    equivalence between $\C$ and $\hatD$: we have
\[\hatD \simeq \Pro(\D_f) \simeq [\Ind(\D_f^{op})]^{op} \simeq [\Ind(\C_f)]^{op} \simeq \C^{op}.\]
We denote the equivalence functors by
    \[P\colon \hatD\xra{\simeq} \C^{op} \quad\text{and}\quad P^{-1}\colon \C^{op}\xra{\simeq} 
    \hatD\]
    and the corresponding natural isomorphism between $\Id_{\C^{op}}$ and $PP^{-1}$  by
 \[ \lambda\colon \Id_{\C^{op}}\xra{\cong} PP^{-1}. \] 
   \item  Denote by $\under{\dash}$ the forgetful functors of $\C$ and $\hatD$ into $\Set$, and by
$\one_\C$ and $\one_\D$ the free objects on one generator both in $\C$ and $\D$, respectively. Note that both $\one_\C$ and $\one_\D$ are finite because $\C$ and $\D$ are assumed to be locally finite.  The two finite objects $O_\C\in\C_f$ and $O_\D\in \D_f$ defined by
\[O_\C \defeq P\one_\D \quad\text{and}\quad O_\D\defeq P^{-1}\one_\C\]
play the role of a 
\emph{dualizing object} (also called a \emph{schizophrenic object} in 
\cite{Johnstone1982}) of $\C$ and $\hatD$. This means that there is a natural 
isomorphism 
   between the functors $\under{\dash}\o P^{op}$ and $\hatD(\dash,O_\D)\colon \hatD^{op}\to \Set$ given for all $D\in \hatD$ by
 \[ \under{PD} \cong \C(\one_\C,PD) \cong \hatD(P^{-1}PD,P^{-1}\one_\C) = \hatD(P^{-1}PD,O_\D) \cong 
 \hatD(D,O_\D).\]   
  Symmetrically, we have a natural isomorphism $\under{\dash}\o {P^{-1}}\cong\C(\dash, 
O_\C)\colon \C^{op}\to \Set$. 
It follows that the objects $O_\C$ and $O_\D$ have the essentially the same underlying set:
\[
      \under{O_\D} \cong \hatD(\one_\D,O_\D) \cong \under{P\one_\D} = \under{O_\C}.
\]
Explicitly, this bijection maps an element $y\colon 1\to \under{O_\D}$ of $\under{O_\D}$ to the element
\[ 1 \xra{\eta} \under{\one_\C} \xra{\lambda} \under{PP^{-1}\one_\C} = \under{PO_\D} \xra{Py^\D} \under{P\one_\D} = \under{O_\C}  \] 
of $\under{O_\C}$, where $\eta\colon 1\to \under{\one_\C}$ is the inclusion of the generator and $y^\D\colon \one_\D\to O_\D$ is the unique extension of $y$ to a morphism in $\D$.
\item Recall from Lemma \ref{lem:episurj} that epimorphisms in $\hatD$ coincide with the surjective morphisms. Accordingly, \emph{subobjects} in the dual category $\C$ are represented by monomorphisms, i.e. injective morphisms.
\end{enumerate}
\end{rem}

\begin{example}\label{ex:ocod} Consider the categories of Example \ref{ex:categories}.
\begin{enumerate}
\item $\C = \BA$ and $\D=\Set$ with $\hatD=\Stone$: the dual equivalence $P\colon \Stone\simeq \BA^{op}$ is the classical Stone duality \cite{Johnstone1982}, mapping a Stone space to the boolean algebra of its clopen subsets. We have 
\[O_\BA = \{0<1\} \quad\text{and}\quad O_\Set=\{0,1\}.\]
\item $\C = \DL$ and $\D=\Pos$ with $\hatD=\Priest$: the dual equivalence $P\colon \Priest \simeq \DL^{op}$ is the classical Priestley duality \cite{priestley72}, mapping a Priestley space to the lattice of its clopen upper sets. We have 
\[O_{\DL} = \{0<1\} \quad\text{and}\quad O_\Pos=\{0<1\} .\]
\item $\C = \D=\JSL$ with $\hatD=\Stone(\JSL)$: the dual equivalence $P\colon \Stone(\JSL) \simeq \JSL^{op}$ maps a Stone semilattice to the semilattice of its clopen ideals. We have 
\[O_\JSL = \{0<1\}.\]
\item $\C= \D =\Vect_\Field$ with $\hatD=\Stone(\Vect_\Field)$: the dual equivalence $P\colon \Stone(\Vect_\Field) \simeq \Vect_\Field^{op}$ maps a Stone vector space $X$ to the space of linear continuous maps from $X$ into $\Field$. We have 
\[O_{\Vect_\Field} = K.\]
\end{enumerate}
\end{example}

\subsection{Recognizable languages}
A language $L\seq\Sigma^*$ of finite words may be identified with its characteristic
function~$L\colon\Sigma^*\to\{0,1\}$. To model languages in our categorical setting, we replace the one-sorted
alphabet $\Sigma$ by an $S$-sorted alphabet $\Sigma$ in 
$\Set_f^S$, and represent it in $\D^S$ via the free object $\FSigma\in\D^S_f$ generated
by~$\Sigma$.
The output set~$\{0,1\}$ is replaced by a finite ``object
of outputs'' in $\D^S_f$, viz.\ the object with $O_\D\in \D_f$ in each sort. By abuse of notation, we
denote this object of $\D^S_f$ also by $O_\D$. This leads to the following definition, unifying concepts in \cite{boj15} and \cite{amu15}.
\begin{definition}
 A
  \emph{language} over the alphabet $\Sigma\in\Set_f^S$ is a morphism \[L\colon 
  T\FSigma \to
  O_\D\] in $\D^S$. It is called \emph{$\F$-recognizable} if there exists a $\MT$-homomorphism  $h\colon \MT\FSigma\to (A,\alpha)$  with $(A,\alpha)\in \F$ and a morphism $p\colon A\to O_\D$ in $\D^S$ with $L=p\o 
  h$. 
  \[
  \xymatrix{
  T\FSigma \ar[r]^L \ar[d]_h & O_\D\\
  (A,\alpha) \ar[ur]_p &
  }
  \]
  In this case, we say that \emph{$L$ is recognized by $h$ (via $p$)}. If $\F=$ all finite $\MT$-algebras, an $\F$-recognizable language is called \emph{$\MT$-recognizable}. We denote by \[\Rec_\F(\Sigma)\quad\text{and}\quad \Rec(\Sigma) \]
    the set of all $\F$-recognizable and $\MT$-recognizable languages over $\Sigma$, respectively.
\end{definition}
Note that, although languages are morphisms in $\D^S$, the set $\Rec_\F(\Sigma)$ is one-sorted.

\begin{example}[Monoids]\label{ex:rec_monoids}
Let $\MT=\MT_*$ on $\Set$ with $O_\Set = \{0,1\}$. A language 
$L\colon T_*\Sigma\to O_\Set$ corresponds to 
a language $L\seq \Sigma^*$ of finite words. It is recognized 
by
a monoid morphism $h\colon \Sigma^*\to A$ iff $L =h^{-1}[Y]$ for some
subset $Y\seq A$. Recognizable languages coincide with regular languages, i.e. those accepted by finite automata; see e.g.~\cite{pin15}.
\end{example}

\begin{example}[$\omega$-semigroups]\label{ex:rec_omegasem}
Let $\MT=\MT_\infty$ on
$\Set^2$ with $O_\Set=\{0,1\}$.  Since 
$T_\infty(\Sigma,\emptyset)=(\Sigma^+,\Sigma^\omega)$, a language $L\colon 
T_\infty(\Sigma,\emptyset)$ $\to
O_\Set$ corresponds to an $\infty$-language 
$L
\seq (\Sigma^+,\Sigma^\omega)$. It is recognized by an 
$\omega$-semigroup morphism $h\colon (\Sigma^+,\Sigma^\omega)\to A$ iff
$L = h^{-1}[Y]$ for some two-sorted subset $Y\seq A$. An $\infty$-language is recognizable iff it is \emph{$\infty$-regular}, i.e.
accepted by some finite B\"uchi automaton; see \cite{pinperrin04}.  
\end{example}

\begin{example}[$\Field$-algebras]\label{ex:rec_kalgs}
Let $\MT=\MT_\Field$ be the free $\Field$-algebra monad in $\Vect_\Field$ with $O_{\Vect_\Field}=\Field$. 
As seen in Example \ref{ex:monads_kalgs} we have $\MT_\Field\FSigma = K[\Sigma]$, the $\Field$-algebra of polynomials over $\Sigma$, and a language $L\colon T_\Field \FSigma \to K$ corresponds to a weighted language $L^@\colon \Sigma^*\to \Field$ over the field $\Field$. It is recognizable if there exists a $\Field$-algebra homomorphism $h\colon \MT_\Field\FSigma\to A$ into some finite $\Field$-algebra and a linear map $p\colon A\to K$ with $L=p\o h$.
A weighted language is recognizable
 iff it is accepted by some finite linear weighted automaton over $\Field$; see \cite{reu80}.
\end{example}
The topological perspective on regular languages rests on the important observation 
that the regular languages over $\Sigma$ correspond exactly to the clopen 
subsets of the Stone space $\widehat{\Sigma^*}$ of profinite words, or equivalently 
to 
continuous maps from $\widehat{\Sigma^*}$ into the discrete space $\{0,1\}$;
see e.g. \cite[Prop. VI.3.12]{pin15}.  This generalizes from the 
monad $\MT_*$ on $\Set$ to arbitrary monads $\MT$:

\begin{theorem}\label{thm:reciso}
$\F$-recognizable languages over $\Sigma\in \Set_f^S$ correspond bijectively to morphisms from 
$\hatt\FSigma$ to~$O_\D$ in $\hatD^S$. The bijection is given by \[(\hatt\FSigma\xra{\hat L} O_\D) \mapsto (T\FSigma \xra{\iota_\FSigma} V\hatt\FSigma \xra{V\hat L} O_\D).\]
\end{theorem}

\begin{proof}
  We first show that the indicated bijection is well-defined, i.e. for any morphism $\hat{L}\colon
  \hatt\FSigma\to O_\D$ in $\hatD^S$ the language $L \defeq V\hat{L}\o\iota_\FSigma\colon
  T\FSigma \to O_\D$ is $\F$-recognizable. Since the finite object $O_\D$ is finitely copresentable in
  $\hatD^S$, see Remark \ref{rem:hatdlfcp}, the morphism $\hat{L}$ factors
  through the cofiltered limit cone defining $\hatt \FSigma$, i.e.\ there exists
  a $\MT$-homomorphism $h\colon \MT\FSigma \to (A,\alpha)$ with $(A,\alpha)\in \F$ and a
  morphism $p\colon A\to O_\D$ in $\hatD^S$ with $\hat{L} = p\o \br{h}$. Then $L$ is recognized by $h$ via $p$:
  \begin{align*}
  L &= V\hat L\o \iota_\FSigma & \text{(def. $L$)}\\
  &= p\o Vh^+ \o \iota_\FSigma & \text{($\hat L = p\o h^+$)}\\
  &= p\o h & \text{(def. $\iota_\FSigma$ in Rem. 2.11)}. 
  \end{align*}
  Therefore $L$ is $\F$-recognizable. Conversely, let $L\colon T\FSigma\to O_\D$ be an $\F$-recognizable language. 
  Choose a $\MT$-ho\-mo\-mor\-phism $h\colon \MT\FSigma\to (A,\alpha)$ with $(A,\alpha)$ and a
  morphism $p\colon A\to O_\D$ with $L=p\o h$. This yields the following 
  morphism in $\hatD^S$:
  \[
    \hat{L} = (\hatt\FSigma \xra{\br{h}} A \xra{p} O_\D).
  \]
  Since $L = V\hat L \o \iota_\FSigma$ and $\iota_\FSigma$ is dense by Remark
  \ref{rem:basic}.\ref{rem:iota}, the morphism $\hat L$ is independent of
  the choice of $h$ and $p$. By construction the maps $\hat{L}\mapsto L$ and $L\mapsto
  \hat{L}$ are mutually inverse, which proves the theorem.
\end{proof}

\begin{rem}[$\C$-algebraic structure on $\Rec_\F(\Sigma)$]\label{rem:reg-as-C} From the above theorem and Remark 
\ref{rem:oc_vs_od}.2, we deduce
\begin{equation}\label{eq:iso}
\Rec_\F(\Sigma) \cong \hatD^S(\hatt\FSigma,O_\D) 
= \prod_s \hatD((\hatt\FSigma)_s,O_\D) \cong \prod_s 
\under{P(\hatt\FSigma)_s}.
\end{equation}
\vspace{-0.4cm}

\noindent Therefore we can view $\Rec_\F(\Sigma)$ as an object of $\C$, by endowing it with the unique $\C$-algebraic structure making the above bijection an isomorphism $\Rec_\F(\Sigma)\cong \prod_s P(\hatt\FSigma)_s$ in $\C$.
\end{rem}

\begin{proposition}\label{prop:reg-as-C}
$\Rec_\F(\Sigma)$ forms a subobject of the product
$\prod_s O_\C^{\under{T\FSigma}_s}$ in $\C$.
\end{proposition}

\begin{rem}
More precisely,  let
$i\colon \Rec_\F(\Sigma)\monoto \prod_s O_\C^{\under{T\FSigma}_s}$ be the injective map
given by\[
  (T\FSigma\xra{\;L\;}O_\D) \quad\mapsto\quad (\under{T\FSigma}_s
  \xra{\;L_s\;}
  \under{O_\D} \xra{\;\cong\;} \under{O_\C})_{s\in S},
\] where
$\under{O_\D} \cong \under{O_\C}$ is the bijection of Remark
\ref{rem:oc_vs_od}.2. We will show that $i$ is a morphism of $\C$.
\end{rem}

\begin{proof}
\begin{enumerate}
\item We first show that, for each sort $s$, the object $P(\hatt\FSigma)_s$
  forms a subobject of $O_\C^{\under{T\FSigma}_s}$ in $\C$.  For each element $x\colon
  \one_\D \to (T\FSigma)_s$ of $\under{T\FSigma}_s$, the $\D$-morphism $\one_\D
  \xra{x} (T\FSigma)_s \xra{\iota_\FSigma} V(\hatt\FSigma)_s$ is continuous,
  because $\one_\D$ is finite and thus discrete. That is, there exists a unique morphism 
  $\hat{x}\colon \one_\D
  \to (\hatt\FSigma)_s$ in $\hatD$ with $V\hat{x} = \iota_\FSigma\o x$. Since the
  morphisms $x$ are jointly surjective, and $\iota_\FSigma$ is dense by Remark
  \ref{rem:basic}.\ref{rem:iota}, the family $(\hat{x})_{x\colon \one_\D\to
    (T\FSigma)_s}$ forms a jointly epimorphic family in
  $\hatD$. Thus the dual family $(P\hat{x}\colon P(\hatt\FSigma)_s\to O_\C)_{x\colon \one_\D\to
    (T\FSigma)_s}$ in $\C$ is
  jointly monomorphic. Let $m_s$ be the unique morphism in $\C$ making the triangle below commute for all $x\colon \one_\D \to (TX)_s$, where $\pi_x$ is the product projection.
  \[ 
  \xymatrix{
  P(\hatt\FSigma)_s \ar[dr]_{P\hat x} \ar@{ >->}[r]^{m_s} & 
  O_\C^{\under{T\FSigma}_s} 
  \ar[d]^{\pi_x} \\
  & O_\C
  }
   \]
   Since the family $(P\hat x)$ is jointly monomorphic, $m_s$ is a monomorphism (i.e. injective).

\item It follows that $\Rec_\F(\Sigma)$ is a subobject of $\prod_s
  O_\C^{\under{T\FSigma}_s}$ via the embedding \[ i\colon \Rec_\F(\Sigma) \cong
    \prod_s {P(T\FSigma)_s}  \xra{\prod_s m_s} \prod_s
    O_\C^{\under{T\FSigma}_s}. \]
Let $\pi_{s,x}\colon \prod_s O_\C^{\under{T \FSigma}_s}\to O_\C$ denote the product projection associated to $s\in S$ and $x\colon \one_\D \to (T\FSigma)_s$.  
By applying the definition of the bijection \eqref{eq:iso} and the morphisms $m_s$, the morphism $\pi_{s,x}\o i\colon \Rec_\F(\Sigma)\to O_\C$ maps a language $L\colon T\FSigma\to O_\D$ to the element
\[ 1 \xra{\eta} \under{\one_\C} \xra{\lambda} \under{PP^{-1}\one_\C} = \under{PO_\D} \xra{P\hat L_s} \under{P(\hatt \FSigma)_s} \xra{P\hat x} \under{O_\C},  \]
where $\hat L$ is the continuous extension of $L$, see Theorem \ref{thm:reciso}.
Now consider the following diagram, where $j\colon \under{O_\D}\cong \under{O_\C}$ is the bijection of Remark \ref{rem:oc_vs_od}.2:
\[
\xymatrix{
\under{\one_\C} \ar[r]^\lambda & \under{PP^{-1}\one_\C} \ar@{=}[r] & \under{PO_\D} \ar[d]^{P(L_s\o x)} \ar[r]^{P\hat L_s} & \under{P(\hatt\FSigma)_s} \ar[dd]^{P\hat x} \\
& & \under{P\one_\D} \ar@{=}[dr] &\\
1 \ar[uu]_\eta \ar[r]_{\eta} & \under{\one_\D} \ar[r]_{L_s\o x} & \under{O_\D} \ar[r]_j & \under{O_\C}   
}  
\]
The left-hand part commutes by the definition of the bijection $j$. For the right-hand part observe that
\[ V\hat L_s \o V\hat x = V\hat L_s \o \iota_\FSigma\o x = L_s\o x \]
by the definition of $\hat x$ and $\hat L$. Therefore $\hat L_s\o \hat x = L_s\o x$ in $\hatD$. It follows that the above diagram commutes. The upper path from $1$ to $\under{O_\C}$ is the element  $\pi_{s,x}\o i(x)$ of $\under{O_\C}$, as seen above. This shows that $i$ maps a language $L\colon T\FSigma\to O_\D$ to the element $(\under{T\FSigma}_s
  \xra{\;L_s\;}
  \under{O_\D} \xra{\;j\;} \under{O_\C})_{s\in S}$ of $\prod_s O_\C^{\under{T\FSigma}_s}$, as claimed.
\end{enumerate}
\end{proof}
From the above proposition it follows that the $\C$-algebraic structure of
$\Rec_\F(\Sigma)$ is completely determined by the one of $O_\C$. 

\begin{example}\label{ex:c_operations} For the categories of Example \ref{ex:categories} we get the following $\C$-algebraic operations on $\Rec_\F(\Sigma)$. In the first three cases we have $\under{O_\C}=\under{O_\D}=\{0,1\}$ and thus view languages $L\colon T\FSigma\to O_\D$ as (certain) subsets of $T\FSigma$.  In the last line, the operator $+$ is the addition of weighted languages over $\Field$, i.e. $(L_0+L_1)(x)=L_0(x)+L_1(x)$, and $0$ is the zero language $0\colon T\FSigma\to \Field$.

\begin{table}[h]
\begin{tabularx}{\columnwidth}{lllll}
\hline \vspace{-0.4cm}\\
$\C$ & $\D$ & $O_\C$ & $O_\D$ & $\C$-algebraic operations on $\Rec_\F(\Sigma)$
\\
\hline
$\BA$ & $\Set$ & $\{0<1\}$ & $\{0,1\}$ & $\cup$, $\cap$, $(\dash)^\complement$, $\emptyset$, $T\Sigma$ \\
$\DL$ & $\Pos$ & $\{0<1\}$ & $\{0<1\}$ & $\cup$, $\cap$, $\emptyset$, $T\Sigma$ \\
$\JSL$ & $\JSL$ & $\{0<1\}$  & $\{0<1\}$ & $\cup$, $\emptyset$ \\
$\Vect_{\Field}$ & $\Vect_{\Field}$ & $\Field$ & $\Field$ & $+$, $0$ \\
\hline
\end{tabularx}
\end{table}
\end{example}

\begin{rem} For the free-monoid monad $\MT=\MT_*$ on $\Set$ and $\F=$ all finite monoids, we recover an important result of Pippenger \cite{Pippenger1997}: the boolean
algebra of regular languages over $\Sigma$ is dual to the Stone space
$\widehat{\Sigma^*}$ of profinite words. In fact, in this one-sorted case the isomorphism \eqref{eq:iso} gives
$\Rec(\Sigma) \cong P(\widehat{\Sigma^*})$.
\end{rem}
\subsection{Syntactic \texorpdfstring{$\MT$}{T}-algebras}
In algebraic language theory, a recurring theme is to characterize languages by properties of their \emph{syntactic algebras}, i.e. their minimal recognizing algebras. For example, Sch\"utzenberger's theorem \cite{schuetzenberger} asserts that a regular language is star-free iff its syntactic monoid is aperiodic. Since the syntactic monoid of language can be computed directly from an accepting finite automaton, and aperiodicity is a decidable property of finite monoids, this implies that star-freeness is decidable.

It turns out that the language-theoretic concept of a syntactic algebra is intimately related to our purely algebraic notion of a unary presentation. The results of this subsection serve to put our concepts into the context of classical algebraic language theory; they are, however, not used in the sequel and may be skipped by readers interested only in the variety theorem and its applications.

\begin{definition} Let $L\colon T\FSigma\to O_\D$ be an $\F$-recognizable language. By a
\emph{syntactic 
$\MT$-algebra for $L$} is meant a 
 $\MT$-algebra $A_L$ together with a surjective $\MT$-homomorphism
$e_L\colon
\MT\FSigma\epito A_L$ (called a \emph{syntactic morphism} for $L$) such that 
\begin{enumerate}[(i)]
\item  $e_L$ recognizes $L$;
\item  $e_L$ 
factors through any surjective $\MT$-homomorphism $e\colon 
\MT\FSigma\epito A$ with $A\in \F$ that recognizes $L$, i.e.\ $e_L = h\o e$ for some $\MT$-homomorphism $h\colon A\epito A_L$.
\[
\xymatrix{
\MT\FSigma \ar@{->>}[dr]_{e_L} \ar@{->>}[r]^e  & A \ar@{-->>}[d]^{h}\\
& A_L
}
\]
\end{enumerate}
\end{definition}

\begin{rem}
The universal property determines $A_L$ and $e_L$ uniquely up to isomorphism. Moreover, it implies that $A_L\in \F$ because $\F$ is closed under quotients.
\end{rem}

\begin{example}\label{ex:synmon}
Let $\MT=\MT_*$ on $\Set$ with $O_\Set=\{0,1\}$. The syntactic monoid of a 
recognizable language $L\colon 
\Sigma^*\to \{0,1\}$ is formed as follows. Let $\equiv_L$ be the monoid congruence on $\Sigma^*$ given by 
\[ x\equiv_L x'\quad\text{iff}\quad L(yxz)=L(y x' z)\text{ for all $y,z\in\Sigma^*$}.\]
Then the quotient monoid $\Sigma^*/\mathord{\equiv_L}$ is finite and the projection
\[e_L\colon \Sigma^*\epito  
\Sigma^*/\mathord{\equiv_L}\]
 is a syntactic morphism for $L$.
\end{example}
The above definition of $\equiv_L$ involves for each $y,z\in \Sigma^*$ the unary operation $x\mapsto yxz$ on $\Sigma^*$. It can be expressed as the composite of the two unary operations $x\mapsto yx$ and $x\mapsto xz$ appearing in the unary presentation of $\Sigma^*$ in Example \ref{ex:unpres_monoids}. This is no coincidence: we will show below that one can always derive a syntactic congruence from a unary presentation, and vice versa.
\begin{notation}\label{not:syncong}
Let $\U$ be a set of unary operations on $T\FSigma$, and denote by ${\ol{\U}}$ its closure under composition and identity morphisms $\id\colon (T\FSigma)_s\to(T\FSigma)_s$. 
\begin{enumerate}
\item Suppose that $\D$ is a variety of algebras. The $S$-sorted equivalence relation $\equiv_{\U,L}$ on $\under{T\FSigma}$ is defined as follows: for any two elements $x,x'\in \under{T\FSigma}_s$, put
\[ x\equiv_{\U,L,s} x' \quad\text{iff}\quad L_{t}\o u(x) = L_{t}\o u(x') \text{ for all $u\colon (T\FSigma)_s\to (T\FSigma)_{t}$ in $\ol{\U}$}. \]
\item Suppose that $\D$ is a variety of ordered algebras. The $S$-sorted preorder $\leq_{\U,L}$ on $\under{A}$ is defined as follows: for any two elements $x,x'\in \under{T\FSigma}_s$, put
\[ x\leq_{\U,L,s} x' \quad\text{iff}\quad L_{t}\o u(x) \leq L_{t}\o u(x') \text{ for all $u\colon (T\FSigma)_s\to (T\FSigma)_{t}$ in $\ol{\U}$}. \]
\end{enumerate} 
\end{notation}
\begin{rem}
If $\D$ is a variety of algebras then $\equiv_{\U,L}$ is a congruence on $T\FSigma$ in $\D^S$. Indeed, $\equiv_{\U,L,s}$ is the intersection of the kernels of the morphisms $L_t\o u$, where $u\colon (T\FSigma)_s\to (T\FSigma)_{t}$ in $\U$. Analogously, if $\D$ is a variety of ordered algebras, $\leq_{\U,L}$ is a stable preorder on $T\FSigma$ in $\D^S$. We write 
\[e_L\colon T\FSigma\epito T\FSigma/\mathord{\equiv_{\U,L}}\quad\text{and}\quad e_L\colon T\FSigma\epito T\FSigma/\mathord{\leq_{\U,L}},\]
respectively, for the corresponding quotients in $\D^S$; see \ref{app:congruences} and \ref{app:stablepreorders}.
\end{rem}

\begin{theorem}\label{thm:synalgconst}
For any set $\U$ of unary operations on $T\FSigma$, the following are equivalent:
\begin{enumerate}[(i)]
\item  $\U$ is a unary presentation of the free $\MT$-algebra $\MT\FSigma$.
\item Every $\F$-recognizable language $L\colon T\FSigma\to O_\D$ has a syntactic $\MT$-algebra, and $e_L$ carries a quotient of $\MT\FSigma$ in $\Alg{\MT}$ that forms a syntactic morphism for $L$.
\end{enumerate}
\end{theorem}
To prepare the proof, we need to establish an important property of the output object $O_\D$:

\begin{lemma}\label{lem:cosepord}
\begin{enumerate}
\item If $\D$ is a variety of algebras, then $O_\D$ is a cogenerator of $\D_f^S$: for any object $D\in \D^S_f$ and any 
two elements $x,y\in
\under{D}_s$ with $s\in S$ we have
\[
  x= y \quad\text{iff}\quad \forall (D\xra{k} O_\D): k(x)= k(y).
\]
\item If $\D$ is a variety of ordered algebras, then $O_\D$ is an order-cogenerator of $\D_f^S$: for any object $D\in \D^S_f$ and any 
two elements $x,y\in
\under{D}_s$ with $s\in S$ we have
\[
  x\leq y \quad\text{iff}\quad \forall (D\xra{k} O_\D): k(x)\leq k(y).
\]
\end{enumerate}
\end{lemma}
\begin{proof}
To prove (b), suppose first that $S=1$. Given $D\in \D_f$, its dual object $PD\in \C_f$ is finite and thus finitely generated, so there exists a surjective morphism (i.e. a strong epimorphism) $e\colon \coprod_{i\in I} \one_\C \epito PD$ in $\C_f$, where $I=\under{PD}$. The dual morphism $m\colon D\monoto \prod_{i\in I} O_\D$ in $\D_f$ is a strong monomorphism and thus order-reflecting. Therefore $x\not\leq y$ in $\under{D}$ implies $m(x)\not\leq m(y)$, and thus $\pi_i m(x)\not \leq \pi_i m(y)$ for some $i\in I$ because the product projections $\pi_i\colon \prod_i O_\D\to O_\D$ are jointly order-reflecting. This shows that the morphism $k := \pi_i \o m\colon D\to O_\D$ separates $x$ and $y$, as desired.

Now let $S$ be arbitrary and $x\not\leq y \in \under{D}_s$. By the above argument there exists a morphism $k_s\colon D_s\to O_\D$
in $\D$ with $k_s(x)\not\leq k_s(y)$. For any sort $t\neq s$, pick an arbitrary
morphism $k_{t}\colon D_{t}\to O_\D$. Such a morphism exists because, by our
Assumption~\ref{ass:catframework}(ii) that the signature of $\C$ contains a constant, we
dually have a morphism $\one_\C \to PD_{t}$ in $\C_f$. Thus $k\colon D\to O_\D$ is
a morphism in $\D^S$ with $k(x)\not\leq k(y)$.

The proof of (a) is completely analogous, replacing inequations by equations.
\end{proof}

\begin{proof}[Proof of Theorem \ref{thm:synalgconst}]
We prove the case where $\D$ is a variety of ordered algebras; for the unordered case, just replace $\leq_{\U,L}$ by $\equiv_{\U,L}$ and  inequations by equations throughout the proof. For brevity, put $A_{L} := T\FSigma/\mathord{\leq_{\U,L}}$.

\medskip
\noindent\textbf{(i)$\To$(ii)} Suppose that $\U$ is a unary 
presentation of $\MT\FSigma$, and 
let 
$L\colon T\FSigma\to O_\D$ be an $\F$-recognizable 
language. 
\begin{enumerate}[(a)]
\item We show that there exists a morphism $p_{L}\colon A_{L}\to O_\D$ in $\D^S$ with 
$L=p_{L}\o 
e_{L}$, using the homomorphism theorem. Let $x,y\in 
\under{T\FSigma}_s$ 
with 
$e_{L}(x)\leq e_{L}(y)$, i.e. $x\leq_{\U,L} y$. Since $\ol \U$ contains 
all identities, putting  $u:=\id_{(T\FSigma)_s}$ in the definition of $\leq_{\U,L}$ (see Notation
\ref{not:syncong})  yields 
$L(x)\leq L(y)$. The homomorphism theorem gives the desired $p_{L}$.
\item Since $L$ is $\F$-recognizable, there is a surjective $\MT$-homomorphism 
$e\colon \MT\FSigma\epito (A,\alpha)$ into a $\MT$-algebra $(A,\alpha)\in \F$ and a morphism $p\colon A\to 
O_\D$ in $\D^S$ with $L=p\o e$. Thus $e$ is an $\F$-quotient, and since $\U$ is a unary presentation of $\MT\FSigma$, it follows that $e$ is a $\U$-quotient, i.e. every $u\colon (T\FSigma)_s\to (T\FSigma)_{t}$ in $\U$ has a lifting $u_A\colon A_s\to A_{t}$ 
along 
$e$ with $e\o u = u_A\o e$. We claim that there exists a morphism $h\colon 
A\to 
A_L$ in $\D^S$ with 
$e_L = h\o e$. This follows from the homomorphism theorem: let 
$x,y\in\under{T\FSigma}_s$ with $e(x)\leq e(y)$. Then, for all sorts $t$ and $u\colon 
(T\FSigma)_s\to 
(T\FSigma)_{t}$ in $\U$,
\[
  L\o u(x) 
  = 
  p\o e \o u(x) 
  = 
  p\o u_A\o e(x) 
  \leq 
  p\o u_A\o e(y) 
  =  
  p\o e \o u(y) 
  = 
  L\o u(y).
\]
Thus $x\leq_{\U,L} y$, or equivalently $e_L(x)\leq e_L(y)$, and the homomorphism 
theorem gives the desired $h$.
\item We show that (I) 
$e_{L}$ is $\F$-refinable, and (II) $e_L$ is a $\U$-quotient. This implies the claim: since
$\U$ is a unary presentation of $\MT\FSigma$ it follows that the morphism is $e_L$ is a $\MT$-quotient, i.e. it can be equipped with a 
$\MT$-algebra 
structure making $e_L$ a $\MT$-homomorphism. And part (a) and (b) show that 
$e_L$ recognizes $L$ and has the universal property of a syntactic morphism.

(I) is clear since $e_L = h\o e$ by (b) and $e\colon \MT\FSigma\epito (A,\alpha)$ is an $\F$-quotient of $\MT\FSigma$. For (II) we need to show that for all $u\colon (T\FSigma)_s\to (T\FSigma)_t$ in $\U$ and all $x,y\in 
\under{T\FSigma}_s$ with $x\leq_{\U,L} y$, we have $u(x)\leq_{\U,L} u(y)$. Note that 
for all sorts $t'$ and all $u'\colon (T\FSigma)_{t}\to (T\FSigma)_{t'}$ in $\U$ we have $u'\o u\in 
\ol \U$ because $\ol \U$ is closed under composition. Thus $x\leq_{\U,L} 
y$ 
implies 
\[ L\o (u'\o u) (x) \leq L\o (u'\o u) (x)\quad\text{for all sorts $t'$ and $u'\colon 
(T\FSigma)_{t}\to (T\FSigma)_{t'}$ in $\U$}, \]
which means precisely that $u(x)\leq_{\U,L} u(y)$. 
\end{enumerate}
\textbf{(ii)$\To$(i)} Suppose that, for any $\F$-recognizable language $L$ 
over $\Sigma$, the morphism $e_L\colon T\FSigma\epito A_L$ is a $\MT$-quotient that forms a syntactic morphism for $L$. We verify that 
$\U$ is a unary 
presentation. Thus let $e\colon T\FSigma\epito A$ be an $\F$-refinable quotient in 
$\D^S$. We need to show that $e$ is a $\MT$-quotient iff $e$ is a $\U$-quotient.

\smallskip \noindent For the ``only if'' direction, suppose that $e$ is a $\MT$-quotient, i.e.  $A$ carries a $\MT$-algebra 
structure $(A,\alpha)$ making $e\colon 
\MT 
\FSigma\epito (A,\alpha)$ a $\MT$-homomorphism. We need to show that $e$ is a $\U$-quotient: every 
unary operation $u\colon (T\FSigma)_s\to (T\FSigma)_{t}$ in $\U$ has a lifting 
along $e$, i.e. there exists a morphism $u_A\colon 
A_s\to A_{t}$ with $e\o u = u_A\o e$. To this end, we again apply the homomorphism theorem.
For any morphism $k\colon A\to O_\D$ in $\D^S$ we have the $\F$-recognizable language 
$L_k 
:= k\o e: T\FSigma\to O_\D$, which by hypothesis has the syntactic 
morphism 
$e_{L_k}\colon 
\MT\FSigma \epito A_{L_k}$. Since $e_{L_k}$ recognizes $L_k$, there exists a 
morphism $p_{L_k}\colon A_{L_k}\to O_\D$ with $L_k= 
p_{L_k} \o e_{L_k}$. Furthermore, the universal property of the syntactic 
morphism $e_{L_k}$ gives a unique $\MT$-homomorphism $h_k\colon A\epito A_{L_k}$ with 
$e_{L_k} = h_k \o e$. Then for all
$x,y\in\under{T\FSigma}_s$ we have the following implications:
\begin{align*}
 e(x) \leq e(y) ~&\To~ \forall(k\colon A\to O_\D): e_{L_k}(x) \leq e_{L_k}(y) & (e_{L_k} 
 = h_k \o e)\\
 ~&\Lra~ \forall k: x \leq_{\U,L_k} y & (\text{def. $e_{L_k}$})\\
 ~&\To~ \forall k: L_k\o u(x) \leq L_k\o u(y) & (\text{def. $\leq_{\U,L_k}$})\\
 ~&\Lra~ \forall k: k\o e \o u(x) \leq k\o e\o u(y) & (\text{def. $L_k$})\\
 ~&\Lra~ e \o u(x) \leq e\o u(y) & (\text{Lemma \ref{lem:cosepord}}).
\end{align*}
The homomorphism theorem  gives the desired lifting $u_A$. Thus $e$ is a $\U$-quotient.

\smallskip \noindent Conversely, for the ``if'' direction suppose that $e$ is a $\U$-quotient. Since $e$ is $\F$-refinable, we can choose a quotient $\ol e\colon \MT\FSigma\epito B$ with $B\in \F$ and a morphism $f\colon B\to A$ with $e=f\o \ol e$. Moreover, by hypothesis every $u\colon 
(T\FSigma)_s\to (T\FSigma)_{t}$ in $\U$ 
has 
a lifting $u_A\colon A_s\to A_{t}$ along $e$. It is our task to show that $A$ is a $\MT$-quotient, i.e. carries a 
$\MT$-algebra structure making $e$ a $\MT$-homomorphism. For each $k\colon A\to 
O_\D$ in $\D^S$ the language $L_k := k\o e\colon T\FSigma\to O_\D$ is $\F$-recognizable; indeed, we have $L_k = k\o f \o \ol e$, so $L_k$ is recognized by the $\MT$-homomorphism $\ol e$. By hypothesis, $L_k$ has the syntactic morphism $e_{L_k}\colon \MT\FSigma\epito 
A_{L_k}$. Since $e_{L_k}$ recognizes $L_k$, there exists a morphism $p_{L_k}\colon 
A_{L_k}\to 
O_\D$ with $L_k = p_{L_k}\o e_{L_k}$.

We claim that $e_{L_k}$ factors through $e$. To see this, we use the 
homomorphism 
theorem. Given $x,y\in \under{T\FSigma}_s$ with $e(x)\leq e(y)$, we have for all 
$u\colon (T\FSigma)_s\to (T\FSigma)_{t}$ in $\U$:
\begin{align*}
L_k\o  u(x) &= k\o e \o u(x) & \text{(def. $L_k$)}\\
&= k\o u_A\o e(x) & \text{(def. $u_A$)}\\
&\leq k\o u_A \o e(y) & \text{($e(x)\leq e(y)$)}\\
&= k\o e \o u(y) & \text{(def. $u_A$)}\\
&= L_k\o  u(y) & \text{(def. $L_k$)}.
\end{align*}
Thus $x\leq_{\U,L_k} y$, or equivalently $e_{L_k}(x)\leq e_{L_k}(y)$. The 
homomorphism 
theorem yields a morphism $h_k$ with $e_{L_k} = h_k \o e$.

We are ready to define the desired $\MT$-algebra structure $(A,\alpha)$ on $A$ for which 
$e$ is a $\MT$-homomorphism. Since $T$ preserves epimorphisms, it suffices to
find a 
morphism $\alpha\colon TA\to A$ in $\D^S$ making the following square commute, see \ref{app:quotients_of_talgs}:
\[
\xymatrix{
TT\FSigma \ar[r]^{\mu_\FSigma}  \ar@{->>}[d]_{Te} & T\FSigma \ar@{->>}[d]^e \\
TA \ar[r]_\alpha & A
}
\]
 The construction of $\alpha$  once 
again 
rests on the homomorphism theorem.
The proof is illustrated by the diagram below, where $\alpha_{L_k}$ is the 
$\MT$-algebra structure of $A_{L_k}$.
\[
\xymatrix{
TT\FSigma \ar[r]^{\mu_\FSigma} \ar@{->>}[d]_{Te} 
\ar@/_3em/@{->>}[dd]_{Te_{L_k}} & T\FSigma \ar@{->>}[d]^e \ar[dr]^{L_k} & \\
TA \ar@{-->}[r]_\alpha \ar@{->>}[d]_{Th_k} & A \ar@{->>}[d]^{h_k} \ar[r]^k & 
O_\D\\
TA_{L_k} \ar[r]_{\alpha_{L_k}} & A_{L_k} \ar[ur]_{p_{L_k}}
}
\]
For all $x,y\in \under{TT\FSigma}_s$ with $Te(x)\leq Te(y)$ we have 
\begin{align*}
 k\o e\o  \mu_\FSigma(x) &= L_k \o  \mu_\FSigma(x) & (\text{def.\ $L_k$})\\
 &= p_{L_k} \o e_{L_k} \o \mu_\FSigma(x) & (\text{def. $p_{L_k}$, $e_{L_k}$})\\
 &= p_{L_k} \o \alpha_{L_k}\o Te_{L_k}(x) & (\text{$e_{L_k}$ is 
 $\MT$-hom.})\\
 &=  p_{L_k} \o \alpha_{L_k}\o Th_k \o Te(x) & (\text{def.\ $h_k$})\\
 &\leq  p_{L_k} \o \alpha_{L_k}\o Th_k \o Te(y) & (\text{$Te(x)\leq Te(y)$})\\
 &= \cdots & (\text{compute backwards})\\ 
 &= k\o e\o  \mu_\FSigma(y).
\end{align*}
Since this holds for all $k\colon A\to O_\D$, Lemma \ref{lem:cosepord} shows that 
$e\o \mu_\FSigma(x)\leq e\o 
\mu_\FSigma(y)$. Thus 
the homomorphism theorem yields the desired $\MT$-algebra structure 
$\alpha$. Therefore $e$ is a $\MT$-quotient. 

This concludes the proof that $\U$ forms a unary presentation of $\MT\FSigma$.
\end{proof}

\begin{example}\label{ex:syncong}
Let $\MT=\MT_\ast$ on $\Set$. Consider the the unary presentation of the free monoid $M=\Sigma^*$ as in Example \ref{ex:unpres_monoids} and let $L\seq \Sigma^*$ be a regular (= monoid-recognizable) language. Then $\equiv_{\U,L}$ is precisely the monoid congruence $\equiv_L$ in Example \ref{ex:synmon}, and thus the above theorem shows that $e_L\colon \Sigma^*\epito \Sigma^*/\mathord{\equiv_L}$ is indeed a syntactic morphism for $L$. Similarly, for the unary presentations in Example \ref{ex:unpres_omegasem} and \ref{ex:unpres_monads}, the theorem gives a description of the syntactic $\omega$-semigroup for an $\infty$-recognizable language \cite{pinperrin04}, and of the syntactic $\MT$-algebra for a monad $\MT$ on $\Set^S$ and a $\MT$-recognizable language \cite{boj15}. Since syntactic algebras are not used in the sequel, we leave the details to the interested reader.
\end{example}
Theorem \ref{thm:synalgconst} explains why syntactic algebras are traditionally presented as a key technique in work on Eilenberg theorems: they implicitly contain the concept of a unary presentation. However,  from a categorical point of view the latter are the ``heart of the matter'', and for deriving Eilenberg-type correspondences it is usually easier to directly work with unary presentations in lieu of considering syntactic algebras. We will demonstrate this in Section \ref{sec:applications}.

\section{The Variety Theorem}\label{sec:varietythm}

In this section we establish our main result, the Variety Theorem for $\F$-recognizable languages. To introduce varieties of languages in our categorical setting, we first need a notion of \emph{language derivatives} extending the classical concept. Here our Assumption \ref{ass:catframework}(ii) comes to use:

\begin{notation}
Recall that the variety $\C$ is assumed to have a constant in the signature. 
Choosing a constant gives a natural transformation from $C_{\one_\C}\colon \C\to
\C$, the 
constant functor on $\one_\C\in\C$, to the identity functor $\Id_\C$. It dualizes 
to a 
natural 
transformation
\[\bot\colon \Id_{\hatD}\to C_{O_\D}.\]
We usually drop subscripts and write $\bot$ for each component $\bot_D\colon D\to O_\D$ ($D\in \hatD$).
\end{notation}

\begin{example}\label{ex:bot} The purpose of $\bot$
is to model the empty set. Recall that for our categories $\D$ of Example~\ref{ex:categories} we have
\[O_{\Set} = \{0,1\},\, O_\Pos = \{0<1\},\, O_{\JSL} = \{0<1\}\text{ and } O_{\Vect_\Field} = \Field,\] see Example \ref{ex:ocod}. In each case, we choose
$\bot\colon D\to O_\D$ to be the constant morphism with
value~$0$, i.e. the characteristic function of the empty set.
\end{example}

\begin{definition}\label{def:derivatives_preimages} Let $L\colon T\FSigma\to 
O_\D$ be a language. Then we define the following languages:
\begin{enumerate}
\item The \emph{derivative} $u^{-1}L$ of $L$ w.r.t. a unary operation $u\colon
(T\FSigma)_s\to(T\FSigma)_{t}$
  is the language over~$\Sigma$ defined by
\[
      \begin{cases}
        (u^{-1}L)_s = (T\FSigma)_s \xra{\;u\;} (T\FSigma)_{t} \xra{\;L_{t}\;} O_\D &\\
        (u^{-1}L)_r = (T\FSigma)_r \xra{\;\iota_\FSigma\;} (V\hatt\FSigma)_r \xra{\;V\bot\;}
        O_\D & \text{for $r \neq s$.}
      \end{cases}
    \]
\item The \emph{preimage} $g^{-1}L$ of $L$ under a $\MT$-homomorphism $g\colon
  \MT\FDelta\to \MT\FSigma$ is the language over~$\Delta$ defined by 
\[ g^{-1}L = (\,T\FDelta \xra{g} T\FSigma \xra{L} O_\D\,).\]
\end{enumerate}
\end{definition}

\begin{rem}
If $S=1$, we simply have $u^{-1}L = L\o u$ and the natural transformation $\bot$ is not used. Therefore, in the single-sorted case our Assumption \ref{ass:catframework}(ii) can be dropped.
\end{rem}

\begin{example}\label{ex:derivatives}
\begin{enumerate}
\item Let $\MT=\MT_*$ on $\Set$. Consider the unary presentation $\U_\Sigma$ of the free monoid $M=\MT_\ast\Sigma = \Sigma^*$ as in Example \ref{ex:unpres_monoids}. Then a language $L\colon T_\ast \Sigma\to  O_\Set$, viewed as a subset $L\seq\Sigma^*$, has the following derivatives w.r.t. the operations in $\U_\Sigma$:
\[y^{-1}L = \{\,x\in\Sigma^*: yx\in L\,\}\quad\text{and}\quad Ly^{-1} = \{\,x\in\Sigma^*: xy\in L\,\}\] for $y\in \Sigma^*$. These are exactly the classical derivatives for languages of finite words. The preimage of $L$ under a monoid morphism $g\colon \Delta^*\to \Sigma^*$ is the language
\[ g^{-1}L = \{\, x\in\Delta^*: g(x)\in L \,\}.\]
\item Let $\MT=\MT_\infty$ on $\Set^2$. Consider the unary presentation $\U_\Sigma$ of the free $\omega$-semigroup $A=\MT(\Sigma,\emptyset)=(\Sigma^+,\Sigma^\omega)$ as in Example \ref{ex:unpres_omegasem}. A language $L\colon \MT_\infty(\Sigma,\emptyset)\to O_\Set = (\{0,1\},\{0,1\})$ may be viewed as the subset $L\seq \Sigma^+\cup \Sigma^\omega$ obtained as the union of the subsets $L_+$ and $L_\omega$ in its two sorts. Then $L$ has the following derivatives w.r.t.~the operations in $\U_\Sigma$:
\begin{align*}
y^{-1} L &= \{\,x\in \Sigma^+: yx\in L\,\}\\
Ly^{-1} &= \{\,x\in \Sigma^+: xy\in L\,\}\\
Lz^{-1} &= \{\,x\in \Sigma^+: xz \in L\,\}\\
\omega^{-1}L &= \{\,x\in \Sigma^+: x^\omega \in L\,\}\\
y^{-1}L_\omega &= \{\,z\in\Sigma^\omega: yz\in L\,\}
\end{align*}
for $y\in\Sigma^*$ and $z\in 
\Sigma^\omega$.  These are the derivatives for
$\infty$-languages studied by Wilke~\cite{wilke91}. The preimage of $L$ under an $\omega$-semigroup morphism $g\colon (\Delta^+,\Delta^\omega)\to (\Sigma^+,\Sigma^\omega)$ is the language
\[ g^{-1}L = \{ x\in \Delta^+\cup \Delta^\omega\colon g(x) \in L \}. \]
\item Let $\MT$ be an arbitrary monad on $\Set^S$, and consider the unary operations $[p]$ of Example \ref{ex:unpres_monads} presenting the free $\MT$-algebra $A=\MT\Sigma$. The induced
derivatives of a language $L\seq T\Sigma$ are the languages 
$p^{-1}L\seq 
T\Sigma$ given by 
\[(p^{-1}L)_s = \{\,x\in (T\Sigma)_s : [p](x)\in L_{t}\,\}\quad\text{and}\quad
(p^{-1}L)_r = \emptyset \text{ for $r\neq s$},\]
where $p\colon 1_{t}\to T(T\Sigma+1_s)$ 
ranges over polynomials over $T\Sigma$. These \emph{polynomial derivatives} were
 studied by Boja\'nczyk \cite{boj15}. The preimage of $L$ under a $\MT$-homomorphism $g\colon \MT\FDelta\to\MT\FSigma$ is the language
\[ g^{-1}L = \{x\in T\Delta : g(x)\in L \}.\]
\end{enumerate}
\end{example}

\begin{proposition}\label{prop:recclosed} Let $L\colon T\FSigma\to O_\D$ be an $\F$-recognizable language.
\begin{enumerate}
\item If $\U_\Sigma$ is a unary presentation of $\MT\FSigma$, every derivative $u^{-1}L$ with $u\in\U_\Sigma$ is $\F$-recognizable.
\item For any $\MT$-homomorphism $g\colon \MT\FDelta\to\MT\FSigma$, the preimage $g^{-1}L$ is $\F$-recognizable.
\end{enumerate}
\end{proposition}

\begin{proof}
Let $L\colon T\FSigma\to O_\D$ be an $\F$-recognizable language. By Theorem 
\ref{thm:reciso} there exists a morphism $\hat L\colon \hatt\FSigma\to O_\D$ 
in $\hatD^S$ with $L=V\hat L\o \iota_\FSigma$.
\begin{enumerate}[(a)]
\item Let $u\colon (T\FSigma)_s\to (T\FSigma)_{t}$ in $\U_\Sigma$,
  and take its continuous extension $\hat u$, see
  Lemma~\ref{lem:extend_fi}. Then we have the commutative diagrams below
  (where $r\neq s$ in the right-hand diagram). 
\[
\xymatrix{
(T\FSigma)_s \ar@/^2em/[rr]^{u^{-1}L} \ar[r]^u \ar[d]_{\iota_\FSigma} & 
(T\FSigma)_{t} \ar[r]^L \ar[d]^{\iota_\FSigma} & O_\D\\
V(\hatt\FSigma)_s \ar[r]_{V\hat u} & 
V(\hatt\FSigma)_{t} \ar[ur]_{V\hat L} &
}
\qquad
\xymatrix@C+1em{
(T\FSigma)_r \ar@/^2em/[rr]^{u^{-1}L} \ar[r]_{\iota_\FSigma} 
 & V(\hatt\FSigma)_r \ar[r]_{V\bot} & O_\D
}
\]
This shows that $u^{-1}L$ corresponds to the morphism
$\hatt\FSigma \to O_\D$ in $\hatD^S$ being $\hat L \cdot \hat u$ in sort $s$ and
$\bot$ in every other sort. By Theorem~\ref{thm:reciso}, $u^{-1}L$ is
$\F$-recognizable.

\item Let $g\colon \MT\FDelta\to\MT\FSigma$ be a $\MT$-homomorphism
  and $\hat g$ its continuous extension, see Lemma \ref{lem:gat}. Then
  we have the commutative diagrams below.
  \[
  \xymatrix{
    T\FDelta \ar[r]^g \ar[d]_{\iota_\FDelta} & T\FSigma 
    \ar[r]^L \ar[d]^{\iota_\FSigma} & O_\D\\
    V\hatt\FDelta \ar[r]_{V\hat g} & V\hatt\FSigma \ar[ur]_{V\hat L} & 
    }
  \]
  This shows that $g^{-1}L = L \cdot g$ corresponds to
  $\hat L \cdot \hat g\colon \hatt\FSigma \to O_\D$, whence $g^{-1}L$ is
  $\F$-recognizable by Theorem~\ref{thm:reciso}.\qedhere
\end{enumerate}
\end{proof}

\takeout{
\begin{rem}\label{rem:admissiblesubobject}
  Recall the isomorphism  $\Rec_\F(\Sigma) \cong \prod_s
  P(\hatt\FSigma)_s$ of
  Remark~\ref{rem:reg-as-C}. In the following we study subobjects 
  $W_\Sigma\seq\Rec_\F(\Sigma)$ in $\C$. However, for technical reasons we  
  restrict ourselves
  to subobjects of the form $\prod_s m_s\colon \prod_s (W_\Sigma')_s 
  \monoto \prod_s P(\hatt\FSigma)_s$, where $m_s\colon (W_\Sigma')_s\monoto 
  P(\hatt\FSigma)_s$ is a monomorphism in $\C$:
      \[
        \xymatrix@C+0,5em@R-1em{
          W_\Sigma \ar@{ >->}[r]^\seq  \ar[d]_{\cong} & \Rec_\F(\Sigma)
          \ar[d]^{\cong} \\
          \prod_s (W_\Sigma')_s \ar@{ >->}[r]_{\prod_s m_s} & \prod_s P(\hatt\FSigma)_s
        }
      \] 
      Such subobjects are called 
        \emph{admissible}. Clearly, for $S=1$, any 
  subobject of $\Rec_\F(\Sigma)$ is admissible. More importantly, if $\C$ is one of the categories of Example 
  \ref{ex:categories} and $\U_\Sigma$ contains all identity morphisms, one can show that any subobject $W_\Sigma\seq \Rec_\F(\Sigma)$ closed under derivatives (i.e. $L\in W_\Sigma$ implies $u^{-1}L\in 
  W_\Sigma$ for all $u\in \U_\Sigma$) is admissible. Thus, in these cases the admissibility condition in Definition \ref{def:variety}.1 below can be dropped. For Definition \ref{def:variety}.2, recall from   the previous section that we work with a fixed class $\A\seq \Set_f^S$ of alphabets.
\end{rem}
}

\begin{definition}
Let $\mathcal{L} = (\,L^s\colon T\FSigma\to O_\D\,)_{s\in S}$ be an $S$-indexed family of languages over $\Sigma$. The \emph{diagonal} of $\mathcal L$ is the language $\Delta\mathcal{L}$ over $\Sigma$ given by \[(\Delta\mathcal{L})_s = L^s_s\colon (T\FSigma)_s\to O_\D\quad\text{for all $s\in S$}.\]
\end{definition}

\begin{notation}
Recall that we work with a fixed class $\A\seq \Set_f^S$ of alphabets.
Fix for each $\Sigma\in \A$ a unary presentation $\U_\Sigma$ of the free $\MT$-algebra $\MT\FSigma$.
\end{notation}

\begin{definition}[Local variety of languages]\label{def:localvariety}
A \emph{local variety of languages} over
    $\Sigma\in \A$ is a subobject
  \[W_\Sigma\seq \Rec_\F(\Sigma)\] in $\C$ closed under $\U_\Sigma$-derivatives and diagonals. That is,
\begin{enumerate}[(i)]
\item for each $L\in W_\Sigma$ and $u\in\U_\Sigma$, one has $u^{-1}L\in W_\Sigma$;
\item for each $S$-indexed family  $\mathcal{L} = (\,L^s\,)_{s\in S}$ of languages in $W_\Sigma$, one has $\Delta \mathcal{L}\in W_\Sigma$.
\end{enumerate}
\end{definition}

\begin{rem}\label{rem:diagclosed} There are two important cases where the closure under diagonals in the above definition is trivially satisfied and thus can be dropped:
\begin{enumerate}
\item  For $S=1$, clearly every subobject $W_\Sigma\seq\Rec_\F(\Sigma)$ is closed under diagonals.
\item If $\C$ is one of the categories of Example 
  \ref{ex:categories} and $\U_\Sigma$ contains all identity morphisms $\id_{(T\FSigma)_s}: (T\FSigma)_s\to (T\FSigma)_s$, then every subobject $W_\Sigma\seq\Rec_\F(\Sigma)$ closed under $\U_\Sigma$-derivatives is closed under diagonals. To see this, let $\mathcal{L} = (L^s)_{s\in S}$ be an $S$-indexed family of languages in $W_\Sigma$. Since $W_\Sigma$ is closed under $\U_\Sigma$-derivatives, the language $\id_{(T\FSigma)_s}^{-1}L^s$ lies in 
$W_\Sigma$ for each $s\in S$. Recall from Example \ref{ex:bot}
that $\bot$ has been chosen as the zero map. Therefore the derivative $\id_{(T\FSigma)_s}^{-1}L^s$ agrees with $L^s$ in sort $s$,
and 
is the zero map (corresponding to the empty set) in all other sorts. Now observe that for the categories $\C=\BA$, $\DL$, $\JSL$, 
the set $W_\Sigma\seq\Rec_\F(\Sigma)$, being a subobject in $\C$, is closed under union by Example \ref{ex:c_operations}. This implies that the diagonal 
language $\Delta\mathcal L = 
\bigcup_s 
\id_{(T\FSigma)_s}^{-1}L^s$ lies in $W_\Sigma$. Analogously for 
$\C=\Vect_{\Field}$ where $W_\Sigma$, being a linear subspace of $\Rec_\F(\Sigma)$, is closed under addition of languages.
\end{enumerate}
Since condition 2 above is satisfied in all our applications in Section \ref{sec:applications}, the varieties considered there will not need closure under diagonals as an explicit property.
\end{rem}

\begin{example}
In contrast to the previous remark, subobjects of $\Rec(\Sigma)$ are not always closed under diagonals. For example, consider the identity monad $\MT=\Id$ on $\Set^2$. Then $\MT$-algebras are just two-sorted sets. Fix a finite set $X\neq \emptyset$ and consider the two-sorted alphabet $\Sigma=(X,X)\in \Set_f^2$.  The boolean algebra of $\MT$-recognizable languages over $\Sigma$ is given by
\[ \Rec(\Sigma)=P(X)\times P(X) = \{\, (K,L): K,L\seq X \,\}. \]
The free $\MT$-algebra $\MT(\Sigma)=\Sigma$  has the trivial unary presentation $\U_{\Sigma}=\emptyset$. Then 
\[ W_\Sigma = \{\, (L,L) : L\seq X \,\} \]
is a boolean subalgebra of $\Rec(\Sigma)$ that is (trivially) closed under $\U_{\Sigma}$-derivatives, but it is not closed under diagonals: for any two distinct subsets $K,L\seq X$ we have $(K,K), (L,L)\in W_\Sigma$, but the diagonal $(K,L)$ does not lie in $W_\Sigma$. 
\end{example}
Next, we characterize the closure of a local variety under diagonals in terms of morphisms in $\C$ (see Lemma \ref{lem:diagonalclosed} below). First, a technical preparation:

\begin{rem}\label{rem:ptsigma}
\begin{enumerate}
\item For each $\Sigma\in\Set_f^S$ and each sort $s$ we have
\begin{align*} 
\under{P(\hatt\FSigma)_s} &\cong \C(\one_\C,P(\hatt\FSigma)_s) \\
&\cong 
\hatD((\hatt\FSigma)_s,O_\D)\\
&\cong  \{\,(T\FSigma)_s\xra{L_s} O_\D : L \in \Rec_\F(\Sigma) \,\}.
\end{align*}
The last bijection is given by $\hat f\mapsto V\hat f\o 
\iota_\FSigma$. To see this, 
observe that for each $\F$-recognizable language $L\colon
T\FSigma\to O_\D$, the language $L'\colon T\FSigma\to O_\D$ with 
$L'_s = L_s$ and $L_r = V\bot\o \iota_\FSigma$ for $r\neq s$ is also 
$\F$-recognizable. Indeed, if $L$ is recognized by the $\MT$-homomorphism $h\colon \MT\FSigma\to A$ via $p\colon A\to O_\D$, then $L'$ is recognized by $h$ via $p'$, where $p'_s = p_s$ and $p'_r = \bot$ for $r\neq s$, using the naturality of $\bot$. From this and Theorem
\ref{thm:reciso} the bijection immediately follows.

Therefore, we assume from now on that $P(\hatt\FSigma)_s$ is carried by the set $\{\,L_s: 
L\in\Rec_\F(\Sigma)\,\}$. With this identification, the isomorphism 
$\Rec_\F(\Sigma)\cong \prod_s P(\hatt\FSigma)_s$ of Remark \ref{rem:reg-as-C} 
maps 
an $\F$-recognizable language $L\colon T\FSigma\to O_\D$ to the $S$-tuple $((T\FSigma)_s\xra{L_s} 
O_\D)_{s\in S}$ of its components.
\item Every subobject $W_\Sigma\seq \Rec_\F(\Sigma)$ in $\C$ contains the ``empty 
language'', i.e.\ the language with $V\bot\o \iota_\FSigma \colon 
(T\FSigma)_r 
\to O_\D$ in each 
sort $r$. Indeed, by the 
definition of $\bot$ (the dual of the natural transformation choosing a 
constant) this language is precisely the constant 
in $\Rec_\F(\Sigma)\cong 
\prod_s P(\hatt\FSigma)_s$, and every subobject of $\Rec_\F(\Sigma)$ in $\C$ must contain the 
constant.
\end{enumerate}
\end{rem}

\begin{lemma}\label{lem:diagonalclosed}
A subobject $W_\Sigma\seq \Rec_\F(\Sigma)$ is closed under diagonals iff it has the form $\prod_s m_s\colon \prod_s (W_\Sigma')_s 
    \monoto \prod_s P(\hatt\FSigma)_s$ where $m_s\colon (W_\Sigma')_s\monoto 
    P(\hatt\FSigma)_s$ is a monomorphism in $\C$. 
\end{lemma}

\begin{proof}
Given any subobject $W_\Sigma\seq \Rec_\F(\Sigma)$ and any sort $s$, consider the morphism
\[ \xymatrix{ W_\Sigma \ar@{>->}[r]^<<<<<\seq & \Rec_\F(\Sigma) \ar[r]^<<<<<\cong & \prod_s P(\hatt\FSigma)_s \ar[r]^{\pi_s} & P(\hatt \FSigma)_s } \]
where $\pi_s$ is the product projection, and factorize it into a surjective morphism $e_s\colon W_\Sigma \epito (W_\Sigma')_s$ follows by an injective morphism
$m_s\colon
(W_\Sigma')_s 
\monoto P(\hatt\FSigma)_s$:
\begin{equation}\label{eq:wsigmasquare}
  \xymatrix@R-1em{
  W_\Sigma \ar@{ >->}[r]^{\seq} \ar@{->>}[dd]_{e_s} & \Rec_\F(\Sigma) 
  \ar[d]^\cong \\
  & \prod_s P(\hatt\FSigma)_s \ar[d]^{\pi_s} \\
  (W_\Sigma')_s \ar@{ >->}[r]_{m_s} & P(\hatt\FSigma)_s
  }
\end{equation}
By Remark \ref{rem:ptsigma}.1 above, $(W_\Sigma')_s$ is (up to isomorphism) carried by the set
\[
  \{\,(T\FSigma)_s\xra{L_s} O_\D \;:\; L\in W_\Sigma\,\}.
\]
From \eqref{eq:wsigmasquare} we get the commutative square below, where $e=\langle e_s \rangle_{s\in 
S}\colon W_\Sigma\to \prod_s (W'_\Sigma)_s$:
\[
\xymatrix@C+3em@R-1em{
W_\Sigma \ar@{ >->}[r]^{\seq} \ar[d]_{e} & \Rec_\F(\Sigma) \ar[d]^\cong \\
\prod_s (W_\Sigma')_s \ar@{ >->}[r]_{\prod_s m_s} & \prod_s P(\hatt\FSigma)_s
}
\]
Clearly $e$ is monic. Out task is to show that $e$ is surjective (i.e. an isomorphism) iff $W_\Sigma$ is closed under diagonals. To this end, note that every element of $\prod_s(W_\Sigma')_s$ is an $S$-tuple $(L^s_s\colon (T\FSigma)_s\to O_\D)_{s\in S}$ for some $S$-indexed family $\mathcal{L}=(L^s)_{s\in S}$ of languages in $W_\Sigma$. That $e$ is surjective means precisely that, for any such family $\mathcal L$, there exists a language $L\in W_\Sigma$ with $L_s = L^s_s$ for all $s$, i.e. with $L=\Delta\mathcal{L}$. This concludes the proof.
\end{proof}

\begin{theorem}[Local Variety Theorem]\label{thm:localeilenberg}
The lattice of local varieties of languages over $\Sigma\in\A$ (ordered by inclusion) is isomorphic to the lattice of local pseudovarieties of $\Sigma$-generated 
$\F$-algebras form isomorphic lattices.
\end{theorem}
To prepare the proof, we need an auxiliary lemma:
\begin{lemma}\label{lem:localvar}
 Suppose that 
$W_\Sigma\seq \Rec_\F(\Sigma)$ is a subobject closed under diagonals,
represented by subobjects $m_r\colon (W_\Sigma')_r \monoto P(\hatt\FSigma)_r$ 
($r\in S$) in $\C$, see Lemma \ref{lem:diagonalclosed}. Let $u\colon (T\FSigma)_s\to (T\FSigma)_{t}$ in $\U_\Sigma$ and $\hat u\colon (\hatt\FSigma)_s\to (\hatt\FSigma)_{t}$ its continuous 
  extension, see Lemma 
  \ref{lem:extend_fi}. Then the 
following statements are equivalent:
\begin{enumerate}[(i)]
  \item $u^{-1}L\in W_\Sigma$ for all $L\in W_\Sigma$.
  \item There exists a morphism $u'$ in $\C$ making the following square commute:
    \begin{equation}\label{eq:derclos}
      \vcenter{
        \xymatrix{
          (W_\Sigma')_{t} \ar[r]^{u'} \ar@{ >->}[d]_{m_{t}} & (W_\Sigma')_{s}
          \ar@{ >->}[d]^{m_s}\\ P(\hatt\FSigma)_{t} \ar[r]_{P\hat u} &
          P(\hatt\FSigma)_s.
        }
      }
    \end{equation}
\end{enumerate}
In particular, $W_\Sigma$ is a local variety iff a morphism
$u'$ with
\eqref{eq:derclos} exists for every $u \in \U_\Sigma$.
\end{lemma}

\begin{proof}
Recall from Remark \ref{rem:ptsigma}.1 that $P(\hatt \FSigma)_r$ is, up to 
isomorphism, carried by the set $\{\,L_r: L\in \Rec_\F(\Sigma)\,\}$, and 
$(W_\Sigma')_r$ by 
the subset $\{\,L_r: L\in W_\Sigma\,\}$. From the definition of $\hat u$ it
follows that $P\hat u$ takes an 
element $L_{t}$ of $P(\hatt\FSigma)_{t}$ to $L_{t}\o u$.
Thus (ii) is equivalent to the statement that
$L_{t}\o u\in (W_\Sigma')_{s}$ for all $L\in W_\Sigma$. From 
this observation the implication (i)$\To$(ii) follows immediately, since
$(u^{-1}L)_s = L_{t}\o u$. 

Conversely, suppose that (ii) holds, and let $L\in W_\Sigma$. By the above 
argument, we have $L_{t}\o u\in (W_\Sigma')_{s}$. Moreover, by Remark 
\ref{rem:ptsigma}.3 the 
``empty language'' with $V\bot\o \iota_\FSigma$ in each sort lies in 
$W_\Sigma$. The closure of $W_\Sigma$
under diagonals thus implies that 
the language with $L_{t}\o u$ in sort $s$ and $V\bot\o \iota_\FSigma$ in 
all sorts $r\neq s$ lies in $W_\Sigma$. But this is precisely the derivative 
$u^{-1}L$, which proves (ii)$\To$(i).
\end{proof}

\begin{proof}[Proof of the Local Variety Theorem]
Duality + Local Reiterman! Suppose that $W_\Sigma\seq \Rec_\F(\Sigma)$ is
a subobject in $\C$ closed under diagonals, represented by a subobject
\[
  m = \left(\xymatrix{(W_\Sigma')_s \ar@{ >->}[r]^-{m_s} &
      P(\hatt\FSigma)_s}\right)
  _{s\in S}
\]
in $\C^S$, see Lemma \ref{lem:diagonalclosed}. From Lemma \ref{lem:localvar} and \ref{lem:profinite_algs}, it follows via duality that $W_\Sigma$ forms a local 
variety of languages iff 
the dual quotient
\[
  \left(
    \xymatrix@+2em{(\hatt\FSigma)_s \ar[r]^-\cong & P^{-1}P
    (\hatt\FSigma)_s \ar@{->>}[r]^-{P^{-1}m_s} &  P^{-1}(W_\Sigma')_s}
\right)_{s\in S}
\]
in $\hatD^S$ carries a $\Sigma$-generated profinite $\hatT$-algebra. Thus the lattice of local varieties of languages over $\Sigma$ is isomorphic to the lattice of $\Sigma$-generated profinite $\hatT$-algebras. Moreover, the Local Reiterman Theorem (see Theorem \ref{thm:localreiterman}) asserts that the latter is isomorphic to the lattice of local pseudovarieties 
of $\Sigma$-generated $\F$-algebras. This concludes the proof.
\end{proof}
Next, we consider non-local varieties of languages and prove the full Variety Theorem.

\begin{definition}[Variety of languages]\label{def:variety}
A \emph{variety of languages} is a family of local varieties 
\[(\,W_\Sigma\seq \Rec_\F(\Sigma)\,)_{\Sigma\in\A}\] closed
  under preimages of $\MT$-homomorphisms, i.e. for each $L\in W_\Sigma$ and each $\MT$-homomorphism $g\colon \MT\FDelta\to \MT\FSigma$ with $\Sigma,\Delta\in \A$ one has $g^{-1}L\in
  W_\Delta$.
\end{definition}
This leads to the main result of our paper, which holds under the Assumptions \ref{ass:catframework}.
\begin{theorem}[Variety Theorem]
\label{thm:eilenberg}
The lattice of varieties of languages (ordered by inclusion) is isomorphic to the lattice of
pseudovarieties of $\F$-algebras.
\end{theorem}
Again, we need a technical lemma to prepare the proof.

\begin{lemma}\label{lem:var}
  Let $W_\Sigma\seq \Rec_\F(\Sigma)$ and 
  $W_\Delta\seq \Rec_\F(\Delta)$ be subobjects closed under diagonals, represented
  by $m^\Sigma_s \colon (W'_\Sigma)_s \monoto P(\hatt\FSigma)_s$ and 
  $m^\Delta_s \colon (W'_\Delta)_s \monoto P(\hatt\FDelta)_s$ ($s \in S$), respectively (see Lemma \ref{lem:diagonalclosed}).
  Then for 
  any $\MT$-homomorphism $g\colon \MT\FDelta \to \MT\FSigma$, the following 
  statements 
  are equivalent:
  \begin{enumerate}[(i)]
    \item $g^{-1}L\in W_\Delta$ for all $L\in W_\Sigma$.
    \item There is a morphism $g'\colon W'_\Sigma \to W'_\Delta$ in $\C^S$ 
    making the
      following square commute for any sort $s$, where $\hat g\colon 
      \hatT\FDelta\to\hatT\FSigma$ is the continuous extension of $g$ (see 
      Lemma \ref{lem:gat}).
  \begin{equation}\label{eq:preclos}
    \vcenter{
    \xymatrix{
    (W_\Sigma')_s   \ar[r]^{g'_s} \ar@{ >->}[d]_{m^\Sigma_s} & (W_\Delta')_s 
    \ar@{ >->}[d]^{m^{\Delta}_s}\\
    P(\hatt\FSigma)_s \ar[r]_{P\hat g_s} & P(\hatt\FDelta)_s
    }
  }
  \end{equation}
  
  \end{enumerate}

\end{lemma}
\begin{proof}
  Again we use that $P(\hatt \FSigma)_s$ can assumed to be carried by
  the set $\{\,L_s: L\in \Rec_\F(\Sigma)\,\}$, and $(W_\Sigma')_s$ is the
  subset $\{L_s: L\in W_\Sigma\}$, see Remark \ref{rem:ptsigma}.1. Analogously for $P(\hatt\FDelta)_s$
  and $(W_\Delta')_s$. From the definition of $\hat g$ it follows that
  $P\hat g_s$ takes an element $L_{s}$ of $P(\hatt\FSigma)_{s}$ to
  $L_{s}\o g_s$.  Thus (ii) is equivalent to the statement that
  $L_{s}\o g_s\in (W_\Delta')_{s}$ for all $L\in W_\Sigma$ and all
  sorts $s$.  From this the implication of (i)$\To$(ii) follows
  immediately, since $(g^{-1}L)_s = L_s \cdot g_s$. Conversely,
  suppose that (ii) holds, and let $L\in W_\Sigma$. By the above
  argument, we have $L_s\o g_s\in (W_\Delta')_s$ for all $s$. Since $W_\Delta$ is closed under diagonals, this implies that $g^{-1}L= L\o g$ lies
  in $W_\Delta$, i.e. (ii)$\To$(i) holds.
\end{proof}

\begin{proof}[Proof of the Variety Theorem]
Duality + Reiterman! As shown in the proof of the Local Variety Theorem, a family $(\,W_\Sigma\seq \Rec_\F(\Sigma)\,)_{\Sigma\in\A}$ of local varieties corresponds via duality to a family of $\Sigma$-generated profinite $\hatT$-algebras in $\hatD^S$. By Lemma \ref{lem:var}, the family $(W_\Sigma)_\Sigma$ forms a variety of 
languages (i.e. is closed under preimages) iff the dual family of 
$\Sigma$-generated profinite $\hatT$-algebras forms a 
pro-$\F$ theory. This shows that the lattice of varieties of languages over $\Sigma$ is isomorphic to the lattice of pro-$\F$ theories. Moreover, the Reiterman Theorem (see Theorem \ref{thm:reiterman}) asserts that the latter is isomorphic to the lattice of pseudovarieties of $\F$-algebras. This proves the Variety Theorem.
\end{proof}
For concrete applications of our (Local) Variety Theorem, see Section \ref{sec:applications}.

\begin{rem}\label{rem:straubing}
Straubing~\cite{straubing02} studied \emph{$\mathsf{C}$-varieties of
  regular languages} which are defined as Eilenberg's varieties of
regular languages, except that closure under preimages is required
only w.r.t.~a given class $\mathsf{C}$ of monoid morphisms. By
making a class $\mathsf{C}$ of $\MT$-homomorphisms an
additional parameter of our framework, Theorem~\ref{thm:eilenberg}
easily generalizes to a monad version of Straubing's variety theorem for $\mathsf{C}$-varieties. In more detail, let $\mathsf{C}$ be a family associating to each pair $(\Sigma,\Delta)\in 
\A^2$ a set $\mathsf{C}(\Delta,\Sigma)$ of $\MT$-homomorphisms from 
$\MT\FDelta$ to $\MT\FSigma$. A \emph{$\mathsf{C}$-variety of languages} is 
given as in Definition \ref{def:variety}, but with $g$ restricted to 
homomorphisms in $\mathsf{C}$. Similarly, a \emph{pro-$\F$ theory} and an \emph{$\F$-theory w.r.t. $\mathsf{C}$} is 
given as in Definition \ref{def:profinitetheory} and \ref{def:ftheory}, but with $g$ again 
restricted to homomorphisms in $\mathsf{C}$. This leads to the following
\end{rem}

\begin{theorem}[Straubing Theorem for $\MT$-algebras]
The lattice of $\mathsf{C}$-varieties of languages is isomorphic to the lattice of 
$\F$-theories w.r.t. $\mathsf{C}$.
\end{theorem}
\begin{proof}
 This is in complete analogy to the first part of the proof of the Variety Theorem: a family of local varieties of languages forms a $\mathsf{C}$-variety iff its dual family of $\Sigma$-generated profinite $\hatT$-algebras forms a pro-$\F$ theory w.r.t. $\mathsf{C}$.  And by an argument analogous to Lemma \ref{lem:profiniteconcrete}, pro-$\F$ theories w.r.t. $\mathsf{C}$ correspond to $\F$-theories w.r.t $\mathsf{C}$.
\end{proof}
The special case $\MT=\MT_*$ on $\Set$ is due to Straubing
\cite{straubing02}.

\section{Reduced $\MT$-algebras and the Reduced Variety Theorem}
We have seen that in a many-sorted setting, a restriction of the alphabets is often necessary to capture the proper  languages. In some cases, it can is also be necessary to restrict the sorts. For example, recall that an $\infty$-language is a two-sorted subset of $T_\infty(\Sigma,\emptyset)=(\Sigma^+,\Sigma^\omega)$. If one is only interested in \emph{$\omega$-languages} (i.e. languages of infinite words), one needs restrict to subsets of $T_\infty(\Sigma,\emptyset)$ that are empty in the first sort.

In this section, we present a variety theorem that is parametric in a subset $S_0\seq S$ of the sorts and considers only recognizable languages that are empty outside of $S_0$. On the algebraic side, this means to restrict to $\MT$-algebras that are \emph{reduced} with respect to their $S_0$-components. As concrete applications of our Reduced Variety Theorem (see Theorem \ref{thm:eilenberg_reduced}), will derive Eilenberg-type correspondences for $\omega$-regular languages and regular tree languages in Section \ref{sec:applications}.

 The results of this section are inspired by the paper \cite{salehi07} where reduced languages and algebras are investigated for the special case of tree languages. They generalize the results of the previous sections, which emerge by taking $S_0=S$. Many arguments are straightforward adaptions, so we will only sketch some of the proofs and emphasize on the new concepts.

\begin{notation}
Throughout this section, we fix a subset $S_0\seq S$ of sorts and continue to work with a fixed class $\A\in \Set_f^S$ of alphabets. Moreover, let us fix for each $\Sigma\in\A$ a unary presentation $\U_\Sigma$ of $\MT\FSigma$, and let $\ol \U_\Sigma$ denote the closure of $\U_\Sigma$ under composition and identity morphisms $\id\colon (T\FSigma)_s\to (T\FSigma)_s$. Then also $\ol\U_\Sigma$ forms a unary presentation of $\MT\FSigma$. We write $\ol \U_\Sigma(s,t)$ for the set of all unary operations $u\colon (T\FSigma)_s\to(T\FSigma)_t$ in $\ol\U_\FSigma$, where $s,t\in S$.
\end{notation}

\subsection{Reduced $\MT$-algebras}
In this subsection, we investigate $\MT$-algebras that are ``minimal'' with respect to their $S_0$-part. They form the algebraic counterpart of languages that are empty in non-$S_0$ sorts.

\begin{definition}
A $\MT$-algebra $(A,\alpha)$ is called \emph{reduced} if there exists no proper quotient $e\colon (A,\alpha)\epito (B,\beta)$ in $\Alg{\MT}$ whose components $e_s: A_s\epito B_s$ for $s\in S_0$ are isomorphisms. By a \emph{$\Sigma$-generated reduced $\F$-algebra} is meant a $\Sigma$-generated $\F$-algebra $e\colon \MT\FSigma\epito (A,\alpha)$ whose codomain $(A,\alpha)$ is reduced.
\end{definition}

\begin{rem}
For $S=S_0$, every $\MT$-algebra is reduced.
\end{rem}

\begin{definition}\label{def:reducedquot}
Let $e\colon T\FSigma\epito A$ be a quotient of $T\FSigma$ in $\D^S$. Then $e$ is called
\begin{enumerate}
\item \emph{reduced} if there exists no proper quotient $e\colon A\epito B$ in $\D^S$ whose components $e_s: A_s\epito B_s$ for $s\in S_0$ are isomorphisms.
\item a \emph{reduced $\U_\Sigma$-quotient} if it is (i) reduced, (ii) $\F$-refinable and (iii) every unary operation $u\in \ol \U(s,t)$ with $s,t\in S_0$ lifts along $e$, i.e. there exists $u_A\colon A_s\to A_t$ with $e\o u = u_A\o e$.
\[
\xymatrix{
(T\FSigma)_s \ar[r]^u \ar@{->>}[d]_e & (T\FSigma)_t \ar@{->>}[d]^e \\
A_s \ar@{-->}[r]_{u_A} & A_t
}
\]
\end{enumerate}
\end{definition}

\begin{rem}\label{rem:reduction}
 A quotient $e\colon T\FSigma\epito A$ in $\D^S$ is reduced iff, for all sorts $s\not\in S_0$, the set $\under{A_s}$ has at most one element. Therefore a reduced quotient $e\colon T\FSigma\epito A$ is completely determined by its components $e_s\colon (T\FSigma)_s\to A_s$ for $s\in S_0$.
\end{rem}

\begin{rem}
In the special case $S_0=S$, the reduced $\U_\Sigma$-quotients $e\colon T\FSigma\epito A$ of $T\FSigma$ are exactly the $\F$-quotients of $\MT\FSigma$. Indeed, in this case the lifting property of $e$ in the above definition states exactly that $e$ is a $\ol \U_\Sigma$-quotient in the sense of Definition \ref{def:uquotient}, and thus corresponds to an $\F$-quotient because $\ol \U_\Sigma$ is a unary presentation of $\MT\FSigma$.
\end{rem}
Our goal in this subsection is to show that $\Sigma$-generated reduced $\F$-algebras correspond to  reduced $\U_\Sigma$-quotients of $T\FSigma$ (see Lemma \ref{lem:usigma_vs_reduced}). The key is the following construction:

\begin{notation}
Let $e\colon T\FSigma\epito A$ be a quotient of $T\FSigma$ in $\D^S$.
\begin{enumerate}
\item If $\D$ is variety of algebras, the congruence $\equiv_e$ on $T\FSigma$ is defined as follows: for any sort $s$ and elements $x,x'\in \under{T\FSigma}_s$, put 
\[ x\equiv_e x' \quad\text{iff}\quad e_t\o u(x) = e_t\o u(x') \text{ for all $u\in \ol\U_\Sigma(s,t)$ with $t\in S_0$}.  \]
\item If $\D$ is variety of ordered algebras, the stable preorder $\preceq_e$ on $T\FSigma$  is defined as follows: for any sort $s$ and elements $x,x'\in \under{T\FSigma}_s$, put 
\[ x\preceq_e x' \quad\text{iff}\quad e_t\o u(x) \leq e_t\o u(x') \text{ for all $u\in \ol\U_\Sigma(s,t)$ with $t\in S_0$}.  \]
\end{enumerate}
We denote by $e_R\colon T\FSigma\epito A_R$ the quotient in $\D^S$ induced by $\equiv_e$ (resp. $\preceq_e$), see \ref{app:congruences} and  \ref{app:stablepreorders}.
\end{notation}
The following lemma collects the key properties of $e_R$:

\begin{lemma}\label{lem:reduction_props} For any $\F$-refinable reduced quotient $e\colon T\FSigma\epito A$, the following holds:
\begin{enumerate}
\item $e_R$ is an $\F$-quotient, i.e. carries a $\Sigma$-generated $\F$-algebra $e_R\colon \MT\FSigma\epito (A_R,\alpha_R)$.
\item $e_R$ is the smallest $\Sigma$-generated $\F$-algebra with $e\leq e_R$.
\item $(A_R,\alpha_R)$ is reduced.
\item If $e$ is a reduced $\U_\Sigma$-quotient, then $e_R$ is the unique $\Sigma$-generated reduced $\F$-algebra with $e_{R,s}= e_s$ for all $s\in S_0$.
\end{enumerate}
\end{lemma}

\begin{proof}
We consider the case where $\D$ is a variety of ordered algebras; for the unordered case, replace equations by inequations and stable preorders by congruences. 

\smallskip
\noindent We first establish an auxiliary statement:

\smallskip
\noindent($\ast$) \emph{Claim.} If $h\colon \MT\FSigma\epito B$ is a $\Sigma$-generated $\F$-algebra  with $e\leq h$ then also $e_R\leq h$.

\smallskip
\noindent\emph{Proof.} Suppose that $e\leq h$, i.e. $e = p\o h$ for some morphism $p\colon B\epito A$. To show that $e_R\leq h$, we apply the homomorphism theorem. Thus suppose that elements $x,x'\in \under{T\FSigma}_s$ are given with $h(x)\leq h(x')$. We need to show that $e_R(x)\leq e_R(x')$, i.e. $x\preceq_e x'$. For any $u\in\ol\U_\Sigma(s,t)$ with $t\in S_0$, we have
\[ e_t\o u(x) = p_t\o h_t\o u(x) = p_t\o u_B\o h_s(x) \leq  p_t\o u_B\o h_s(x') =  p_t\o h_t\o u(x') =  e_t\o u(x'). \]
Here $u_B$ is the lifting of $u$ along $h$, which exists because $\ol \U_\Sigma$ is a unary presentation of $\MT\FSigma$ and $h$ is an $\F$-quotient. This shows $x\preceq_s x'$, as claimed. Thus  by the homomorphism theorem  $e_R$ factors through $h$, proving ($\ast$).

\smallskip
\noindent Now for the proof of the four statements in the lemma.
\begin{enumerate}
\item Since the quotient $e$ is $\F$-refinable by hypothesis, $e_R$ is also $\F$-refinable by ($\ast$). To show that $e_R$ is an $\F$-quotient (equivalently, that $\preceq_e$ is an $\F$-congruence), it suffices to verify that $\preceq_e$ is a $\U_\Sigma$-congruence, since $\U_\Sigma$ is a unary presentation. Thus suppose that $x,x'\in \under{T\FSigma}_s$ are given with $x\preceq_e x'$ and let $u\in\U_\Sigma(s,t)$. We need to show $u(x)\preceq_e u(x')$. For all $v\in \ol\U_\Sigma(t,t')$ with $t'\in S_0$, we have $v\o u\in\ol \U_\FSigma(s,t')$ by closure under composition, and thus $e_{t'}\o v\o u(x)\leq e_{t'}\o v\o u(x')$ since $x\preceq_e x'$. But this means precisely that $u(x)\preceq_e u(x')$ by the definition of $\preceq_e$, as desired.
\item We first show that $e\leq e_R$, using the homomorphism theorem. Let $x,x'\in\under{T\FSigma}_s$ with $e_R(x)\leq e_R(x')$ (i.e. $x\preceq_e x'$). We need to show that $e_s(x)\leq e_s(x')$. The latter is trivial for $s\not\in S_0$ because $e$ is reduced, i.e. the codomain of $e_s$ has at most one element. For $s\in S_0$, we have $e_t\o u (x)\leq e_t\o u(x')$ for all $u\in \ol\U_\Sigma(s,t)$ with $t\in S_0$ by the definition of $\preceq_e$. In particular, putting $u=\id_{(T\FSigma)_s}$ yields $e_s(x)\leq e_s(x')$. Thus $e\leq e_R$ by the homomorphism theorem.

That $e_R$ is the smallest $\Sigma$-generated $\F$-algebra  with $e\leq e_R$ now follows from ($\ast$).
\item Suppose that $q\colon (A_R,\alpha_R)\epito (B,\beta)$ is a quotient of $(A_R,\alpha_R)$ such that $q_s$ is an isomorphism for every $s\in S_0$. We need to show that $q$ is an isomorphism. To this end, we first show that $e\leq q\o e_R$, using once again the homomorphism theorem: given $x,x'\in \under{T\FSigma}_s$ with $q\o e_R(x)\leq q\o e_R(x')$, we need to show that $e(x)\leq e(x')$. For $s\not\in S_0$ this holds trivially because $e$ is reduced, i.e. the codomain of $e_s$ has at most one element. Thus let $s\in S_0$. Then $q\o e_R(x)\leq q\o e_R(x)$ implies $e_R(x)\leq e_R(x')$ because $q_s$ is an isomorphism, and thus $e(x)\leq e'(x)$ because $e\leq e_R$ by part 2 of the lemma. Thus, the homomorphism theorem gives $e\leq q\o e_R$. By part 2 again, we conclude $e_R\leq q\o e_R$. On the other hand, we trivially have $q\o e_R\leq e_R$. Thus $q\o e_R \cong e_R$, which shows that $q$ is an isomorphism.

\item Suppose that $e$ is a reduced $\U_\Sigma$-quotient. We first show that $e_{R,s}= e_s$ for $s\in S_0$. To this end, it suffices to show that $e_{R,s}$ and $e_s$ have the same ordered kernel: for all $x,x'\in \under{T\FSigma}_s$ one has $x\preceq_e x'$ iff $e_s(x)\leq e_s(x')$. The ``only if'' direction  follows from ($\ast$). For the ``if'' direction, let $e_s(x)\leq e_s(x')$. To show that $x\preceq_e x'$, let $u\in \ol \U(s,t)$ with $t\in S_0$. Then 
\[ e_t\o u(x) = u_A\o e_s(x)\leq u_A\o e_s(x') = e_t\o u(x'), \]
where $u_A$ is the lifting of $u$ along $e$. It exists because $s,t\in S_0$ and $e$ is a reduced $\U_\Sigma$-quotient. Thus, by the definition of $\preceq_e$ we have $x\preceq_e x'$, as claimed.

For the uniqueness statement, suppose that $h\colon \MT\FSigma\epito (B,\beta)$ is a $\Sigma$-generated reduced $\F$-algebra with $h_s = e_s$ for all $s\in S_0$. Then $e\leq h$; indeed, we have $e=p\o h$ where $p\colon B\epito A$ is the morphism with $p_s=\id_{A_s}$ for $s\in S_0$, and otherwise the unique morphism into $A_s$. Thus, by ($\ast$), we also have $e_R\leq h$, i.e. $e_R = q\o h$ for some $q\colon B\epito A_R$. Clearly $q_s = \id$ for $s\in S_0$ because $h_s = e_{R,s} = e_s$. Since $(B,\beta)$ is reduced, it follows that $q$ is an isomorphism. Thus $h\cong e_R$.
\qedhere
\end{enumerate}
\end{proof}

\begin{lemma}\label{lem:subdirectproductreduced}
Let $e_0$ and $e_1$ be two reduced $\U_\Sigma$-quotients of $T\FSigma$.
\begin{enumerate}
\item The subdirect product $e$ of $e_0$ and $e_1$ in $\D^S$ is a reduced $\U_\Sigma$-quotient, and moreover
\item $e_R$ is the subdirect product of $e_{0,R}$ and $e_{1,R}$ in $\Alg{\MT}$. 
\item The pushout of $e_0$ and $e_1$ in $\D^S$ is a reduced $\U_\Sigma$-quotient.
\end{enumerate}
\end{lemma}

\begin{proof}
\begin{enumerate}
\item By the definition of the subdirect product (see \ref{app:subdirectproducts}), clearly $e$ is a reduced quotient. Next, we show that $e$ is $\F$-refinable. Since $e_0$ and $e_1$ are $\F$-refinable, they factor through some (w.l.o.g. the same) $\F$-quotient $h$ of $\MT\FSigma$, i.e. $e_0,e_1\leq h$. But since $e$ is the smallest quotient of $T\FSigma$ above $e_0$ and $e_1$, we also have $e\leq h$, showing that $e$ is $\F$-refinable. That $e$ satisfies the lifting property is shown as in the proof of Lemma \ref{lem:subdirectproduct}.1.
\item We need to prove that $e_R$ is the smallest $\Sigma$-generated $\F$-algebra with $e_{i,R}\leq e_R$ ($i=0,1$). First, by Lemma \ref{lem:reduction_props}.2 we have $e\leq e_R$ and thus also $e_i\leq e_R$, since $e_i\leq e$. By Lemma \ref{lem:reduction_props}.2 again, $e_{i,R}$ is the smallest $\Sigma$-generated $\F$-algebra with $e_i\leq e_{i,R}$, and therefore $e_{i,R}\leq e_R$. In remains to show that $e_R$ is the smallest $\Sigma$-generated $\F$-algebra satisfying the latter inequality. Thus let $q$ be any $\Sigma$-generated $\F$-algebra with $e_{i,R}\leq q$ for $i=0,1$. Since $e_{i}\leq e_{i,R}$, we get $e_i\leq q$ for $i=0,1$, and thus $e\leq q$ by the minimality of $e$. But since $e_R$ is the smallest $\Sigma$-generated $\F$-algebra with $e\leq e_R$, it follows that  $e_R\leq q$, as desired.
\item Let $p$ be the pushout of $e_0$ and $e_1$ in $\D^S$. Clearly $p$ is reduced and $\F$-refinable because $e_0$ is and $p\leq e_0$. The lifting property is established as in the proof of Lemma \ref{lem:subdirectproduct}.2.\qedhere
\end{enumerate}
\end{proof}

\begin{notation}\label{not:downarrow}
For any morphism $h\colon A\to B$ in $\D^S$, we denote by $h_\downarrow\colon A\to B_\downarrow$ the morphism where (i) $h_{\downarrow,s} = h_s$ for $s\in S_0$ and (ii) $h_{\downarrow,s}$ is the image of the unique morphism $A_s\to 1$ for $s\not\in S_0$. In particular, $B_{\downarrow,s}=B_s$ for $s\in S_0$, and $B_{\downarrow,s}$ has at most one element for $s\not\in S_0$. Moreover, if $h$ is a quotient, then $h_\downarrow$ is a reduced quotient in $\D^S$. 
\end{notation}

\begin{rem}
If $h\colon \MT\FSigma\epito (A,\alpha)$ is a $\Sigma$-generated reduced $\F$-algebra, then $h_\downarrow$ is a reduced $\U_\Sigma$-quotient of $T\FSigma$. Indeed, that $h_\downarrow$ satisfies the lifting property of Definition \ref{def:reducedquot}.2 follows immediately from the fact that $\ol \U_\Sigma$ forms a unary presentation of $\MT\FSigma$ and $h$ is an $\F$-quotient.
\end{rem}

\begin{lemma}\label{lem:usigma_vs_reduced} The maps
\[ h\mapsto h_\downarrow\quad\text{and}\quad e\mapsto e_R\]
define an isomorphism between the poset of $\Sigma$-generated reduced $\F$-algebras and the poset of reduced $\U_\Sigma$-quotients of $T\FSigma$.
\end{lemma}

\begin{proof}
\begin{enumerate}[(a)]
\item For every reduced $\U_\Sigma$-quotient $e\colon T\FSigma\epito A$, we have $e = e_{R,\downarrow}$. Indeed, clearly both quotients agree on sorts $s\not\in S_0$ because they are reduced. On sorts $s\in S_0$, the quotient $e$ agrees with $e_R$ by Lemma \ref{lem:reduction_props}.4, and $e_R$ agrees with $e_{R,\downarrow}$ by the definition of $(\dash)_\downarrow$.
\item For every $\Sigma$-generated reduced $\F$-algebra $h\colon \MT\FSigma\epito (A,\alpha)$, we have $h=h_{\downarrow,R}$. Indeed, since both $h$ and $h_{\downarrow,R}$ agree with the reduced $\U_\Sigma$-quotient $e:= h_\downarrow$ on sorts $s\in S_0$, this follows from Lemma \ref{lem:reduction_props}.4.
\item It remains to show that the two maps $h\mapsto h_\downarrow$  and $e\mapsto e_R$ are monontone. Clearly the map $h\mapsto h_\downarrow$ is monotone. For the monotonicity of $e\mapsto e_R$, suppose that $e_i\colon T\FSigma\epito A_i$, $i=0,1$, are reduced $\U_\Sigma$-quotients with $e_0\leq e_1$, i.e. $e_0$ factorizes through $e_1$. Since $e_1\leq e_{1,R}$ by Lemma \ref{lem:reduction_props}.2,  we also have $e_0\leq e_{1,R}$. But since (again by Lemma \ref{lem:reduction_props}.2) $e_{0,R}$ is the smallest $\Sigma$-generated $\F$-algebra with $e_0\leq e_{0,R}$, it follows that $e_{0,R}\leq e_{1,R}$.\qedhere
\end{enumerate}
\end{proof}
With the help of unary presentations, finite reduced $\F$-algebras can be described as follows:

\begin{lemma}\label{lem:reducedconcrete}
Let $A\in \F$ be an $\A$-generated algebra and let $e\colon \MT\FSigma\epito A$ be a surjective $\MT$-homomorphism with $\Sigma\in \A$. Then the following statements are equivalent:
\begin{enumerate}[(i)]
\item $A$ is reduced.
\item For each sort $s\not \in S_0$ and any two elements $a,a' \in \under{A_s}$, 
\[ \text{if $u_A(a) = u_A(a')$ for all $u\in \ol\U_\Sigma(s,t)$ with $t\in S_0$}\qquad \text{then}\qquad a=a', \]
and in the ordered case
\[ \text{if $u_A(a) \leq u_A(a')$ for all $u\in \ol\U_\Sigma(s,t)$ with $t\in S_0$}\qquad \text{then}\qquad a\leq a',\] 
where $u_A$ is the lifting of $u$ along $e$.
\end{enumerate}
\end{lemma}

\begin{proof}
We consider the case where $\D$ is a variety of ordered algebras. That $A$ is reduced means precisely that $e=e_{\downarrow,R}$ by Lemma \ref{lem:usigma_vs_reduced}.  Since $e_{\downarrow,R}\leq e$ by Lemma \ref{lem:reduction_props}.2, the latter equation is equivalent to $e\leq e_{\downarrow,R}$. By the homomorphism theorem, this is the case precisely if, for all $x,x'\in \under{T\FSigma}_s$,
\[x\preceq_{e_\downarrow} x'\quad\text{implies}\quad e(x)\leq e(x')\]
i.e. 
 \[ \text{if $e_t\o u(x) \leq e_t\o u(x')$ for all $u\in \ol\U_\Sigma(s,t)$ with $t\in S_0$}\qquad \text{then}\qquad e(x)\leq e(x').\] 
For $s\in S_0$ this implication is trivially satisfied (take $u=\id_{(T\FSigma)_s}$). Thus it suffices to restrict to sorts $s\not\in S_0$. But then the above implication means exactly that, for $a,a'\in \under{A_s}$ with $s\in S_0$,
 \[ \text{if $u_A(a) \leq u_A(a')$ for all $u\in \ol\U_\Sigma(s,t)$ with $t\in S_0$}\qquad \text{then}\qquad a\leq a',\] 
 because $e$ is surjective and $u_A$ is the lifting of $u$ along $e$.
\end{proof}

\begin{example}\label{ex:reducedalg_omegasem}
Let $\MT=\MT_\infty$ on $\Set^S$ with $S=\{+,\omega\}$, $S_0=\{\omega\}$, $\A= \{(\Sigma,\emptyset): \Sigma\in\Set_f\}$, $\F =$ all finite $\omega$-semigroups. Given an $\A$-generated (= complete) finite $\omega$-semigroup $A$, express $A$ as a quotient $e\colon (\Sigma^+,\Sigma^\omega)\epito (A_+,A_\omega)$ with $\Sigma\in \Set_f$. For the unary presentation $\U_\Sigma$ of $(\Sigma^+,\Sigma^\omega)$ as in Example \ref{ex:unpres_omegasem},  the unary operations in $\ol\U_\Sigma(+,\omega)$ (i.e. from $\Sigma^+$ to $\Sigma^\omega$) have the form $x\mapsto yxz$ and $x\mapsto y(xy')^\omega$ with $y,y'\in \Sigma^*$ and $z\in \Sigma^\omega$. The corresponding lifted operations from $A_+$ to $A^\omega$ have the form $a\mapsto bac$ and $a\mapsto b(ab')^\omega$ with $b,b'\in A_+\cup \{1\}$ and $c\in A_\omega$. Thus, by the above lemma the $\omega$-semigroup $A$ is reduced iff, for all $a,a'\in A_+$,
\[ \text{if $bac = ba'c$ and $b(ab')^\omega = b(a'b')^\omega$ for all $b,b'\in A_+\cup\{1\}$ and $c\in A_\omega$} \qquad\text{then}\qquad a=a'. \]
In other words, $A$ is reduced iff any two distinct elements of $A_+$ can be separated by some iterated unary operation from $A_+$ into  $A_\omega$.
\end{example}
\subsection{Local pseudovarieties of reduced $\MT$-algebras}
In this section we consider reduced $\U_\Sigma$-quotients of $\hatt\FSigma$ and show that they correspond to local pseudovarieties of $\Sigma$-generated reduced $\F$-algebras. This is the reduced analogue of the Local Reiterman Theorem \ref{thm:localreiterman}.

\begin{notation}
Recall that we denote the continuous extension of $u\in \U_\Sigma(s,t)$ by $\hat u\colon (\hatt \FSigma)_s\to (\hatt\FSigma)_t$, see Lemma \ref{lem:extend_fi}. Then every $u=u_n\o\ldots\o u_1$ in $\ol \U_\Sigma$ has the continuous extension $\hat u := \hat u_n\o\ldots\o \hat u_1$.
Moreover, the continuous extension of an $\F$-refinable quotient $e\colon T\FSigma\epito A$ in $\D^S$ is denoted by $\hat e\colon \hatt\FSigma\epito A$ (see Lemma \ref{lem:extensiblevsrefinable}).
\end{notation}
The following definition gives the analogous concepts of Definition \ref{def:reducedquot} for $\hatt$:
\begin{definition}\label{def:reducedquothatt}
Let $\phi\colon \hatt\FSigma\epito P$ be a quotient of $\hatt\FSigma$ in $\hatD^S$. Then $\phi$ is called
\begin{enumerate}
\item \emph{reduced} if there exists no proper quotient $\phi\colon P\epito P'$ in $\hatD^S$ whose components $\phi_s$ for $s\in S_0$ are isomorphisms.
\item a \emph{reduced $\U_\Sigma$-quotient} if it is (i) reduced and (ii) for every unary operation $u\in \ol \U_\Sigma(s,t)$ with $s,t\in S_0$, the continuous extension $\hat u$ lifts along $\phi$, i.e. there exists $u_P\colon P_s\to P_t$ with $\phi\o  \hat u = u_P\o \phi$.
\[
\xymatrix{
(\hatt\FSigma)_s \ar[r]^{\hat u} \ar@{->>}[d]_\phi & (\hatt\FSigma)_t \ar@{->>}[d]^\phi \\
P_s \ar@{-->}[r]_{u_P} & P_s
}
\]
\end{enumerate}
\end{definition}

\begin{rem}
If $S_0=S$, a reduced $\U_\Sigma$-quotient of $\hatt\FSigma$ corresponds precisely to a $\Sigma$-generated profinite $\hatT$-algebra, see Lemma \ref{lem:profinite_algs}.
\end{rem}

\begin{rem}\label{rem:uextendreduced}
In analogy to Remark \ref{rem:uextend}, if the object $P$ in the above definition is finite, then $\phi$ forms a reduced $\U_\Sigma$-quotient of $\hatt\FSigma$ iff its restriction $V\phi\o \iota_\FSigma\colon T\FSigma\epito P$ forms a $\U_\Sigma$-quotient of $T\FSigma$.
\end{rem}

\begin{definition}\label{def:localpseudovar_reduced}
A \emph{local pseudovariety of reduced $\U_\Sigma$-quotients} is an ideal in the poset of reduced $\U_\Sigma$-quotients of $T\FSigma$ (cf. Definition \ref{def:localpseudovariety}).
\end{definition}
We aim to show that local pseudovarieties of reduced $\U_\Sigma$-quotients correspond to reduced $\U_\Sigma$-quotients of $\hatt\FSigma$. The following construction gives the translation between the two concepts, generalizing Construction \ref{rem:localvarconst}.

\begin{construction}\label{rem:usigma_to_localpseudovar}
\begin{enumerate}
\item To every reduced $\U_\Sigma$-quotient $\phi\colon\hatt\FSigma\epito P$ we associate the class $\P^\phi$ of all quotients of $T\FSigma$ of the form
\[ e = (\,T\FSigma \xra{\iota_\FSigma} V\hatt\FSigma \xra{V\phi} P \xra{Ve'} A\,), \]
where $e'\colon P\epito A$ is a finite quotient in $\hatD^S$ and $e'\o \phi$ is a reduced $\U_\Sigma$-quotient of $\hatt\FSigma$. Note that any such morphism $e$ is a reduced $\U_\Sigma$-quotient of $T\FSigma$: it is clearly reduced, it is $\F$-refinable by Lemma \ref{lem:extensiblevsrefinable}, and satisfies the lifting property by Remark \ref{rem:uextendreduced}.
\item To every local pseudovariety $\P$ of reduced $\U_\Sigma$-quotients we associate a quotient \[\phi_\P\colon \hatt\FSigma\epito P^\P\] in $\hatD^S$ as follows. View $\P$ as a full subcategory of the comma category $(T\FSigma\downarrow \D_f)$ and form the cofiltered limit $P^\P$ of the diagram 
\[ \P\mapsto \hatD^S,\quad (e\colon T\FSigma\epito A)\mapsto A,  \]
with limit projections denoted by
\[ e_\P^*\colon P\epito A. \] 
Note that the projections are surjective by Lemma \ref{lem:projection-is-surjective}. The morphisms $\hat e\colon \hatt\FSigma\epito A$ (where $e$ ranges over elements of $\P$) form a compatible family over the same diagram, and thus there exists a unique $\phi^\P\colon \hatt\FSigma\epito P^\P$ making the triangle below commute:
\[
\xymatrix{
\hatt\FSigma \ar@{->>}[d]_{\hat e}\ar@{->>}[r]^{\phi^\P} & P^\P \ar@{->>}[dl]^{e_\P^*}\\
A
}
\] The morphism $\phi^\P$ is surjective by Lemma \ref{lem:surjections-between-cofiltered-diagrams}.
\end{enumerate}
\end{construction}

\begin{lemma}
\begin{enumerate}
\item For every reduced $\U_\Sigma$-quotient $\phi\colon \hatt\FSigma\epito P$, the class $\P^\phi$ forms a local pseudovariety of reduced $\U_\Sigma$-quotients of $T\FSigma$.
\item For every local pseudovariety $\P$ of reduced $\U_\Sigma$-quotients of $T\FSigma$, the morphism $\phi^\P\colon \hatt\FSigma\epito P$ is a reduced $\U_\Sigma$-quotient of $\hatt\FSigma$.
\end{enumerate}
\end{lemma}

\begin{proof}
\begin{enumerate}
\item Let $\phi\colon \hatt\FSigma\epito P$ be a reduced $\U_\Sigma$-quotient. Clearly $\P^\phi$ is downwards closed. To show that $\P^\phi$ is directed,  let
\[ e_i = (T\FSigma \xra{\iota_\FSigma} V\hatt\FSigma \xra{V\phi} P \xra{Ve_i'} A_i), \quad i=0,1,\]
be two elements of $\P^\phi$. We need to find an element $e$ with $e_0,e_1\leq e$. Form the subdirect product $e'\colon P\epito A$ of $e_0'$ and $e_1'$ in $\hatD^S$.  Then $e'\o \phi$ is a reduced $\U_\Sigma$-quotient; the argument is completely analogous to the proof of Lemma \ref{lem:subdirectproduct}.1. Therefore 
\[ e := (T\FSigma \xra{\iota_\FSigma} V\hatt\FSigma \xra{V\phi} P \xra{Ve'} A), \quad i=0,1,\]
lies in $\P^\phi$, and moreover $e_i\leq e$ because $e_i'\leq e'$.
\item Let $\P$ be a local pseudovariety of reduced $\U_\Sigma$-quotients of $T\FSigma$, and let $u\in \ol\U_\Sigma(s,t)$ with $s,t\in S_0$. We need to show that $\hat u$ has a lifting along $\phi^\P$. To this end, observe that the morphisms $P_s\xra{e_\P^*} A_s \xra{u_A} A_t$ (where $e\colon T\FSigma\epito A$ ranges over $\P$ and $u_A$ is the lifting of $u$ along $e$) form a compatible family over the diagram defining the limit $P_t$. Therefore there exists a unique $u_P\colon P_s\to P_t$ with $e_\P^*\o u_P = u_A\o e_\P^*$ for all $e\colon T\FSigma\epito A$ in $\P$. This implies $u_P\o \phi^\P = \phi^\P \o \hat u$, as this equation holds when precomposed with the dense map $\iota_\FSigma$ and postcomposed with the jointly monomorphic limit projections $e_\P^*$.\qedhere
\end{enumerate}
\end{proof}

\begin{lemma}\label{lem:localpseudovar_localprofiniteth_1_reduced}
For any local pseudovariety $\P$ of reduced $\U_\Sigma$-quotients of $T\FSigma$, we have 
$\P = 
\P^\phi$ where $\phi := \phi^\P$.
\end{lemma}

\begin{proof}
$\P \seq 
\P^{\phi}$: Let $(e\colon T\FSigma\epito A)\in \P$. Then for the 
corresponding limit projection $e^*_\P\colon P^\P \epito A$ we have
\[ e = (\xymatrix{T\FSigma \ar[r]^{\iota_\FSigma} & V\hatt\FSigma 
\ar@{->>}[r]^{V\phi} & VP^\P \ar@{->>}[r]^{Ve^*_\P} & A}) 
\]
since $\hat e = e^*_\P \cdot \phi^\P$ by the definition of $\phi=\phi^\P$
and since $V\hat e \cdot \iota_\FSigma = e$ by the definition of $\hat e$. Therefore $e\in \P^{\phi}$ by the
definition of $\P^{\phi}$.

$\P^{\phi}\seq \P$: Let $(e\colon \MT\FSigma\epito A)\in 
\P^{\phi}$. Thus there exists a quotient 
$e'\colon P^\P \epito A$ such that
\[ e = (\xymatrix{T\FSigma \ar[r]^{\iota_\FSigma} & V\hatt\FSigma 
\ar@{->>}[r]^{V\phi} & VP^\P \ar@{->>}[r]^{Ve'} & A}). 
\]
Since the finite object $A$ is finitely copresentable in $\hatD^S$, the morphism $e'$
factors through the the limit cone defining $P^\P$; that is, there 
exists an $h\colon 
\MT\FSigma\epito A'$ in $\P$ and a $\hatT$-homomorphism $s\colon A'\epito A$ with 
$e' = s\o h^*_\P$. Then the commutative diagram below shows that $e\leq h$, and thus $e\in \P$ because $h\in \P$ and $\P$ is downwards closed.
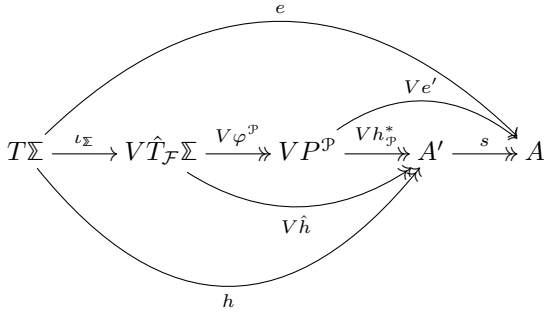
\[
\xymatrix{
T\FSigma \ar[r]^-{\iota_\FSigma} \ar@{->>}@/_5em/[rrr]_h \ar@/^5em/[rrrr]^e & 
V\hatt\FSigma \ar@{->>}@/_2em/[rr]_{V\hat h} \ar@{->>}[r]^-{V\phi^\P} &  
VP^\P \ar@/^2em/[rr]^{Ve'} \ar@{->>}[r]^-{Vh_\P^*} & A' \ar@{->>}[r]^-s & 
A
}
\vspace*{-15pt}
\]
\end{proof}

\begin{lemma}\label{lem:subdirectproducthatdreduced}
Let $e_0$ and $e_1$ be two reduced $\U_\Sigma$-quotients of $\hatt\FSigma$.
\begin{enumerate}
\item The subdirect product of $e_0$ and $e_1$ in $\hatD^S$ is a reduced $\U_\Sigma$-quotient. 
\item The pushout of $e_0$ and $e_1$ in $\hatD^S$ is a reduced $\U_\Sigma$-quotient.
\item Given a reduced $\U_\Sigma$-quotient $\phi\colon \hatt\FSigma\epito P$ and a finite quotient $e\colon P\epito A$, there exists a finite quotient $h\colon P\epito B$ such that $h\o \phi$ is a reduced $\U_\Sigma$-quotient of $\hatt D$ and $e\leq h$.
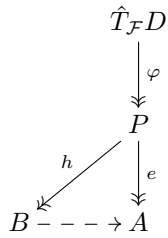
\[
\xymatrix{
& \hatt D \ar@{->>}[d]^\phi\\
& P \ar@{->>}[d]^e \ar@{->>}[dl]_h\\
B \ar@{-->}[r]& A
}
\]
\item Given a reduced $\U_\Sigma$-quotient $\phi\colon \hatt\FSigma\epito P$ and a finite quotient $e\colon P\epito A$, let $\S$ be the full subcategory of $(P\downarrow \hatD_f^S)$ on all finite quotients $e\colon P\epito A$ such that $e\o \phi$ is a reduced $\U_\Sigma$-quotient. Then $P$ is the cofiltered limit of the diagram
\[ \pi \colon \S\to \hatD^S,\quad (e\colon P\epito A)\mapsto A, \]
with limit projections $e$.
\end{enumerate}
\end{lemma}

\begin{proof}
The first two statements follow as in  proof of Lemma \ref{lem:subdirectproductreduced}. The last two  statements follow as in the proof of Lemma \ref{lem:subdirectproducthatd}.
\end{proof}

\begin{lemma}\label{lem:localpseudovar_localprofiniteth_2_reduced}
 For each reduced $\U_\Sigma$-quotient $\phi\colon 
 \hatt\FSigma\epito 
P$, we have $\phi = \phi^\P$ where $\P := \P^\phi$.
\end{lemma}

\begin{proof}
Consider the diagram $\pi$ as in Lemma \ref{lem:subdirectproducthatdreduced}.3. Then $\S\cong \P=\P^\phi$, and thus by the definition of $P^\P$ (see Construction \ref{rem:usigma_to_localpseudovar}) the object $P^\P$ is the limit of the diagram $\pi$. On the other hand, by  Lemma \ref{lem:subdirectproducthatdreduced}.4 also the object $P$ is the limit of the diagram $\pi$ with limit projections $e$. The uniqueness of limits gives an isomorphism $j\colon P\xra{\cong} P^\P$, and as in the proof of Lemma \ref{lem:localpseudovar_localprofiniteth_2} this implies $\phi\cong\phi^\P$.
\end{proof}

\begin{lemma}\label{lem:localreitermanusigma}
The maps
\[ \phi\mapsto \P^\phi\quad\text{and}\quad \P\mapsto \phi^\P \]
give an isomorphism between the lattice of reduced $\U_\Sigma$-quotients of $\hatt\FSigma$ and the lattice of pseudovarieties of reduced $\U_\Sigma$-quotients.
\end{lemma}

\begin{proof}
The two maps give mutually inverse bijections by Lemma \ref{lem:localpseudovar_localprofiniteth_1_reduced} and \ref{lem:localpseudovar_localprofiniteth_2_reduced}. That they are order-preserving is shown in analogy to the proof of Theorem \ref{thm:localreiterman}).
\end{proof}
Let us rephrase the above lemma in terms of reduced algebras:

\begin{definition}
By a \emph{local pseudovariety of $\Sigma$-generated reduced $\F$-algebras} is meant a filter in the poset of $\Sigma$-generated reduced $\F$-algebras.
\end{definition}
From the above lemma and Lemma \ref{lem:usigma_vs_reduced} we get

\begin{theorem}[Local Reiterman Theorem for Reduced $\F$-algebras]\label{thm:localreiterman_reduced}
The lattice of reduced $\U_\Sigma$-quotients of $\hatt\FSigma$ is isomorphic to the lattice of pseudovarieties of $\Sigma$-generated reduced $\F$-algebras.
\end{theorem}

\subsection{Pseudovarieties of reduced $\MT$-algebras}
In this subsection, we consider pseudovarieties of reduced $\F$-algebras and show that they correspond to reduced pro-$\F$ theories. This is the reduced version of the Reiterman Theorem \ref{thm:reiterman}. We continue to work with a fixed class $\A\seq\Set_f^S$ of alphabets.

\begin{definition}\label{def:pseudovar_reduced}
We say that a $\MT$-algebra $A$ \emph{divides} a $\MT$-algebra $B$ if $A$ is a quotient of a subalgebra of $B$, i.e. there exist $\MT$-homomorphisms $e$ and $m$ as shown below:
\[
\xymatrix{
C \ar@{>->}[r]^m \ar@{->>}[d]& B\\
A&
}
\]
 A \emph{pseudovariety of reduced $\F$-algebras} is a full subcategory $\V\seq \F$ such that all elements of $\V$ are $\A$-generated and reduced, and $\V$ is closed under division of finite products. That is, whenever $A_1,\ldots A_n\in \V$ and $A$ is an $\A$-generated reduced $\F$-algebra that divides $A_1\times\ldots\times A_n$, then $A\in \V$.
\end{definition}

\begin{rem}
For $S=S_0$, every $\MT$-algebra is reduced and thus a {pseudovariety of reduced $\F$-algebras} is just a pseudovariety of $\F$-algebras, see Definition \ref{def:pseudovar}.
\end{rem}

\begin{definition}\label{def:profinitetheoryreduced}
A \emph{reduced pro-$\F$ theory} is a family \[\rho = (\,\rho_\Sigma\colon \hatt\FSigma 
\epito P_\Sigma\,)_{\Sigma\in \A}\] such that $\rho_\Sigma$ is a reduced $\U_\Sigma$-quotient of $\hatt\FSigma$, and for 
every 
$\MT$-homomorphism $g\colon \MT \FDelta\to \MT\FSigma$ with 
$\Sigma,\Delta\in\A$ there exists a  
morphism $g_P\colon P_\Delta\to P_\Sigma$ with $\rho_\Sigma \o \hat g = 
g_P\o \rho_\Delta$.
\[ 
\xymatrix{
\hatt\FDelta \ar[r]^{\hat g} \ar@{->>}[d]_{\rho_\Delta} & \hatt\FSigma \ar@{->>}[d]^{\rho_\Sigma}\\
P_\Delta \ar@{-->}[r]_{g_P} & P_\Sigma
}
\]
Here $\hat g$ is the continuous extension of $g$, see Lemma \ref{lem:gat}. Reduced pro-$\F$ theories are ordered by $\rho\leq \rho'$ iff 
$\rho_\Sigma\leq \rho_\Sigma'$ for each $\Sigma\in \A$.
\end{definition}
Using the local Reiterman theorem for reduced $\F$-algebras, pro-$\F$ theories can be equivalently described in terms of local pseudovarieties:

\begin{definition}
A \emph{reduced $\F$-theory} is a family $\T=(\,\T_\Sigma\,)_{\Sigma\in \A}$ such that (i) $\T_\Sigma$ is a local pseudovariety of reduced $\U_\Sigma$-quotients of $T\FSigma$ for every $\Sigma\in \A$, and (ii) for every $\MT$-homomorphism $g\colon \MT\FDelta\to \MT\FSigma$ with $\Sigma,\Delta\in \A$ and every $e\colon \MT\FSigma\epito A$ in $\T_\Sigma$, there exists an $\ol e\colon \MT\FDelta\epito B$ in $\T_\Delta$ and a morphism $\ol g\colon B\to A$ with $\ol g \o \ol e = e\o g$.
\[ 
\xymatrix{
\MT\FDelta \ar[r]^{g} \ar@{->>}[d]_{\ol e} & \MT\FSigma \ar@{->>}[d]^{e}\\
B \ar@{-->}[r]_{\ol g} & A
}
\]
Reduced $\F$-theories are ordered by $\T\leq \T'$ iff $\T_\Sigma\seq \T_\Sigma'$ for all $\Sigma \in \A$. 
\end{definition}

\begin{notation}
\begin{enumerate}
\item For any reduced pro-$\F$ theory $\rho$, form the family \[\T^\rho = (\T_\Sigma^\rho)_{\Sigma\in\A}\] where $\T_\Sigma^\rho$ is the local pseudovariety of $\U_\Sigma$-quotients of $\T\FSigma$ associated to $\rho_\Sigma$ (see Construction \ref{rem:usigma_to_localpseudovar}.1).
\item For any reduced $\F$-theory $\T$, form the family \[\rho^\T = (\rho_\Sigma^\T\colon \hatT\FSigma\epito P_\Sigma^\T)\] where $\rho_\Sigma^\T$ is the reduced $\U_\Sigma$-quotient of $\hatt\FSigma$ associated to $\T_\Sigma$ (see Construction \ref{rem:usigma_to_localpseudovar}.2).
\end{enumerate}
\end{notation}
In analogy to Lemma \ref{lem:profiniteconcrete} we get
\begin{lemma}\label{lem:profiniteconcretereduced}
\begin{enumerate}
\item For any reduced pro-$\F$ theory $\rho$, the family $\T^\rho$ forms a reduced $\F$-theory.
\item For any reduced $\F$-theory $\T$, the family $\rho^\T$ forms a reduced pro-$\F$ theory.
\item The maps $\rho\mapsto \T^\rho$ and $\T\mapsto \rho^\T$ give an isomorphism between the poset of reduced pro-$\F$ theories and the poset of reduced $\F$-theories.
\end{enumerate}
\end{lemma}

\begin{proof}
\begin{enumerate}
\item Suppose that $\rho$ is a reduced pro-$\F$ theory. Let $g\colon \MT\FDelta\to\MT\FSigma$ be a $\MT$-homomorphism with $\Sigma,\Delta\in \A$, and  let $e\colon \MT\FSigma\epito A$ in $\T_\Sigma^\rho$. By the definition of $\T_\Sigma^\rho$, the quotient $e$ has the form
\[ e = (T\FSigma \xra{\iota_\FSigma} V\hatt\FSigma \xra{V\rho_\Sigma} VP_\Sigma \xra{Vq} A) \]
for some $q$ in $\hatD^S$. Factorize the morphism $q\o g_P$ into a surjective morphism $\ol q\colon P_\Delta\to C$ followed by an injective/order-reflecting morphism $\ol g\colon C\to A$. By Lemma \ref{lem:subdirectproducthatdreduced}.3, there exists a finite quotient $h\colon P_\Delta \epito C$ such that $h\o \rho_\Delta$ is a reduced $\U_\Sigma$-quotient and  $\ol q\leq h$. Then the reduced $\U_\Sigma$-quotient
\[ \ol e := (T\FDelta \xra{\iota_\FDelta} V\hatt\FDelta \xra{V\rho_\Delta} VP_\Delta \xra{V h} C)\]
lies in $\T_\Delta^\rho$ (by the definition of $\T_\Delta^\rho$) and $e\o g$ factors through $\ol e$, as shown by the diagram below:
\[
\xymatrix{
&T\FDelta \ar@/_3ex/[dddl]_{\ol e} \ar[r]^g \ar[d]_{\iota_\FDelta} & T\FSigma \ar[d]^{\iota_\FSigma} \ar@/^7ex/[ddd]^e \\
&V\hatt\FDelta \ar[r]_{V\hat g}\ar@{->>}[d]_{V\rho_\Delta} & V\hatt\FSigma \ar@{->>}[d]^{V\rho_\Sigma}\\
&VP_\Delta \ar[dl]_{h} \ar@{->>}[d]_{V\ol q} \ar[r]_{Vg_P} & VP_\Sigma \ar@{->>}[d]^{Vq} \\
C \ar@{-->}[r] &B \ar[r]_{\ol g} & A
}
\]
This shows that $\T^\rho$ is a reduced $\F$-theory.
\item Completely analogous to the proof of Lemma \ref{lem:profiniteconcrete}.2.
\item By Lemma \ref{lem:localreitermanusigma}, the two maps $\rho\mapsto \T^\rho$ and $\T \mapsto \rho^\T$ give mutually inverse and order-preserving bijections between the two posets.\qedhere
\end{enumerate}
\end{proof}
 Recall from Lemma \ref{lem:reduction_props} that every reduced $\U_\Sigma$-quotient $e\colon T\FSigma\epito A$ has an associated reduced $\Sigma$-generated $\F$-algebra $e_R\colon \MT\FSigma\epito A_R$, viz. the smallest $\Sigma$-generated $\F$-algebra with $e\leq e_R$. Recall also the morphism $h_\downarrow$ associated to a morphism $h$ in $\D^S$, see Notation \ref{not:downarrow}.

In analogy to Lemma \ref{lem:theory_to_pseudovar_property} we get:
\begin{lemma}\label{lem:theory_to_pseudovar_property_reduced}
Let $\T$ be a reduced $\F$-theory, and let $A$ be an 
$\A$-generated reduced algebra in $\F$. Then the following statements are equivalent:
\begin{enumerate}[(i)]
\item There exists $e\in \T_\Sigma$ for some $\Sigma\in \A$ such that $e_R$ has codomain $A$.
\item For every $\MT$-homomorphism $h\colon \MT\FDelta\to A$ with $\Delta\in \A$, the morphism $h_\downarrow$ factors through some $e\in \T_\Delta$.
\end{enumerate}
\end{lemma}
\begin{proof}
(i)$\Ra$(ii) Let $e\colon T\Sigma\epito A'$ in $\T_\Sigma$ such that $e_R$ has codomain $A$. Given a $\MT$-homomorphism $h\colon \MT\FDelta\to A$ with $\Delta\in \A$, choose a 
$\MT$-homomorphism $g\colon \MT\FDelta\to \MT\FSigma$ with $h = e_R\o g$, see Lemma \ref{lem:projective}. Since $\T$ is an $\F$-theory, we have $e\o g = \ol g \o \ol e$ for some $\ol e\in \T_\Delta$ and some morphism $\ol g$. Thus we get the commutative diagram below, where the morphism $q$ witnesses that $e\leq e_R$, see Lemma \ref{lem:reduction_props}.2.
\[
\xymatrix{
 \MT\FDelta \ar[r]^g \ar[dr]^h \ar@{->>}[d]_{\ol e} & \MT\FSigma \ar@{->>}[d]^{e_R} \ar@{->>}[dr]^{e} & \\
 B \ar@/_5ex/[rr]_{\ol g} & A \ar[r]_q & A'
}
\]
Since $e_{R,s}=e_s$ for $s\in S_0$, see Lemma \ref{lem:reduction_props}.4, we have $q_s = \id$ for $s\in S_0$. Therefore $h_\downarrow = q\o h$, and thus $h_\downarrow$ factors through $\ol e$ via $\ol g$.
This proves (ii).

(ii)$\Ra$(i)  Since $A$ is $\A$-generated, there exists a surjective 
$\MT$-homomorphism $h\colon \MT\FSigma\epito A$ for some $\Sigma\in\A$. 
Thus, by (ii), $h_\downarrow$ factors through some $q\in \T_\Sigma$, i.e. $h_\downarrow\leq q$. Since $\T_\Sigma$ is downwards closed, we have $h_\downarrow\in \T_\Sigma$. Moreover, since $h$ is a reduced quotient of $\MT\FSigma$,  Lemma \ref{lem:usigma_vs_reduced} shows that $h_{\downarrow,R}=h$ and thus $h_{\downarrow,R}$ has codomain $A$. Therefore (i) holds with $e:=h_\downarrow$.
\end{proof}

\begin{notation}
For any reduced $\F$-theory $\T$, let $\V^\T$ be the class of all $\A$-generated reduced algebras $A\in \F$ satisfying the equivalent properties of the above lemma.
\end{notation}

\begin{lemma}\label{lem:theory_to_pseudovar_welldef_reduced}
If $\T$ is a reduced $\F$-theory, then $\V^\T$ is a pseudovariety of $\F$-algebras.
\end{lemma}

\begin{proof}
Let $A_i$ ($i\in I$) be a finite family of algebras in $\V^\T$ and suppose that $A$ is an $\A$-generated reduced $\MT$-algebra that divides $\prod_i A_i$. That is, there exists an injective/order-reflecting $\MT$-homomorphism $m\colon B\monoto \prod_i A_i$ and a surjective $\MT$-homomorphism $q\colon B\epito A$ for some $\MT$-algebra $B$. We denote the product projections by
$\pi_i\colon \prod_{i} A_i\to A_i$. To show that $A$ lies in $\V^\T$, we verify the criterion of Lemma \ref{lem:theory_to_pseudovar_property_reduced}(ii). Thus let $h\colon \MT\FDelta\to A$ be a $\MT$-homomorphism. We need to show that $h_\downarrow$
factors through some $e\in \T_\Delta$. First, since $\MT\FDelta$ is projective (see Lemma \ref{lem:projective}), we can choose a $\MT$-homomorphism $\ol h$ with $h=q\o \ol h$.
 For each $i$, the $\MT$-homomorphism $h_i := \pi_i\o m\o \ol h$ has codomain $A_i\in \V^\T$. Thus, by the definition of $\V^\T$, the morphism $h_{i,\downarrow}$ factors through some $e_i\in \T_\Delta$; that is, $h_{i,\downarrow} = k_i\o e_{i}$ for some $k_i$. Since $\T_\Delta$ is a local pseudovariety, there exists $e\in \T_\Delta$ with $e_i\leq e$ for all $i$, i.e. $e_{i}= g_i\o e$ for some $g_i$. Thus we get the commutative diagram below, where $r_i$ is the unique morphism with $h_{i,\downarrow} = r_i \o h_i \,(= r_i \o \pi_i\o m\o \ol h)$.
\[  
\xymatrix@R+1em{
T\FDelta \ar@{->>}@/^5ex/[rr]^{e_{i}}\ar@/_5ex/[dd]_h \ar@{->>}[r]^{e} \ar[d]_{\ol h} & C  
\ar@{->>}[r]^{g_i} & C_i \ar[dr]^{k_i}\\
B \ar@{->>}[d]_q \ar@{ >->}[r]_m & \prod_i A_i \ar[r]_{\pi_i} & A_i \ar[r]_{r_i}& A_{i,\downarrow}\\
A &&&
}
\]
From this we get the following commutative square, using that $\prod_i A_{i,\downarrow} = (\prod_i A_{i})_\downarrow $:
\[
\xymatrix{
T\Delta \ar@{->>}[r]^{e} \ar[d]_{\ol h} & C \ar[d]^{\langle k_i\o g_i\rangle} \\
B \ar@{>->}[r]_{m_\downarrow} & \prod_i A_{i,\downarrow}
}
\]
Diagonal fill-in gives a morphism $d\colon C\to B$ with $d\o e = \ol h$. But then $h_\downarrow = q_\downarrow\o \ol h = q_\downarrow\o d\o e$, showing that $h_\downarrow$ factors through $e\in \T_\Delta$.
By the definition of $\V^\T$, this proves that $A\in \V^\T$.
\end{proof}

\begin{notation}
For any pseudovariety $\V$ of reduced $\F$-algebras and $\Sigma\in \A$ form the family
\[  \T^\V = (\,\T_\Sigma^\V\,)_{\Sigma\in \A}\]
where, for each $\Sigma\in \A$,
\[  \T_\Sigma^\V = \{\,e\colon T\FSigma\epito A \,:\, \text{$e$ is a reduced $\U_\Sigma$-quotient with $A_R\in \V$}\,\}.\]
\end{notation}

\begin{lemma}
For any pseudovariety $\V$ of reduced $\F$-algebras, $\T^\V$ is an $\F$-theory.
\end{lemma}

\begin{proof}
\begin{enumerate}[(a)]
\item We first show that $\T_\Sigma^\V$ is a local pseudovariety of reduced $\U_\Sigma$-quotients for each $\Sigma\in \A$.
 To show that $\T_\Sigma^\V$ is downwards closed, let $e\colon T\FSigma\epito A$ and $q\colon T\FSigma\epito B$ be two reduced $\U_\Sigma$-quotients with $e\in \T_\Sigma^\V$ and $q\leq e$. Then $q_R\leq e_R$ by Lemma \ref{lem:usigma_vs_reduced}, i.e. $B_R$ is a quotient of $A_R$. Since $B_R$ is reduced and $A_R\in \V$, this implies that $B_R\in \V$ because $\V$ is closed under reduced quotients. Therefore $q\in \T_\Sigma^\V$ by the definition of $\T_\Sigma^\V$.

To show that $\T_\Sigma^\V$ is directed, let $e_i\colon T\FSigma\epito A_i$ ($i=0,1$) be two reduced $\U_\Sigma$-quotients in $\T_\Sigma^\V$. We need to show that there exists a reduced $\U_\Sigma$-quotient $e$ in $\T_\Sigma^\V$ with $e_0,e_1\leq e$. Let $e\colon T\FSigma\epito A$ be the subdirect product of $e_0$ and $e_1$ in $\D^S$, i.e. that smallest quotient of $T\FSigma$ with $e_0,e_1\leq e$. By Lemma \ref{lem:subdirectproductreduced}, $e$ is a reduced $\U_\Sigma$-quotient, and moreover $e_{R}$ is the subdirect product of the two $\Sigma$-generated $\F$-algebras $e_{i,R}\colon \MT\FSigma\epito A_{i,R}$ in $\Alg{\MT}$. Thus $A_R$ is a subalgebra of the product $A_{0,R}\times A_{1,R}$. Moreover, $A_R$ is clearly $\A$-generated, and $A_R$ is reduced by Lemma \ref{lem:reduction_props}.3. It follows that $A_R\in \V$ because $A_{0,R}, A_{1,R}\in \V$ and $\V$ is closed under division of finite products. Thus $e$ lies in $\T_\Sigma^\V$. This concludes the proof that $\T_\Sigma^\V$ is a local pseudovariety.
\item To show that $\T^\V$ forms an $\F$-theory, let $g\colon \MT\FDelta\to\MT\FSigma$ be a $\MT$-homomorphism with $\Sigma,\Delta\in \A$ and let $e\colon T\FSigma\epito A$ in $\T_\Sigma^\V$. Factorize $e_R\o g$ into a surjective $\MT$-homomorphism $h\colon \MT\Delta\epito B$ followed by an injective/order-reflecting $\MT$-homomorphism $m\colon B\monoto A_R$. Let $\ol e := h_\downarrow\colon T\FDelta\epito B'$. Now consider the following diagram:
\[
\xymatrix{
T\FDelta \ar[r]^g \ar@{->>}[d]_{\ol e} \ar@/_5ex/[dd]_{\ol e_R} \ar@{->>}@/_10ex/[ddd]_h  & T\FSigma \ar@{->>}[d]^e \ar@{->>}@/^5ex/[ddd]^{e_R}\\
B' \ar[r]^{\ol g} & A \\
B'_R \ar@{-->>}[u] &\\
B \ar@{-->>}[u] \ar@{>->}[r]_m & A_R \ar@{-->>}[uu]
}
\]
Here the dashed arrows are witnessing that $\ol e\leq \ol e_R\leq h$ and $e\leq e_R$, see Lemma \ref{lem:reduction_props}.2. Note that by Lemma \ref{lem:reduction_props}.4, on sorts $s\in S_0$ the morphisms $\ol e$, $\ol e_R$ and $h$ agree, as do the morphisms $e$ and $e_R$, and thus the dashed arrows are carried by identity morphisms on these sorts. Therefore there exists a unique morphism $\ol g$ making the lower part of the diagram commute; specifically, $\ol g$ agrees with $m$ on sorts $s\in S_0$, and is uniquely determined on sorts $s\not\in S_0$ because $A_s$ has at most one element. Then the outer part of the diagram and all parts except possibly the upper square commute, which then implies that also the upper square commutes. Moreover, we have $\ol e\in \T_\Sigma^\V$ since the algebra $B_R'$ divides  $A_R\in \V$ and thus lies in $\V$.

\end{enumerate}
\end{proof}

\begin{lemma}\label{lem:theory_pseudovar_equivalence1_reduced}
For any pseudovariety $\V$ of reduced $\F$-algebras we have $\V = \V^\T$ where $\T := \T^\V$.
\end{lemma}

\begin{proof}
To show $\V\seq\V^\T$, let $A\in \V$. Since $A$ is $\A$-generated, there exists a surjective $\MT$-homomorphism $e\colon \MT\FSigma\epito A$ with $\Sigma\in \A$, and since $A$ is reduced, we have $e=e_{\downarrow,R}$ by Lemma \ref{lem:usigma_vs_reduced}. Then $e_\downarrow\in \T_\Sigma$ by the definition of $\T_\Sigma = \T_\Sigma^\V$ and therefore $A\in \V^\T$ by the definition of $\V^\T$.

For the reverse inclusion $\V^\T\seq \V$, let $A\in \V^\T$. Then, for some $\Sigma\in\A$, there exists $e\in \T_\Sigma$ such that $e_R$ has codomain $A$. But by the definition of $\T_\Sigma = \T_\Sigma^\V$, this means that $A\in \V$. 
\end{proof}

\begin{lemma}\label{lem:theory_pseudovar_equivalence2_reduced}
For any reduced $\F$-theory $\T$
we 
have $\T = \T^\V$, where $\V:=\V^\T$.
\end{lemma}

\begin{proof}

 To show $\T\seq\T^\V$, let $e\colon \MT\FSigma\epito A$ in $\T_\Sigma$. Then $A_R\in \V^\T=\V$ by the definition of $\V^\T$. But then $e\in \T^\V$ by the definition of $\T^\V$.

For the reserve inclusion $\T^\V\seq \T$, let $e\colon \MT\FSigma\epito A$ in $\T_\Sigma^\V$. By the definition of $\T_\Sigma^\V$, this means that $A_R\in \V=\V^\T$. By the definition of $\V^\T$, there exists $\Delta\in \A$ and $q\colon T\FDelta\epito B$ in $\T_\Delta$ with $B_R=A_R$. Using projectivity of free $\MT$-algebras (see Lemma \ref{lem:projective}), choose a $\MT$-homomorphism $g\colon \MT\FSigma\to \MT\FDelta$ with $q_R\o g = e_R$. This implies \[q\o g = q_{R,\downarrow}\o g = (q_R\o g)_\downarrow = e_{R,\downarrow}=e\] by Lemma \ref{lem:usigma_vs_reduced}.  Since $\T$ is a reduced $\F$-theory, there exists $\ol q\in \T_\Sigma$ and $\ol g$ with $q\o g = \ol g \o \ol q$.
\[
\xymatrix{
T\FSigma \ar[r]^g \ar@{->>}[d]_{\ol q} \ar@{->>}[dr]^e & T\FDelta \ar@{->>}[d]^{q}\\
C \ar@{->>}[r]_{\ol g}& B
}
\]
Therefore $e=\ol g \o \ol q$, i.e. $e\leq \ol q$. Since $\T_\Sigma$ is downwards closed, it follows that $e\in \T_\Sigma$.
\end{proof}

\begin{lemma}\label{lem:ftheory_vs_pseudovar_reduced}
The maps 
\[\T\mapsto \V^\T\quad\text{and}\quad \V\mapsto \T^\V\]
give an isomorphism between the lattice of reduced $\F$-theories and the lattice of pseudovarieties of reduced $\F$-algebras (ordered by inclusion).
\end{lemma}

\begin{proof}
By Lemma \ref{lem:theory_pseudovar_equivalence1_reduced} and \ref{lem:theory_pseudovar_equivalence2_reduced}, the two maps are mutually inverse. Moreover, they are clearly order-preserving, which proves the claim.
\end{proof}
From this and Lemma \ref{lem:profiniteconcretereduced} we get the main result of the present subsection:

\begin{theorem}[Reiterman Theorem for Reduced $\F$-algebras]\label{thm:reiterman_reduced}
The lattice of
pseudovarieties of reduced $\F$-algebras is isomorphic to the 
lattice of reduced pro-$\F$
theories.
\end{theorem}

\subsection{Reduced varieties of languages and the Reduced Variety Theorem}
In thus subsection, we present a variety theorem dealing with languages that are ``empty'' outside of the sorts in $S_0$.

\begin{definition}
A \emph{reduced language} over $\Sigma\in\Set_f^S$ is a family of morphisms in $\D$
\[ L = (\,L_s\colon (T\FSigma)_s\to O_\D\,)_{s\in S_0}. \]
It is called \emph{$\F$-recognizable} if there exists a $\MT$-homomorphism $h\colon \MT\FSigma\epito (A,\alpha)$ with codomain in $\F$ and a family $(p_s\colon A_s\to O_\D)_{s\in S_0}$ in $\D$ with $L_s = p_s\o h_s$ for all $s\in S_0$. We write $\Rec_{\F,R}(\Sigma)$ for the set of all reduced $\F$-recognizable languages over $\Sigma$.
\end{definition}

\begin{example}\label{ex:reducedlang_omegasem}
Let $\MT=\MT_\infty$ in $\Set^S$ (where  $S=\{+,\omega\}$) with $O_\Set=\{0,1\}$ and $\F=$ all finite $\omega$-semigroups. Choose $S_0=\{\omega\}$. Since $T_\infty(\Sigma,\emptyset) = (\Sigma^+,\Sigma^\omega)$, a reduced language corresponds precisely to an \emph{$\omega$-language}, i.e. a language $L\seq \Sigma^\omega$ of infinite words. The language $L$ is recognizable if there exists an $\omega$-semigroup homomorphism $h\colon (\Sigma^+,\Sigma^\omega)\to A$ into a finite $\omega$-semigroup $A$ and a subset $Y\seq A_\omega$ with $L=h^{-1}[Y]$. Recognizable languages are precisely the $\omega$-regular languages, i.e. $\omega$-languages accepted by finite B\"uchi automata \cite{pinperrin04}.
\end{example}

\begin{rem}
In complete analogy to Remark \ref{rem:reg-as-C} and Proposition \ref{prop:reg-as-C}, we have
\[ \Rec_{\F,R}(\Sigma)\cong \prod_{s\in S_0} P(\hatt\FSigma)_s, \]
and $\Rec_{\F,R}(\Sigma)$ forms a subobject of $\prod_{s\in S_0} O_\C^{\under{T\FSigma}_s}$.
\end{rem}

\begin{definition}
Given an $S_0$-indexed family $\mathcal{L} = ( L^s )_{s\in S_0}$ of reduced languages over $\Sigma$, the \emph{diagonal} of $\mathcal{L}$ is the reduced language $\Delta\mathcal{L}$ with $(\Delta\mathcal{L})_s = L_s^s: (T\FSigma)_s\to O_\D$ for $s\in S_0$.
\end{definition}

\begin{rem}
In analogy to Lemma \ref{lem:diagonalclosed}, a subobject $W_\Sigma\seq \Rec_{\F,R}(\Sigma)$ is closed under diagonals iff it has the form $\prod m_s\colon \prod_{s\in S_0} W_\Sigma' \monoto \prod_{s\in S_0} P(\hatt\FSigma)_s$ for monomorphisms $m_s$ in $\C$.
\end{rem}

\begin{definition}\label{def:derivatives_preimages_reduced} Let $L$ be a reduced language over $\Sigma\in\Set_f^S$. Then we define the following languages:
\begin{enumerate}
\item The \emph{derivative} $u^{-1}L$ of $L$ w.r.t. a unary operation $u\colon
(T\FSigma)_s\to(T\FSigma)_{t}$ with $s,t\in S_0$ is the reduced language over~$\Sigma$ defined by
\[
      \begin{cases}
        (u^{-1}L)_s = (T\FSigma)_s \xra{\;u\;} (T\FSigma)_{t} \xra{\;L_{t}\;} O_\D &\\
        (u^{-1}L)_r = (T\FSigma)_r \xra{\;\iota_\FSigma\;} (V\hatt\FSigma)_r \xra{\;V\bot\;}
        O_\D & \text{for $r\in S_0\setminus\{ s\}$.}
      \end{cases}
    \]
\item The \emph{preimage} $g^{-1}L$ of $L$ under a $\MT$-homomorphism $g\colon
  \MT\FDelta\to \MT\FSigma$ is the reduced language over~$\Delta$ defined for all $s\in S_0$ by 
\[ (g^{-1}L)_s = (T\FDelta)_s \xra{g_s} (T\FSigma)_s \xra{L_s} O_\D.\]
\end{enumerate}
\end{definition}

\begin{definition}[Local variety of reduced languages]\label{def:localvariety_reduced}
A \emph{local variety of reduced languages} over
    $\Sigma\in \A$ is a subobject
  \[W_\Sigma\seq \Rec_{\F,R}(\Sigma)\] in $\C$ closed under $\ol \U_\Sigma$-derivatives and diagonals. That is,
\begin{enumerate}[(i)]
\item for each $L\in W_\Sigma$ and $u\in\ol \U_\Sigma(s,t)$ with $s,t\in S_0$, one has $u^{-1}L\in W_\Sigma$;
\item for each $S_0$-indexed family  $\mathcal{L} = (\,L^s\,)_{s\in S_0}$ of reduced languages in $W_\Sigma$, one has $\Delta \mathcal{L}\in W_\Sigma$.
\end{enumerate}
\end{definition}

\begin{rem}\label{rem:diagclosed_reduced}
Again, if $\C$ is one of the four categories of Example \ref{ex:categories}, the closure under diagonals in the above definition can be dropped; see the argument in Remark \ref{rem:diagclosed}.
\end{rem}

\begin{theorem}[Local Variety Theorem for Reduced Languages]\label{thm:localeilenberg_reduced}
For each $\Sigma\in \A$, the lattice of local varieties of reduced languages over $\Sigma$ is isomorphic to the lattice local pseudovarieties of $\Sigma$-generated reduced $\F$-algebras.
\end{theorem}

\begin{proof}
Duality + local Reiterman! In analogy to the proof of the Local Variety Theorem (see Theorem \ref{thm:localeilenberg}), local varieties of reduced languages over $\Sigma$ dualize precisely to the concept of a reduced $\U_\Sigma$-quotient of $\hatt\FSigma$ in $\hatD^S$. And by the local Reiterman theorem for reduced $\F$-algebras (see Theorem \ref{thm:localeilenberg_reduced}), the latter correspond to local pseudovarieties of $\Sigma$-generated reduced $\F$-algebras.
\end{proof}

\begin{definition}[Variety of languages]\label{def:variety_reduced}
A \emph{variety of reduced languages} is a family of local varieties of reduced languages
\[(\,W_\Sigma\seq \Rec_{\F,R}(\Sigma)\,)_{\Sigma\in\A}\] closed
  under preimages, i.e. for each $L\in W_\Sigma$ and each $\MT$-homomorphism $g\colon \MT\FDelta\to \MT\FSigma$ with $\Sigma,\Delta\in \A$ one has $g^{-1}L\in
  W_\Delta$.
\end{definition}
Then we get the following version of the Variety Theorem:

\begin{theorem}[Variety Theorem for Reduced Languages]
\label{thm:eilenberg_reduced}
The lattice of varieties of reduced languages (ordered by inclusion) is isomorphic to the lattice of
pseudovarieties of reduced $\F$-algebras.
\end{theorem}

\begin{proof}
Duality + Reiterman! In analogy to the proof of Theorem \ref{thm:eilenberg}, varieties of reduced languages dualize precisely to pro-$\F$ theories, and the Reiterman theorem for reduced $\F$-algebras (see Theorem \ref{thm:reiterman_reduced}) asserts that the latter correspond to pseudovarieties of reduced $\F$-algebras.
\end{proof}

\section{Applications}\label{sec:applications}
In this section, we demonstrate that a  number of concrete Eilenberg-type correspondences, including \emph{new} results, emerge as special instances from our categorical framework. To derive a variety theorem for a given type of languages, we apply the following three-step procedure, as indicated in the introduction:

\smallskip
\noindent\textbf{Step 1.} Choose a monad $\MT$ and a class $\F$ of finite $\MT$-algebras such that the $\F$-recognizable languages coincide with the given type of languages.

\smallskip
\noindent\textbf{Step 2.} Choose a class $\A\seq\Set_f^S$ of alphabets and, for each $\Sigma\in\A$, a unary presentation $\U_\Sigma$ of the free $\MT$-algebra $\MT\FSigma$.

\smallskip
\noindent\textbf{Step 3.} Spell out what (local) varieties of languages and (local) pseudovarieties of $\F$-algebras are by instantiating our definitions to $\MT$, $\F$, $\A$ and $\U_\Sigma$. Then invoke the (reduced) Variety Theorem.

\subsection{The classical Eilenberg theorem}
As our first example, we show how to recover the original Eilenberg theorem for regular languages. It follows from the Variety Theorem by choosing the following parameters:
\begin{table}[h]
\begin{tabularx}{\columnwidth}{lX}
\hline \vspace{-0.4cm}\\
$\C/\D$         & $\BA$/$\Set$ \\
$\MT$ & the monad $\MT_*$ on $\Set$ with $\Alg{\MT_*} \cong$ monoids \\
$\F$ & all finite monoids\\
$\A$ &  $\Set_f$\\
$\U_\Sigma$ & $\{ x\mapsto yx,\, x\mapsto xy \;:\; y\in \Sigma^* \}$\\
\hline
\end{tabularx}
\end{table}

\smallskip
\noindent\textbf{Step 1.} Choose the free-monoid monad $\MT=\MT_\ast$ on $\Set$ and $\F=$ all finite monoids. Regular languages are precisely the languages recognizable by finite monoids, cf. Example \ref{ex:rec_monoids}.

\smallskip
\noindent\textbf{Step 2.} Choose $\A=\Set_f$. For each $\Sigma\in \A$, choose the unary presentation $\U_\Sigma$ of $\Sigma^*$ containing the operations $x\mapsto yx$ and $x\mapsto xy$, where $y\in\Sigma^*$; see Example \ref{ex:unpres_monoids}. Recall from Example \ref{ex:derivatives} that the $\U_\Sigma$-derivatives of a language $L\seq \Sigma^*$ are the classical derivatives.

\smallskip
\noindent\textbf{Step 3.} Instantiating Definition \ref{def:localvariety} and \ref{def:variety} to $\MT$, $\F$, $\A$ and $\U_\Sigma$ gives the following concepts:
\begin{enumerate}
\item A \emph{local variety of regular languages} is a subset $W_\Sigma\seq \Rec(\Sigma)$ languages closed under boolean operations and derivatives.
\item A \emph{variety of regular languages} is a family of local varieties $(W_\Sigma\seq \Rec(\Sigma))_{\Sigma\in\Set}$ closed under preimages, i.e. $g^{-1}L\in W_\Delta$ for every language $L\in W_\Sigma$ and every monoid morphism $g\colon \Delta^*\to\Sigma^*$. This is precisely Eilenberg's original concept.
\end{enumerate} 
Since $\A=\Set_f$, all finite monoids are $\A$-generated by Example \ref{ex:agenerated_talgs}. Therefore, instantiating Definition \ref{def:localpseudovariety} and \ref{def:pseudovar} gives:
\begin{enumerate}
\item A \emph{local pseudovariety of monoids over $\Sigma$} is class of $\Sigma$-generated finite monoids $e\colon \Sigma^*\epito M$ closed under quotients and subdirect products.
\item A \emph{pseudovariety of monoids} is a class of finite monoids closed under quotients, submonoids and finite products.
\end{enumerate}
Theorem \ref{thm:localeilenberg} and \ref{thm:eilenberg} give Eilenberg's theorem \cite{eilenberg76} and its local version due to Gehrke, Grigorieff and Pin \cite{ggp08}:
\begin{theorem}
The lattice of (local)  varieties of regular languages is isomorphic to the lattice of (local) pseudovarieties of monoids.
\end{theorem}

\subsection{The variety theorem for $\D$-monoids}
Eilenberg's theorem \cite{eilenberg76} has several variants in the literature that consider classes of languages with modified closure properties (e.g. dropping closure under complement) and replace monoids by ordered monoids \cite{pin95}, idempotent semirings \cite{polak01} and $\Field$-algebras \cite{reu80}.  In our previous work \cite{ammu14,ammu15,cu15} we unified these results to an Eilenberg theorem for $\D$-monoids in an arbitrary commutative variety $\D$. Our approach in loc. cit.  was automata-theoretic in nature and exploited the fact that finite automata can be dually modeled as algebras and coalgebras for suitable functors. Here we give another, completely algebraic proof: we show how to deduce the variety theorem for $\D$-monoids in our categorical framework.

In this subsection, we assume that $\D$ is a \emph{commutative} variety of algebras or ordered algebras. This means that for any $A,B\in \D$ 
the 
hom-set $\D(A,B)$ carries a subobject of $B^\under{A}$ in $\D$. For example, the varieties $\D=\Set$, $\Pos$, $\JSL$ and $\Vect_{\Field}$ of Example \ref{ex:categories} are commutative. Given $A,B,C\in \D$, by a \emph{$\D$-bimorphism} from $A,B$ to $C$ is meant a function $f\colon \under{A}\times\under{B}\to \under{C}$ such that the maps $f(a,\dash)\colon B\to C$ and $f(\dash,b)\colon A\to C$ are morphisms of $\D$ for every $a\in \under{A}$ and $b\in \under{B}$. The \emph{tensor product} of $A$ and $B$ is an object $A\otimes B$ in $\D$ equipped with a universal bimorphism $t_{A,B}\colon \under{A}\times \under{B}\to \under{A\otimes B}$, i.e. such that for every bimorphism $f\colon \under{A}\times\under{B}\to\under{C}$ there exists a unique morphism $\ol f\colon A\otimes B\to C$ in $\D$ making the following triangle commutative:
\[
\xymatrix{
\under{A}\times\under{B} \ar[r]^{t_{A,B}} \ar[dr]_f & \under{A\otimes B} \ar[d]^{\ol f}\\
& \under{C}
}
\]
Tensor products exist in every commutative variety $\D$, making $(\D,\otimes,\one_\D)$  a symmetric monoidal closed 
category; see e.g. \cite[Theorem 3.10.3]{Bor09}. Since the tensor product $\otimes$ represents bimorphisms, the monoid objects in the category $(\D,\otimes,\one_\D)$ have the following concrete algebraic description: a \emph{$\D$-monoid} $(D,\bullet,1)$ consists of an
object $D\in\D$ equipped with a monoid structure $(\under{D},\bullet,1)$ on the underlying set whose multiplication $\bullet\colon \under{D}\times \under{D}\to \under{D}$ is
a \emph{$\D$-bimorphism}. A \emph{morphism} of $\D$-monoids is a 
morphism in $\D$ that preserves the monoid structure. We denote by $\Mon{\D}$ the category of $\D$-monoids and their morphisms. Note that $\D$-monoids in $\D=\Set$, $\Pos$, $\JSL$ and $\Vect_{\Field}$ are precisely classical monoids, ordered monoids, idempotent semirings, and $\Field$-algebras.

To apply our framework, we need to investigate the construction of free $\D$-monoids. Consider the following diagram of adjunctions, where $\CatMon$ is the category of monoids (in $\Set$), the functors $U$ , $U'$, $\under{\dash}$ are the forgetful functors and $F$, $F'$, $(\dash)^*$, $\Psi$ are their left adjoints. By the uniqueness of adjoints, both the outer and the inner square commute.
\[
\xymatrix{
\D \ar@<-1ex>[rr]_{F'} \ar@<-1ex>[d]_{\under{\mathord{-}}} && \Mon{\D} 
\ar@<1ex>[d]^U_\dashv \ar@<-1.5ex>[ll]_{U'}^\top \\
\Set \ar@<-1.5ex>[u]_\Psi^\vdash \ar@<1ex>[rr]^{(\mathord{-})^*}_\bot && 
\CatMon \ar@<1.5ex>[ll]^{\under{\mathord{-}}} \ar@<1.5ex>[u]^F
}
\]
For notational simplicity, let us assume that $X\seq\under{\Psi X}$ for every set $X$ and that the unit $\eta_X\colon X\to \under{\Psi X}$ of the adjunction $\Psi\dashv \under{\dash}\colon \D\to\Set$ is the inclusion map. The left adjoint $F$ sends a monoid $M=(M,\o,e)$ to the $\D$-monoid $FM = 
(\Psi\under{M},\bullet,e)$, where the multiplication
$\bullet\colon \Psi\under{M}\times \Psi\under{M} \to \Psi\under{M}$ is 
the unique 
bimorphism in $\D$ that extends the multiplication $\o\colon \under{M}\times 
\under{M}\to\under{M}$ of $M$. Moreover, $F$ sends a monoid morphism $h\colon M\to M'$ to the $\D$-monoid morphism
$\Psi \under{h}\colon \Psi \under{M}\to \Psi \under{M'}$. This implies that the 
free $\D$-monoid on a 
free object $\FSigma=\Psi \Sigma$ in $\D$ is given by
\begin{equation}\label{eq:tm}
 F'\FSigma = F'\Psi\Sigma = F\Sigma^* = 
(\Psi\Sigma^*,\bullet,\epsilon), \end{equation}
where $\epsilon$ is the empty word and $\bullet$ extends the concatentation of 
words. In the following, we write $\FSigma^*$ for the object $\Psi\Sigma^*\in \D$.

\noindent After these preparations, let us apply our three-step plan to derive the variety theorem for $\D$-monoids. We choose the parameters as follows:
\begin{table}[h]
\begin{tabularx}{\columnwidth}{lX}
\hline \vspace{-0.4cm}\\
$\C/\D$         & varieties satisfying the Assumptions \ref{ass:dt} and \ref{ass:catframework}(i),(iii) with $\D$ commutative \\
$\MT$ & the monad $\MT_M$ on $\D$ with $\Alg{\MT_M} \cong \Mon{\D}$ \\
$\F$ & all finite $\D$-monoids\\
$\A$ &  $\Set_f$\\
$\U_\Sigma$ & $\{\, \FSigma^* \xra{y\bullet \dash} \FSigma^*,\,\FSigma^* \xra{\dash\bullet y} \FSigma^* : y\in\Sigma^*\,\}$\\
\hline
\end{tabularx}
\end{table}

\smallskip
\noindent\textbf{Step 1.} Let $\MT=\MT_M$ be the monad on $\D$
associated to the adjunction $F'\dashv U'$. Since both $\D$ and $\Mon{\D}$ are varieties of (ordered) algebras, the forgetful functors $\D\to \Set$ and $\Mon{\D}\to \Set$ (resp. $\D\to \Pos$ and $\Mon{\D}\to \Pos$) are monadic. Therefore, $U'$ is monadic by~\cite[Corollary~4.5.7]{Bor09}, and thus $\Alg(\MT_M) \cong \Mon{\D}$. A language $L\colon 
T_M\FSigma = \FSigma^* \to O_\D$ in $\D$ 
corresponds (via 
the adjunction $\Psi\dashv \under{\mathord{-}}\colon \D\to \Set$) to a map 
$L^@\colon \Sigma^*\to \under{O_\D}$, and conversely a map $L\colon \Sigma^*\to \under{O_\D}$ can be uniquely extended to a language $L^\D\colon \FSigma^*\to O_\D$. Let us call a map of the form $L\colon \Sigma^*\to \under{O_\D}$ a \emph{behavior}. It is \emph{regular} if is computed by some finite  automaton with output set $\under{O_\D}$. Equivalently, there exists a monoid morphism $h\colon \Sigma^*\to M$ with finite codomain and a function $p\colon M\to \under{O_\D}$ with $L=p\o h$. This is completely analogous to the characterization of regular languages via finite monoids, cf. Example \ref{ex:rec_monoids}.

\medskip
\noindent\textbf{Claim.}
A language $L\colon \FSigma^*\to O_\D$ is recognizable iff the behavior $L^@$ is regular.

\begin{proof}
If $L$ is recognizable, there exists a $\D$-monoid 
morphism $h\colon \FSigma^*\to D$, where $D$ is finite, and a morphism $p\colon D\to 
O_\D$ in $\D$ with $L=p\o h$. Then $h$ restricts to a monoid morphism 
\[h' = (\Sigma^* \monoto U\FSigma^* \xra{Uh} UD ) \]
that recognizes $L^@$ via $\under{p}$. Therefore $L^@$ is regular.

Conversely, suppose that $L^@$ is a regular behavior. Then $L^@$ is
monoid-recognizable (in $\Set$), so there exists a monoid morphism
$h\colon \Sigma^*\to M$, where $M$ is a finite monoid, and a function
$p\colon M\to \under{O_\D}$ such that $L^@=p\o h$. Let $p'\colon \Psi M\to O_\D$
in $\D$ be the adjoint transpose of $p$ (via the adjunction
$\Psi \dashv \under{\mathord{-}}\colon \D\to \Set$). Then
$\Psi h\colon \FSigma^*\to \Psi M$ is a $\D$-monoid morphism that
recognizes $L$ via $p'$, where $\Psi M$ is finite because $\D$ is
assumed to be a locally finite variety (see
Assumptions~\ref{ass:dt}).
\end{proof}
\noindent\textbf{Step 2.} Let $\A=\Set_f$. For each $\Sigma\in \A$, a unary presentation of the free $\D$-monoid $\FSigma^*$ is given by 
\[\U_\Sigma = 
\{\, \FSigma^* \xra{y\bullet \dash} \FSigma^*,\,\FSigma^* \xra{\dash\bullet y} \FSigma^* : y\in\Sigma^*\,\}.\]
To see this, we need to show that a $\D$-congruence $\equiv$ in $\FSigma^*$ forms a $\D$-monoid congruence (i.e. $x\equiv x'$ and $y\equiv y'$ implies $x\bullet y \equiv x'\bullet y'$) iff $\equiv$ is stable under left and right multiplication by words in $\Sigma^*\seq \FSigma^*$. The ``only'' if direction is clear. For the ``if'' direction, suppose that $\equiv$ is stable under left and right multiplication by words. It suffices to show that $\equiv$ is stable under left and right multiplication: $x\equiv x'$ implies $y\bullet x \equiv y\bullet x'$ and $x\bullet y \equiv x'\bullet y$ for all $y\in\FSigma^*$. By symmetry, we only show the stability under right multiplication. Thus let $x,x'\in\FSigma^*$ with $x\equiv x'$ and $y\in\FSigma^*$. Since $\FSigma^*$ is generated by the set $\Sigma^*$ as a $\D$-algebra, we can express $y$ as a term over the generators; that is, there is a term $t$ over the signature of $\D$ with $y=t^{\FSigma^*}(y_1,\ldots,y_n)$, where $y_1,\ldots, y_n\in\Sigma^*$. It follows that
\begin{align*}
x\bullet y &= x\bullet t^{\FSigma^*}(y_1,\ldots,y_n)\\
&= t^{\FSigma^*}(x\bullet y_1,\ldots,x\bullet y_n)\\
&\equiv t^{\FSigma^*}(x'\bullet y_1,\ldots, x'\bullet y_n)\\
&=  x'\bullet t^{\FSigma^*}(y_1,\ldots,y_n)\\
&= x'\bullet y.
\end{align*}
In the second and fourth step we use that $x\bullet\dash$ and $x'\bullet \dash$ are morphisms on $\FSigma^*$ in $\D$, and that morphisms preserve term operations. In the third step, we use that $x\bullet y_i \equiv x'\bullet y_i$ ($i=1,\ldots,n)$ since $x\equiv x'$ and $\equiv$ is stable under right multiplication by words, and that $\equiv$ is a $\D$-congruence (and thus stable under term operations).

\smallskip
\noindent\textbf{Step 3.} 
Instantiating Definition \ref{def:variety} to $\MT$, $\F$, $\A$ and $\U_\Sigma$ (and identifying recognizable languages with regular behaviors $L\colon \Sigma^*\to \under{O_\D}$ as explained above) gives the following concepts:
\begin{enumerate}
\item A \emph{local variety of regular behaviors in $\C$ over $\Sigma\in\Set_f$} is a set $W_\Sigma$ of regular behaviors $L\colon \Sigma^*\to \under{O_\D}$ closed under $\C$-algebraic operations (cf. Remark \ref{rem:reg-as-C} and Example \ref{ex:ocod}) and derivatives; that is, for every behavior $L\colon \Sigma^*\to \under{O_\D}$ in $W_\Sigma$ and every word $y\in\Sigma^*$, the behaviors $y^{-1}L$ and $Ly^{-1}$ defined by
\[ y^{-1}L(x) = L(yx)\quad\text{and}\quad Ly^{-1}(x)=L(xy) \]
lie in $W_\Sigma$.
\item A \emph{variety of regular behaviors in $\C$} is a family $(W_\Sigma)_{\Sigma\in\Set}$ of local varieties of regular behaviors closed under preimages of $\D$-monoid morphisms; that is, for every $\D$-monoid morphism $h\colon \FDelta^*\to \FSigma^*$ and $L\colon \Sigma^*\to \under{O_\D}$ in $W_\Sigma$, the behavior 
\[ h^{-1}L = (\Delta^*\monoto \FDelta^*\xra{h}\FSigma^*\xra{L^\D} O_\D)\]
lies in $W_\Delta$.
\end{enumerate}
Since $\A=\Set_f$, all finite $\D$-monoids are $\A$-generated by Example \ref{ex:agenerated_talgs}. Therefore, instantiating Definition \ref{def:localpseudovariety} and \ref{def:pseudovar} gives:
\begin{enumerate}
\item A \emph{local pseudovariety $\Sigma$-generated $\D$-monoids} is class of $\Sigma$-generated finite $\D$-monoids $e\colon \FSigma^*\epito M$ closed under quotients and subdirect products.
\item A \emph{pseudovariety of $\D$-monoids} is a class of finite $\D$-monoids closed under quotients, $\D$-submonoids and finite products.
\end{enumerate}

\smallskip
\noindent Theorem \ref{thm:localeilenberg} and \ref{thm:eilenberg} then specialize to 
\begin{theorem}[see \cite{ammu14,ammu15,cu15}]\label{thm:eilenberg_finitewords}
Let $\C$ and $\D$ be varieties satisfying the Assumptions \ref{ass:dt} and \ref{ass:catframework}(i),(iii), and suppose that $\D$ is commutative. Then the lattice of (local)  varieties of regular behaviors in $\C$ is isomorphic to the lattice of (local) pseudovarieties of $\D$-monoids.
\end{theorem}
Four concrete instances are listed below. The fourth column states what
$\D$-monoids are, and the third one
describes the $\C$-algebraic operations under which local varieties are closed. In addition, local varieties are closed under derivatives, and varieties are closed preimages of $\D$-monoid morphisms. All these correspondences are known in
the literature, and are uniformly covered by Theorem
\ref{thm:eilenberg_finitewords}.

\begin{table}[h]
\begin{tabularx}{\columnwidth}{llllll}
\hline \vspace{-0.4cm}\\
$\C$ & $\D$ & (local) var. of behav. closed under  & $\cong$ & (local) pseudovarieties of 
& proved in 
\\
\hline
$\BA$ & $\Set$ & boolean operations && monoids & \cite{eilenberg76,ggp08} \\
$\DL$ & $\Pos$ & finite union and finite intersection && ordered monoids &  
\cite{pin95,ggp08} \\
$\JSL$ & $\JSL$ & finite union && idempotent semirings & 
\cite{polak01} \\
$\Vect_{\Field}$ & $\Vect_{\Field}$ & addition of weighted languages && $\Field$-algebras 
& \cite{reu80} \\
\hline
\end{tabularx}
\end{table}

\subsection{Variety theorems for $\Gamma$-algebras}
In \cite{almeida90,sal04}, a generic Eilenberg correspondence for algebras over a finitary $S$-sorted signature $\Gamma$ is established. It follows from our Variety Theorem via the following parameters:
\begin{table}[h]
\begin{tabularx}{\columnwidth}{lX}
\hline \vspace{-0.4cm}\\
$\C/\D$         & $\BA$/$\Set$ \\
$\MT$ & the monad $\MT_\Gamma$ on $\Set^S$ with $\Alg{\MT_\Gamma} \cong$ $\Gamma$-algebras \\
$\F$ & all finite $\Gamma$-algebras\\
$\A$ &  $\Set_f^S$\\
$\U_\Sigma$ & all elementary translations\\
\hline
\end{tabularx}
\end{table}

\smallskip
\noindent\textbf{Step 1.} 
Let $\Gamma$ be a finitary $S$-sorted signature and let $\MT=\MT_\Gamma$ be the free $\Gamma$-algebra monad on $\Set^S$. Thus $\Alg{\MT_\Gamma}$ is the category of $\Gamma$-algebras. Put $\F$ = all finite $\Gamma$-algebras. Then $T_\Gamma(\Sigma)$ is the $S$-sorted set of $\Gamma$-trees with leaves labeled by constants from $\Gamma$ or elements of $\Sigma$. A language $L\colon T_\Gamma(\Sigma)\to O_\Set$ over $\Sigma$ corresponds to a tree language $L\seq T_\Gamma(\Sigma)$. It is recognizable iff it is accepted by some finite tree automaton, see e.g. \cite{tata2007}.

\smallskip
\noindent\textbf{Step 2.} 
Let $\A=\Set_f^S$. For each $\Sigma\in\A$, choose the unary presentation of $\MT_\Gamma\Sigma$ given by the elementary translations on $T_\Gamma(\Sigma)$, see Example \ref{ex:unpres_monads}.

\smallskip
\noindent\textbf{Step 3.} 
 Instantiating Definition \ref{def:localvariety} and \ref{def:variety} to $\MT$, $\F$, $\A$ and $\U_\Sigma$ gives the following concepts:
\begin{enumerate}
\item A \emph{local variety of recognizable tree languages over $\Sigma\in\Set_f^S$} is a set $W_\Sigma$ of recognizable tree languages $L\seq T_\Gamma(\Sigma)$ closed under boolean operations (finite union, finite intersection, complement) and derivatives w.r.t. elementary translations.
\item A \emph{variety of recognizable tree languages} is a family $(W_\Sigma\seq\Rec(\Sigma))_{\Sigma\in\A}$ of local varieties closed under preimages: for every $\Gamma$-homomorphism $g\colon T_\Gamma(\Delta)\to T_\Gamma(\Sigma)$ and $L\in W_\Sigma$, the language
$g^{-1}L\seq T_\Gamma(\Delta)$ with $(g^{-1}L)_s = \{x\in (T_\Gamma(\Delta))_s : g_s(x)\in L_s \}$ lies in $W_\Delta$.
\end{enumerate}
Since $\A=\Set_f^S$, every finite $\MT$-algebra is $\A$-generated by Example \ref{ex:agenerated_talgs}. Therefore, instantiating Definition \ref{def:localpseudovariety} and \ref{def:pseudovar} gives:
\begin{enumerate}
\item A \emph{local pseudovariety of $\Sigma$-generated $\Gamma$-algebras} is a class of finite quotient $\Gamma$-algebras $e\colon \MT_\Gamma\Sigma\epito A$ closed under quotients and subdirect products.
\item A \emph{pseudovariety of $\Gamma$-algebras} is a class of finite $\Gamma$-algebras closed under quotients, subalgebras, and finite products.
\end{enumerate}

\smallskip
\noindent As a special case of Theorem \ref{thm:localeilenberg} and \ref{thm:eilenberg} we get:

\begin{theorem}\label{thm:gammaalgs}
The lattice  of (local) varieties of recognizable $\Gamma$-tree languages is 
isomorphic to the lattice of (local) pseudovarieties of $\Gamma$-algebras.
\end{theorem}
The non-local version is due to Almeida \cite{almeida90} in the single-sorted case, and due to Salehi and Steinby \cite{sal04} in the many-sorted case. The local version is a new result. Similarly, by taking $\D=\Pos$ and the monad $\MT_{\Gamma,\leq}$ on $\Pos^S$ representing \emph{ordered} $\Gamma$-algebras, one gets the concept of a \emph{positive variety of recognizable tree languages}, which emerges from the above concept by dropping closure under complement (since $\C=\DL$). This leads to the following variety theorem:

\begin{theorem}\label{thm:gammaalgs_positive}
The lattice  of (local) positive varieties of recognizable $\Gamma$-tree languages is 
isomorphic to the lattice of (local) pseudovarieties of ordered $\Gamma$-algebras.
\end{theorem}
To the best of our knowledge, the positive variant is a new result.

\subsection{The polynomial variety theorem}
In \cite{boj15}, Boja\'nczyk's established a generic variety theorem \cite{boj15} for arbitrary monads on $\Set^S$, extending the work of Almeida \cite{almeida90} and Salehi and Steinby \cite{sal04} discussed in the previous subsection. Boja\'nczyk's result follows from another application of our three-step procedure using the following parameters:
\begin{table}[h]
\begin{tabularx}{\columnwidth}{lX}
\hline \vspace{-0.4cm}\\
$\C/\D$         & $\BA$/$\Set$ \\
$\MT$ & an arbitrary monad on $\Set^S$\\
$\F$ & all finite $\MT$-algebras\\
$\A$ &  $\Set_f^S$\\
$\U_\Sigma$ & $\{\,[p] : \text{ $p$ a polynomial over $T\Sigma$} \,\}$\\
\hline
\end{tabularx}
\end{table}

\smallskip
\noindent\textbf{Step 1.} 
Let $\MT$ be an arbitrary monad on $\Set^S$, and put $\F =$ all finite $\MT$-algebras.

\smallskip
\noindent\textbf{Step 2.} 
Let $\A=\Set_f^S$. For each $\Sigma\in\A$, choose the polynomial presentation of  $\MT\Sigma$ given by
$ \U_\Sigma = \{\,[p] : \text{ $p$ a polynomial over $T\Sigma$} \,\} $, see Example \ref{ex:unpres_monads}.

\smallskip
\noindent\textbf{Step 3.} 
 Instantiating Definition \ref{def:localvariety} and \ref{def:variety} to $\MT$, $\F$, $\A$ and $\U_\Sigma$ gives the following concepts:
\begin{enumerate}
\item A \emph{local variety of $\MT$-recognizable languages over $\Sigma\in\Set_f^S$} is a set $W_\Sigma$ of $\MT$-recognizable languages over $\Sigma$ closed under boolean operations (finite union, finite intersection, complement) and polynomial derivatives (see Example \ref{ex:derivatives}.3).
\item A \emph{variety of $\MT$-recognizable languages} is a family $(W_\Sigma\seq\Rec(\Sigma))_{\Sigma\in\A}$ of local varieties closed under preimages: for every $\MT$-homomorphism $g\colon \MT\Delta\to \MT\Sigma$ and $L\in W_\Sigma$, the language
$g^{-1}L\seq T\Delta$ with $(g^{-1}L)_s = \{x\in (T\Delta)_s : g_s(x)\in L_s \}$ lies in $W_\Delta$.
\end{enumerate}

\smallskip
\noindent Since $\A=\Set_f^S$, every finite $\MT$-algebra is $\A$-generated by Example \ref{ex:agenerated_talgs}. Therefore, instantiating Definition \ref{def:localpseudovariety} and \ref{def:pseudovar} gives:
\begin{enumerate}
\item A \emph{local pseudovariety of $\Sigma$-generated $\MT$-algebras} is a class of finite quotients $e\colon \MT\Sigma\epito A$ in $\Alg{\MT}$ closed under quotients and subdirect products.
\item A \emph{pseudovariety of $\MT$-algebras} is a class of finite $\MT$-algebras closed under quotients, subalgebras, and finite products.
\end{enumerate}
Therefore we get as a special case of Theorem \ref{thm:localeilenberg} and \ref{thm:eilenberg}:

\begin{theorem}\label{thm:polynomialvars}
The lattice  of (local) polynomial varieties of $\MT$-recognizable languages is 
isomorphic to the lattice of (local) pseudovarieties of $\MT$-algebras.
\end{theorem}
The non-local part is due to 
Boja\'nczyk \cite{boj15}, while the local part is a new result.

\smallskip
In the following subsections, we present variety theorems for $\infty$-languages, $\omega$-languages, cost functions and binary tree languages that do follow from the previous categorical Eilenberg theorems in \cite{ammu14,ammu15,boj15}, but are direct instances of our present Variety Theorem. 

\subsection{Variety theorems for $\infty$-languages and $\omega$-languages}
We derive two variety theorems of Wilke \cite{wilke91} and Pin \cite{pin98} for $\infty$-languages and present a new variety theorem for $\omega$-languages, based on our Variety Theorem for reduced languages.

\subsubsection{Wilke's theorem for $\infty$-languages}
We use the following parameters:
\begin{table}[H]
\begin{tabularx}{\columnwidth}{lX}
\hline \vspace{-0.4cm}\\
$\C/\D$         & $\BA$/$\Set$ \\
$\MT$ & the monad $\MT_\infty$ on $\Set^2$ with $\Alg{\MT_\infty}\cong$ $\omega$-semigroups\\
$\F$ & all finite $\omega$-semigroups\\
$\A$ &  $\{\,(\Sigma,\emptyset):
\Sigma\in\Set_f\,\}$\\
$\U_{(\Sigma,\emptyset)}$ & see Example \ref{ex:unpres_omegasem}\\
\hline
\end{tabularx}
\end{table}

\smallskip
\noindent\textbf{Step 1.} 
 Choose the $\omega$-semigroup monad $\MT=\MT_\infty$ on $\Set^2$, and let $\F=$ all finite $\omega$-semigroups. $\MT_\infty$-recognizable languages are precisely regular $\infty$-languages, see Example \ref{ex:rec_omegasem}.

\smallskip
\noindent\textbf{Step 2.} 
 Let $\A=\{\,(\Sigma,\emptyset):
\Sigma\in\Set_f\,\}$. For each $\Sigma\in\Set_f$, choose the unary presentation of $\MT_\infty(\Sigma,\emptyset)=(\Sigma^+,\Sigma^\omega)$ as in Example \ref{ex:unpres_omegasem}.

\smallskip
\noindent\textbf{Step 3.} 
 Instantiating Definition \ref{def:localvariety} and \ref{def:variety} to $\MT$, $\F$, $\A$ and $\U_\Sigma$ gives the following concepts:
\begin{enumerate}
\item A \emph{local variety of $\infty$-languages over $\Sigma\in\Set_f$} is a set of regular $\infty$-languages $L\seq \Sigma^+\cup \Sigma^\omega$ closed under boolean operations and the derivatives of Example \ref{ex:derivatives}.2.
\item A \emph{variety of $\infty$-languages} is a family $(W_\Sigma)_{\Sigma\in\Set_f}$ such that (i) $W_\Sigma$ is a local variety of $\infty$-languages over $\Sigma$ and (ii) $(W_\Sigma)_\Sigma$ is closed under preimages of $\omega$-semigroup morphisms. That is, for every $L\in W_\Sigma$ and every $\omega$-semigroup morphism $g\colon (\Delta^+,\Delta^\omega)\to (\Sigma^+,\Sigma^\omega)$, the preimage $g^{-1}L = \{ x\in \Delta^+\cup\Delta^\omega: g(x)\in L \}$ lies in $W_\Delta$.
\end{enumerate}
Recall from Example \ref{ex:agenerated_talgs} that $\A$-generated $\omega$-semigroups are precisely the complete $\omega$-semigroups. Therefore Definition \ref{def:localpseudovariety} and \ref{def:pseudovar} give:
\begin{enumerate}
\item For any $\Sigma\in\Set_f$, a \emph{local pseudovariety of $\Sigma$-generated $\omega$-semigroups} is a class of finite $\omega$-semigroup quotients $e\colon (\Sigma^+,\Sigma^\omega)\epito A$ closed under quotients and subdirect products.
\item A \emph{pseudovariety of $\omega$-semigroups} is a class of finite complete $\omega$-semigroups closed under quotients, complete $\omega$-subsemigroups and finite products.
\end{enumerate}
From Theorem \ref{thm:localeilenberg} and \ref{thm:eilenberg} we deduce
\begin{theorem}\label{thm:eilenberg_infty}
The lattice of (local) varieties of $\infty$-languages is isomorphic to the 
lattice of (local) pseudovarieties of $\omega$-semigroups.
\end{theorem}
The non-local part is due to Wilke \cite{wilke91} (in the version of \cite{pinperrin04}), while the local part is a new result, extending the 
local variety theorem of Gehrke, Grigorieff, and Pin~\cite{ggp08} from finite to infinite 
words. 

\subsubsection{Pin's variety theorem for $\infty$-languages}
\noindent In \cite{pin98}, Pin considered a ``positive'' version of Wilke's variety theorem where in the definition of a variety the closure under complements is dropped, and $\omega$-semigroups are replaced by \emph{ordered $\omega$-semigroups}, i.e. $\omega$-semigroups $A$ with a poset structure in each sort such that the products $A_+\times A_+ \xra{\o} A_+$, $A_+\times A_\omega \xra{\o} A_\omega$ and $A_+^\omega \xra{\pi} A_\omega$ are monotone. Pin's result also follows from our three steps via the following parameters:
\begin{table}[h]
\begin{tabularx}{\columnwidth}{lX}
\hline \vspace{-0.4cm}\\
$\C/\D$         & $\DL$/$\Pos$ \\
$\MT$ & the monad $\MT_{\infty,\leq}$ on $\Pos^2$ with $\Alg{\MT_{\infty,\leq}}\cong$ ordered $\omega$-semigroups\\
$\F$ & all finite ordered $\omega$-semigroups\\
$\A$ &  $\{\,(\Sigma,\emptyset):
\Sigma\in\Set_f\,\}$\\
$\U_{(\Sigma,\emptyset)}$ & see Example \ref{ex:unpres_omegasem}\\
\hline
\end{tabularx}
\end{table}

\smallskip
\noindent\textbf{Step 1.}  Let $\D=\Pos$ and choose the monad $\MT_{\infty,\leq}$ on $\Pos^2$ representing \emph{ordered $\omega$-semigroups}. Recall that for any $\Sigma\in\Set_f$, the associated free object $\FSigma\in\Pos$ is just $\Sigma$ with the discrete order. The free ordered $\omega$-semigroup generated by the two-sorted discrete poset $(\FSigma,\FGamma)$ is $(\Sigma^+, \Sigma^\omega+\Sigma^*\times\Gamma)$ (discretly ordered), as in the unordered case; thus $T_{\infty,\leq}(\FSigma,\FGamma) = (\Sigma^+, \Sigma^\omega+\Sigma^*\times\Gamma)$

\smallskip
\noindent\textbf{Step 2.} Let $\A=\{\,(\Sigma,\emptyset):
\Sigma\in\Set_f\,\}$. As in the unordered case, every ordered $\omega$-semigroup $A$ has a unary presentation given by the operations (i) $x\mapsto yx$ and $x\mapsto xy$ on $A_+$, (ii) $x\mapsto xz$ and $x\mapsto x^\omega=\pi(x,x,x,\ldots)$ from $A_+$ to $A_\omega$, and (iii) $z\mapsto yz$ on $A_\omega$, where $y\in A_+\cup \{1\}$ and $z\in A_\omega$. The proof is completely analogous to the unordered case, see Example \ref{ex:unpres_omegasem}: one just replaces the equivalence $\equiv$ by a stable preorder $\preceq$ and equations by inequations. For each $\Sigma\in\Set_f$, choose this unary presentation of $\MT_{\infty,\leq}(\FSigma,\emptyset)$.

\smallskip
\noindent\textbf{Step 3.} 
 As in the unordered case, $\A$-generated ordered $\omega$-semigroups are precisely the complete ones; the argument is identical to the one in Example \ref{ex:agenerated_omegasem}.
Since $\D=\Pos$ and $\C=\DL$, the $\C$-algebraic operations on $\Rec(\Sigma)$ are finite union and finite intersection, see Example \ref{ex:c_operations}. Definition \ref{def:variety} therefore instantiates as follows:
\begin{enumerate}
\item A \emph{local positive variety of $\infty$-languages over $\Sigma\in\Set_f$} is a set of regular $\infty$-languages $L\seq \Sigma^+\cup\Sigma^\omega$ closed under finite union, finite intersection, and the derivatives of Example \ref{ex:derivatives}.2.
\item A \emph{positive variety of $\infty$-languages} is a family $(W_\Sigma)_{\Sigma\in\Set_f}$ such that (i) $W_\Sigma$ is a local positive variety of $\infty$-languages over $\Sigma$ and (ii) $(W_\Sigma)_\Sigma$ is closed under preimages of $\omega$-semigroup morphisms.
\end{enumerate}

\smallskip
\noindent Definition \ref{def:localpseudovariety} and \ref{def:pseudovar} give:
\begin{enumerate}
\item For any $\Sigma\in\Set_f$, a \emph{local pseudovariety of $\Sigma$-generated ordered $\omega$-semigroups} is a class of finite ordered $\omega$-semigroup quotients $e: (\Sigma^+,\Sigma^\omega)\epito A$ closed under quotients and subdirect products.
\item A \emph{pseudovariety of ordered $\omega$-semigroups} is a class of finite complete ordered $\omega$-semigroups closed under quotients, complete ordered $\omega$-subsemigroups and finite products.
\end{enumerate}
From Theorem \ref{thm:eilenberg} we get

\begin{theorem}\label{thm:lastex}
The lattice of (local) positive varieties of $\infty$-languages is
isomorphic to the lattice of (local) pseudovarieties of ordered
$\omega$-semigroups.
\end{theorem}
The non-local part is due to Pin \cite{pin98}, and the local part is again a new result.

\subsubsection{A variety theorem for $\omega$-languages}
Next, we present an Eilenberg-type correspondence for languages of infinite words, an application of our Reduced Variety Theorem. We choose the following parameters:
\begin{table}[H]
\begin{tabularx}{\columnwidth}{lX}
\hline \vspace{-0.4cm}\\
$\C/\D$         & $\BA$/$\Set$ \\
$\MT$ & the monad $\MT_\infty$ on $\Set^{\{+,\omega\}}$ with $\Alg{\MT_\infty}\cong$ $\omega$-semigroups\\
$\F$ & all finite $\omega$-semigroups\\
$S_0$ & $\{\omega\}$\\
$\A$ &  $\{\,(\Sigma,\emptyset):
\Sigma\in\Set_f\,\}$\\
$\U_{(\Sigma,\emptyset)}$ & see Example \ref{ex:unpres_omegasem}\\
\hline
\end{tabularx}
\end{table}

\smallskip
\noindent\textbf{Step 1.} 
 Choose the $\omega$-semigroup monad $\MT=\MT_\infty$ on $\Set^S$, where $S=\{+,\omega\}$. Choose $S_0=\{\omega\}$ and let $\F=$ all finite $\omega$-semigroups. Then reduced $\F$-recognizable languages are precisely $\omega$-regular languages, see Example \ref{ex:reducedlang_omegasem}.

\smallskip
\noindent\textbf{Step 2.} 
 Let $\A=\{\,(\Sigma,\emptyset):
\Sigma\in\Set_f\,\}$. For each $\Sigma\in\Set_f$, choose the unary presentation $\U_\Sigma := \U_{(\Sigma,\emptyset)}$ of $\MT_\infty(\Sigma,\emptyset)=(\Sigma^+,\Sigma^\omega)$ as in Example \ref{ex:unpres_omegasem}.
Note that the only unary operations in $\ol \U_\Sigma(\omega,\omega)$ (i.e. of the form $\Sigma^\omega\to \Sigma^\omega$) are the maps $z\mapsto yz$ for $y\in\Sigma^*$. Therefore the $\ol\U_\Sigma$-derivatives of an $\omega$-language $L\seq \Sigma^\omega$ are the languages
\[ y^{-1}L = \{z\in \Sigma^\omega: yz\in L\} \qquad (y\in \Sigma^*),\]
i.e. left derivatives w.r.t. finite words.

\smallskip
\noindent\textbf{Step 3.} 
 Instantiating Definition \ref{def:localvariety_reduced} and \ref{def:variety_reduced} to $\MT$, $\F$, $\A$, $S_0$ and $\U_\Sigma$ gives the following:
\begin{enumerate}
\item A \emph{local variety of $\omega$-regular languages over $\Sigma\in\Set_f$} is a set of $\omega$-regular languages $L\seq \Sigma^\omega$ closed under boolean operations and left derivatives w.r.t. finite words.
\item A \emph{variety of $\omega$-regular languages} is a family $(W_\Sigma)_{\Sigma\in\Set_f}$ such that (i) $W_\Sigma$ is a local variety of $\omega$-languages over $\Sigma$ and (ii) $(W_\Sigma)_\Sigma$ is closed under preimages of $\omega$-semigroup morphisms. That is, for every $L\in W_\Sigma$ and every $\omega$-semigroup morphism $g\colon (\Delta^+,\Delta^\omega)\to (\Sigma^+,\Sigma^\omega)$, the preimage $g^{-1}L = \{ x\in\Delta^\omega: g(x)\in L \}$ lies in $W_\Delta$.
\end{enumerate}

\smallskip\noindent
Recall from Example \ref{ex:agenerated_talgs} that $\A$-generated $\omega$-semigroups are precisely the complete $\omega$-semigroups. Therefore Definition \ref{def:localpseudovar_reduced} and \ref{def:pseudovar_reduced} give:
\begin{enumerate}
\item For any $\Sigma\in\Set_f$, a \emph{local pseudovariety of $\Sigma$-generated reduced $\omega$-semigroups} is an ideal in the poset of $\Sigma$-generated reduced $\omega$-semigroups $e\colon (\Sigma^+,\Sigma^\omega)\epito A$; cf. Example \ref{ex:reducedalg_omegasem}.
\item A \emph{pseudovariety of reduced $\omega$-semigroups} is a class of finite complete reduced $\omega$-semigroups closed under division of finite products.
\end{enumerate}

\smallskip\noindent
From Theorem \ref{thm:localeilenberg_reduced} and \ref{thm:eilenberg_reduced} we deduce
\begin{theorem}\label{thm:eilenberg_omega}
The lattice of (local) varieties of $\omega$-regalar languages is isomorphic to the 
lattice of (local) pseudovarieties of reduced $\omega$-semigroups.
\end{theorem}
Similarly, by taking the monad $\MT_{\leq,\infty}$ of ordered $\omega$-semigroups, we obtain a version for \emph{positive varieties of $\omega$-regular languages} where the closure under complement is dropped, and pseudovarieties of reduced ordered $\omega$-semigroups. This is completely analogous to the case of $\infty$-languages discussed above.

\begin{theorem}
The lattice of (local) positive varieties of $\omega$-regular languages is isomorphic to the 
lattice of (local) pseudovarieties of reduced ordered $\omega$-semigroups.
\end{theorem}
To the best of our knowledge, the above two Eilenberg correspondences are not known in the literature.

\subsection{The variety theorem for regular cost functions}
In this subsection, we derive the variety theorem for cost functions proved by Daviaud, Kuperberg, and Pin \cite{pin16} in our framework. Regular cost functions were introduced by Colcombet \cite{col09} as a quantitative extension of regular languages, inheriting many of their algebraic and logical properties and decidability results.
A \emph{cost function} over a finite alphabet $\Sigma\in\Set_f$ is a function $f\colon \Sigma^*\to \Nat\cup\{\infty\}$. Two cost functions $f$ and $g$ are identified if, for every subset $A\seq\Nat$, the function $f$ is bounded on $A$ iff $g$ is bounded on $A$. 

Algebraically, cost functions are represented by \emph{stabilization monoids} \cite{col09}.
By the latter is meant an ordered monoid $(M,\leq,1,\o)$ together with a monotone operation $(\dash)^\#\colon E(M)\to E(M)$, where $E(M)$ is the set of idempotent elements of $M$, subject to the following axioms:
\begin{enumerate}[(S1)]
\item For all $s,t\in M$ with $st, ts\in E(M)$, one has $(st)^\#s = s(ts)^\#$.
\item For all $e\in E(M)$, one has $(e^\#)^\# = e^\#e = ee^\# = e^\# \leq e$.
\item $1^\#=1$.
\end{enumerate}
Since the operation $(\dash)^\#$ is only partially defined on $M$, stabilization monoids cannot be directly treated with the standard tools of algebraic language theory. For this reason, in \cite{pin16} an equivalent algebraic theory was introduced that is easier to work with:

\begin{definition}
A \emph{weak stabilization algebra} is an ordered algebra $M$ over the signature $\Gamma=\{1,\o, (\dash)^\#, (\dash)^\omega\}$ satisfying all inequations $t\leq \ol t$ between $\Gamma$-terms that hold in every finite stabilization monoid, where $x^\omega$ is interpreted as the unique idempotent power of $x$. Morphisms of weak stabilization algebras are monotone $\Gamma$-homomorphisms.
\end{definition}
For example, a weak stabilization algebra satisfies the inequation $(x^\omega)^\#\leq x^\omega$, because this inequation holds in every finite stabilization monoid by (S2). Moreover, every weak stabilization algebra is a monoid. Weak stabilization algebras form a variety of ordered algebras.

\begin{definition}
A weak stabilization algebra is called a \emph{stabilization algebra} if satisfies the implication
\begin{itemize}[(S4)]
\item $x\o x = x\quad\To\quad x^\omega = x$.
\end{itemize}
\end{definition}
An important observation is given by the following

\begin{lemma}\label{lem:stabquotients}
Every quotient of a finite stabilization algebra is a stabilization algebra.
\end{lemma}

\begin{proof} Let $e\colon M\epito N$ be a surjective morphism of weak stabilization algebras and suppose that $M$ is a finite stabilization algebra. We need to show that $N$ is a stabilization algebra.
\begin{enumerate}[(a)]
\item We first show that, for every $s\in M$, the element $s^\omega$ is the unique idempotent power of $s$. Choose $n\in\Nat$ such that $s^n$ is the unique idempotent power of $s$. Since the equation $(x^n)^\omega = x^\omega$ holds in every finite monoid (and hence in every finite stabilization monoid), it follows that $(s^n)^\omega = s^\omega$. Moreover, since $s^n\o s^n = s^n$ and (S4) holds in $M$, we have $(s^n)^\omega = s^n$. Therefore $s^\omega = s^n$, as claimed.
\item We show that $N$ satisfies (S4). Let $t\in N$ with $t\o t =t$ and choose $s\in M$ with $t=e(s)$. Since $e$ preserves idempotent powers (being a monoid morphism) and $s^\omega$ is the idempotent power of $s$ by part (a), we have that $t^\omega = e(s)^\omega$ is the idempotent power of $t$. But since $t\o t = t$, the element $t$ is idempotent and thus $t^\omega=t$, proving (S4).\qedhere
\end{enumerate}
\end{proof}
To derive the variety theorem for cost functions, we apply our framework to the following parameters: 
\begin{table}[H]
\begin{tabularx}{\columnwidth}{lX}
\hline \vspace{-0.4cm}\\
$\C/\D$         & $\DL$/$\Pos$ \\
$\MT$ & the monad $\MT_S$ on $\Pos$ with $\Alg{\MT_S}\cong$ weak stabilization algebras\\
$\F$ & all finite stabilization algebras\\
$\A$ &  $\Set_f$\\
$\U_\Sigma$ & all elementary translations\\
\hline
\end{tabularx}
\end{table}

\smallskip
\noindent\textbf{Step 1.} Let $\MT=\MT_S$ be the monad on $\Pos$ representing weak stabilization algebras, i.e. such that $\Alg{\MT_S}$ is the variety of weak stabilization algebras. As shown in \cite[Proposition 4.3]{pin16}, regular cost functions over $\Sigma\in\Set_f$ correspond to languages $L\colon T_S\FSigma\to \{0<1\}=O_\Pos$ (equivalently, to down-sets of the poset $T_S\FSigma$) that are recognizable by finite stabilization algebras. Accordingly, we choose $\F= $ all finite stabilization algebras. Note that, in contrast to all previous examples, $\F$ is a proper subclass of $\FAlg{\MT_S}$. The class $\F$ is closed under quotients by Lemma \ref{lem:stabquotients}, and under subalgebras and finite products since this holds for any implicational class of ordered algebras. Thus $\F$ satisfies Assumption \ref{ass:dt}(iii).

\smallskip
\noindent\textbf{Step 2.} Let $\A=\Set_f$. For every $\Sigma\in\Set_f$, let $\U_\Sigma$ be the unary presentation of $\MT_S\FSigma$ given by the elementary translations, i.e. the operations $x\mapsto yx$, $x\mapsto xy$, $x\mapsto x^\omega$ and $x\mapsto x^\#$, where $y\in \MT_S\FSigma$; see Example \ref{ex:unpres_gammaalgs}.

\smallskip
\noindent\textbf{Step 3.} 
 Instantiating Definition \ref{def:localvariety} and \ref{def:variety} to $\MT=\MT_S$, $\F$, $\A$ and $\U_\Sigma$ gives the following:
\begin{enumerate}
\item A \emph{local variety of recognizable down-sets over $\Sigma\in\Set_f^S$} is a set $W_\Sigma$ of $\F$-recognizable languages over $\Sigma$ (that is, of down-sets $I\seq T_S\FSigma$ recognizable by finite stabilization algebras) closed under finite union, finite intersection, and $\U_\Sigma$-derivatives. The latter means that, for every $I\in W_\Sigma$, the following down-sets lie in $W_\Sigma$:
\begin{align*}
 y^{-1}I &= \{ x\in T_S\FSigma : yx\in I \}\quad (y\in T_S\FSigma)\\
 Iy^{-1} &= \{ x\in T_S\FSigma : xy\in I \}\quad (y\in T_S\FSigma)\\
 \omega^{-1}I &= \{x\in T_S\FSigma : x^\omega\in I \}\\
 \#^{-1}I &= \{x\in T_S\FSigma : x^\#\in I \}
\end{align*}
\item A \emph{variety of $\MT$-recognizable languages} is a family $(W_\Sigma\seq\Rec(\Sigma))_{\Sigma\in\A}$ of local varieties of recognizable down-sets closed under preimages: for every morphism $g\colon \MT_S\FDelta\to\MT_S\FSigma$ of weak stabilization algebras and $I\in W_\Sigma$, the down-set
$g^{-1}I = \{x\in T_S\FDelta : g(x)\in I\}$ lies in $W_\Delta$.
\end{enumerate}
Since $\A=\Set_f$, every finite stabilization algebra is $\A$-generated by Example \ref{ex:agenerated_talgs}. Therefore, instantiating Definition \ref{def:localpseudovariety} and \ref{def:pseudovar} gives:
\begin{enumerate}
\item A \emph{local pseudovariety of $\Sigma$-generated stabilization algebras} is a class of finite quotient stabilization algebras $e\colon \MT_S\FSigma\epito A$ of $\MT_S\FSigma$ closed under quotients and subdirect products.
\item A \emph{variety of $\MT$-algebras} is a class of finite stabilization algebras closed under quotients, subalgebras, and finite products.
\end{enumerate}
Therefore we get as a special case of Theorem \ref{thm:localeilenberg} and \ref{thm:eilenberg}:

\begin{theorem}[Daviaud, Kuperberg, Pin \cite{pin16}]
The lattice of (local) varieties of recognizable down-sets is
isomorphic to the lattice of (local) pseudovarieties of ordered
stabilization algebras.
\end{theorem}

\subsection{The variety theorem for binary tree languages}

As our final application, we deduce the variety theorem for binary trees of Salehi and Steinby \cite{salehi07} and a new local version of it, another instance of our Reduced Variety Theorem.  The algebraic theory corresponding to binary trees is given by Wilke's tree algebras \cite{wilke96}. By a \emph{tree algebra} is meant a three-sorted algebra $A=(A_l, A_t, A_c)$ with finitary operations
\begin{align*}
\iota\colon& A_l\to A_t\\
\kappa\colon& A_l\times A_t\times A_t\to A_t\\
 \lambda\colon& A_l\times A_t\to A_c,\\
  \rho\colon& A_l\times A_t\to A_c\\
   \eta\colon& A_c\times A_t\to A_t\\
  \sigma\colon& A_c\times A_c\to A_c\\
\end{align*}
subject to the following equations:
\begin{align*}
\sigma(\sigma(p,q),r) &= \sigma(p,\rho(q,r))\\
\eta(\sigma(p,q),t) &= \eta(p,\eta(q,t))\\
\eta(\lambda(a,s),t) &= \kappa(a,t,s)\\
\eta(\rho(a,s),t) &= \kappa(a,s,t)
\end{align*}
The free tree algebra generated by the three-sorted set $(\Sigma,\emptyset,\emptyset)$ is carried by  $(\Sigma,T_\Sigma,C_\Sigma^+)$ where $T_\Sigma$ is the set of (nonempty, full, finite and ordered) binary trees with nodes labeled by elements of $\Sigma$, and $C_\Sigma^+$ is the set of nonempty \emph{contexts}, i.e. $\Sigma+\{\ast\}$-labeled binary trees with more than one node where the label $\ast$ appears only at a single leaf. The tree algebra operations on $(\Sigma,T_\Sigma,C_\Sigma^+)$ are defined as follows:
\begin{itemize}
\item $\iota(a)$ is the single-node tree with label $a$.
\item $\kappa(a,t_0,t_1)$ is the tree with root labeled by $a$ and left and right subtrees $t_0$ and $t_1$.
\item $\lambda(a,t)$ is the context with root labeled by $a$, the left child of the root is $\ast$-labeled, and the right child is the tree $t$.
\item $\rho(a,t)$ is the context with root labeled by $a$, the right child of the root is $\ast$-labeled, and the left child is the tree $t$.
\item $\eta(c,t)$ is the tree emerging from $c$ by substituting $t$ for $\ast$.
\item $\rho(c_0,c_1)$ is the context emerging from $c_0$ by substituting $c_1$ for $\ast$.
\end{itemize}
By extension, for the empty context $c_\ast$ (i.e. the context consisting only of a single node with label $\ast$), we define $\eta(c_\ast,t) := t$ and $\rho(c_\ast,c) := c$ and $\rho(c,c_\ast):= c$. A binary tree language is a set $L\seq T_\Sigma$ of trees over $\Sigma$. To obtain the variety theorem for binary tree languages, we apply our Reduced Variety Theorem via the following parameters:
\begin{table}[H]
\begin{tabularx}{\columnwidth}{lX}
\hline \vspace{-0.4cm}\\
$\C/\D$         & $\BA$/$\Set$ \\
$\MT$ & the monad $\MT_B$ on $\Set^{\{l,t,c\}}$ with $\Alg{\MT_B}\cong$ tree algebras\\
$\F$ & all finite tree algebras\\
$S_0$ & $\{t\}$\\
$\A$ &  $\{(\Sigma,\emptyset,\emptyset):\Sigma\in\Set_f\}$\\
$\U_{(\Sigma,\emptyset,\emptyset)}$ & all elementary translations\\
\hline
\end{tabularx}
\end{table}

\smallskip
\noindent\textbf{Step 1.} 
Let $\MT=\MT_B$ be the monad on $\Set^S$ (where $S=\{l,t,c\}$) associated to the algebraic theory of tree algebras. Choose $S_0= \{t\}$, $\A= \{(\Sigma,\emptyset,\emptyset):\Sigma\in\Set_f\}$ and $\F=$ finite tree algebras. By the above, we have $T_B(\Sigma,\emptyset,\emptyset) = (\Sigma, T_\Sigma, C_\Sigma^+)$, and thus a reduced language $L\colon T_B(\Sigma,\emptyset,\emptyset)\to O_\Set$ corresponds precisely to a binary tree language $L\seq T_\Sigma$. Recognizable binary tree languages are exactly regular tree languages, i.e. those accepted by finite binary tree automata; see \cite{wilke96}.

\smallskip
\noindent\textbf{Step 2.} 
 For each $\Sigma\in\Set_f$, the tree algebra $T_B(\Sigma,\emptyset,\emptyset) = (\Sigma, T_\Sigma, C_\Sigma^+)$ has the unary presentation $\U_\Sigma :=\U_{(\Sigma,\emptyset,\emptyset)}$ given by the elementary translations, i.e. all maps of the form 
 $a \mapsto \gamma(a_1,\ldots, a_i, a, a_{i+1},\ldots, a_n)$ where $\gamma$ is one of the six tree algebra operations and $a_1,\ldots, a_n$ are elements of appropriate sort; see Example \ref{ex:unpres_gammaalgs}. Observe that $\ol \U_\Sigma(t,t)$  consists precisely of those unary operations $u\colon T_\Sigma\to T_\Sigma$ of the form $u(t)=\eta(c,t)$, where $c$ is a (possibly empty) context. Thus the $\ol \U_\Sigma$-derivatives of a language $L\seq T_\Sigma$ are the languages
 \[ c^{-1}L = \{ t\in T_\Sigma : \eta(c,t)\in L \}, \]
 where $c$ ranges over all contexts.

\smallskip
\noindent\textbf{Step 3.} 
 Instantiating Definition \ref{def:localvariety_reduced} and \ref{def:variety_reduced} to $\MT$, $\F$, $\A$, $S_0$ and $\U_\Sigma$ gives the following:
\begin{enumerate}
\item A \emph{local variety of binary tree languages over $\Sigma\in\Set_f$} is a set of regular binary tree languages $L\seq T_\Sigma$ closed under boolean operations and derivatives w.r.t. contexts (see Step 2).
\item A \emph{variety of $\omega$-languages} is a family $(W_\Sigma)_{\Sigma\in\Set_f}$ such that (i) $W_\Sigma$ is a local variety of binary tree languages over $\Sigma$ and (ii) $(W_\Sigma)_\Sigma$ is closed under preimages of tree algebra morphisms. That is, for every $L\in W_\Sigma$ and every tree algebra morphism $g\colon (\Delta,T_\Delta,C_\Delta^+)\to (\Sigma,T_\Sigma,C_\Sigma^+)$, the preimage $g^{-1}L = \{ x\in T_\Delta: g(x)\in L \}$ lies in $W_\Delta$.
\end{enumerate}

\smallskip
\noindent By Lemma \ref{lem:reducedconcrete}, an $\A$-generated tree algebra $A$ is reduced iff for all elements $a,a'\in A_l$,
\[ \text{if $\iota(a)=\iota(a')$ and $\kappa(a,s,t)=\kappa(a',s,t)$ for all $s,t\in A_t$}\quad\text{then}\quad a=a',   \]
and for all elements $c,c'\in A_c$,
\[  \text{if $\eta(c,t)=\eta(c',t)$ for all $t\in A_t$}\quad\text{then}\quad c=c'.\]
 Therefore Definition \ref{def:localpseudovar_reduced} and \ref{def:pseudovar_reduced} give:
\begin{enumerate}
\item For any $\Sigma\in\Set_f$, a \emph{local pseudovariety of $\Sigma$-generated reduced tree algebras} is an ideal in the poset of $\Sigma$-generated reduced tree algebras $e\colon (\Sigma,T_\Sigma,C_\Sigma^+)\epito A$.
\item A \emph{pseudovariety of reduced tree algebras} is a class of finite reduced $\A$-generated tree algebras closed under division of finite products.
\end{enumerate}
From Theorem \ref{thm:localeilenberg_reduced} and \ref{thm:eilenberg_reduced} we deduce the variety theorem of Salehi and Steinby \cite{salehi07} and a new local version:
\begin{theorem}\label{thm:eilenberg_trees}
The lattice of (local) varieties of binary tree languages is isomorphic to the 
lattice of (local) pseudovarieties of reduced tree algebras.
\end{theorem}
Again, all these arguments can be easily adapted to ordered algebras. By taking the monad $\MT_{B,\leq}$ on $\Pos^3$ of representing ordered tree algebras (i.e. tree algebras on a three-sorted poset with monotone operations), we then obtain the notion of a \emph{positive variety of  tree languages} where the closure under complement is dropped, and of a \emph{pseudovariety of reduced ordered tree algebras}. Therefore we get the following result that, to the best of our knowledge, has not appeared in the literature:

\begin{theorem}
The lattice of (local) positive varieties of binary tree languages is isomorphic to the 
lattice of (local) pseudovarieties of reduced ordered tree algebras.
\end{theorem}

\section{Conclusions and Future Work}
\enlargethispage{11pt} 
We presented a categorical framework for
algebraic language theory that captures, as special instances, the bulk of Eilenberg theorems in the literature for pseudovarieties of finite algebras and varieties of recognizable languages. Let us mention some directions for future work.

First, we aim to investigate if it is possible to obtain a variety theorem for data languages
based on nominal Stone duality \cite{gabbay11}. On a similar note, it would also
be interesting to see whether dualities modeling probabilistic phenomena (e.g.\ Gelfand or Kadison duality) lead to a meaningful algebraic language theory for probabilistic automata and languages.

Secondly, although \emph{finite} structures are of most relevance from the
automata-theoretic perspective, there has been some work on variety theorems with relaxed finiteness restrictions. One example is Reutenauer's theorem \cite{reu80} for weighted languages over arbitrary fields $\Field$. As already anticipated in \cite{camu16} and the first version \cite{uacm16old} of the present paper, to cover this in our setting the results of Section~\ref{sec:pseudovar} should be presented for $(\E,\M)$-structured categories $\D$ in lieu of varieties. This has been worked out in\cite{u17} and, independently, in the recent preprint \cite{s16}. In the latter, also a formal ``Eilenberg correspondence'' is stated for dual $(\E,\M)$-categories. An important conceptual difference to our present work is that in loc.\ cit.\ one uses
discrete dualities (e.g.\ complete atomic boolean algebras/sets instead
of boolean algebras/Stone spaces) and that unary presentations do not
appear. This makes the concept of a variety of languages (called a
\emph{coequational theory}) and the Eilenberg correspondence in \cite{s16} easy to state, but much harder to apply in practice. The results of \cite{s16} and of our paper do not entail each other, and we leave it for future work to find a common roof.

\newpage

\bibliographystyle{plain}
\bibliography{coalgebra,ourpapers}

\begin{thebibliography}{10}

\bibitem{ahs09}
J.~Ad{\'{a}}mek, H.~Herrlich, and G.~E. Strecker.
\newblock {\em Abstract and Concrete Categories: The Joy of Cats}.
\newblock Dover Publications, 2nd edition, 2009.

\bibitem{ammu14}
J.~Ad\'amek, S.~Milius, R.~Myers, and H.~Urbat.
\newblock Generalized {E}ilenberg {T}heorem {I}: {L}ocal {V}arieties of
  {L}anguages.
\newblock In A.~Muscholl, editor, {\em Proc.~FoSSaCS'14}, volume 8412 of {\em
  LNCS}, pages 366--380. Springer, 2014.
\newblock Full version: \url{http://arxiv.org/pdf/1501.02834v1.pdf}.

\bibitem{amu15}
J.~Ad\'amek, S.~Milius, and H.~Urbat.
\newblock Syntactic monoids in a category.
\newblock In {\em Proc.~CALCO'15}, LIPIcs. Schloss Dagstuhl--Leibniz-Zentrum
  f\"ur Informatik, 2015.
\newblock Full version: \url{http://arxiv.org/abs/1504.02694}.

\bibitem{ammu15}
J.~Ad\'{a}mek, R.~Myers, S.~Milius, and H.~Urbat.
\newblock {Varieties of languages in a category}.
\newblock In {\em 30th Annual ACM/IEEE Symposium on Logic in Computer Science}.
  IEEE, 2015.

\bibitem{ar}
J.~Ad\'{a}mek and J.~Rosick\'y.
\newblock {\em Locally Presentable and Accessible Categories}.
\newblock Cambridge University Press, 1994.

\bibitem{almeida90}
J.~Almeida.
\newblock On pseudovarieties, varieties of languages, filters of congruences,
  pseudoidentities and related topics.
\newblock {\em Algebra Universalis}, 27(3):333--350, 1990.

\bibitem{ban72}
B.~Banaschewski.
\newblock On profinite universal algebras.
\newblock In {\em General Topology and its Relations to Modern Analysis and
  Algebra}, pages 51--62. Academia Publishing House of the Czechoslovak Academy
  of Sciences, 1972.

\bibitem{Bedon1998}
N.~Bedon and O.~Carton.
\newblock An {E}ilenberg theorem for words on countable ordinals.
\newblock In {\em Proc. LATIN'98}, volume 1380 of {\em LNCS}, pages 53--64.
  Springer, 1998.

\bibitem{Bedon2005}
N.~Bedon and C.~Rispal.
\newblock {S}ch{\"u}tzenberger and {E}ilenberg theorems for words on linear
  orderings.
\newblock In {\em Proc. DLT'05}, volume 3572 of {\em LNCS}, pages 134--145.
  Springer, 2005.

\bibitem{birkhoff37}
G.~Birkhoff.
\newblock Rings of sets.
\newblock {\em Duke Mathematical Journal}, 3(3):443--454, 1937.

\bibitem{boj15}
M.~Boja\'nczyk.
\newblock Recognisable languages over monads.
\newblock In I.~Potapov, editor, {\em Proc.\ DLT'15}, volume 9168 of {\em
  LNCS}, pages 1--13. Springer, 2015.
\newblock \url{http://arxiv.org/abs/1502.04898}.

\bibitem{Bor09}
F.~Bor\c{c}eux.
\newblock {\em Handbook of Categorical Algebra II: Categories and Structures},
  volume~51 of {\em Encyclopedia of Mathematics and its Applications}.
\newblock Cambridge University Press, 1994.

\bibitem{camu16}
L.-T. Chen, J.~Ad\'amek, S.~Milius, and H.~Urbat.
\newblock Profinite monads, profinite equations and {R}eiterman's theorem.
\newblock In B.~Jacobs and C.~L\"oding, editors, {\em Proc.~FoSSaCS'16}, volume
  9634 of {\em LNCS}. Springer, 2016.
\newblock \url{http://arxiv.org/abs/1511.02147}.

\bibitem{cu15}
L.-T. Chen and H.~Urbat.
\newblock A fibrational approach to automata theory.
\newblock In L.~S. Moss and P.~Sobocinski, editors, {\em Proc.~CALCO'15},
  LIPIcs. Schloss Dagstuhl--Leibniz-Zentrum f\"ur Informatik, 2015.

\bibitem{col09}
T.~Colcombet.
\newblock The theory of stabilisation monoids and regular cost functions.
\newblock In {\em Proc. ICALP'08}, pages 139--150. Springer, 2009.

\bibitem{tata2007}
H.~Comon, M.~Dauchet, R.~Gilleron, C.~L\"oding, F.~Jacquemard, D.~Lugiez,
  S.~Tison, and M.~Tommasi.
\newblock Tree automata techniques and applications.
\newblock Available on: \url{http://www.grappa.univ-lille3.fr/tata}, 2007.
\newblock October 2007.

\bibitem{pin16}
L.~Daviaud, D.~Kuperberg, and J.-\'E. Pin.
\newblock Varieties of cost functions.
\newblock In N.~Ollinger and H.~Vollmer, editors, {\em Proc. STACS 2016},
  volume~47 of {\em LIPIcs}, pages 30:1--30:14. Schloss
  Dagstuhl--Leibniz-Zentrum f\"ur Informatik, 2016.

\bibitem{eilenberg76}
S.~Eilenberg.
\newblock {\em {Automata, Languages, and Machines Vol.~B}}.
\newblock Academic Press, 1976.

\bibitem{gabbay11}
M.~J. Gabbay, T.~Litak, and D.~Petri{\c{s}}an.
\newblock Stone duality for nominal boolean algebras with new.
\newblock In A.~Corradini, B.~Klin, and C.~C{\^i}rstea, editors, {\em Proc.
  CALCO'11}, volume 6859 of {\em LNCS}, pages 192--207. Springer, 2011.

\bibitem{ggp08}
M.~Gehrke, S.~Grigorieff, and J.-\'E. Pin.
\newblock Duality and equational theory of regular languages.
\newblock In L.~Aceto and al., editors, {\em Proc.~ICALP'08, Part II}, volume
  5126 of {\em LNCS}, pages 246--257. Springer, 2008.

\bibitem{Johnstone1982}
P.~T. Johnstone.
\newblock {\em {Stone spaces}}.
\newblock Cambridge University Press, 1982.

\bibitem{Linton1969}
F.~E.~J. Linton.
\newblock {An outline of functorial semantics}.
\newblock In B.~Eckmann, editor, {\em Semin. Triples Categ. Homol. Theory},
  volume~80 of {\em LNM}, pages 7--52. Springer Berlin Heidelberg, 1969.

\bibitem{maclane}
S.~{Mac Lane}.
\newblock {\em Categories for the Working Mathematician}.
\newblock Springer, 2nd edition, 1998.

\bibitem{manes76}
E.~G. Manes.
\newblock {\em Algebraic Theories}, volume~26 of {\em Graduate Texts in
  Mathematics}.
\newblock Springer, 1976.

\bibitem{mat76}
G.~Matthiessen.
\newblock {Theorie der heterogenen Algebren}.
\newblock Technical report, Universität Bremen, 1976.

\bibitem{mat79}
G.~Matthiessen.
\newblock A heterogeneous algebraic approach to some problems in automata
  theory, many-valued logic and other topics.
\newblock In {\em Proc. Klagenfurt Conf.}, 1979.

\bibitem{pinperrin04}
D.~Perrin and J.-\'E. Pin.
\newblock {\em {Infinite Words}}.
\newblock Elsevier, 2004.

\bibitem{pin95}
J.-\'E. Pin.
\newblock A variety theorem without complementation.
\newblock {\em Russ. Math.}, 39:80--90, 1995.

\bibitem{pin98}
J.-\'E. Pin.
\newblock Positive varieties and infinite words.
\newblock In {\em LATIN 98}, volume 1380 of {\em LNCS}, pages 76--87. Springer,
  1998.

\bibitem{pin15}
J.-{\'E}. Pin.
\newblock Mathematical foundations of automata theory.
\newblock Available at
  \url{http://www.liafa.jussieu.fr/~jep/PDF/MPRI/MPRI.pdf}, October 2015.

\bibitem{Pippenger1997}
N.~Pippenger.
\newblock {Regular languages and {S}tone duality}.
\newblock {\em Th. Comp. Sys.}, 30(2):121--134, 1997.

\bibitem{polak01}
L.~Pol\'ak.
\newblock Syntactic semiring of a language.
\newblock In J.~Sgall, A.~Pultr, and P.~Kolman, editors, {\em Proc.\ MFCS'01},
  volume 2136 of {\em LNCS}, pages 611--620. Springer, 2001.

\bibitem{priestley72}
H.~A. Priestley.
\newblock {Ordered topological spaces and the representation of distributive
  lattices}.
\newblock {\em Proc. London Math. Soc.}, 3(3):507, 1972.

\bibitem{Reiterman1982}
J.~Reiterman.
\newblock {The Birkhoff theorem for finite algebras}.
\newblock {\em Algebra Universalis}, 14(1):1--10, 1982.

\bibitem{reu80}
C.~Reutenauer.
\newblock S\'eries formelles et alg\`ebres syntactiques.
\newblock {\em J.~Algebra}, 66:448--483, 1980.

\bibitem{Ribes2010}
L.~Ribes and P.~Zalesskii.
\newblock {\em {Profinite Groups}}.
\newblock Springer Berlin Heidelberg, 2010.

\bibitem{s16}
J.~Salam\'anca.
\newblock {Unveiling Eilenberg-type Correspondences: Birkhoff's Theorem for
  (finite) Algebras + Duality}.
\newblock \url{https://arxiv.org/abs/1702.02822}, February 2017.

\bibitem{sal04}
S.~Salehi and M.~Steinby.
\newblock Varieties of many-sorted recognizable sets.
\newblock TUCS Technical Report 626, Turku Center for Computer Science, 2004.

\bibitem{salehi07}
S.~Salehi and M.~Steinby.
\newblock Tree algebras and varieties of tree languages.
\newblock {\em Theor. Comput. Sci.}, 377(1-3):1--24, 2007.

\bibitem{sch65}
M.~P. Sch\"utzenberger.
\newblock On finite monoids having only trivial subgroups.
\newblock {\em Inform. and Control}, 8:190--194, 1965.

\bibitem{schuetzenberger}
Marcel~Paul Sch\"utzenberger.
\newblock On the definition of a family of automata.
\newblock {\em Inform.~and Control}, 4(2--3):275--270, 1961.

\bibitem{straubing02}
H.~Straubing.
\newblock On logical descriptions of regular languages.
\newblock In S.~Rajsbaum, editor, {\em LATIN 2002 Theor. Informatics}, volume
  2286 of {\em LNCS}, pages 528--538. Springer, 2002.

\bibitem{therien80}
D.~Th\'erien.
\newblock {\em Classification of Regular Languages by Congruences}.
\newblock PhD thesis, University of Waterloo, 1980.

\bibitem{u17}
H.~Urbat.
\newblock {A note on HSP theorems}.
\newblock \url{https://www.tu-braunschweig.de/Medien-DB/iti/hspnote.pdf}.

\bibitem{uacm16old}
H.~Urbat, J.~Ad\'amek, L.-T. Chen, and S.~Milius.
\newblock {One Eilenberg Theorem to Rule Them All}.
\newblock \url{https://arxiv.org/abs/1602.05831v1}. February 2016.

\bibitem{wilke91}
T.~Wilke.
\newblock An {E}ilenberg theorem for $\infty$-languages.
\newblock In {\em Proc.~ICALP'91}, volume 510 of {\em LNCS}, pages 588--599.
  Springer, 1991.

\bibitem{wilke96}
T.~Wilke.
\newblock An algebraic characterization of frontier testable tree languages.
\newblock {\em Theor. Comput. Sci.}, 154(1):85--106, 1996.

\end{thebibliography}

\clearpage
\appendix
\counterwithin{theorem}{section}

\section{Categorical toolkit}\label{sec:categories}
We review some standard concepts from category theory that we will use throughout
this paper. For details we refer to textbooks such as
\cite{maclane,ahs09}, and also to \cite{ar} for an introduction to locally presentable categories.

\newcommand{\mysubsec}[1]{%
  \refstepcounter{subsection}
  \par\medskip
  \noindent\textbf{\thesubsection\hspace{1.5ex}#1.\quad}%
}

\mysubsec{Monads}\label{app:monads} A \emph{monad} on a category $\ACat$ 
is 
a triple $\MT = (T,\eta,\mu)$ consisting of an endofunctor $T\colon \ACat\to\ACat$
and two natural transformations $\eta\colon \Id\to T$ and $\mu\colon TT\to T$ (called the
\emph{unit} and \emph{multiplication} of $\MT$) such that the following 
diagrams commute:
\[
\xymatrix{
T \ar[r]^{\eta T} \ar[dr]_{\id} & TT \ar[d]^\mu \\
& T
}
\qquad
\xymatrix{
T \ar[r]^{T\eta} \ar[dr]_{\id} & TT \ar[d]^\mu \\
 & T
}
\qquad
\xymatrix{
TTT \ar[r]^{T\mu} \ar[d]_{\mu T} & TT \ar[d]^\mu \\
TT \ar[r]_\mu & T
}
\]

\mysubsec{Algebras for a monad}\label{app:algmonad} Let $\MT = 
(T,\eta,\mu)$ be a monad on a 
category $\ACat$. By a \emph{$\MT$-algebra} is meant a pair $(A,\alpha)$ of an 
object $A\in \ACat$ and a morphism $\alpha\colon TA\to A$ satisfying the 
\emph{unit} and \emph{associative laws}, i.e. making the following diagrams commute:
\[
\xymatrix{
A\ar[dr]_{\id} \ar[r]^{\eta_A} & TA  \ar[d]^\alpha \\
 & A
}
\qquad
\xymatrix{
TTA \ar[d]_{T\alpha} \ar[r]^{\mu_A} & TA \ar[d]^\alpha \\
TA \ar[r]_\alpha & A
}
\] 
Given two $\MT$-algebras $(A,\alpha)$ and $(B,\beta)$, a
\emph{$\MT$-homomorphism} $h\colon (A,\alpha)\to (B,\beta)$ is a morphism
$h\colon A\to B$ in $\ACat$ such that $h \o \alpha = \beta \o Th$. Denote by 
$\Alg{\MT}$ the category of $\MT$-algebras and 
$\MT$-homomorphisms.
There is a forgetful functor $U\colon\Alg{\MT} \to \ACat$ given by
$(A,\alpha)\mapsto A$ on objects and $h\mapsto h$ on morphisms. It has a left
adjoint assigning to each object $A$ of $\ACat$ the $\MT$-algebra $\MT A = 
(TA,\mu_A)$, called the \emph{free $\MT$-algebra on~$A$}, and to each morphism 
$h\colon A\to B$ the
$\MT$-homomorphism $Th\colon \MT A \to \MT B$. Note that for any $\MT$-algebra
$(A,\alpha)$ the associative law states precisely that $\alpha\colon \MT A \to
(A,\alpha)$ is a $\MT$-homomorphism. Moreover, the unit law implies that $\alpha$ is a (split) epimorphism in $\ACat$.

\mysubsec{Quotients of $\MT$-algebras}\label{app:quotients_of_talgs}
Let $\MT=(T,\eta,\mu)$ be a monad, and suppose that $T$ preserves epimorphisms. For any commutative square as shown below, if $(A,\alpha)$ is a $\MT$-algebra and $e$ is an epimorphism, then $(B,\beta)$ is a $\MT$-algebra.
\begin{equation}\label{eq:square}
\vcenter{
\xymatrix{
TA \ar[r]^\alpha \ar@{->>}[d]_{Te} & A \ar@{->>}[d]^e\\
TB \ar[r]_\beta & B
}}
\end{equation}
Indeed, for the associative law for $(B,\beta)$ consider the diagram below:
 \[
\xymatrix{
TTA \ar[rrr]^{\mu_{A}} \ar@{->>}[dr]^{TTe} \ar[ddd]_{T\alpha} &&& TA \ar[ddd]^{\alpha} \ar[dl]^{Te}\\
& TTB \ar[r]^{\mu_B} \ar[d]_{T\beta} & TB \ar[d]^\beta &\\
& TB \ar[r]_\beta & B & \\
TA \ar[ur]^{Te} \ar[rrr]_{\alpha} &&& A \ar[ul]_{e}
}
\]
The outside square commutes by the associative law for $(A,\alpha)$, the upper part because $\mu$ is natural and the left-hand, right-hand and lower parts by \eqref{eq:square}. Therefore the central square also commutes (since it does when precomposed with the epimorphism $TTe$), giving the associative law for $(B,\beta)$. Analogously, to establish the unit law $\beta\o \eta_B = \id$ one proves that this holds when precomposed with the epimorphism $e$:
\[
\beta \o \eta_B \o e = \beta \o Te \o \eta_A = e \o \alpha \o \eta_A = e.
\]

\mysubsec{Limits of $\MT$-algebras}\label{app:limits_of_talgs} For any monad $\MT$ on a category $\ACat$, the
forgetful functor $U\colon \Alg{\MT} \to \ACat$ preserves limits, being a 
right adjoint (see \ref{app:algmonad}). More importantly, it also 
\emph{creates 
limits}. That is, given a
diagram $D\colon \S\to \Alg{\MT}$ and a limit cone $(p_s\colon A\to UD_s)_{s\in
  \S}$ over $UD$ in $\ACat$, there exists a unique $\MT$-algebra structure
$(A,\alpha)$ on $A$ such that all $p_s$ are $\MT$-homomorphisms, and moreover
$(p_s\colon (A,\alpha)\to D_s)_{s\in \S}$ forms a limit cone over $D$ in $\Alg{\MT}$.
In the case where the category $\ACat$ is complete, it follows that also
$\Alg{\MT}$ is complete and that $U$ \emph{reflects limits}. That is, a cone 
$(p_s\colon (A,\alpha)\to D_s)_{s\in S}$ over $D$ is a limit cone whenever $(p_s\colon A\to 
UD_s)_{s\in S}$ is a limit cone over $UD$. 

\mysubsec{Comma categories}\label{app:commacat}
 Let $F\colon \B\to \ACat$ be a functor and $A$ an object in $\ACat$. The
 \emph{comma category} $(A\downarrow F)$ has as objects all morphisms $(A\xra{f}
 FB)$ in $\ACat$ with $B\in \B$, and its morphisms from $(A\xra{f_1} FB_1)$ to
 $(A\xra{f_2} FB_2)$ are morphisms $h\colon B_1\to B_2$ in $\B$ with $f_2 = Fh\o
 f_1$. If $F\colon \B \hookrightarrow \ACat$ is the inclusion  of a 
 subcategory $\B$,
 we write $(A\downarrow \B)$ for $(A\downarrow F)$.

\mysubsec{Kan extensions}\label{app:kanext} The \emph{right Kan 
extension} of a functor 
$F\colon \ACat \to \C$ 
along~$K\colon \ACat \to
\B$ is a functor~$R\colon \B \to \C$ together with a universal natural
transformation $\epsilon\colon RK \to F$, i.e.\ for every functor
$G\colon \B \to \C$ and every natural transformation~$\gamma\colon GK \to F$ 
there
exists a unique natural transformation $\gamma^\dagger\colon G \to R$ with
$\gamma = \epsilon \o \gamma^\dagger K$. The universal property determines $R$ and $\epsilon$ uniquely up to isomorphism. If $\ACat$ is small and $\C$ is 
complete, the right Kan extension $R$ exists, and the object $RB$ for $B\in \B$ can be computed as the limit of the diagram
\[
  (B\downarrow K) \xra{\;Q^B\;} \ACat \xra{\;F\;} \C,
\]
where $Q^B$ is the projection functor that maps $(B\xra{f} KA)$ to $A$ and~$h\colon (B\xra{f_1} KA_1) \to
(B\xra{f_2} KA_2)$ to~$h\colon A_1\to A_2$.

\mysubsec{Codensity monads}\label{app:codensity_monad}
Let $\epsilon\colon RK\to K$ be the right Kan extension of a functor $K\colon
\ACat \to\B$ along itself.  Then $R$ can be equipped with a monad structure $\MR
= (R,\eta^\MR, \mu^\MR)$ where the unit $\eta^\MR$ is given by $(\id_K)^\dagger\colon
\Id\to R$ and the multiplication $\mu^\MR$ by $(\epsilon\o
R\epsilon)^\dagger\colon RR\to R$.  The monad $\MR$ is called the
\emph{codensity monad} of $K$, see e.g.~\cite{Linton1969}.

\mysubsec{Final functors}\label{app:final_func} Let $\K$ be a cofiltered category (see \ref{app:cofiltered_limits} below). A functor $F\colon \K\to
\B$ is called \emph{final} if 
\begin{enumerate}[(i)]
  \item for any object $B$ of $\B$, there exists a 
  morphism $f\colon FK\to B$ for some $K\in\K$, and

  \item given two parallel morphisms $f,g\colon FK\to B$ with $K\in \K$, there 
  exists a morphism $k\colon K'\to K$ in $\K$ with $f\o Fk = g\o Fk$.
\end{enumerate}
The importance of final functors is that they facilitate the 
construction of limits. If $F\colon \K\to \B$ is final, then a 
diagram 
$D\colon \B\to \ACat$ has a limit iff the diagram $DF\colon
\K\to \ACat$ has a limit, and in this case the two limit objects agree. Specifically,
any limit 
cone $(p_B\colon A\to D_B)_{B\in \B}$ over $D$ restricts to a limit cone 
$(p_{FK}\colon A\to D_{FK})_{K\in \K}$ over $DF$. 

\mysubsec{Cofiltered limits and inverse limits}\label{app:cofiltered_limits} A category 
$\K$ is \emph{cofiltered} if for every finite subcategory $D\colon 
\K'\monoto\K$ 
there exists a cone over $D$. This is equivalent to the following three 
conditions:
\begin{enumerate}[(i)]
\item $\K$ is nonempty.
\item For any two objects $Y$ and $Z$ of $\K$, there exist two
morphisms $f\colon X\to Y$ and  $g\colon X\to Z$ with a common domain $X$.
\item For any two parallel morphisms $f,g\colon Y\to Z$ in $\K$, there exists a 
morphism $e\colon X\to Y$ with $f\o e = g\o e$.
\end{enumerate}
A \emph{cofiltered limit} in a category $\ACat$ is a limit of a diagram $\K \to
\ACat$ with cofiltered scheme $\K$. It is also called an \emph{inverse limit} 
if $\K$ is a (co-directed) poset. For any small cofiltered category $\K$, there exists a final functor $F\colon \K_0\to \K$ where $\K_0$ is a small co-directed poset. Consequently, a category has cofiltered limits iff it has inverse limits, and a functor preserves cofiltered limits iff it preserves inverse limits.

The dual concept of a cofiltered limit is a \emph{filtered colimit}, see~\cite{ar}.

\mysubsec{Finitely copresentable 
objects}\label{app:fin_copres} An 
object $A$ of a category 
$\ACat$ is called \emph{finitely copresentable} if the hom-functor 
$\ACat(\dash,A)\colon \ACat\to \Set^{op}$ preserves cofiltered limits. 
Equivalently, for any cofiltered limit cone $(p_i\colon B\to B_i)_{i\in I}$ in
$\ACat$  the following two statements hold:
 \begin{enumerate}[(i)]
  \item Every morphism $f\colon B\to A$ factors through some $p_i$.

  \item For any $i\in I$ and any two morphisms 
  $s,s'\colon B_i\to A$ with 
  $s\o p_i = s'\o p_i$, there exists a connecting morphism $b_{ji}\colon B_j\to B_i$ in the given
  diagram with $s\o b_{ji}  = s' \o b_{ji}$.
\end{enumerate}

\mysubsec{Locally finitely copresentable categories}\label{app:lfcp_cat} 
A category $\ACat$ is called
\emph{locally finitely copresentable} if it satisfies the following three 
properties:
\begin{enumerate}[(i)]
\item $\ACat$ is complete;
\item the full subcategory $\ACat_f$ of finitely copresentable objects is 
essentially small, i.e.\ the objects of $\ACat_f$ (taken up to isomorphism) 
form a set;
\item any object $A$ of $\ACat$ is a cofiltered limit of finitely 
copresentable objects; that is, there exists a cofiltered limit cone $(A\to
A_i)_{i\in I}$ in $\ACat$ with $A_i\in 
\ACat_f$ for all $i\in I$.
\end{enumerate}
This implies that $\ACat$ is also cocomplete.
If $\ACat$ is locally finitely copresentable, so is any functor category $\ACat^\S$, where $\S$ is an arbitrary small category. In particular, this holds for any product category $\ACat^S$ (where $S$ is set) and for the arrow category 
$\ACat^\to$. The latter has as objects all morphisms of $\ACat$, and as morphisms from $(A\xra{f}B)$ to $(C\xra{g}D)$ all pairs of morphisms $(a: A\to C, b: B\to D)$ in $\ACat$ with $b\o f = g\o a$. The finitely copresentable objects of~$\ACat^\to$ are precisely the
arrows with finitely copresentable domain and codomain.

\mysubsec{Cofiltered limits in locally finitely copresentable 
categories}\label{app:cof_lim_in_lfcp_cat} Let $\ACat$ be a locally finitely copresentable category.  A  cone $(p_i\colon B\to
B_i)_{i\in I}$ over a cofiltered diagram in $\ACat$ is a limit
cone iff
\begin{enumerate}[(i)]
  \item every morphism $f\colon B\to A$ with $A\in \ACat_f$ factors through some
    $p_i$, and
  \item this factorization is essentially unique: given $i\in I$ and $s,s'\colon
    B_i\to A$ with $s\o p_i = s'\o p_i$, there exists a morphism $b_{ji}\colon
    B_j\to B_i$ in the diagram with $s\o b_{ji}  = s' \o b_{ji}$.
\end{enumerate}
Note that if all $p_i$'s are epimorphisms, condition (ii) is trivial.

\mysubsec{Canonical diagrams}\label{app:can_diagram} Let $\ACat$ be a 
locally finitely 
copresentable 
category.  Then for each object $A\in\ACat$ the comma category $(A\downarrow
\ACat_f)$ is essentially small and cofiltered.  The \emph{canonical diagram of
  $A$} is the cofiltered diagram $K_A\colon (A\downarrow \ACat_f)\to \ACat$
that maps an object $(A\xra{f} A_1)$ to $A_1$ and a morphism $h\colon (A\xra{f_1}
A_1) \to (A\xra{f_2} A_2)$ to $h\colon A_1\to A_2$. Every object $A$ of $\ACat$ is
the cofiltered limit of its canonical diagram, that is, $K_A$ has the limit cone
\[
  (f\colon A\to K_A f)_{f\in (A\downarrow \ACat_f)}.
\]

\mysubsec{Pro-completions}\label{app:procomp} Let $\B$ be a small 
category. By 
a \emph{pro-completion} (or a \emph{free completion under cofiltered limits}) 
of $\B$ is meant a category $\Pro\B$ together with a 
full embedding $I\colon \B\monoto \Pro\B$ such that 
\begin{enumerate}[(i)]
\item $\Pro\B$ has cofiltered limits.
\item For any functor $F\colon \B\to \C$ into a category $\C$ with cofiltered 
limits, there exists a functor $\overline{F}\colon \Pro\B\to \C$, unique up to natural 
isomorphism, such that $\ol F$ preserves cofiltered limits and $\overline{F}\o I$ is
naturally isomorphic to~$F$.
\end{enumerate}
The universal property determines $\Pro\B$ uniquely up to equivalence of 
categories. If the category $\B$ has finite limits, then $\Pro\B$ is locally 
finitely copresentable, and its finitely copresentable objects are up to 
isomorphism the objects $IB$ ($B\in\B$). Conversely, every locally finitely 
copresentable 
category $\ACat$ arises in this way: we have $\ACat = \Pro\ACat_f$.

The dual concept of a pro-completion is an \emph{ind-completion}, i.e. the 
free 
completion under filtered colimits.  For more on pro-completions see~\cite{Johnstone1982}, and for ind-completions see~\cite{ar}.

\mysubsec{Factorization systems}\label{app:factsystem}  A 
\emph{factorization system} in a 
category 
$\ACat$ is 
a pair $(\E,\M)$ where $\E$ and $\M$ are classes of morphisms of $\ACat$ with 
the following properties:
\begin{enumerate}[(i)]
\item Both $\E$ and $\M$ are closed under composition and contain all 
isomorphisms.
\item Every morphism $f$ of $\ACat$ has a factorization $f=m\o e$ with $e\in 
\E$ 
and $m\in \M$.
\item The \emph{diagonal fill-in} property holds: given a commutative square 
as shown below with $e\in\E$ and $m\in\M$, there exists a unique morphism 
$d$ making both triangles commute.
\[
\xymatrix{
A \ar@{->>}[r]^e \ar[d]_f & B \ar[d]^g \ar@{-->}[dl]_d \\
C \ar@{ >->}[r]_m & D
}
\]
\end{enumerate}
Factorization systems are discussed at length in~\cite{ahs09}. We will use three standard facts about factorization systems:
\begin{enumerate}[(a)]
\item Suppose that $\M$ is a class of monomorphisms. If $(p_i\colon A\to A_i)_{i\in I}$ is a limit cone in $\ACat$, then the 
factorization 
$\xymatrix{p_i = (A\ar@{->>}[r]^<<<<<{e_i} & A_i' \ar@{ >->}[r]^{m_i} & A_i})$ with 
$e_i\in 
\E$ 
and 
$m_i\in\M$ yields another limit cone $(e_i\colon A\epito A_i')_{i\in I}$ over the 
same scheme. 
\item Suppose that $\E$ is a class of epimorphisms. If $\MT$ is a monad on 
$\ACat$ that preserves $\E$, i.e.\ $e\in \E$ 
implies
$Te\in \E$, then $\Alg{\MT}$ has the factorization system of $\E$-carried and
$\M$-carried $\MT$-homomorphisms.
\item Every locally finitely copresentable category $\ACat$ has the (epi, 
strong mono)
factorization system. Its arrow category $\ACat^\to$, see \ref{app:lfcp_cat}, has the 
factorization system of componentwise epimorphic and strongly 
monomorphic morphisms.
\end{enumerate}

\mysubsec{Quotients and subdirect products}\label{app:subdirectproducts}
Suppose that $\ACat$ is a category with a factorization system $(\E,\M)$, where $\E$ is a class of epimorphisms. For any two quotients (= $\E$-morphisms) $e_i\colon A\epito A_i$ ($i=0,1$) put $e_0\leq e_1$ if $e_0$ factorizes through $e_1$, i.e. $e_0 = h\o e_1$ for some morphisms $h$. This defines a preorder on the class of quotients of $A$, which can be turned into a partial order by identifying any two quotients $e_0$ and $e_1$ with $e_0\leq e_1\leq e_0$.

If $\ACat$ has binary products, then the \emph{subdirect product} of $e_0$ and $e_1$ is
 the image $e\colon \MT\FSigma\epito A$ of the morphism $\langle e_0, e_1\rangle\colon \MT\FSigma\to A_0\times A_1$.
\[
\xymatrix{
\MT\FSigma \ar@{->>}[r]^e \ar@/^5ex/[rr]^{\langle e_0,e_1\rangle} \ar@/^10ex/[rrr]^{e_i} & A \ar@{>->}[r]^m & A_0\times A_1 \ar[r]^{\pi_i} & A_i
}
\]
It can be characterized as the smallest quotient of $A$ with $e_0,e_1\leq e$.

\section{Topological toolkit}
In this section, we recall some important properties of cofiltered limits in the 
category $\mathbf{CHaus}$ of
compact 
Hausdorff spaces and continuous maps. The proofs of the first three lemmas below can be found in Chapter~1 
of~\cite{Ribes2010}.
\begin{lemma}
  \label{lem:surjections-between-cofiltered-diagrams}
  Let $\tau\colon D_1 \to D_2$ be a natural transformation between cofiltered
  diagrams (over the same scheme) in $\mathbf{CHaus}$. If each $\tau_i\colon D_1i \twoheadrightarrow D_2i$ is surjective, 
  so is the mediating map $\Lim \tau \colon \Lim D_1 \to \Lim D_2$.
  In particular, if $(\tau_i\colon X \twoheadrightarrow Di)$ is a cone of
  surjections over a cofiltered diagram $D$, then the mediating map $X\to \Lim D$ is surjective.
\end{lemma}

\begin{lemma}
  \label{lem:projection-is-surjective}
  Let $D$ be a cofiltered diagram in $\mathbf{CHaus}$. 
  If all connecting maps $D(i \xrightarrow{f} j)$
  are surjective, so is each limit projection $\rho_i\colon \Lim D \to Di$.
\end{lemma}

\begin{lemma}\label{lem:limit-nonempty}
  Let $D$ be a cofiltered diagram of non-empty spaces in $\mathbf{CHaus}$. 
Then $\Lim D$ is non-empty.
\end{lemma}

\begin{lemma}\label{lem:connecting-map-image}
Let $(p_i: X\to X_i)_{i\in I}$ be a cofiltered limit in $\mathbf{CHaus}$, where all $X_i$'s are finite (and thus discrete). 
For 
each $i\in I$, there is some $j\in I$ and a connecting map $g_{ji}: X_j\to 
X_i$ with $p_i[X] = g_{ji}[X_j]$.
\end{lemma}

\begin{proof}
By \ref{app:cofiltered_limits} we may assume that $I$ is a codirected poset.

\begin{enumerate}[(i)]
\item For each $x\in X_i\setminus p_i[X]$,  there exists some $j\leq i$ 
such that $x\not\in g_{ji}[X_j]$. To see this, suppose the contrary. 
Then, for each $j\leq i$, the set $X_j' := \{\,y\in X_j : g_{ji}(y)=x\,\}$ is 
non-empty. Moreover, for $k\leq j\leq i$, the connecting map $g_{kj}: X_k\to X_j$ 
restricts to $X_k'$ and $X_j'$. Thus $(X_j')_{j\leq i}$ forms a subdiagram of 
$(X_j)_{j\leq i}$, and by Lemma \ref{lem:limit-nonempty} its limit $(p_j':  
X'\to X_j')_{j\leq i}$ is non-empty. Consider the unique map $g: X'\to X$ with 
$p_j\o g = s_j\o p_j'$ for all $j\leq i$, where $s_j: X_j'\monoto X_j$ is the 
inclusion. Then, 
choosing any $x\in X'$, we have $x=p_i(g(x'))$, contradicting the assumption 
that
$x\not\in p_i[X]$.
\item Since $X_i\setminus p_i[X]$ is a finite, by (i) and codirectedness of 
$I$ there is some $j\leq i$ such that $g_{ji}[X_j]\seq p_i[X]$. Moreover, we 
have
$g_{ji}[X_j]\supseteq p_i[X]$ because $(p_i)$ 
is a cone.  This proves the claim.
\end{enumerate}
\end{proof}
Two important classes of compact Hausdorff spaces are {Stone spaces} and {Priestley spaces}. A \emph{Stone space} is a compact topological space $X$ such that for any two elements $x\neq y$ of $X$ there is a clopen (= simultaneously open and closed) set containing $x$ but not $y$. A \emph{Priestley space} is an ordered compact space $X$ (i.e. a compact topological space endowed with a partial order $\leq$ on $X$) such that for any two elements $x\not\leq y$ of $X$ there is a clopen $\leq$-upper set containing $x$ but not $y$. We denote by $\Stone$ the category of Stone spaces and continuous maps, and by $\Priest$ the category of Priestley spaces and continuous monotone maps.

\section{Algebraic toolkit}

In this section, we review some concepts of many-sorted universal algebra.

\mysubsec{Signatures}
Given a set $S$ of sorts, by an \emph{$S$-sorted signature} is meant a set $\Gamma$ of \emph{operation symbols} with specified input and output sorts from $S$. Formally, to each $\gamma\in\Gamma$ one associates a triple $(\kappa,\mathbf{s}, s)$ where $\kappa$ is a cardinal number, $\mathbf{s}\in S^\kappa$ and $s\in S$. We call  $(\kappa,\mathbf{s}, s)$ the \emph{type} and $\kappa$ the \emph{arity} of $\gamma$. If $\kappa=n$ is finite and $\mathbf{s} = (s_0,\ldots, s_{n-1})$, we write $\gamma\colon s_0\times \ldots \times s_{n-1} \to s$ to specify the type of $\gamma$. If all operation symbols in $\Gamma$ are of finite arity, $\Gamma$ is called a \emph{finitary signature}.

\mysubsec{Algebras and ordered algebras}\label{app:algordalg}
Let $\Gamma$ be an $S$-sorted signature.
A \emph{$\Gamma$-algebra} is an $S$-sorted set $A\in \Set^S$ with associated $\Gamma$-operations. That is, to each operation symbol $\gamma\in \Gamma$ of type  $(\kappa,\mathbf{s}, s)$, where $\mathbf{s} = (s_i)_{i<\kappa}$, one associates an operation $\gamma^A\colon \prod_{i<\kappa}{A_{s_i}} \to A_s$. A \emph{morphism} (also called a \emph{$\Gamma$-homomorphism}) between two $\Gamma$-algebras $A$ and $B$ is an $S$-sorted function $f\colon A\to B$ such that the following diagram commutes for all $\gamma\in \Gamma$ of type  $(\kappa,\mathbf{s}, s)$:
\[
\xymatrix{
 \prod_{i<\kappa}{A_{s_i}} \ar[r]^<<<<<{\gamma^A} \ar[d]_{\prod_{i} h_{s_i}} &  A_s \ar[d]^{h_s} \\
 \prod_{i<\kappa}{B_{s_i}} \ar[r]_<<<<<{\gamma^B} &  B_s\\
}
\]
\emph{Quotients} and \emph{subalgebras} of $\Gamma$-algebras are represented by sortwise surjective and injective morphisms, respectively.

 An \emph{ordered $\Gamma$-algebra} is a $\Gamma$-algebra $A$ with an additional poset structure $\leq$ on each $A_s$ such that all $\Gamma$-operations $\gamma^A\colon \prod_{i<\kappa}{A_{s_i}} \to A_s$ are monotone maps. A \emph{morphism} between ordered $\Gamma$-algebras is a sortwise monotone $\Gamma$-homomorphism. \emph{Quotients} and \emph{subalgebras} of ordered $\Gamma$-algebras are represented by sortwise surjective and order-reflecting morphisms, respectively. Here a morphism $h$ is called \emph{order-reflecting} if $hx\leq hx'$ implies $x\leq x'$ (and thus $hx\leq hx'$ iff $x\leq x'$ by monotonicity). Note that every order-reflecting morphism is injective.

\mysubsec{Congruences}\label{app:congruences}
Let $\Gamma$ be an $S$-sorted signature and $A$ a $\Gamma$-algebra. An \emph{$S$-sorted equivalence relation on $A$} is family $\mathord{\equiv} = (\equiv_s)_{s\in S}$, where $\equiv_s$ is an equivalence relation on $A_s$. If the sort $s$ is clear from the context, we drop the subscript and write $\equiv$ for $\equiv_s$. We call an $S$-sorted equivalence relation $\equiv$ a \emph{congruence} on $A$ if it is stable under all $\Gamma$-operations: for any $\gamma\in \Gamma$ of type  $(\kappa,\mathbf{s}, s)$ and any two elements $a,b\in \prod_{i<\kappa} A_{s_i}$, 
\[ a_i \equiv_{s_i} b_i \text{ for all $i<\kappa$} \qquad\text{implies}\qquad \gamma^A(a) \equiv_s \gamma^A(b). \]
For any congruence $\equiv$, the $S$-sorted set $A/\mathord{\equiv} = (A_s/\mathord{\equiv_s})_{s\in S}$ of $\equiv$-congruence classes can be equipped with a $\Gamma$-algebra structure: for any $\gamma\in \Gamma$ of type $(\kappa,\mathbf{s}, s)$, put
\[\gamma^{A/\mathord{\equiv}}\colon \prod_{i<\kappa}{(A_{s_i}/\mathord{\equiv_{s_i}})} \to (A_s/\mathord{\equiv_s}),\qquad ([a_i]_{\equiv_{s_i}})_{i<\kappa} \mapsto [\,\gamma^A((a_i)_{i<\kappa})\,]_{\equiv_s}.\]
This is the unique $\Gamma$-algebra structure on $A/\mathord{\equiv}$ making the projection
 $e_\equiv\colon A\epito A/\mathord{\equiv}$ a $\Gamma$-homomorphism. 

Conversely, every quotient $e\colon A\epito B$ of $A$ yields a $\Gamma$-congruence on $A$, viz. the \emph{kernel congruence} $\equiv_e$ given on sort $s$ by
\[ a \equiv_{e,s} b \quad\text{iff}\quad e_s(a) = e_s(b).\]
The maps $\equiv\, \mapsto\, e_\equiv$ and $e\,\mapsto\, \equiv_e$ give a bijective correspondence between quotients of $A$ and congruences on $A$.

\mysubsec{Stable preorders}\label{app:stablepreorders}
Let $\Gamma$ be an $S$-sorted signature and $(A,\leq)$ an ordered $\Gamma$-algebra.  An \emph{$S$-sorted preorder on $A$} is a family $\mathord{\preceq} = (\preceq_s)_{s\in S}$ where $\preceq_s$ is a preorder on $A_s$. If the sort $s$ is clear from the context, we drop the subscript and write $\preceq$ for $\preceq_s$. We call an $S$-sorted preorder $\preceq$  a \emph{stable preorder} on $A$ if (i) $\preceq$ is coarser than $\leq$, i.e. $\mathord{\leq} \seq \mathord{\preceq}$, and (ii) $\preceq$ is stable under all $\Gamma$-operations: for any $\gamma\in \Gamma$ of type  $(\kappa,\mathbf{s}, s)$ and any two elements $a,b\in \prod_{i<\kappa} A_{s_i}$,
\[ 
  a_i \preceq_{s_i} b_i \text{ for all $i<\kappa$} 
  \qquad\text{implies}\qquad 
  \gamma^A(a) \preceq_s \gamma^A(b). 
\]
If $\preceq$ is a stable preorder then the equivalence relation
$\equiv \,=\, \preceq \cap \succeq$ is a congruence, see
\ref{app:congruences}. The $\Gamma$-algebra $A/\mathord{\equiv}$ can
be turned into an ordered algebra by putting
$[a]_{\equiv_s} \leq_s [b]_{\equiv_s}$ iff $a\preceq_s b$. We denote
this ordered algebra by $A/\mathord{\preceq}$. The projection
$e_\preceq\colon A\epito A/\mathord{\preceq}$ is a monotone surjective
$\Gamma$-homomorphism.

Conversely, every quotient $e\colon A\epito B$ of $A$ yields a stable preorder on $A$, viz. the \emph{ordered kernel} $\preceq_e$ given on sort $s$ by
\[ a \preceq_{e,s} b \quad\text{iff}\quad e_s(a) \leq e_s(b).\]
The maps $\mathord{\preceq} \mapsto e_\preceq$ and $e \mapsto \mathord{\preceq}_e$ give a bijective correspondence between quotients of $(A,\leq)$ and stable preorders on $(A,\leq)$.

\mysubsec{Elementary translations}\label{app:elementarytranslations}
Let $\Gamma$ be a finitary signature and $A$ an $\Gamma$-algebra. By an \emph{elementary translation} is meant a map $u\colon A_s\to A_t$ ($s,t\in S$) of the form
\begin{equation}\label{eq:translation}u(a) = \gamma^A(a_1,\ldots, a_{i}$, $a,a_{i+1},\ldots, a_n) \quad\text{for all $a\in A_s$},
\end{equation}
where $n\geq 0$, $i\in\{0,\ldots,n\}$, $\gamma\colon s_1\times \ldots \times s_{i}\times s\times s_{i+1}\times\ldots \times s_n\to s$ is an $(n+1)$-ary operation symbol in $\Gamma$, and $a_i\in A_{s_i}$ for $i=1,\ldots, n$. A $S$-sorted equivalence relation $\equiv$ on $A$ forms a congruence iff it is stable under all elementary translations, that is, for every elementary translation \eqref{eq:translation} and any two elements $a,a\in A_s$ with $a\equiv_s a'$, one has $u(a)\equiv_s u(a')$; see e.g. \cite{mat76,mat79}. Analogously, an $S$-sorted preorder on an ordered algebra $(A,\leq)$ is a stable preorder iff it is stable under all elementary translations.

\mysubsec{Varieties of algebras and ordered algebras}\label{app:varieties}
Let $\Gamma$ be a finitary $S$-sorted signature. A \emph{variety of $\Gamma$-algebras} is a class of $\Gamma$-algebras specified by equations $s=t$ between $\Gamma$-terms. Similarly, a \emph{variety of ordered $\Gamma$-algebras} is a class of ordered $\Gamma$-algebras specified by inequations $s\leq t$ between $\Gamma$-terms. By an \emph{$S$-sorted variety of (ordered) algebras} is meant a variety $\ACat$ of (ordered) $\Gamma$-algebras for some finitary $S$-sorted signature $\Gamma$. For every single-sorted variety $\ACat$ and set of sorts $S$, one can view the product category $\ACat^S$ as an $S$-sorted variety. 

Categorically, varieties can be modeled by monads: for every $S$-sorted variety $\ACat$ of (ordered) algebras, there exists a monad $\MT$ on $\Set^S$ (resp. $\Pos^S$) with $\ACat\simeq \Alg{\MT}$. 

\mysubsec{Topological algebras}\label{app:topalg}
Let $\Gamma$ be an $S$-sorted signature. A \emph{topological $\Gamma$-algebra} is a $\Gamma$-algebra $A$ endowed with a topology on each component $A_s$ such that for any  $\gamma\in \Gamma$ of type  $(\kappa,\mathbf{s}, s)$ the operation $\gamma^A\colon \prod_{i<\kappa}{A_{s_i}} \to A_s$ is continuous. An \emph{ordered topological $\Gamma$-algebra} $A$ is additionally endowed with a partial order on each sort such that every $\gamma^A$ is continuous and monotone. Morphisms of (ordered) topological $\Gamma$-algebras are (monotone) continuous $\Gamma$-homomorphisms. Given a variety $\ACat$ of $\Gamma$-algebras, we denote by $\Stone(\ACat)$ the category of topological $\Gamma$-algebras carrying a Stone topology such that the underlying $\Gamma$-algebra (forgetting the topology) lies in $\ACat$. Similarly, if $\ACat$ is a variety of ordered $\Gamma$-algebras, we denote by $\Priest(\ACat)$ the category of ordered topological $\Gamma$-algebras carrying a Priestley topology such that the underlying ordered $\Gamma$-algebra lies in $\ACat$.

\mysubsec{Homomorphism theorem}\label{app:homtheorem}
Let $\Gamma$ be an $S$-sorted signature, and let $e\colon A \epito B$ and 
$f\colon A \to 
        C$ be two morphisms of $\Gamma$-algebras with $e$ sortwise surjective. Then the following are equivalent:
\begin{enumerate}[(1)]
\item There exists a morphism $g\colon B\to C$ 
        with 
        $g\o e = f$.
\item For all sorts $s$ and
            $a,a'\in \under{A}_s$, if
    $e(a)= e(a')$ 
    then $f(a)=f(a')$.
 \end{enumerate}
 Similarly, if $e$ and $f$ are morphisms between ordered $\Gamma$-algebras, then (i) is equivalent to
 \begin{enumerate}[(1)]
 \item[(2')] For all sorts $s$ and $a,a'\in \under{A}_s$, if $e(a)\leq e(a')$ then $f(a)\leq f(a')$.
 \end{enumerate}
If moreover $e$ and $f$ are morphisms between (ordered) topological $\Gamma$-algebras whose topologies are compact Hausdorff, then the map $g$ in (i) is continuous. This follows from the fact that for every continuous surjection $e\colon A\epito B$ between compact Hausdorff spaces, the space $B$ carries the final topology w.r.t. $e$.

\end{document}